\documentclass[
a4paper,
11pt, 
twoside, 
openright 
]{memoir}
\chapterstyle{veelo}

\usepackage{microtype}
\usepackage{amssymb}
\usepackage{amsmath}
\usepackage{amsthm}
\usepackage{thmtools}
\usepackage{booktabs}
\usepackage[T1]{fontenc}

\usepackage{csquotes}

\usepackage[backend=biber,style=numeric-comp,isbn=false,doi=false,url=false,maxnames=10, giveninits=true]{biblatex}

\usepackage{tikz}
\usetikzlibrary{positioning}
\usetikzlibrary{decorations.pathmorphing}
\usetikzlibrary{shapes.geometric}
\usetikzlibrary{shapes.misc}
\usetikzlibrary{patterns}

\pgfdeclarepatternformonly{del lines}{%
    \pgfqpoint{-1pt}{-1pt}}{\pgfqpoint{21pt}{21pt}}{\pgfqpoint{16pt}{16pt}}%
    {
        \pgfsetlinewidth{0.4pt}
        \pgfpathmoveto{\pgfqpoint{0pt}{0pt}}
        \pgfpathlineto{\pgfqpoint{16pt}{16pt}}
        \pgfusepath{stroke}
    }
\pgfdeclarepatternformonly{added lines}{\pgfpointorigin}{\pgfqpoint{50pt}{10pt}}{\pgfqpoint{50pt}{10pt}}%
    {
        \pgfsetlinewidth{0.2pt}
        \pgfpathmoveto{\pgfqpoint{0pt}{0.5pt}}
        \pgfpathlineto{\pgfqpoint{100pt}{0.5pt}}
        \pgfusepath{stroke}
    }

\usepackage{gnuplot-lua-tikz}

\usepackage[procnumbered,ruled,linesnumbered,vlined,algochapter]{algorithm2e}

\usepackage{paralist}
\usepackage{adjustbox}
\usepackage{xspace}
\usepackage{hyphenat} 

\usepackage[NoDate]{currvita}

\usepackage[hidelinks,unicode,psdextra]{hyperref}
\usepackage{cleveref} 

\hypersetup{
    pdftitle = {Leveraging the Power of Graph Algorithms: Efficient Algorithms for Computer-Aided Verification},
    pdfauthor = {Alexander Svozil},
    pdfsubject = {Dissertation},
    pdfkeywords = {theoretical computer science, efficient algorithms}
}

\title{Leveraging the Power of Graph Algorithms: Efficient Algorithms for Computer-Aided Verification }
\author{Dipl.\,Ing.~Alexander Svozil, BSc}

\setsecnumdepth{subsection}
\settocdepth{subsection}

\usetikzlibrary{decorations.pathreplacing,positioning,shapes.misc,calc}

\DontPrintSemicolon
\SetProcNameSty{textsc}
\SetFuncSty{textsc}
\SetAlFnt{\small}

\allowdisplaybreaks[1]

\declaretheorem[numberwithin=section]{theorem}
\declaretheorem[numberlike=theorem]{lemma}
\declaretheorem[numberlike=theorem]{corollary}
\declaretheorem[numberlike=theorem]{proposition}
\declaretheorem[numberlike=theorem]{claim}
\declaretheorem[numberlike=theorem]{definition}
\declaretheorem[style=remark,numberlike=theorem]{remark}

\declaretheorem[numberlike=theorem]{observation}

\declaretheorem[numberlike=theorem]{datastructure}
\declaretheorem[numberlike=theorem]{invariant}
\declaretheorem[numberlike=theorem]{example}
\declaretheorem[numberlike=theorem]{conjecture}
\declaretheorem[numberlike=theorem]{reduction}

\SetKwProg{Procedure}{Procedure}{}{}

\newcommand{\GG}{\ensuremath{\Gamma}}
\newcommand{\Out}[1]{\ensuremath{\mathit{Out}(#1)}}
\newcommand{\In}[1]{\ensuremath{\mathit{In}(#1)}}
\newcommand{\SCCs}[1]{\ensuremath{\mathtt{SCCs}(#1)}}

\newcommand{\edgeset}[1]{E(#1)}
\newcommand{\bits}{\mathit{bits}}

\newcommand{\MEC}{\mathsf{MEC}}

\newcommand{\Inf}[1]{\ensuremath{\mathit{Inf}(#1)}}
\newcommand{\reach}[1]{\ensuremath{\text{Reach}(#1)}}
\newcommand{\safety}[1]{\ensuremath{\text{Safety}(#1)}}
\newcommand{\buchi}[1]{\ensuremath{\text{B\"{u}chi}(#1)}}
\newcommand{\cobuchi}[1]{\ensuremath{\text{coB\"{u}chi}(#1)}}
\newcommand{\pbuchi}[1]{\ensuremath{\text{boundedB\"{u}chi}(#1)}}
\newcommand{\pcobuchi}[1]{\ensuremath{\text{boundedcoB\"{u}chi}(#1)}}
\newcommand{\parity}[1]{\ensuremath{\text{Parity}(#1)}}
\newcommand{\streett}[1]{\ensuremath{\text{Streett}(#1)}}
\newcommand{\mpayoff}[1]{\ensuremath{\text{MeanPayoff}(#1)}}
\newcommand{\mpayoffp}[1]{\ensuremath{\text{MeanPayoffParity}(#1)}}
\newcommand{\MP}[1]{\ensuremath{\mathit{MP}(#1)}}
\newcommand{\MPP}[1]{\ensuremath{\mathit{MPP}(#1)}}

\newcommand{\W}[2]{\ensuremath{W_{#1}(#2)}}
\newcommand{\ASW}[1]{\ensuremath{\langle\!\langle\text{1}\rangle\!\rangle_{a.s.}(#1)}}

\newcommand{\attr}[3]{\ensuremath{\mathit{attr}_{#1}(#2,#3)}}
\newcommand{\GraphReach}[1]{\ensuremath{\mathit{GraphReach}(#1)}}

\newcommand{\decrSCC}{\mathit{T_d}}
\newcommand{\condense}[1]{\ensuremath{\mathsf{CONDENSE}(#1)}}

\newcommand{\True}{\mathbf{true}}
\newcommand{\False}{\mathbf{false}}

\renewcommand{\O}{\ensuremath{\widetilde{O}}}
\newcommand{\para}[1]{\smallskip\noindent\emph{#1}}
\newcommand{\ls}{\langle}
\newcommand{\rs}{\rangle}
\newcommand{\D}{\mathcal{D}}
\newcommand{\A}{\mathcal{A}}
\renewcommand{\L}{\mathcal{L}}
\newcommand{\set}[1]{\{#1\}}
\newcommand{\V}{\mathcal{\nu}}
\newcommand{\N}{\mathbb{N}}
\newcommand{\Z}{\mathbb{Z}}
\newcommand{\Q}{\mathbb{Q}}

\newcommand{\val}{\mathit{val}}

\newcommand{\restr}{\upharpoonright}

\newcommand{\numprio}{d\xspace}
\newcommand{\Lift}{\mathit{Lift}}
\DeclareMathOperator{\slift}{lift}

\newcommand{\CP}{\mathit{CPre}}
\newcommand{\mina}{\mathit{min}}

\newcommand{\best}{\mathit{best}}

\newcommand{\+}{\mathit{+}}
\newcommand{\z}{{\bar{z}}}
\renewcommand{\S}{\mathbf{S}}

\newcommand{\E}{\ensuremath{\mathcal{E}}\xspace}

\newcommand{\os}{symbolic one-step }
\newcommand{\trho}{\widetilde{\rho}}
\newcommand{\PG}{\mathcal{P}}

\newcommand{\WOrder}{\mathcal{W}}

\newcommand{\Continue}{\textbf{continue}}

\newcommand{\EC}[1]{\ensuremath{\mathsf{EC}(#1)}}
\SetKwFunction{CollapseEC}{CollapseEC}

\newcommand{\Attr}[3]{\ensuremath{\mathsf{Attr}_{#1}^{#2}(#3)}}

\newcommand{\Pre}[2]{\ensuremath{\mathit{Pre}_{#1}(#2)}}
\newcommand{\Post}[2]{\ensuremath{\mathit{Post}_{#1}(#2)}}
\newcommand{\Set}[1]{\ensuremath{\{#1}\}}
\SetKwFunction{Separator}{Separator}
\newcommand{\Diam}[1]{\ensuremath{\mathsf{diam}(#1)}}
\newcommand{\Pick}[1]{\ensuremath{\mathit{Pick}(#1)}}
\newcommand{\ROut}[1]{\ensuremath{\mathit{ROut}(#1)}}
\newcommand{\WinPEC}[2]{\ensuremath{\mathsf{WinParityEC}{(#1,#2)}}}

\newcommand{\SymASReach}[2]{\ensuremath{\mathsf{SymASReach}_{#1}(#2)}}
\newcommand{\minPriority}[1]{\ensuremath{\mathsf{MinPriority}_(#1)}}
\newcommand{\WE}{\ensuremath{\mathsf{WE}}}
\renewcommand{\P}{\ensuremath{\mathcal{P}}}
\newcommand{\aspace}{\ensuremath{\mathit{space}}}

\newcommand{\mectime}{\textsc{mec}\xspace}
\DeclareMathOperator*{\argmax}{argmax}

\SetKwProg{Function}{function}{}{}
\newcommand{\mcount}{\textit{count}}
\newcommand{\nill}{\textit{null}}

\newcommand{\poly}{\mathit{poly}}
\newcommand{\Seq}[1]{\ensuremath{\mathit{Seq}(#1)}}
\newcommand{\Coverage}[1]{\ensuremath{\mathit{Coverage}(#1)}}

\newcommand{\ProcessVertex}[1]{\textsc{ProcessVertex}\ensuremath{(#1)}}
\newcommand{\Targets}{\mathcal{T}}

\newcommand{\copies}[1]{\ensuremath{\mathit{Copies}(#1)}}

\DeclareMathOperator{\dist}{dist}

\DeclareMathOperator{\proj}{Proj}
\DeclareMathOperator{\lift}{Lift}
\DeclareMathOperator{\move}{Shift}
\DeclareMathOperator{\nextl}{NxtLyr}
\newcommand{\event}{\mathcal{E}}

\begin{document}

	\frontmatter

	\makepagestyle{titlepage}
\makeoddhead{titlepage}{}{}{~~~~~~~~~~~~~~~~~~~~~~~~~~~~~~~~~~~~~~~~~~~~~~~~~~~~~~~~~~~~~~~~~~~~~~~~~~~~~~~~~\includegraphics{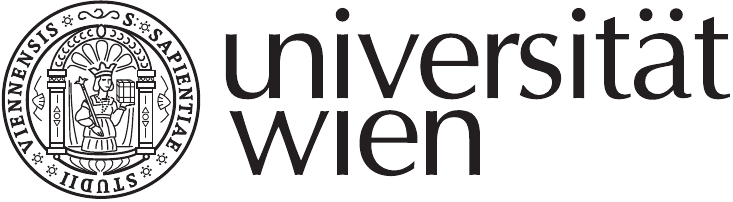}}

\calccentering{\unitlength}                         
\begin{adjustwidth*}{\unitlength}{-\unitlength}     
\begin{adjustwidth}{-1cm}{-1cm}

\thispagestyle{titlepage}
{\centering

\sffamily

~
\vfill

\vfill

\HUGE \textbf{\textsc{Dissertation / Doctoral Thesis}}\\

\vfill

\normalsize Titel der Dissertation / Title of the doctoral thesis \\

\huge \textbf{\thetitle}

\vfill

\normalsize verfasst von / submitted by\\
\Large \theauthor

\vfill

\normalsize angestrebter akademischer Grad / in partial fulfillment of the requirement for the degree of \\
\Large Doktor der Technischen Wissenschaften (Dr.\,techn.)

\vfill

\normalsize

\begin{tabbing}
Dissertationsgebiet laut Studienblatt:blablablabla \hspace{2.5em} \= Informatik \kill
\\
Wien, 2021 / Vienna, 2021\\ 
\\
Studienkennzahl lt. Studienblatt: /\\ degree programme code as it appears on the student\\ record sheet:\> A 786 880 \\
Dissertationsgebiet lt. Studienblatt: /\\ field of study as it appears on the student record sheet: \> Informatik \\
Betreuerin: / Supervisor: \> Univ.-Prof.\ Dr.\ Monika Henzinger \\
\end{tabbing}

\vspace{-7ex}

~
} 

\end{adjustwidth}
\end{adjustwidth*}

\normalfont


	\cleardoublepage
	\begin{abstract}
		\emph{Model checking} verifies whether a given \emph{model} of a system ensures a \emph{specification}.  
In \emph{reactive synthesis} the input is a specification and the goal is to construct a 
correct \emph{reactive system}.

We consider a variety of models: 
\emph{(state-transition) graphs} where vertices represent the
states of a system and edges are the transitions from one state to the next. 
\emph{Markov Decision Processes (MDPs)} extend graphs to model probabilistic systems, e.g.\, randomized
protocols. 
\emph{Game graphs} extend graphs to model interactions with the environment and are a central tool
for reactive synthesis.

\emph{Objectives} are sets of traces in the model and represent the specification.
The class of \emph{$\omega$-regular objectives} expresses all commonly used functional specifications for reactive systems in model checking and
synthesis. Examples for $\omega$-regular objectives we consider are \emph{reachability} objectives, 
\emph{B\"uchi} objectives, \emph{bounded B\"uchi} objectives, \emph{Streett} objectives and
\emph{parity} objectives. To model functional and efficient reactive systems we use the combination of 
\emph{quantitative (mean-payoff)} objectives and $\omega$-regular objectives.
For the verification of MDPs with $\omega$-regular objectives, we consider computing the
\emph{maximal end-component (MEC)} decomposition of MDPs. 
 

The goal of the thesis is to leverage fast graph algorithms
and modern algorithmic techniques for problems in model checking and synthesis on graphs, MDPs, and game graphs.
The results include \emph{symbolic algorithms}, a well-known class of algorithms in model checking
that trades limited access to the input model for an efficient representation.
In particular, we present the following results:
\begin{compactitem}
	\item Algorithms for game graphs with 
		mean-payoff B\"uchi objectives and mean-payoff coB\"uchi objectives 
		which match one of the best 
		running time bounds for mean-payoff objectives. 
	\item A near-linear time randomized algorithm for Streett objectives in graphs and MDPs.
	\item A sub-cubic time algorithm for bounded B\"uchi objectives in graphs and a cubic time algorithm for game graphs.
	\item Conditional lower bounds for queries of reachability objectives in game graphs and MDPs. Linear and near-linear time algorithms for sequential reachability objectives in graphs and MDPs respectively.
	\item The first quasi-polynomial time symbolic algorithm for parity objectives in game graphs.
	\item We break a long-standing running time bound for MEC decomposition from the '90s
		by providing a sub-quadratic time symbolic algorithm. 
\end{compactitem}

	\end{abstract}

		\begin{abstract}
			\begin{sloppypar}
	Ein \emph{Modellprüfer} kontrolliert ob ein gegebenes \emph{Modell} eines Systems 
	eine \emph{Anforderung} erfüllt. In der \emph{reaktiven Synthese} ist die Eingabe eine Anforderung und 
	das Ziel ist ein korrektes \emph{reaktives System} zu erzeugen.
	Wir betrachten folgende Modelle: \emph{(Zustand-Transition) Graphen} wo Knoten die Zustände des Systems darstellen und Kanten die Transitionen von einem Zustand zum Nächsten sind. 
	\emph{Markow-Entscheidungsprozesse} (MEPs) erweitern Graphen 
	um probabilistische Systeme modellieren zu können, z.B. randomisierte Protokolle. 
	\emph{Spielgraphen} erweitern Graphen um Interaktionen mit der Umwelt zu 
	modellieren und sind ein zentrales Werkzeug für die reaktive Synthese.
\end{sloppypar}
\emph{Zielvorgaben} sind Mengen von Abläufen in einem Modell und 
repräsentieren die Anforderungen. Die Klasse der \emph{$\omega$-regulären 
Zielvorgaben} drückt alle üblichen funktionalen 
Anforderungen für reaktive Systeme in der Modellprüfung 
und in der Synthese aus. Beispiele für $\omega$-reguläre Zielvorgaben die wir betrachten sind \emph{Erreichbarkeits}-, \emph{Büchi}-, \emph{eingeschränkte B\"uchi}-, 
\emph{Streett}- und \emph{Paritäts}-Zielvorgaben. 
Um funktionale und effiziente reaktive Systeme zu modellieren benutzen wir die 
Kombination von \emph{quantitativen (Mittelwert-)} Zielvorgaben und $\omega$-regulären Zielvorgaben.
Für die Prüfung von MEPs mit $\omega$-regulären Zielvorgaben betrachten wir die Berechnung von der \emph{maximalen Schluss-Komponenten (MSK) Dekomposition} eines MEPs.

Das Ziel der Arbeit ist es schnelle 
Graphalgorithmen und moderne algorithmische Techniken für 
Probleme in der Modellprüfung und Synthese in Graphen, MEPs und 
Spielgraphen wirksam einzusetzen. 
Unter den Resultate sind auch \emph{symbolische Algorithmen}, 
eine bekannte Klasse von Algorithmen in der Modellprüfung die einen
begrenzten Zugang zum Eingabemodell für eine effiziente 
Repräsentation eintauschen. Wir stellen folgende Ergebnisse vor:
\begin{compactitem}
\item Algorithmen für Spielgraphen mit Mittelwert-B\"uchi-Zielvorgaben und Mittelwert-coB\"uchi-Zielvorgaben 
	die einem der schnellsten Algorithmen für Mittelwert-Zielvorgaben in der Laufzeit gleichziehen. 
\item Ein randomisierter Algorithmus für Streett-Zielvorgaben in Graphen und MEPs in fast linearer Laufzeit.
\item Ein Algorithmus für eingeschränkte B\"uchi Zielvorgaben in sub-kubischer Laufzeit für Graphen und kubischer Laufzeit in Spielgraphen.
\item Konditionale untere Schranken für Anfragen von 
	Erreichbarkeits-Zielvorgaben in Spielgraphen und MEPs. 
	Algorithmen für sequenzielle Erreichbarkeits-Zielvorgaben in 
	Graphen und MEPs in jeweils linearer Zeit und fast linearer Zeit.
\item Der erste symbolische Algorithmus für 
	Paritäts-Zielvorgaben in quasi-polynomieller Zeit.
\item Wir durchbrechen eine langstehende Laufzeitschranke für die MSK Dekomposition 
	von den 90er-Jahren indem wir einen symbolischen Algorithmus in sub-quadratischer Laufzeit präsentieren.

\end{compactitem}

		\end{abstract}

	\cleardoublepage
	\renewcommand{\abstractname}{Acknowledgments}
	\begin{abstract}
		I am deeply grateful to my advisor Monika Henzinger.
The years as a Ph.D. student passed by extremely fast due to her outstanding guidance and support. 
During all these years I was very proud to have her as my advisor: 
She continuously pushed me into the right challenges which sharpened my problem-solving skills and
taught me essential research skills in computer science. I hope that a small portion of her brilliance which  
I witnessed many times during our research sessions and convinced me to do the Ph.D. under 
her guidance in the first place rubbed off on me.

I am deeply thankful to Krishnendu Chatterjee for being a great mentor and for 
all the invaluable guidance throughout this research.
I thank Wolfgang Dvo\v{r}\'{a}k for all the hours we spent discussing research problems 
and writing papers together: Collaborating with him is an awesome experience. 
I thank Gramoz Goranci for the collaboration and
for being a cheerful roommate. I thank Christian for the collaboration and for introducing me to algorithm engineering. 
I thank Sagar Kale for the collaboration and for improving my writing style. 

Special thanks go to Christel Baier and V\'{e}ronique Bruy\`{e}re who have agreed to review this
thesis and be part of my thesis committee, especially in these crazy times: I could have not wished for a better thesis committee. 

Finally I would like to thank all my colleagues over the years: Stefan Neumann,
Alexander Noe, Kathrin Hanauer, Sebastian Forster, Pan Peng, Dariusz Leniowski, Veronika Loitzenbauer, Shahbaz Khan, Sebastian Lamm, Marcelo Faraj, and 
Xiaowei Wu who contributed with questions, suggestions, clarifications and attention.

\vspace{6ex}
{\small
The research leading to these results has received funding from the European Research Council under the European Union’s
Seventh Framework Programme (FP/2007--2013) / ERC Grant Agreement
no. 340506.\ and from the Vienna Science and Technology Fund (WWTF) through
project ICT15--003.
}

	\end{abstract}

	\cleardoublepage
	\renewcommand{\abstractname}{Bibliographic Note}
	\begin{abstract}
		Several results of this thesis were already published in conference papers and thus the chapters of this thesis are based on the following papers:
\begin{sloppypar}
	\begin{itemize}
		\item \textbf{\Cref{cha:mfcs}:} \fullcite{CDHS17MFCS}. \nocite{CDHS17MFCS}
		\item \textbf{\Cref{cha:concur}:} \fullcite{CDHS19CONCUR}. \nocite{CDHS19CONCUR}
		\item \textbf{\Cref{cha:icalp}:} \fullcite{CHKS21ICALP}. \nocite{CHKS21ICALP}
		\item \textbf{\Cref{cha:icaps}:} \fullcite{CDHS18ICAPS}. \nocite{CDHS18ICAPS}
		\item \textbf{\Cref{cha:lpar}:} \fullcite{CDHS18LPAR}. \nocite{CDHS18LPAR}
		\item \textbf{\Cref{cha:lics}:} \fullcite{CDHS21LICS}. \nocite{CDHS21LICS}
	\end{itemize}
\end{sloppypar}

	\end{abstract}

	\cleardoublepage
	\tableofcontents*

	\mainmatter
	\chapter{Introduction}
	In the last decades, computer systems enriched the world in many aspects: 
	Transportation (e.g.\ self-driving cars), entertainment (e.g.\ decades of online video content) and 
	production (e.g.\ enterprise resource planning) are just a few examples.
	Because humans create most of the aforementioned systems, 
	coding mistakes and conceptual errors are frequent and hard to
	avoid. Overlooked side cases, misspelled variable names and uninitialized pointers are part 
	and parcel of the systems produced nowadays. Especially in safety-critical environments,
	bugs have detrimental outcomes. For example, in 1996 the launch of the Ariane 5
	rocket failed because its navigation system tried to convert a 64-bit number to a 16-bit number without
	checking if 16 bits are sufficient to hold the value~\cite{Ben-Ari01}.
	Other examples of similar disasters include bugs in medical equipment, electric power
	transmission, and automated trading systems which
	cost human lives and millions of dollars~\cite{softwarebugslist}.

	In practice, computer scientists usually write \emph{tests} to avoid errors.
	A test succeeds if the input matches the expected output. On the upside, 
	tests are efficient and easy to write. On the downside, tests cannot, in general, certify the
	correctness of a system. In particular, bugs in a concurrent system are notoriously hard to find
	with tests: In one instance, a test of a concurrent system may pass but, on the next instance, 
	the same test may fail due to concurrency issues. In contrast to the testing paradigm, 
	the field of computer-aided verification \emph{proves} that a system is correct. 

	Due to the undecidability of the Halting Problem~\cite{turing1937computable,church1936unsolvable} and Rice's
	Theorem~\cite{ricethm}, we know that certifying if a given arbitrary system has a 
	property is undecidable. Even when we accept these harsh limitations and shift 
	the focus to systems with finitely many states where the verification of properties 
	becomes decidable, many questions in computer-aided verification are immensely difficult to resolve 
	due to results in complexity theory.

	Despite these sobering general results, Clarke, Emerson, Sifakis and others invented a 
	practical approach to tackle the challenges in computer-aided verification 
	called ``model checking''~\cite{Pnueli77,ClarkeE81,QueilleS82,ClarkeES09}.
	Given a system and a system specification the ``model checking workflow'' is as follows:
	First, we compile the system into a \emph{finite state-transition graph} (modeling)
	and translate the system specification into \emph{properties} of the model.
	The \emph{model checker} takes the state-transition graph and the properties as
	input. The output is either (a) that the given state-transition graph fulfills the desired 
	properties, i.e., the system ensures the system specification or (b) that there is a 
	counterexample of one of the properties because, e.g., the system has a bug or the 
	model does not sufficiently represent the system. 

	In what follows, we consider \emph{reactive systems}, i.e., systems that continuously 
	interact with the environment (other parallel processes, user input, etc.). 
	A state of a reactive system consists of the variable valuations at a point in time.
	An example of a rudimentary reactive system is a semaphore, a reactive system that 
	makes sure that only one out of many parallel processes enters a critical section at 
	each point in time.
	For reactive systems \emph{terminal states} are undesirable (e.g., the semaphore is in a deadlock)
	and, thus, every state has a successor state. More sophisticated examples for reactive systems
	are, for example, an aircraft or a flight-control systems.

	\para{Models.}
	The standard model is a \emph{(state-transition) graph} where a vertex represents a state of the
	reactive system. 
	The edges of the graph represent 
	transitions between states, e.g., when a variable is incremented.
	A \emph{play} is an infinite path starting at an \emph{initial vertex} in the graph 
	which represents the initial state of the reactive system. A play represents a run of the reactive system.
	While many reactive systems can be modeled using state-transition graphs, 
	there are cases where graphs are not able to express the desired properties of a reactive
	system: For example, if the state transition depends on the action of an adversarial actor 
	or an uncertain environment. 
	When a reactive system interacts with an adversarial environment 
	(e.g., user input), \emph{game graphs} extend the notion of 
	graphs~\cite{AbadiLW89,Dill89book}.  
	Another prominent use-case for game graphs is \emph{Synthesizing}~\cite{Bruyere17,BohyBFJR12}: 
	Synthesizing is the ambitious approach to build a reactive system from scratch given a specification~\cite{Church62,Buchi62}. The process of synthesizing is a turn-based game where the environment controls possible inputs for the system and the system controls the output~\cite{ModelCheckingBook}. 	
	If a state transition depends on a given probability distribution or non-deterministic
	behavior, 
	\emph{Markov Decision Processes (MDPs)} extend graphs~\cite{Vardi85, CY95}. 
	Model Checking in MDPs is a key component in establishing the correctness of randomized distributed
	algorithms~\cite{KwiatkowskaNP12}, studying	biological processes~\cite{HeathKNPT08}, 
	analyzing security protocols~\cite{NormanS06}, optimizing power management~\cite{NormanPKSG05} 
	and many more~\cite{prismcasestudies,BaierHK19,BaierAFK18,BolligC04}.
	
	\para{Games Graphs.}
	In game graphs, there are two players. Player~1 represents the choices of the system and 
	player~2 represents the environment. \emph{Player-1 vertices} and \emph{player-2 vertices} partition the
	vertices in the graph. At the beginning of a play, a token is placed on the
	vertex which represents the initial state. When the token
	is on a player-1 vertex, player~1 moves the token along one of its outgoing edges and 
	the system goes to the next state. When the token is on a player-2 vertex, 
	player~2 (the environment) moves the token along one of the outgoing edges. 

	\para{Markov Decision Processes (MDPs).}
	When the system interacts with a non-deterministic or an uncertain component we model it with
	MDPs. 
	That is, vertices are partitioned into player-1 vertices and \emph{random vertices}. 
	Again, at the start of a play, we place a token on the initial state. 
	For player-1 vertices, the system has the choice to go to the next state along its outgoing edges. 
	When the token is on a random vertex, a probability distribution determines the 
	next vertex along one of its outgoing edges. 
	For the problems we consider in MDPs, we can assume without loss of generality that the
	probability distribution is uniform. 

	\para{Specification.}
	The system specification defines the desired properties of a model and 
	is input in both model checking and reactive synthesis. An example of such a property is the
	\emph{safety property} where a system must not reach a set of undesired states. \emph{Objectives} define 
	the desired properties of a reactive system as a set of plays.
	\emph{Qualitative objectives}, such as $\omega$-regular objectives 
	describe the functional behavior of a reactive system.
	Throughout, we consider mostly $\omega$-regular objectives which are canonical because they express
	most functional requirements for model checking and synthesis~\cite{MannaP92, PnueliR89}. 
	\emph{Quantitative objectives} describe desired behavior for the performance or 
	resource consumption of a reactive system~\cite{ChakrabartiAHS03} and are often combined with qualitative objectives~\cite{CJH05,BohyBFR13}.
	We consider the following objectives:

	\para{Reachability and Safety objectives.}
	Given a set of ``good vertices'', a
	\emph{reachability objective} is defined by the set of plays, 
	which include a good vertex. 
	Dually, given a set of ``safe vertices'', a \emph{safety objective} is the set of plays,  
	which visit only safe vertices. 
	
	\para{Sequential Reachability objectives.}
	Given $k$ sets of vertices, the \emph{sequential reachability objective} 
	is the set of plays that visit at some point in time a vertex in the first set, then, later on, a vertex
	in the second set and so on until the play visits a vertex in the $k$th set of vertices. 
	
	\para{B\"uchi and coB\"uchi objectives.}
	Given a set of ``B\"uchi vertices'', \emph{B\"uchi objectives} are the set of plays which visit
	a B\"uchi vertex infinitely often.
	Dually, given a set of ``coB\"uchi vertices'', \emph{coB\"uchi objectives}
	describe the set of plays which visits only the coB\"uchi vertices infinitely often. 

	\para{Bounded B\"uchi and bounded coB\"uchi objectives.}
	\emph{Bounded B\"uchi objectives} and \emph{bounded coB\"uchi objectives} extend the B\"uchi and coB\"uchi objectives:
	Given an integer $d$ and a set of B\"uchi vertices a play in the bounded B\"uchi objective includes 
	a B\"uchi vertex every
	at most $d$ steps after finitely many steps. Dually, the bounded coB\"uchi objectives contains a play 
	if it visits $d$ coB\"uchi vertices consecutively infinitely often. 

	\para{Parity objectives.}
	For the \emph{parity objective}, every vertex in the model has an integer priority. 
	The parity objective includes plays where the vertex with minimum priority visited 
	infinitely often is even. 
	
	\para{Streett objectives.}
	\emph{Streett objectives} have a set of requests and corresponding grants. Each request and
	corresponding grant is a set of vertices in the model. 
	Streett objectives include a play if it reaches for every request occurring infinitely often the corresponding grant infinitely often. 

	\para{Mean-Payoff objectives.}
	\emph{Mean-Payoff objectives} are quantitative objectives where every 
	edge in the model has an associated reward. 
	The \emph{payoff} of a play is the limit of the long-run average of the rewards in the play.
	Mean-payoff objectives have a threshold and include a play if the payoff is above the threshold.

	\para{Mean-Payoff Parity objectives.}
	\emph{Mean-payoff parity objectives} combine mean-payoff objectives and parity objectives:
	A play is in the mean-payoff parity objective if it is in the parity objective and the mean-payoff objective.

	\para{Algorithmic Questions.}
	We perform an algorithmic study of the following two questions:
	\begin{compactitem}
	\item Given a model, an objective, and an initial vertex in the model we
		compute whether player~1 can ``force'' a play in the objective (despite the non-determinism in an MDP or
		the adversarial choices of player~2 in a game graph). 
		We say that the play \emph{ensures} the objective and that the vertex is \emph{winning} for player~1. 
		A natural extension of this question is to compute \emph{all} vertices for which 
		player~1 can ensure the given objective in the model, i.e., the \emph{winning set}. 

	\item Given an MDP, we compute the \emph{maximal end-component (MEC) decomposition} of an MDP\@. 
		Intuitively, a MEC describes a maximal (under set inclusion) 
		set of vertices in the MDP for which player~1 can from any vertex in the MEC ``force'' 
		to reach every other vertex in the MEC despite the probabilistic elements of the model,
		i.e., it generalizes strongly connected components in graphs. 
		Computing the MECs of an MDP is a key component for computing winning sets 
		of $\omega$-regular objectives~\cite{deAlfaroThesis,ChaThesis}. 
	\end{compactitem}

	\para{Symbolic Algorithms and the state explosion problem.}
	A state of a reactive system consists of a valuation for each variable at a point in time. 
	Note that this causes huge numbers of vertices in the induced model because the number of states grow exponentially with the number variables, e.g., 
	for a bit array with 20 entries and two variables with values in $\{0,\dots,9\}$ we have at least $2^{20}\cdot10^2$ states.
	This huge induced model is a major drawback in model-checking because the model might not fit into the
	main memory. To tackle this problem, a successful 
	approach is to express the sets of states and transition relations \emph{implicitly} instead of
	\emph{explicitly} in terms of BDDs~\cite{bryant1986graph,BurchCMD90,bry92,BaierCHKR97}. We consider a theoretical model for algorithms
	that works for this implicit representation without considering specifics of the representation 
	called \emph{symbolic model of computation}~\cite{ChatterjeeHLOT18,ChatterjeeHJS13,ChatterjeeDHL18,GentiliniPP03,abs-2001-04333}. 
	A \emph{symbolic algorithm} can use the same operations as a regular RAM algorithm, 
	except for the access to model: To access the input graph a symbolic algorithm must use the following two types of \emph{symbolic operations}: 
	\begin{compactenum}
	\item \emph{One-step operations $\Pre{}{\cdot}$ and $\Post{}{\cdot}$}. 
			Given a set of vertices $X$, the predecessor operation $\Pre{}{X}$ returns the set of vertices
			with an edge to a vertex in $X$. Similarly, the successor operation $\Post{}{X}$ returns the set of vertices with an edge from some vertex in $X$.
		\item  \emph{Basic set operations.} Basic set operations have as input one or two sets of
			vertices or edges and perform, for example, the union, intersection, and complement on
			the set(s).
	\end{compactenum}
	In the symbolic model, running time is defined as the \emph{number of symbolic operations}.
	The symbolic model defines one unit of space as one set (not the size of the set) due to the implicit representation as BDD. \emph{Symbolic space} is the number of sets stored simultaneously at any point of the algorithm.

	\para{Modern Graph Algorithms.}
	We leverage the power of modern graph algorithms to
	obtain faster algorithms for the algorithmic questions. We also provide negative results called ``conditional lower bounds'', i.e., 
	we show that an improvement in running time for the algorithmic questions implies an algorithm
	that breaks a long-standing running time barrier for well-studied problems like SAT\@.
	We describe three algorithmic concepts which contributed to the results over the algorithmic questions:

	\para{Dynamic Graph Algorithms.}
	A \emph{dynamic graph algorithm}, in contrast to a \emph{static} algorithm, maintains a property of a graph, e.g.\ strongly connected
	components (SCCs), while edges of a graph are removed and added. The most na\"ive dynamic graph algorithm
	computes the property from scratch every time an edge is removed or added.
	Finding
	fast dynamic algorithms is an exciting mathematical challenge because they are relevant in
	practice~\cite{hanauer2021recent} and can be used as subroutines in static algorithms, for
	example, when vertices are repeatedly deleted in a static algorithm. 

	\para{Hierarchical Graph Decomposition.}
	The \emph{hierarchical graph decomposition} is a technique originally introduced for dynamic graph
	algorithms~\cite{HenzingerKW99} but a breakthrough result showed that a similar technique can be used to
	obtain faster algorithms for computing winning sets in game graphs with B\"uchi objectives~\cite{CH14} and other problems~\cite{CHL17,HenzingerKL15} where we repeatedly remove sets of vertices. 
	Given a model with $n$ vertices,
	the technique consists of iteratively constructing $\log n$ subgraphs 
	with only $O(2^i)$ $(1 \leq i \leq \log n)$ edges for each vertex of the original graph (the vertices remain
	unmodified in the subgraphs). The subgraph $G_i$ is defined by the edges in iteration $i$.
	The	goal is to find a set of vertices that fulfills a given desired 
	property and to show that the size
	of this set is proportional to the number of edges in $G_i$. 
	We save running time by repeatedly looking at 
	subgraphs of the original graph with a smaller amount of edges to find vertices with the desired
	property and making sufficient progress by removing at least $O(2^i)$ vertices from
	the graph. 

	\para{Conditional Lower Bounds.}
	Similar to reductions from an NP-hard problem, \emph{conditional lower bounds} give running time
	guarantees conditioned on the fact that there is no significant improvement in running time over decades for
	certain well-studied problems~\cite{VW2018survey}. Examples for the ``well-studied'' problems 
	include \emph{all-pair shortest paths in graphs}, \emph{3-SUM} and \emph{CNF-SAT}\@. In particular, 
	if we have a conditional lower bound for 
	an algorithmic question, we obtain
	better algorithms for a long-standing well-studied problem if there is a substantial improvement in running 
	time for the algorithmic question which is highly unlikely. 

	\section{Related Work}
	In this section, we present an overview of the results for the considered algorithmic questions. 
	We describe the running time of the algorithm for a model with $n$ vertices and $m$ edges. 
	For parity objectives, we consider models with $d$ priorities and for mean-payoff objectives 
	$W$ denotes the maximum weight of an edge. For Streett objectives, we denote with $b$ the size of
	the input sets consisting of the requests and grants.

	For explicit algorithms, in \emph{graphs}, the following results are known: 
		\begin{compactitem}
			\item The winning set of reachability, safety, B\"uchi and
				coB\"uchi objectives can be computed in linear time using
				depth first search and strongly connected components~\cite{Tarjan72}. 

			\item For Streett objectives, the best algorithms compute the winning set in time $O(m^{1.5} \sqrt{\log
					n})$ and $O(n^2)$~\cite{CDHL16,HT96}.
		\end{compactitem}
		\begin{sloppypar}
			In \emph{MDPs}: 
			\begin{compactitem}
			\item The MEC-decomposition can be determined in $O(min(m^{1.5}, n^2))$~\cite{CH14}.
			\item Computing the winning set of reachability objectives can be done in time $O(m
				+ \MEC)$ where $\MEC$ denotes the running time of the best algorithm for
				MEC-decomposition~\cite{CDHL16}.
			\item Computing the winning set of Streett objectives can be done in time
				$O(m^{1.5} \sqrt{\log n})$, $O(n^2)$~\cite{CDHL16}.
			\end{compactitem}
		\end{sloppypar}
		In \emph{game graphs}: 
		\begin{compactitem}
		\item Computing the winning set of reachability and safety objectives can be done in
			linear time~\cite{I81,B80}.
		\item For B\"uchi and coB\"uchi objectives, the winning set can be computed in
			$O(n^2)$ time~\cite{CH14}.
		\item For computing the winning set of parity objectives 
			there is a long line of work which improves the running
			time~\cite{Jurdzinski00,JurdzinskiPZ08,Schewe17}. 
			The most important recent development is an algorithm in 
			quasi-polynomial running time ($O(n^{\log d + 6})$)~\cite{calude2017stoc}. 
			Since this breakthrough, many other quasi-polynomial
			algorithms were
			discovered~\cite{FearnleyJS0W17,lazic2017qp,Parys19,LehtinenB20,DaviaudJT20}. 
			The long-standing open question 
			is if there is a polynomial-time algorithm for computing
			the winning set of parity objectives in game graphs.
		\item Similarly, for computing the winning set of mean-payoff objectives
			it is not known whether there exists a polynomial time algorithm. 
			The two fastest algorithms are $O(mn2^{n/2})$~\cite{DorfmanKZ19,K21}
			and $O(mnW)$~\cite{brim2011}. 

		\item For the intersection of mean-payoff objectives and parity objectives, i.e.,
			mean-payoff parity objectives the best algorithm for computing the winning set
			is in $O(dmn^{\lg(d/\lg n)+2.45}W)$~\cite{DaviaudJL18}. 
		\end{compactitem}

		For \emph{symbolic algorithms} the following results are known:
		\begin{compactitem}
		\item In MDPs, there are two results for computing the MEC-decomposition:
			First, the classical symbolic algorithm which 
			performs $O(n)$ symbolic SCC decompositions which can be done in $O(n)$
			symbolic operations and $O(\log n)$ symbolic space~\cite{GentiliniPP03}.
			Second, an algorithm which uses $O(n\sqrt{m})$ symbolic operations and $O(\sqrt{m})$ symbolic
			space~\cite{ChatterjeeHJS13}. 
		\item In game graphs, the best algorithm to compute parity objectives uses $O(\lg d)$ 
			symbolic space and $n^{2\lg(d/\lg n) + O(1)}$ symbolic operations~\cite{abs-2001-04333}.
			Another, older algorithm for parity games before the first explicit 
			quasi-polynomial algorithm~\cite{calude2017stoc} 
			uses $\min\{n^{O(\sqrt{n})},O(n^{d/3+1})\}$ symbolic operations and $O(n)$ symbolic space~\cite{chatterjee2017symbolic}.
		\end{compactitem}

	\section{Results and Outline.}
	In this section, we give an overview of the chapters in the thesis. The $\O(\cdot)$ notation hides
	poly-logarithmic factors.
	\begin{compactitem}
		\item In Chapter~\ref{cha:prelim}, we introduce the necessary definitions to describe the results.
			
		\item In Chapter~\ref{cha:mfcs}, we present algorithms which improve the running time for
			computing the winning sets of mean-payoff B\"uchi objectives and
			mean-payoff co-B\"uchi objectives from $O(n^3mW)$ to $O(nmW)$. 
			Furthermore, we present an $O(n^{d-1}mW)$ time algorithm for the winning set of mean-payoff parity
			objectives. 
		\item In Chapter~\ref{cha:concur}, we present near-linear algorithms, i.e., a $\O(m + b)$
			time algorithm for computing the winning set of a 
			Streett objective in graphs and MDPs. Towards that goal, we present $\O(m)$ time
			algorithms for computing the MEC decomposition of an MDP, computing the MEC decomposition
			under edge deletions, and computing the winning set of a reachability objective in MDPs. 

		\item We present the first sub-cubic time algorithm  
			for computing the winning set of bounded B\"uchi objectives in graphs in Chapter~\ref{cha:icalp}.
			For the same problem in game graphs, we present a new $O(n^2d)$ time algorithm using
			hierarchical graph decomposition.

		\item We study reachability problems with $k$ sets of vertices in graph games and MDPs 
			in Chapter~\ref{cha:icaps}: 
			In particular, we present algorithms for computing the winning set of
			sequential reachability objectives and the \emph{coverage problem} where, instead of
			one reachability objective we are given $k$ reachability objectives which must be
			satisfied at the same time. 
			For computing the winning set of sequential reachability objectives
			we present conditional lower bounds for game graphs. 
			For MDPs, we present a subcubic time algorithm which rules out conditional lower bounds.
			For the coverage problem we provide novel conditional lower bounds for game graphs and
			MDPs and argue why the problem can be solved in linear time in graphs. 
		
		\item In Chapter~\ref{cha:lpar}, we shift the focus to symbolic algorithms. 
			We present the first \emph{symbolic} algorithm for computing the winning set of parity
			objectives in graph games with a quasi-polynomial
			number of symbolic operations and $O(d \log n)$ symbolic space. 
			
		\item In Chapter~\ref{cha:lics}, we present two symbolic algorithms for MDPs. 
			The first algorithm computes the MEC decomposition with a symbolic space
			and symbolic operation trade-off: It requires $\O(n^{2-\epsilon})$
			symbolic operations and $\O(n^{\epsilon})$ symbolic space for $0 < \epsilon \leq 1/2$.
			The second algorithm computes the winning set of parity objectives with $\log d$ 
			computations of the MEC decomposition and improves the time-space product 
			from $\O(n^2d)$ to $\O(n^2)$.
	\end{compactitem}

	\chapter{Preliminaries}\label{cha:prelim}
	In this chapter, we introduce necessary definitions for the following chapters.
Since the notation and definitions are standard, we base them on similar definition
sections~\cite{Loitzenbauer16, CH14, ChatterjeeDHL18,chatterjee2016separation}.
\section{Models}

\para{Game Graphs.}
Game graphs $\GG = ((V,E), \ls V_1, V_2 \rs)$ consist of a finite set of
vertices $V$, a finite set of edges $E$ and partitions $V$ into
player-1 vertices $V_1$ and the adversarial player-2 vertices $V_2$.

\para{Markov decision process (MDP).} 
An MDP $P\! =\! ((V,E), \ls V_1, V_R \rs,\delta)$ 
has a finite set of vertices $V$ which we partition into 
player-1 vertices $V_1$ and random vertices $V_R$, 
a finite set of edges $E \subseteq (V \times V)$, 
and a probabilistic transition function $\delta$.  
The probabilistic transition function maps $V_R$ to $\D(V)$, i.e., the set of probability distributions 
over the set of vertices $V$.  There is an edge from a random vertex $v$ to a vertex $w$, i.e.\
$(v,w) \in E$ if and only if $\delta(v)[w] > 0$. 
We call an edge $e = (u,v)$ \emph{random edge} if $u \in V_R$. Otherwise it is a \emph{player-1	edge}.
For simplicity, we let $\delta(v)$ be the uniform distribution over vertices $u$ with $(v,u) \in E$:
this common technical assumption is without loss of generality for the qualitative analysis of MDPs.

\begin{remark}\label{icaps:rem:mdps}
	A standard way to define MDPs, e.g.~\cite{GuestrinKPV03}, is to 
	consider vertices with actions to define a probabilistic transition function 
	for every vertex and action. 
	In our model, the choice of actions is represented as the choice of edges at
	player-1 vertices and the probabilistic transition function is represented 
	by the random vertices.
	This allows us to treat MDPs and game graphs uniformly, 
	and graphs can be described easily as a special case of MDPs.
\end{remark}

We also follow the common technical assumption that vertices in game graphs and MDPs do not have self-loops and
that every vertex has an outgoing edge.

\para{Graphs.} 
\emph{Graphs} are the special case of MDPs with $V_R = \emptyset$ and game graphs with $V_2 =
\emptyset$. The set $\Out{v} = \{ u \in V \mid (v,u) \in E\}$ describes the set of successors
of $v$ and the set $\In{v} = \{u \in V \mid (u,v) \in E \}$ describes the set of predecessors
of $v$.
When $U$ is a set of vertices, we define $\edgeset{U}$
to be the set of all edges incident to the vertices in $U$, i.e., 
$\edgeset{U} = \{(u,v) \in E \mid u \in U \text { or } v \in U\}$. 
With $G[S]$ we denote the subgraph of $G = (V,E)$ induced by the set of vertices $S \subseteq V$,
i.e., $G[S] = (S,E_S)$ where $E_S = \{(u,v) \mid u,v \in S\}$. A graph is \emph{strongly connected} 
if there is a path between every pair of vertices.
We denote with $n = |V|$ the number of vertices and with $m = |E|$ the number of edges.

\section{Plays and Strategies.}

\para{Plays.}
An infinite sequence $\omega = \ls v_0, v_1, v_2, \dots \rs$ of
vertices such that each $(v_{i-1},v_i) \in E$ for all $i \geq 1$ is called \emph{play}. 
The set of all plays is denoted with $\Omega$. A \emph{finite play} $V^*$ is a finite prefix of a play.

\para{Strategies.}
We call the recipes that extend finite plays \emph{strategies} and there are player-1 strategies and player-2
strategies.
A player-1 \emph{strategy} is a function $\sigma: V^* \cdot V_1 \mapsto V$ 
and maps every finite play $\omega \in V^* \cdot V_1$ that ends in a 
player-1 vertex $v$ to a successor vertex $\sigma(\omega)$, i.e., 
$(v, \sigma(\omega)) \in E$. We define player-2 strategies analogously.
A player-1 strategy is \emph{memoryless} if $\sigma_1(\omega) = \sigma_1(\omega')$ 
for all $\omega, \omega' \in V^* \cdot V_1$ that end in the same vertex $v \in V_1$, that is, 
the strategy does not depend on the entire finite play, but only on the last vertex. We define memoryless player-2
strategies analogously.
We denote with $\Sigma$ and $\Pi$ the sets of all strategies for player~1 and player~2 respectively.

\para{Outcome of Strategies.}
The \emph{outcome of strategies} is a unique play starting at an initial vertex which we define for all models:
In graphs, given an initial vertex, a player-1 strategy induces a unique play in the graph. 
In MDPs, given an initial vertex $v$ and a player-1 strategy $\sigma$, we obtain a set of possible
plays $\omega(v,\sigma)$ if player~1 follows $\sigma$ since random vertices choose their successor according to a
probability distribution. An event is a measurable subset of $\omega(v,\sigma)$ and the
probabilities of events are uniquely defined~\cite{Vardi85}. 
For a vertex $v$, strategy $\sigma$ and an event $\A \subseteq \Omega$, we
denote by $\Pr^\sigma_v(\A)$ the probability that a play belongs to
$\A$ if the game starts at $v$ and player~1 follows $\sigma$.
In game graphs, given a starting vertex $v$ and strategies player-1 strategy $\sigma$ and player-2
strategy $\pi$, the unique play $\omega(v,\sigma,\pi) = \ls v_0, v_1,v_2,
\dots, \rs$ is defined as $v_0 = v$ and for all $i>0$ if $v_i \in V_1$ then
$\sigma(\ls v_0,\dots,v_i \rs) = v_{i+1}$ and if $v_i \in V_2$, then $\pi(\ls v_0, \dots, v_i \rs) =
v_{i+1}$.

\section{Objectives.}
An \emph{objective} $\Phi \subseteq \Omega$ is the set of ``winning plays''.
The play $\omega \in \Omega$ \emph{satisfies} the objective if $\omega \in \Phi$. 
For a play $\omega = \ls v_0, v_1, v_2, \dots \rs$ define $\Inf{\omega} = \{ v \in V \mid v_i = v \text{ for infinitely many } i\geq 0\}$ 
to be the set of vertices that occur infinitely often in $\omega$.
For the definitions of the following objectives let $\Gamma$ be an MDP or a game graph.

\begin{enumerate}
	\item \emph{Reachability and Safety objectives.} A \emph{Reachability objective}
		$\reach{T,\Gamma}$ requires, given a set $T$ of vertices, 
		that a play visits \emph{at least one} vertex in $T$. Dually, for a set of vertices $C$, 
		a play in the \emph{safety objective} $\safety{C,\Gamma}$ visits \emph{only} vertices in $C$. 
		Formally, $\reach{T,\Gamma} = \{\langle v_0, v_1, v_2, \dots \rangle \in \Omega \mid \exists
			k \geq 0: v_k \in T\}$ and $\safety{C, \Gamma} = \{\langle v_0, v_1, v_2, \dots \rangle \in \Omega \mid \forall k \geq 0: v_k \in C\}$. 
		The two objectives are dual, i.e., $\reach{T, \Gamma} = \Omega \setminus \safety{V\setminus T, \Gamma}$.

\item \emph{Sequential Reachability.} For a tuple of vertex sets 
	$\Targets = (T_1,T_2,\dots,T_k)$ in $\Gamma$ the \emph{sequential reachability objective} is the 
	set of infinite plays that contain a vertex of $T_1$ followed by a vertex of
	$T_2$ and so on up to a vertex of $T_k$, i.e.,
	$\Seq{\Targets, \Gamma} = \{ \ls v_0, v_1, v_2, \dots \rs \in \Omega  \mid \exists j_1, j_2,
		\dots j_k : v_{j_1} \in T_1, v_{j_2} \in T_2, \dots, v_{j_k} \in T_k \text{ and }
		j_1 \leq j_2 \leq \cdots \leq j_k \}$.
	 
	\item \emph{B\"uchi and coB\"uchi objectives.} Given a set $B$ of vertices called ``B\"uchi
		vertices'', the \emph{B\"uchi objective} $\buchi{B,\Gamma}$ contains all plays which visit a
		vertex in $B$ infinitely often. Dually, the \emph{coB\"uchi objective} $\cobuchi{C,\Gamma}$
		contains a play if it visits only vertices in $C$ infinitely often. Formally, 
		$\buchi{B,\Gamma} = \{ \omega \in \Omega \mid \Inf{\omega} \cap B \neq \emptyset \}$ and 
		$\cobuchi{C,\Gamma} = \{ \omega \in \Omega \mid \Inf{\omega} \subseteq C\}$.
		The two objectives are dual, i.e., $\buchi{B,\Gamma}  = \Omega \setminus \cobuchi{V	\setminus B, \Gamma}$.

	\item \begin{sloppypar}
			\emph{Bounded B\"uchi and bounded coB\"uchi objectives.} Given a set of ``B\"uchi
			vertices'' $B$ and an integer $d \geq 0$, the \emph{bounded B\"uchi objective} $\pbuchi{B,d,\Gamma}$ 
			includes a play if, after visiting finitely
			many arbitrary vertices, the distance between any two consecutive B\"uchi vertices in
			the play is at most $d$. 
			Dually, the \emph{bounded coB\"uchi objective} $\pcobuchi{C,d, \Gamma}$ includes a play if it
			visits at least $d$ consecutive vertices in $C$ infinitely often. 
		\end{sloppypar}
		Formally, the sets of ``winning plays'' are
		{\small
		\begin{align*}
			\pbuchi{B,d, \Gamma} &= \{ \omega \in \Omega \mid \exists i\geq 0 
				\forall j\geq i: \{v_j,v_{j+1}, \dots, v_{j+d-1}\} \cap B \neq \emptyset  \}\\ 
			\pcobuchi{C,d, \Gamma} &= \{ \omega \in \Omega \mid \forall i \geq 0 \exists j \geq i: \{v_j, v_{j+1}, \dots, v_{j+d-1}\}
				\subseteq C \}.
		\end{align*}
	}%

\item \emph{Parity objectives.} Given a \emph{priority function} $p$ that maps every vertex 
		to a non-negative integer priority, a play satisfies the \emph{parity objective} if 
		the minimum priority vertex that appears infinitely often is even. 
		Formally, the parity objective is the set $\parity{p,\Gamma} = \{ \omega \in \Omega \mid \min \{ p(v) \mid 
				v \in \Inf{\omega}\} \text{ is even} \}$.
		The Büchi and coBüchi objectives are special cases of parity objectives with two priorities:
		For Büchi objectives the image of $p$ is $\{0,1\}$ and for coBüchi objectives it is $\{1,2\}$. 
		Given a game graph $\Gamma$ with parity function $p$, we call $(\Gamma,p)$ a \emph{parity game}.

	\item  \emph{Threshold mean-payoff objectives.} 
		A \emph{weight function} $w: E \mapsto \Z$ maps edges to integers. 

		\para{Mean-Payoff function.}
		Given $(\Gamma,w)$, the \emph{mean-payoff function} $\MP{\omega,w,\Gamma}$ maps a play $\omega$ 
		and a weight function $w$ to the long-run average weight of the play:
		$\MP{\omega,w,\Gamma} = \liminf\limits_{n\mapsto \infty} \frac{1}{n} \cdot \sum_{i=0}^{n-1} w(v_i,v_{i+1})$.

		\para{Mean-payoff objectives.}
		Given a threshold $\V \in \Q$ and a weight function $w$,
		the \emph{threshold mean-payoff objective} includes plays with a mean-payoff value of at
		least $\V$, i.e., $\mpayoff{\V,w,\Gamma} = \{ \omega \in \Omega \mid \MP{\omega ,w,\Gamma} \geq \V \}$. 
		When $\Gamma$ is a game graph, $(\Gamma,w)$ is a \emph{mean-payoff game}.

	\item \emph{Threshold mean-payoff parity objectives.}
		Given a priority function $p$ and a weight function $w$ in a game graph $\Gamma$ we call the triple 
		$(\Gamma,p,w)$ a \emph{mean-payoff parity game}.

		\para{Mean-Payoff Parity function.}
		Given $(\Gamma,p,w)$ the \emph{mean-payoff parity 
			function} maps every play in $\Gamma$ to a real-number or $-\infty$ as follows:
		If the play satisfies the parity objective, then the value of the play is the mean-payoff value, 
		else it is $-\infty$.
		Formally, for a play $\omega$ we have
		\[
			\MPP{\omega ,p,w,\Gamma}=\begin{cases}
				\MP{\omega,w,\Gamma} & \text{ if } \omega \in \parity{p,\Gamma} \\
				- \infty & \text{ if }  \omega \not\in \parity{p,\Gamma}.
			\end{cases}
		\]
		\para{Mean-Payoff Parity objective.}
		Given a priority function $p$ for $\Gamma$, a weight function $w$ for $\Gamma$ and a
		threshold $\V$,
		the \emph{threshold mean-payoff parity objective} combines parity and mean-payoff
		objectives, i.e., 
		$\mpayoffp{\V,p,w,\Gamma} = \{\omega \in \Omega \mid \MPP{\omega,p,w,\Gamma} \geq \V\}$.

	\item \emph{$k$-pair Streett objective.} 
		In the \emph{Streett objective}	we are given a set of $k$ pairs of vertex sets, 
		i.e., $\{(L_1, U_1), \dots,	(L_k,U_k)\}$ 
		such that $L_i, U_i \subseteq V$ for $1 \leq i \leq k$. 
		The objective includes a play if, whenever some vertex in $L_i$ is visited
		infinitely often, then also some vertex of $U_i$ is visited
		infinitely often (for all $1\leq i \leq k$). More formally,
		$\streett{\{(L_i,U_i) \mid 1 \leq i \leq k\}, \Gamma} = \{\omega \in \Omega \mid L_i \cap \Inf{\omega} = \emptyset \text{ or }
			U_i(\omega)\neq \emptyset \text{ for all } 1 \leq i \leq k\}$.
\end{enumerate}

We sometimes omit the MDP or game graph $\Gamma$ from the objective when it is obvious on which game graph
the objective is defined on.

\section{Winning Strategies, Winning Sets, Queries, and Value Functions }
\para{Winning Strategies and Winning Sets.}
In MDPs, a player-1 strategy $\sigma$ is \emph{almost-sure (a.s.) winning} 
from a starting vertex $v$ for an objective $\phi$ iff
$\Pr_v^\sigma(\phi) = 1$. The \emph{winning set} $\ASW{\phi}$
for player~1 is the set of vertices from which player~1 has an almost-sure
winning strategy. 

In game graphs, given an objective $\Phi \subseteq \Omega$ for player~1, a strategy $\sigma \in \Sigma$ is a
\emph{winning strategy for player 1} from vertex $v$ if for all player-2 strategies $\pi \in \Pi$
the play $\omega(v,\sigma,\pi)$ satisfies $\Phi$. 
We define the winning strategies of player~2 analogously. 
A \emph{vertex is winning} for player~1 for $\Phi$ if 
player~1 has a winning strategy from $v$. 
Formally, the \emph{set of winning vertices for player~1} for the objective $\Phi$ is 
$\W{1}{\Phi}= \{v \in V \mid
	\exists \sigma \in \Sigma \text{ s.t. } \forall \pi \in \Pi:  \omega(v,\sigma, \pi) \in
	\Phi \}$.
We define the set of winning vertices for player~2 with respect to the objective $\Phi$ with 
$\W{2}{\Phi}= \{v \in V \mid
	\exists \pi \in \Pi \text{ s.t. } \forall \sigma \in \Sigma:  \omega(v,\sigma, \pi)
	\in \Phi
\}$. Due to a seminal result by Martin~\cite{Martin75} every vertex in $V$ belongs to the winning
set of player~1 or the winning set of player~2 and, thus, the two sets form a partition of $V$.
We say that a vertex is either \emph{winning for player~1} if it is in $\W{1}{\cdot}$ or
\emph{winning for player~2} if it is in $\W{2}{\cdot}$.

\para{Coverage and AllCoverage.}
In the \emph{coverage problem} the input consists of the vertex sets $T_1, \dots,T_k$ and a start vertex $s$. 
The coverage problem is not a single objective but is a query involving several objectives.
That is, for $k$ different vertex sets, namely $T_1, T_2, \dots, T_k$, the coverage 
query $\Coverage{T_1, \dots, T_k}$ asks whether $s$ is a winning vertex for player~1 for all reachability
objective $\reach{T_i}$ $(1 \leq i \leq k)$.
If $s$ is winning for player~1 for all the reachability objectives we say that $s$ winning for player~1 regarding
the query $\Coverage{T_1, \dots, T_k}$.
In the \emph{AllCoverage} problem the input are target sets $T_1, T_2, \dots,T_k$ and we determine 
the player-1 winning set of the coverage problem, i.e., all vertices with a player-1 winning strategy for $\Coverage{T_1, \dots, T_k}$.

\para{Value Functions.}
Given a payoff function $f$ and a vertex $v$ (such as the mean-payoff function, or the mean-payoff parity function),
the value for player~1 at $v$ is the maximal payoff that she can guarantee against all strategies of 
player~2.
Formally, \[
\val(f)(v)= \sup_{\sigma \in \Sigma} \inf_{\pi \in \Pi} f(\omega(v,\sigma,\pi)). 
\]

\section{Basic Algorithmic Results}\label{sec:basic:algorithms}
In this section, we introduce special subsets of vertices for the models and basic algorithmic concepts 
which we need to prove the theoretical results in the following chapters.

\subsection{Graphs: SCCs, decremental algorithms and Graph Reachability}
For the following definitions we are given a graph $G = (V,E)$.
\para{Strongly connected components (SCCs), bottom SCCs and the condensation of a Graph}
A set of vertices $X \subseteq V$ forms a \emph{strongly connected subgraph} (SCS) if 
the induced subgraph $G[X]$ is strongly connected. An SCS is \emph{trivial} if it contains a single
vertex only and all other SCSs are \emph{non-trivial}.
A \emph{strongly connected component} (SCCS) is a set of vertices $C$
such that $G[C]$ is an SCS and $C$ is a maximal set (under set inclusion) in $V$ such that $G[C]$ is
an SCS\@. The SCC $C$ is a \emph{bottom SCC} if no vertex $v \in C$ has an edge to a
vertex in $V \setminus C$. 
The \emph{SCC decomposition} partitions the vertices of $V$ into the corresponding
SCCs and can be determined in $O(m)$ time~\cite{Tarjan72}.
The condensation of $G$, denoted by $\condense{G}$ is the graph where all
vertices in the same SCC in $G$ are contracted: We call the vertices of $\condense{G}$ 
\emph{nodes} to distinguish them from the vertices in $G$.

\para{Graph Reachability.}
Let $\GraphReach{S,G}$ be the set of vertices in $G$ that can reach a vertex of $S \subseteq V$. 
Depth-first search determines the set $\GraphReach{S,G}$ in linear time~\cite{Tarjan72}.

\para{Decremental Graph Algorithm.}
A \emph{decremental graph algorithm} is a data structure which supports the deletion of \emph{player-1 edges} 
while it maintains the solution to a graph problem. A decremental graph algorithm usually allows three kinds of
operations: (1) \emph{preprocessing}, which is computed when it receives the initial
input graph,
(2)~\emph{delete}, which deletes a player-1 edge and updates the data structure, and
(3)~\emph{query}, which computes the answer to the problem. 
The \emph{query time} is the time that the decremental graph algorithm
needs to compute the answer to the query.
The \emph{update time} of a decremental algorithm is the running time for a single \emph{delete} operation.
We sometimes refer to the delete operations as \emph{update operation}.
The running time of a decremental algorithm is characterized by the \emph{total update time}, i.e., 
the sum of the update time over the worst-case sequence of deletions. 
Sometimes a decremental algorithm is randomized and the provided running time guarantees 
hold for an \emph{oblivious adversary} who fixes the sequence of updates in advance. 
When we use a randomized decremental algorithm assuming an oblivious adversary as a subprocedure, 
the sequence of deleted edges must not depend on the random choices of the decremental algorithm.

\subsection{MDPs: MECs and Random Attractors}

For the following definitions let $P = (V,E, \langle V_1, V_R \rangle,\delta)$ be an MDP\@.

\begin{sloppypar}
\para{Maximal End-Components.}
An \emph{end-component} is a set of vertices $X \subseteq V$
such that (1) $P[X]$ strongly connected and (2)~all \emph{random vertices} have their outgoing
edges in $X$, that is, for all $v \in X \cap V_R$ and all $(v,u) \in E$ we have $u \in X$.
An end-component is \emph{trivial} if it has size one. All other end-components are \emph{non-trivial}.
An end-component maximal under set inclusion is a \emph{maximal end-component (MEC)}. 
MECs generalize SCCs in graphs and, in a MEC $X$, player-1 can \emph{almost-surely reach} (reach
with probability~1) all vertices $u
\in X$ from every vertex $v \in X$ because random vertices do not leave $X$. 
The MEC-decomposition of an MDP is the partition of $V$ into MECs and the set of vertices which do
not belong to any MEC\@. 
Every bottom SCC $C$ in $(V,E)$ is a MEC because no vertex (and thus no random vertex) has an outgoing edge. 
\end{sloppypar}

\begin{figure}
	\centering
	\includegraphics{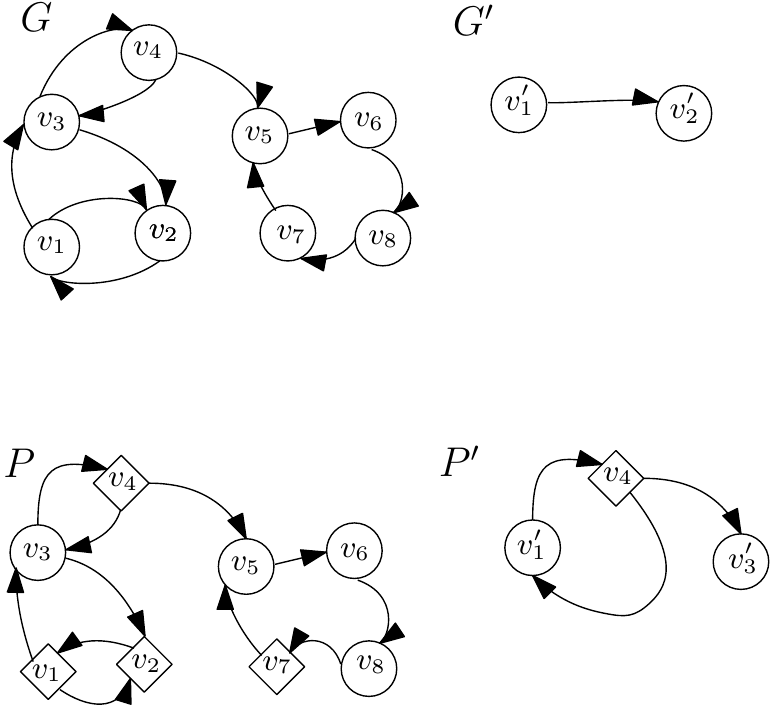}
	\caption{A graph $G$ and an MDP $P$ where we contract all SCCs and MECs respectively into player-1 vertices (removing self-loops) to obtain $G'$ and $P'$. The graph $G'$ is acyclic whereas the MDP $P'$ contains a cycle.}\label{icaps:fig:diff_SCCS_MECS}
\end{figure}

\begin{example}[Key difference of SCCs and MECs]\label{icaps:ex:diff_SCCS_MECS}
In the SCC decomposition each vertex belongs to exactly one SCC 
(which might be a trivial SCC just containing that vertex) but for the MEC decompositions 
a non-empty set of random vertices which do not belong to any MEC can exist (still each vertex belongs to at most one MEC). 
Consequently, if we contract each MECs into a player-1 vertex, 
the resulting MDP is not necessarily acyclic which is in contrast to 
the acyclic condensation graph that we obtain from contracting the SCCs.
In Figure~\ref{icaps:fig:diff_SCCS_MECS} we demonstrate this key difference: 
When we contract all SCCS of the graph $G$ into player-1 vertices yields the DAG $G'$.
Consider the MDP $P$ where the vertices $\{v_1,v_2,v_4,v_7\}$ of $G$ 
are random vertices and all edges remain unchanged. 
If we contract the MECS, i.e. $\{\{v_1,v_2,v_3\}, \{v_5,v_6,v_7,v_8\}\}$ 
into player-1 vertices $\{v'_1, v'_3\}$ we obtain the MDP $P'$ which has a cycle. 
Note that this is because $v_4$ does not belong to a MEC and is strongly connected with $v'_1$ or $v_3$ respectively in $P'$ and $P$.\@
\end{example}

\para{Random attractor.}
A \emph{random attractor} $\attr{R}{T}{P}$ of a vertex set
$T$ in $P$ is a vertex set. It includes all vertices (1) in $T$, (2) random vertices with an edge to the random
attractor and player-1 vertices	with all outgoing edges in the random attractor. 
We define the random attractor $A = \attr{R}{T}{P}$ inductively as
follows: $A_0 = T$ and $A_{i+1} = A_i \cup \{v \in V_R \mid \Out{v} \cap
		A_i \neq \emptyset \} \cup \{v \in V_{1} \mid \Out{v} \subseteq A_i 
	\}$ for all $i>0$.
	Due to~\cite{I81,B80} we can compute the random attractor
	$A = \attr{R}{T}{P}$ of a set $T$ in time $O(\sum_{v \in A} \In{v})$.

\para{Reachability in MDPs.}
Given a vertex set $T$, the set of vertices
from which $T$ can be \emph{reached almost-surely} (i.e., with probability~1) can be computed 
in $O(m)$ time given the MEC-decomposition of $P$~\cite[Theorem 4.1]{chatterjee2016separation}. 

\subsection{Game Graphs: Closed Sets, Attractors}
For the following definitions, let $\Gamma = (V,E,\langle V_1, V_2 \rangle)$ be a game graph.

\para{Closed Sets.}
A set $U \subseteq V$ of vertices is a \emph{closed set for player~1} if the following two conditions hold.
\begin{enumerate}
	\item For all vertices $u \in (U \cap V_1)$ we have $\Out{u} \subseteq U$, that is, all successors of
		player-1 vertices are again in $U$ and 
	\item For all $u \in (U \cap V_2)$ we have that $\Out{u} \cap V_2 \neq \emptyset$, that is, every
		player-2 vertex in $U$ has a successor in $U$.
\end{enumerate}
Player-2 closed sets are defined analogously by exchanging roles of player~1 and player~2. Every
closed set $U$ for player $x \in \{1,2\}$ induces a subgame graph which we denote $\Gamma \upharpoonright U$. 

The following proposition establishes the connection between closed sets, winning for safety,
reachability, and coB\"uchi objectives. The proof of the proposition is straightforward and can be found
in~\cite{CH14}.

\begin{proposition}[\cite{CH14}]\label{prop:closedsets}
	Consider a game graph $\Gamma$, and a closed set $U$ for player~1. Then, the following
	assertions hold:
	\begin{enumerate}
		\item Player~2 has a winning strategy for the objective $\safety{U}$ for all vertices in
			$U$, that is, player~2 can ensure that if the play starts in $U$, then the play never
			leaves the set $U$.
		\item For all $T \subseteq V \setminus U$, we have $\W{1}{\reach{T}} \cap U = \emptyset$,
			that is, for any set $T$ of vertices outside $U$, player~1 does not have a strategy from
			vertices in $U$ to ensure to reach $T$.
		\item If $U \cap B = \emptyset$ (i.e., there is no B\"uchi vertex in $U$), then every vertex
			in $U$ is winning for player~2 for the coB\"uchi objective.
	\end{enumerate}
\end{proposition}

\para{Attractors and Reachability in Game Graphs.}
Given a set of vertices $T \subseteq V$, the set of vertices from which player~$x$ can reach $T$
against all strategies of the other player, is the \emph{player-$x$ attractor} of $T$, i.e., 
$\attr{x}{T}{\Gamma} = \W{x}{\reach{T, \Gamma}}$. 
Formally, the player-$x$ attractor ($x \in \{1,2\}$) $\attr{x}{T}{\Gamma}$ of a given vertex set $T$ is
the limit of the sequence $A_0 = T; A_{i+1} = A_i \cup \{v \in V_x
	\mid \Out{v} \cap A_i \neq \emptyset \} \cup \{v \in V_{\bar{x}} \mid
	\Out{v} \subseteq A_i \}$ for all $i\geq 0$.  
The running time for computing an attractor $A = \attr{x}{T}{\Gamma}$ is $O(m)$~\cite{I81,B80}.

The following observation connects closed sets and attractors.
\begin{observation}[\cite{CH14}]\label{obs:attractorclosedset}
	For all game graphs $\Gamma$, all players $\ell \in \{1,2\}$, and all sets $U \subseteq V$ we
	have the following:
	The set $V \setminus \attr{\ell}{U}{\Gamma}$ is a closed set for player $\ell$, i.e., no player
	$\ell$ vertex in $V \setminus \attr{\ell}{U}{\Gamma}$ has an edge to $\attr{\ell}{U}{\Gamma}$
	and every vertex of the other player in $V \setminus \attr{\ell}{U}{\Gamma}$ has an edge in $V
	\setminus \attr{\ell}{U}{\Gamma}$.
\end{observation}

\section{Symbolic Model of Computation}
In the \emph{(set-based) symbolic model of computation}, we store the model 
with implicitly represented sets of vertices and edges.
An \emph{symbolic algorithm} accesses vertices and edges of the MDP not explicitly but with
\emph{(set-based) symbolic operations}. We characterize the resources in the
symbolic model of computation by the \emph{number of (set-based) symbolic operations} and
\emph{(set-based) symbolic space}. 

\para{Set-Based Symbolic Operations.}
A symbolic algorithm can use the same mathematical, memory, access, and logical
operations as a regular RAM algorithm, except for the access to the graph.

An input model with vertices $V$ and edges $E$ can be accessed only by the following types of operations:
\begin{enumerate}
	\item The algorithm can combine two sets of vertices or edges with \emph{basic set operations}:
		$\cup,\cap,\subseteq,\setminus$, $\times$ and $=$.
	\item To obtain the predecessors/successors of a vertex set $S$ with regard to\@ the edge set
		$E$, the algorithm uses the \emph{one-step operation}. The predecessor and successor
		operation of a vertex set $S$ over a specific edge set $E$ are:
		\begin{align*}
			\Pre{E}{S} &= \{ v \in V \mid \Out{v} \cap S \neq \emptyset \} \text{ and }\\
			\Post{E}{S} &= \{ v \in V \mid \In{v} \cap S \neq \emptyset \}
		\end{align*}

	\item For a vertex set $S$, the $\Pick{S}$ operation which returns an arbitrary vertex and the
		cardinality operation $|S|$ which returns cardinality of $S$.
	\item When the model is a game graph, an algorithm can use the 
		\emph{controllable predecessor operation} of a vertex set $S$ ($z$ and $\bar{z}$ denote the two players) 
		\begin{equation*}
			\CP_z(S) = \{ v \in V_z \mid \Out{v} \cap S \neq \emptyset \} \cup \{ v \in
				V_{\bar{z}} \mid \Out{v} \subseteq S \}.
		\end{equation*}
		We can express the set $\CP_z(S)$ using only $\Pre{E}{\cdot}$ and
		basic set operations.
		
\end{enumerate}

Sometimes we omit the subscript $E$ of the one-step operations when the edge set is clear from the context.

\begin{figure}
	\begin{center}
		\includegraphics[scale=1.7]{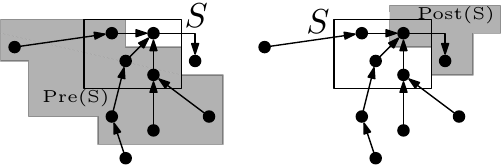}
	\end{center}
	\caption{Illustration of the one-step operations.}
\end{figure}

\para{Set-based Symbolic Space.}
The basic unit of space for a symbolic algorithm is a set~\cite{BrowneCJLM97,chatterjee2017symbolic}.  
For example, a BDD can represent a set symbolically~\cite{bryant1986graph,bry92,BurchCMDH90,ClarkeMCH96,Somenzi99,ClarkeBook,ClarkeGJLV03,GentiliniPP08,ChatterjeeHJS13}. 
It is hard to correlate the size of the set with the size of the BDD\@: Consider for example a model
whose state-space consists of valuations of $N$-boolean variables. The
set of all vertices is simply represented as a true BDD\@. Also, the set of all
vertices where the $k$th bit is false is represented by a BDD which depending
on the value of the $k$th bit chooses true or false. Again, a constant-size BDD represents this set\@. 
Therefore, even large sets of vertices can sometimes be represented as constant-size BDDs.
In general, the size of the smallest BDD representing a set is computationally hard to
determine and depends on the variable reordering~\cite{ClarkeBook}.
We represent each set, thus, as a unit data structure to obtain a clean theoretical model for the
algorithmic analysis. It follows that the \emph{symbolic space requirements} of a symbolic
algorithm, is the \emph{maximal number of sets} the algorithm stores
simultaneously.

\section{Conjectured Lower Bounds}\label{icaps:ss:lowerbounds}
Many results from classical complexity are based on standard complexity-theoretic
assumptions, e.g., P $\neq$ NP.\@ Likewise, we derive
polynomial lower bounds which are based on widely believed, conjectured lower
bounds on well-studied algorithmic problems. The lower bounds we
derive depend on the popular conjectures below: 

We first consider the conjectures on Boolean Matrix
Multiplication~\cite[Theorem 6.1]{WW18} and
\emph{triangle detection} in graphs~\cite[Conjecture 2]{abboud2014popular}. 
A \emph{triangle} in a graph is a triple $x,y,z$ of vertices such that
$(x,y),(y,z),(z,x) \in E$. 
In \emph{triangle detection} the input is a graph and the question is 
if a triangle exists in the graph. 
We remove, with linear-time preprocessing, all self-loops in instances of \emph{triangle detection}. 

\begin{remark}[Combinatorial algorithm]\label{icaps:remark:combinatorial}
    The notion of combinatorial algorithm is often used in the field of
    fine-grained complexity community~\cite{AbboudBW18,LincolnWW18, BringmannFK19},
    despite the lack of a formal definition.
    Intuitively, combinatorial algorithms 
	do not use fast matrix multiplication~\cite{williams2012multiplying,le2014powers}. 
	While non-combinatorial algorithms have
	the matrix multiplication exponent $\omega$ in the running time.
	To the best of our knowledge, 
	all algorithms for deciding (almost-sure) winning 
	conditions in game graphs and MDPs are combinatorial.
	Therefore, lower bounds for combinatorial algorithms are of 
	particular interest in our setting. 
	For more details on the notion of a combinatorial algorithm, we direct the reader to~\cite{BallardDHS12,henzinger2015hardness}.
\end{remark}

\begin{conjecture}[Comb. Boolean Matrix Multiplication Conjecture
	(BMM)]\label{icaps:conj:bmm}
	A $O(n^{3-\epsilon})$ time combinatorial algorithm for computing the
	boolean product of two $n \times n$ matrices for any $\epsilon > 0$ does not exist.
\end{conjecture}
\begin{sloppypar}
\begin{conjecture}[Strong Triangle Conjecture (STC)]\label{icaps:conj:stc}
	Neither a $O(\min \{ n^{\omega - \epsilon}, m^{2\omega/(\omega+1)-\epsilon} \})$
	expected time nor a $O(n^{3-\epsilon})$ ($\epsilon > 0$) time combinatorial algorithm
	that can detect whether a graph contains a triangle exist.
	($\omega < 2.373$ is the matrix multiplication exponent.)
\end{conjecture}
\citeauthor{WW18}~\cite[Theorem 6.1]{WW18} showed that the BMM is equivalent to the
combinatorial part of STC.\@ Also, if we do not restrict ourselves to
combinatorial algorithms, STC, still gives a super-linear lower bound.
\end{sloppypar}

We also consider the Strong Exponential Time Hypothesis
(SETH) used in~\cite[Conjecture 1]{abboud2014popular} and introduced by~\cite{impagliazzo1999complexity,impagliazzo1998problems} for the satisfiability problem of propositional logic and the Orthogonal Vector Conjecture. 

\begin{conjecture}[Strong Exponential Time Hypothesis (SETH)] For
	$\epsilon>0$ there is a $k$ such that $k$-CNF-SAT on $n$ variables and $m$
	clauses cannot be solved in $O(2^{(1-\epsilon)n} \poly(m))$ time.
\end{conjecture}

\para{The Orthogonal Vectors Problem (OV).} Given two sets $S_1,S_2$ of
$d$-bit vectors with $|S_1| =|S_2| = N$ and $d = \omega(\log N)$, are there $u \in
S_1$ and $v \in S_2$ such that $\sum_{i=1}^d u_i \cdot v_i = 0$? 

\begin{conjecture}[Orthogonal Vectors Conjecture (OVC)]\label{icaps:conj:ovc}
	A $O(N^{2-\epsilon})$ time algorithm for the Orthogonal Vectors
	Problem for any $\epsilon > 0$ does not exist.
\end{conjecture}

The SETH implies the OVC~\cite[Theorem 5]{williams2005satisfaction}.
An explicit reduction is given in the survey article~\cite[Theorem 3.1]{VW2018survey}.
Whenever a problem is provably hard assuming OVC, is also hard when assuming
SETH.\@ 

\begin{remark}
	The conjectures make sure that no \emph{polynomial} improvements over the best-known running
	times are possible but do not exclude improvements by sub-polynomial factors such as
	poly-logarithmic factors or factors like $2^{\sqrt{\log n}}$.
\end{remark}

	\chapter{Faster Algorithms for Mean-Payoff Parity Games}\label{cha:mfcs}
	In this chapter, we consider computing the winning region for threshold mean-payoff parity objectives,
and the value function for mean-payoff parity objectives.

\section{Introduction}
\para{Graph games in Reactive Synthesis.}
There has been a long history of using graph games for modeling and 
synthesizing reactive processes~\cite{BuchiL69,PnueliR89,RamadgeW87}:
a reactive system and its environment represent the two players, whose states 
and transitions are specified by the vertices and edges of a game graph.
Consequently, graph games provide the theoretical foundation
for modeling and synthesizing reactive processes.

\para{Qualitative and quantitative objectives.}
For reactive systems, the objective is given as a set of desired paths 
(such as $\omega$-regular specifications), or 
as a quantitative optimization objective with a payoff function on the paths.
The class of $\omega$-regular specifications provides a robust framework
to express all commonly used specifications for reactive systems in 
verification and synthesis. 
Parity objectives are a canonical way to express $\omega$-regular objectives~\cite{Thomas97},
where an integer priority is assigned to every vertex, and a path satisfies the 
parity objective for player~1 if the minimum priority visited infinitely often is even.
One of the classical and most well-studied quantitative objectives is the mean-payoff 
objective, where a reward is associated with every edge, and the payoff of a path is 
the long-run average of the rewards of the path.

\para{Mean-payoff parity objectives.}
Traditionally the verification and the synthesis problems were considered
with qualitative objectives.
However, recently combinations of qualitative and quantitative objectives have 
received a lot of attention.
Qualitative objectives such as $\omega$-regular objectives specify
the functional requirements of reactive systems, whereas the quantitative 
objectives specify resource consumption requirements (such as for embedded 
systems or power-limited systems).
Combining quantitative and qualitative objectives is crucial in the 
design of reactive systems with both resource constraints and functional 
requirements~\cite{CAHS03,CJH05,BFLMS08,BCHJ09}.
For example, mean-payoff parity objectives are relevant in the synthesis of 
optimal performance lock-synchronization for programs~\cite{CCHRS11}, 
where one player is the synchronizer, the opponent is 
the environment; the performance criteria are specified as mean-payoff objective; 
and the functional requirement (e.g., data-race freedom or liveness) 
as an $\omega$-regular objective.
Mean-payoff parity objectives have been used in several other applications,
e.g., defining permissivity for parity games~\cite{BMOU11} and 
for the robustness in synthesis~\cite{BCGHHJKK14}.

\para{Threshold and value problems.}
For graph games with mean-payoff and parity objectives, there are two variants 
of the problem.
First, the \emph{threshold} problem, where a threshold $\V$ is given for 
the mean-payoff objective and player~1 must ensure the parity objective and 
that the mean-payoff is at least $\V$.
Second, the \emph{value} problem, where player~1 maximizes the mean-payoff 
value while ensuring the parity objective.
In the sequel of this section, we will refer to graph games with mean-payoff and
parity objectives as mean-payoff parity games.

\para{Previous results.}
Mean-payoff parity games were first studied in~\cite{CJH05},
and algorithms for the value problem were presented.
It was shown in~\cite{CD10a} that the decision problem for mean-payoff parity games
lies in NP $\cap$ coNP (similar to the status of mean-payoff games and parity games).
For game graphs with $n$ vertices, $m$ edges, parity objectives with $d$ priorities, 
and maximal absolute reward value $W$ for the mean-payoff objective, the previous 
known algorithmic bounds for mean-payoff parity games are as follows:
For the threshold problem, the results of~\cite{CD10a} give an 
$O(n^{d+4}m dW)$-time algorithm.
This algorithmic bound was improved in~\cite{BMOU11} where an $O(n^{d+2} m W)$-time 
algorithm was presented for the value problem. 
The result of~\cite{BMOU11} does not explicitly present any other better bound for the 
threshold problem. However, the recursive algorithm of~\cite{BMOU11} uses value mean-payoff games
as a sub-routine, and replacing value mean-payoff games with threshold mean-payoff games gives 
an $O(n)$-factor saving, and yields an $O(n^{d+1}mW)$-time algorithm for the 
threshold problem for mean-payoff parity games.

\para{Contributions.} 
In this chapter, our main contributions are faster algorithms to solve mean-payoff parity games.
Previous and our results are summarized in Table~\ref{mfcs:tab:complexity}.

\begin{enumerate}
	\item \emph{Threshold problem.} We present an $O(n^{d-1}mW)$-time algorithm
for the threshold problem for mean-payoff parity games, improving the previous 
$O(n^{d+1}mW)$ bound.
The important special case of parity objectives with two priorities correspond to B\"uchi 
and coB\"uchi objectives. 
Our bound for mean-payoff B\"uchi games and mean-payoff coB\"uchi games is 
$O(nmW)$, which matches the best-known bound to solve the threshold
problem for mean-payoff objectives~\cite{brim2011}, and improves the previous known $O(n^3mW)$
bound~\cite{BMOU11}.

\item \emph{Value problem.} We present an $O(n^{d}mW\log(nW))$-time 
algorithm for the value problem for mean-payoff parity games, improving the previous 
$O(n^{d+2}mW)$ bound.
Our bound for mean-payoff B\"uchi games and mean-payoff coB\"uchi games is 
$O(n^2mW\log(nW))$, which matches the bound of~\cite{brim2011} 
to solve the value problem for mean-payoff objectives, and improves the previous 
known $O(n^4mW)$ bound.
\end{enumerate}

\para{Technical contributions.}
Our main technical contributions are as follows:
\begin{enumerate}
\item First, for the threshold problem, we present a decremental algorithm for 
mean-payoff games that supports a sequence of vertex-set deletions along with 
their player-2 reachability set. 
We show that the total running time is $O(n  m  W)$, which matches the 
best-known bound for the static algorithm to solve mean-payoff games.
We show that using our decremental algorithm we can solve the threshold problem for 
mean-payoff B\"uchi games in time $O(n  m  W)$.

\item Second, for mean-payoff coB\"uchi games, the decremental approach does not work.
We present a new static algorithm for threshold mean-payoff games that identifies subsets $X$ 
of the winning set for player~1, where the time complexity is $O(|X|  m  W)$, i.e., 
it replaces $n$ with the size of the set identified. We show that with our new static algorithm 
we can solve the threshold problem for mean-payoff coB\"uchi games in time $O(n  m  W)$. 

\item Finally, we show for all mean-payoff parity objectives, given an algorithm for the 
threshold problem, that the value problem can be solved in time $n  \log (n W)$
times the complexity of the threshold problem. 
\end{enumerate}

\para{Related works.}
The problem of graph games with mean-payoff parity objectives was first 
studied in~\cite{CJH05}. The NP $\cap$ coNP complexity bound was established 
in~\cite{CD10a}, and an improved algorithm for the problem was given 
in~\cite{BMOU11}.
The mean-payoff parity objectives have also been considered in other
stochastic setting such as Markov decision processes~\cite{CD11,CD11b} 
and stochastic games~\cite{CDGO14}.
The algorithmic approaches for stochastic games build on the results for
non-stochastic games. In this chapter, we present faster algorithms for 
mean-payoff parity games.

\begin{table}[t]
	\small
	\centering
	\begin{tabular}{l c l l c l l}
		\toprule
		& &\multicolumn{2}{c}{threshold problem} & & \multicolumn{2}{c}{value problem}\\
		\cmidrule{3-4}  \cmidrule{6-7}
		&&	Previous & New & & Previous & New \\
		\midrule
		MP-B\"uchi   &&  $O(n^3  m  W)$ & $O(n  m  W)$ & & $O(n^4  m  W)$ & $O(n^2  m  W  \log(nW))$\\
		MP-coB\"uchi &&   $O(n^3  m  W)$ & $O(n  m  W)$ & & $O(n^4  m  W)$ & $O(n^2  m  W  \log(nW))$\\
		MP-parity && $O(n^{d+1}  m  W)$ & $O(n^{d-1}  m  W)$& & $O(n^{d+2}  m  W)$ & $O(n^{d}  m  W  \log(nW))$\\
		\bottomrule

	\end{tabular}
	\caption{Algorithmic bounds for mean-payoff (MP) and parity objectives:
		In a game graph $\Gamma$ with weight function $w$, $n$ denotes the number of vertices, $m$ denotes the number of edges, $d$ denotes the number of 
		priorities of the parity function $p$, and $W$ is the maximum absolute value of the weight function $w$.
	}\label{mfcs:tab:complexity}

\end{table}

\section{Decremental Algorithm for Threshold Mean-Payoff Games}\label{mfcs:sec:decr}
In this section, we present a decremental algorithm for threshold mean-payoff games
that supports deleting a sequence of sets of vertices along with their player-2
attractors.
The overall running time of the algorithm is $O(nmW)$.

\para{Key idea.} 
A static algorithm based on the notion of progress measure for mean-payoff games 
was presented in~\cite{brim2011}.
We show that the progress measure is monotonic with regard to the deletion of vertices and their
player-2 attractors.
We use an amortized analysis to obtain the running time of our algorithm.

\para{Mean-payoff progress measure.} 
Let $(\Gamma,w)$ be a mean-payoff game with threshold $\V$.
Progress measure is a function $f$ which maps every vertex in $\Gamma$ to an element of the set 
$C_\Gamma = \{i \in \N \mid i \leq nW \} \cup \{\top\}$, i.e., $f: V \mapsto C_\Gamma$. 
Let $(\preceq,C_\Gamma)$ be a total order, where $x \preceq y$ for $x,y \in C_\Gamma$
holds iff $x \leq y \leq nW$ or $y = \top$. We define the operation $\ominus:
C_\Gamma \times \Z \mapsto C_\Gamma$ for all $a \in C_\Gamma$ and $b \in \Z$ as
follows:
\[ 
a \ominus b = \begin{cases} 
\max(0,a-b) & \text{if $a \neq \top$ and $a-b \leq nW$,}\\
\top        & \text{otherwise.}
\end{cases}
\]
A player-1 vertex $v$ is \emph{consistent} if $f(v) \succeq f(v') \ominus
w(v,v')$ for \emph{any} $v' \in Out(v)$. A player-2 vertex $v$ is \emph{consistent} 
if $f(v) \succeq f(v') \ominus w(v,v')$ for \emph{all} $v' \in Out(v)$.
Let $v \in V$, we define $\lift(f,v) : [V \mapsto C_\Gamma \times V] \mapsto [V \mapsto C_\Gamma]$ as 
$\lift(f,v) = g$ where: 
\[ g(u) = 
\begin{cases}
f(u) & \text{ if $u \neq v$},\\
\min\{f(v') \ominus w(v,v') \mid (v,v') \in E \} & \text{ if $u = v$ and $v \in V_1$},\\
\max\{f(v') \ominus w(v,v') \mid (v,v') \in E \} & \text{ if $u = v$ and $v \in V_2$}.
\end{cases}
\]

\smallskip\noindent\emph{Static Algorithm.}
The static algorithm by Brim et al.~\cite{brim2011} is an iterative algorithm that 
maintains and returns a progress measure $f$ and a list $L$ of vertices that are not consistent. 
The initial progress measure of every vertex is set to zero.
Also, all the weights of all edges are subtracted by the value $\V$, i.e., $w(e) \gets w(e) - \V$ for all edges $e$ in $E$.
The list $L$ is initialized with the vertices which are not consistent considering the 
initial progress measure. 
Then the following steps are executed in a while-loop:
\begin{enumerate}
\item If $L$ is empty, return $f$.
\item Take out a vertex $v$ of $L$.
\item Perform the $\lift$-operation on the vertex, i.e., $f \leftarrow \lift(f,v)$. 
\item If a vertex $v'$ in $\In{v}$ is not consistent, put $v'$ into $L$.
\end{enumerate}
If every vertex is consistent, i.e., the list $L$ is empty, Brim et al.\ show that
the winning region of player~1 is the set of vertices which are not set
to $\top$ in $f$, i.e., $\W{1}{\mpayoff{\V,w,\Gamma}} = \{ v \in V \mid f(v) \neq \top \}$.

\smallskip\noindent\emph{Decremental input/output.} 
Let $(\Gamma,w)$ be a mean-payoff game with threshold $\V$. 
The \emph{input} to the decremental algorithm is a sequence of sets $A_1, A_2, \ldots,A_k$,
such that each $A_i$ is a player-2 attractor of a set $X_i$ in the game 
$\Gamma_i=\Gamma \restr (V \setminus \bigcup_{j<i} A_j)$.
The \emph{output requirement} is the player-1 winning set after the deletion of
$\bigcup_{j<i}A_j$ for $i=1,\dots,k$,
i.e., the output requirement is the sequence $Z_1,Z_2, \ldots,Z_k$, where 
$Z_i=\W{1}{\mpayoff{\V,w,\Gamma_i}}$ in $\Gamma_i=\Gamma \restr (V \setminus \bigcup_{j<i} A_j)$.
In other words, we repeatedly delete a vertex set $X_i$ along with its player-2 attractor $A_i$ 
from the current game graph $\Gamma_i$, and require the winning set of player~1 
for the mean-payoff objective as an output after each deletion.

\para{Decremental algorithm.} 
We maintain a progress measure $f_i$, $1 \leq i \leq k$, during the whole sequence of deletions.
The initial progress measure $f_1$ for the mean-payoff game $(\Gamma,w)$ with threshold 
mean-payoff objective $\mpayoff{\V,w,\Gamma}$ is computed with the static algorithm of Brim et
al~\cite{brim2011}.

For all edges $e$ in $E$, we set $w(e) \gets w(e) - \V$.
In iteration $i$ with 
input $A_i$, in the game $\Gamma_i$ with its corresponding vertex set $V_i$ the following steps are executed:
\begin{enumerate}
    \item If a vertex in the set $\{v \in V_i\setminus A_i \mid \exists v': \  v' \in \Out{v} \text{
				and } v' \in A_i \}$
    is not consistent in $f_i$ without the set $A_i$, put it in the list $L_i$.
    \item Delete the set $A_i$ from $\Gamma_i$ to receive $\Gamma_{i+1}$ (and thus $V_{i+1}$).
    \item Execute the while-loop with steps (1) --- (4) of the above described iterative 
	algorithm by Brim et al.~\cite{brim2011} initialized with
	$\Gamma_{i+1}$, $L_i$ and $f_{i}$ restricted to the vertices in $V_{i+1}$ to obtain $f_{i+1}$.
    \item Output the set $\{ v \in V_{i+1} \mid f(v) \neq \top \}$ from $f_{i+1}$.
\end{enumerate}

\smallskip\noindent\emph{Correctness.}
Let $(\Gamma, w)$ be a mean-payoff game, $\mpayoff{\V,w\Gamma}$ be a threshold objective 
and $A_1, A_2 ,\dots, A_k$ a sequence of sets, 
such that each $A_i$ is a player-2 attractor in the game 
$\Gamma_i = \Gamma \restr (V \setminus \bigcup_{j<i} A_j)$. To show the correctness of the decremental
algorithm we show that the condition that the list $L$ contains all vertices which are not consistent
is an invariant of the decremental algorithm at line~3. This property was proved for the static algorithm in~\cite{brim2011}.
\begin{lemma}
The condition that $L_i$ contains all vertices which are not consistent with the
progress measure $f_i$ restricted to $V_{i+1}$ in $\Gamma_{i+1}$ is an invariant of the
static algorithm called in step~3 of the decremental algorithm for $1 \leq i
\leq k-1$. 
\end{lemma}
\begin{proof}
The fact that the static algorithm correctly returns a progress 
measure with only 
consistent vertices when the invariant holds was shown in~\cite{brim2011}. It
was also shown in~\cite{brim2011} that the invariant is maintained in the loop.
It remains to show that the condition holds when we call the static algorithm at
step~3.
For the base case, let $i = 1$. In the initial progress measure $f_1$ and the initial game graph $\Gamma_1$, 
every vertex is consistent. By the definition of a player-2 attractor, deleting the set $A_1$ potentially
removes edges $(v,v')$ where $v$ is a player-1 vertex in $V\setminus A_1$
and $v'$ is in $A_1$. (Note that $v$ cannot be a player-2 vertex.) 
All of the vertices not consistent anymore are added to $L_i$ in step~1 of the decremental algorithm. For the inductive step let $i = j$. By the induction hypothesis,
all vertices which were not consistent with the progress measures $f_{h-1}$ restricted to $V_{h}$ for 
$2\leq h \leq j$ were added to the corresponding lists. Thus by the correctness of the static algorithm, 
it correctly computes 
the new progress measure $f_{h}$ for the game graph $\Gamma_{h}$ where every vertex is consistent. 
Thus also every vertex in the progress measure $f_j$ restricted to $V_{j}$ is consistent. 
Again the player-2 attractor is removed and vertices that are not consistent with progress measure $f_j$ 
restricted to $V_{j+1}$ are put into $L_j$ by step~1 of the algorithm. 
\end{proof}
Thus we proved that the static algorithm always correctly updates to the new
progress measure in each iteration. 
The winning region of player-1 is obtained by the returned progress measure (step~4).
The decremental algorithm thus correctly computes the sequence 
$Z_1,Z_2, \dots Z_k$, where $Z_i = W_1(\mpayoff{\V,w,\Gamma_i})$.

\smallskip\noindent\emph{Running Time.}
The calculation of the initial progress measure for the mean-payoff game $\Gamma$ with
threshold $\V$ is in time $O(nmW)$.
The vertices which are not consistent anymore after the deletion of $A_i$ 
can be found in time $O(m)$ (step~1). As at most $n$ such sets $A_i$ exist, the running time is
$O(mn)$. In step~3 the static algorithm is executed with
the current progress measure $f_i$:
Every time a vertex $v$ is picked from 
the list $L_i$ it costs $O(|\Out{v} + \In{v}|)$ time to
use $\lift$ on it and to look for vertices in $\In{v}$ which are not consistent anymore
(steps~1--3 in the static algorithm).
We charge the cost for each vertex to its incident edges. Note that deleting a set of vertices and 
their corresponding player-2 attractor 
will only potentially \emph{increase} the progress measure of some player-1 vertices.
As we can increase the progress measure of every vertex only $nW$ times 
before it is set to $\top$ where it is always consistent, 
we get the desired time bound of $O(mnW)$.

\noindent Thus our decremental algorithm for threshold mean-payoff games works as desired and we obtain 
the following result:

\begin{theorem}~\label{mfcs:thm:decr}
	Given a mean-payoff game $(\Gamma,w)$, a threshold mean-payoff objective $\phi$ and a
sequence of sets $A_1, A_2, \dots, A_k$ such that each $A_i$ is a player-2 
attractor of a set $X_i$ in the game $\Gamma_i=\Gamma \restr (V \setminus \bigcup_{j<i} A_j)$, 
the sequence $Z_1,Z_2,\dots, Z_k$, where $Z_i = W_1(\phi)$ in $\Gamma_i$ can be computed in
$O(nmW)$ time.
\end{theorem}

\begin{remark}\label{mfcs:rem:nodeladd}
Note that the running time analysis of our decremental algorithm crucially
depends on the monotonicity property of the progress measure.
If edges are both added and deleted, then the monotonicity property does not hold.
Hence obtaining a fully dynamic algorithm that supports both addition/deletion of 
vertices/edges with running time $O(nmW)$ is an interesting open problem.
However, we will show that for solving mean-payoff parity games, the decremental 
algorithm plays a crucial part.
\end{remark}

\section{Threshold Mean-Payoff Parity Games}
In this section, we present algorithms for threshold mean-payoff parity games.
Our most interesting contributions are for the base case of 
mean-payoff Büchi objectives and mean-payoff coBüchi objectives, and the general case follows
a standard recursive argument. 

\subsection{Threshold Mean-Payoff Büchi Games}
In this section, we consider threshold mean-payoff B\"uchi games.

\para{Algorithm for threshold mean-payoff B\"uchi games.}
The basic algorithm is an iterative algorithm that deletes player-2 attractors.
The algorithm proceeds in iterations. In iteration $i$, let $D_i$ be the set of 
vertices already deleted. Consider the subgame $\Gamma_i=\Gamma \restr (V\setminus D_i)$.
Then the following steps are executed:
\begin{enumerate}
\item Let $V^i= V \setminus D_i$ and $B_i$ denote the set of B\"uchi vertices (or vertices with 
priority~0) in $\Gamma_i$.
Compute $Y_i=\attr{1}{B_i}{\Gamma_i}$ the player-1 attractor to $B_i$ in $\Gamma_i$.

\item Let $X_i=V^{i} \setminus Y_i$. If $X_i$ is non-empty, remove $A_i=\attr{2}{X_i}{\Gamma_i}$
	from $\Gamma_i$,
	and proceed to the next iteration in Step~1 otherwise go to Step~3.

\item Else $V^{i}=Y_i$. Compute $U_i=W_1(\mpayoff{\V,w,\Gamma_i})$, i.e., the winning 
region for the threshold mean-payoff objective in $\Gamma_i$. 
Let $X_i=V^{i} \setminus U_i$. 
If $X_i$ is non-empty, remove $A_i=\attr{2}{X_i}{\Gamma_i}$  from the game graph $\Gamma_i$,
and proceed to the next iteration.
If $X_i$ is empty, then the algorithm stops and all the remaining vertices are winning for player~1
for the threshold mean-payoff B\"uchi objective. 
\end{enumerate}

\para{Correctness.}
Since the correctness argument has been used before~\cite{CJH05}, we only present a brief sketch:
The basic correctness argument shows for $i > 0$ that all vertices removed from $\Gamma_i$ 
do not belong to the winning set for player~1. 
In the end, for the remaining vertices, player~1 can ensure to reach the B\"uchi vertices,
and ensures the threshold mean-payoff objectives.
A strategy that plays for the threshold mean-payoff objectives longer and longer, and in between
visits the B\"uchi vertices, ensures that the threshold mean-payoff B\"uchi objective is satisfied.

\para{Running time analysis.}
We observe that the total running time to compute all attractors is at most $O(nm)$,
since the algorithm runs for $O(n)$ iterations and each attractor computation is linear time.
In step~3, the algorithm needs to compute the winning region for threshold mean-payoff 
objective.
The algorithm always removes a set $X_i$ and its player-2 attractor $A_i$, and requires
the winning set for player~1. 
Thus we can use the decremental algorithm from Section~\ref{mfcs:sec:decr}, which precisely supports
these operations.  
Hence using Theorem~\ref{mfcs:thm:decr} in the algorithm for threshold mean-payoff B\"uchi games, 
we obtain the following result.

\begin{theorem}\label{mfcs:thm:buchi}
	Given a mean-payoff game graph $(\Gamma,w)$ and a threshold mean-payoff B\"uchi objective
	$\mpayoff{\V,w,\Gamma}$, the winning set $\W{1}{\mpayoff{\V,w,\Gamma}}$ can be computed in $O(mnW)$ time.
\end{theorem}

\subsection{Threshold Mean-Payoff coBüchi Games}
In this section, we will present an $O(nmW)$-time algorithm for threshold
mean-payoff coB\"uchi games. 
We start with the description of the basic algorithm for threshold mean-payoff 
coB\"uchi games.

\para{Algorithm for threshold mean-payoff coB\"uchi games.}
The basic algorithm is an iterative algorithm that deletes player-1 attractors.
The algorithm proceeds in iteration. In iteration $i$, let $D_i$ be the set of 
vertices already deleted. Consider the subgame $\Gamma_i=\Gamma \restr (V\setminus D_i)$.
Then the following steps are executed:
\begin{enumerate}
\item Let $V^i= V \setminus D_i$ and $C_i$ denote the set of non coB\"uchi vertices (the set of
	vertices player~1 must avoid to visit infinitely often, i.e., priority-1
	vertices) in $\Gamma_i$.
Compute $Y_i=\attr{2}{C_i}{\Gamma_i}$ the player-2 attractor to $C_i$ in $\Gamma_i$.
\begin{sloppypar}
\item Let $X_i=V^{i} \setminus Y_i$. Consider the subgame $\widehat{\Gamma}_i=\Gamma_i \restr X_i$
	and compute the winning region of the threshhold mean-payoff objective
	$Z_i=\W{1}{\mpayoff{\V,w,\widehat{\Gamma}}}$ for player~1.
\end{sloppypar}

\item If $Z_i$ is non-empty, remove $\attr{1}{Z_i}{\Gamma_i}$ from $\Gamma_i$, and proceed to the next
iteration. Else if $Z_i$ is empty, then all remaining vertices are winning for player~2.
\end{enumerate}

\para{Correctness argument.}
Consider the subgame $\Gamma_i$. In each subgame $\widehat{\Gamma}_i$ of $\Gamma_i$ 
all edges of player~2 are intact since it is obtained after removing a player-2 attractor $Y_i$.
Moreover, there is no priority-1 vertex in $\widehat{\Gamma}_i$. 
Hence, ensuring the threshold mean-payoff objective in $\widehat{\Gamma}_i$ for player~1 implies 
satisfying the threshold mean-payoff coB\"uchi objective. The set $Z_i$ and its player-1 
attractor belongs to the winning set of player~1 and can be removed.
Thus, all vertices removed are part of the winning region for player~1.
Upon termination, in $\widehat{\Gamma}_i$, player~1 cannot satisfy the threshold mean-payoff 
condition from any vertex.
Consider a player-2 strategy, where in $\widehat{\Gamma}_i$ player~2 falsifies the threshold 
mean-payoff condition, and in $Y_i$ plays an attractor strategy to reach $C_i$ (the non coB\"uchi
vertices, i.e., priority-1 vertices).
Given such a strategy, either 
(a)~$Y_i$ is visited infinitely often, and then the coB\"uchi objective is violated; or
(b)~from some point on the play stays in $\widehat{\Gamma}_i$ forever, and then the threshold
mean-payoff objective is violated.
This shows the correctness of the algorithm. 

However, the running time of this algorithm is not $O(nmW)$. 
We now present the key ideas to obtain an $O(nmW)$-time algorithm.

\para{First intuition.}
Our first intuition is as follows. 
In step~2 of the above algorithm, instead of obtaining the whole winning region
$\W{1}{\mpayoff{\V,w,\widehat{\Gamma}}}$ it suffices to identify a subset $X_i \subseteq Z_i$ of the winning region
(if it is non-empty) and remove its player-1 attractor.
We call this the modified algorithm for threshold mean-payoff coB\"uchi games.
We first describe why we cannot use the decremental approach in the following remark.

\begin{remark}
Consider the subgames for which the threshold mean-payoff objective must be solved.
Consider Figure~\ref{mfcs:fig:nodecr}. 
The first player-2 attractor removal induces subgame $\widehat{\Gamma}_1$.
After identifying a winning region $X_1$ of $\widehat{\Gamma}_1$ we remove its player-1 attractor $A_1$.
After removal of $A_1$, we consider the second player-2 attractor to the
priority-1 vertices.
The removal of this attractor induces $\widehat{\Gamma}_2$.
We observe comparing $\widehat{\Gamma}_1$ and $\widehat{\Gamma}_2$ that certain 
vertices are removed, whereas other vertices are added.
Thus the subgames to be solved for threshold mean-payoff objectives do not satisfy
the condition of decremental or incremental algorithms (see Remark~\ref{mfcs:rem:nodeladd}).

\begin{figure}[ht]
\centering
\includegraphics{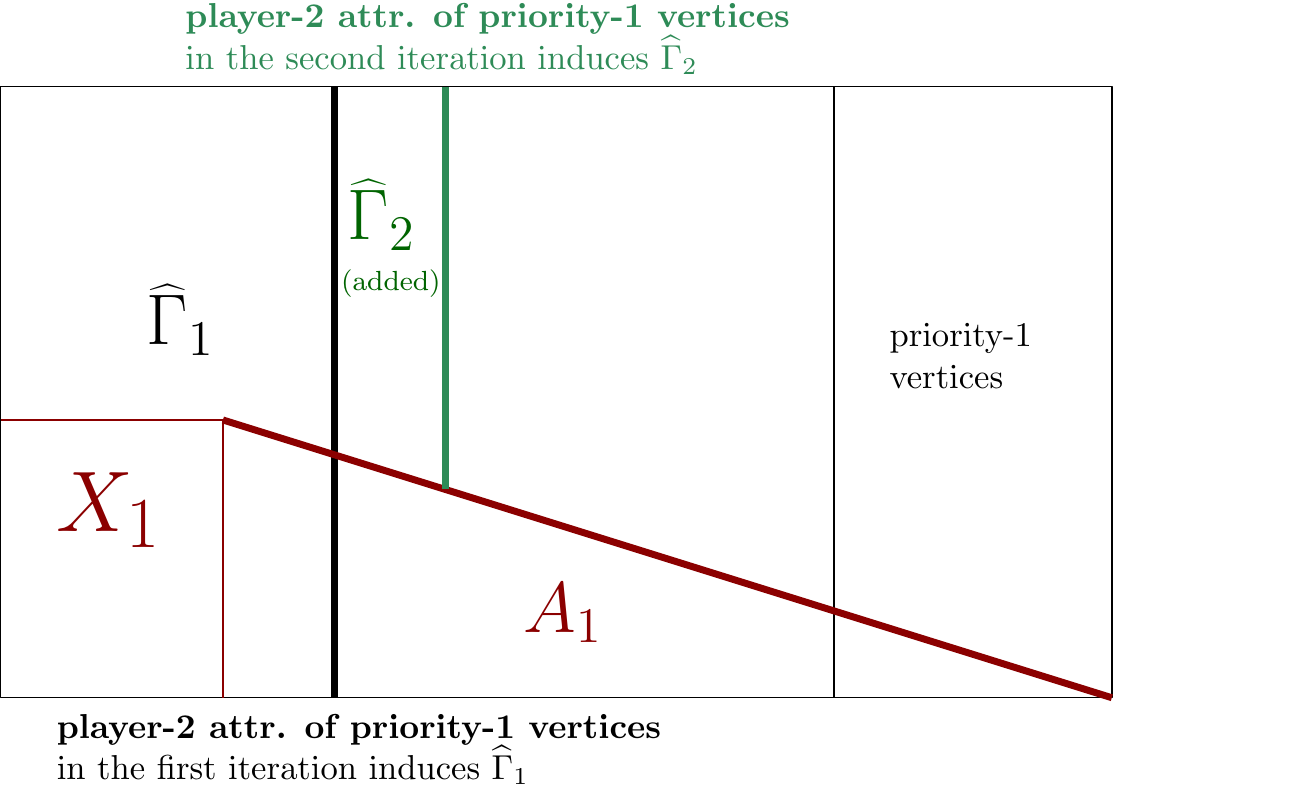}
%
%
%
\caption{Pictorial illustration of threshold mean-payoff coBüchi games. The
    subgames $\widehat{\Gamma}_1$ and $\widehat{\Gamma}_2$ are shown. We observe that $\widehat{\Gamma}_2$ is
        obtained both by addition and deletion of game parts to $\widehat{\Gamma}_1$.}\label{mfcs:fig:nodecr}
\end{figure}
\end{remark}

\para{Second intuition.} 
While we cannot use the decremental algorithm, we can solve the problem in 
$O(n m W)$ time, if we have a modified static algorithm for threshold mean-payoff 
games, with the following property:
(a)~it identifies a subset of the winning region $X$ for player~1, if the winning region
is non-empty, in time $O(|X| m W)$;
(b)~if the winning region is empty, it returns the empty set, and then it takes
time $O(n m W)$.
With such an algorithm we analyze the running time of the above modified algorithm for threshold 
mean-payoff coB\"uchi games.
The total time required for all attractor computations is again $O(n m)$. 
Otherwise, we use the modified static algorithm to remove vertices of player-1 and to remove a set of size 
$|X|$ we take $O(|X|  m W)$ time, and thus we can charge each vertex
$O(m W)$ time. 
Hence the total time required is $O(n  m W)$.
In the rest of the section, we present this modified static algorithm for threshold mean-payoff games.

\para{Problem Statement.}
\begin{framed}
\begin{tabular}{ l l}
	\textbf{Input:}    & Mean-payoff game $(\Gamma,w)$ with threshold $\V$.\\
	\textbf{Question:} & If $\W{1}{\mpayoff{\V}}$ is non-empty, return a nonempty set\\
	&$X \subseteq \W{1}{\mpayoff{\V}}$ in time $O(|X|mW)$,\\
&else return $\emptyset$ in time $O(n  m W)$.
\end{tabular}
\end{framed}

\para{Modified static algorithm for threshold mean-payoff games.}
The basic algorithm for threshold mean-payoff games computes a progress measure, with a 
defined top element value $\top$. 
If the progress measure has the value $\top$ for a vertex, then the vertex is 
declared as winning for player~2. 
With value $\top=n\cdot W$, the correct winning region for both players can be 
identified. 
Moreover, for a given value $\alpha$ for $\top$, the progress measure algorithm 
requires $O(\alpha \cdot m)$ time.
Our modified static algorithm is based on the following idea:
\begin{enumerate}
\item Consider a value $\alpha \leq n \cdot W$ for the top element. 
With this reduced value for the top element, if a winning region is identified
for player~1, then it is a subset of the whole winning region for player~1.
\item We will iteratively double the value for the top element. 
\end{enumerate}
Given the above ideas our algorithm is an iterative algorithm defined as follows:
Initialize top value $\top_0=W$. The $i$-th iteration is as follows:
\begin{enumerate}
\item Run the progress measure algorithm with top value $\top_i$.

\item If a winning region $X$ for player~1 is identified, return $X$.

\item Else $\top_{i+1}=2 \top_i$ (i.e., the top value is doubled).

\item If $\top_{i+1} \geq 2  n  W$, stop the algorithm and return
$\emptyset$, else proceed to the next iteration.
\end{enumerate}

\para{Correctness and running time analysis.}
The key steps of the correctness argument and the running time analysis are as follows:

\begin{enumerate}

\item The above algorithm is correct, since if it returns a set $X$ then it is a subset of 
the winning set for player~1. 

\item If the algorithm returns a winning set with top value $\alpha$, then the total 
running time till this iteration is $m \cdot (\alpha + \alpha/2 + \alpha/4 + \cdots)$, because the 
progress with top value $\alpha$ requires time $O(\alpha m)$. 
Hence the total running time if a set $X$ is returned with top value $\alpha$ is
$O(\alpha \cdot m)$.

\item Let $Z$ be a set of vertices such that no player-2 vertex in $Z$ has an edge out of $Z$,
and the whole subgame $\Gamma \restr Z$ is winning for player~1. Then a winning strategy in 
$Z$ ensures that a progress measure with top value $|Z| W$ would identify the set 
$Z$ as a winning set.

\item From above it follows that if the winning set $X$ is identified at top value $\alpha$, 
but no winning set was identified with top value $\alpha/2$, then the size of the winning set 
is at least $\alpha/(2W)$.

\item It follows from above that if a set $X$ is identified, then the total running time 
to obtain set $X$ is $O(|X|  m  W)$.

\item Moreover, the total running time of the algorithm when no set $X$ is identified is in 
$O(nmW)$, and in this case, the winning region is empty.

\end{enumerate}
Thus we solved the modified static algorithm for threshold mean-payoff games as desired
and obtain the following result.

\begin{theorem}\label{mfcs:thm:winningset}
	Let
	$Z=\W{1}{\mpayoff{(\Gamma,w,\V)}}$. 
If $Z\neq \emptyset$, then a non-empty set $X \subseteq Z$ can be computed in
time $O(|X|  m W)$, 
else an empty set is returned if $Z=\emptyset$, which takes time $O(n m W)$.
\end{theorem}

Using the above algorithm to compute the winning set for player~1 in the
subgames, we obtain an algorithm for threshold mean-payoff coBüchi
games in time $O(n m W)$. 

\begin{theorem}\label{mfcs:thm:cobuchi}
	Given a mean-payoff cob\"uchi game the winning set of a threshold mean-payoff coB\"uchi objective,
can be computed in $O(nmW)$ time.
\end{theorem}

\subsection{Threshold Mean-Payoff Parity Games}
The algorithm for threshold mean-payoff parity games is the standard recursive
algorithm~\cite{CJH05}
(classical parity game-style algorithm) that generalizes the B\"uchi and coB\"uchi cases
(which are the base cases).
The running time recurrence is as follows: 
$T(n,d,m,w) = n (T(n,d-1,m) + O(m)) + O(nmW)$.
Using our approach we obtain the following result.

\begin{theorem}~\label{mfcs:thm:tmpp}
	Given a mean-payoff parity game the winning set of a threshold mean-payoff parity objective can be computed in 
	$O(n^{d-1}mW )$ time.
\end{theorem}

\section{Optimal Values for Mean-payoff Parity Games}
In this section, we present an algorithm that computes the value function for 
mean-payoff parity games. For mean-payoff games a dichotomic search approach was 
presented in~\cite{brim2011}.
We show that such an approach can be generalized to mean-payoff parity games.

\smallskip\noindent\emph{Range of Values for the Dichotomic Search.}
To describe the algorithm we recall a lemma about the possible range of optimal 
values of a mean-payoff parity game.
The lemma is an easy consequence of the characterization of~\cite{CJH05} 
that the mean-payoff parity value coincides with the mean-payoff value,
and the possible range of value for mean-payoff games.

\begin{lemma}[\cite{CJH05,Ehrenfeucht1979,Lifshits2007}]\label{mfcs:lem:optval}
	Let $(\Gamma,p,w)$ be a mean-payoff parity game. For each vertex $v$, the optimal value  
	$\val(\MPP{\cdot,p,w,\Gamma})(v)$ is a rational number $\frac{y}{z}$ such that $1 \leq z \leq n$ 
and $|y| \leq z \cdot W$.
\end{lemma}

By Lemma~\ref{mfcs:lem:optval} the value of each vertex $v \in V$, is contained in the following set of rationals 
\begin{align*}
	S^{(\Gamma,p,w)} = \bigg\{ \frac{y}{z}\  \bigg | \ y,z \in \Z, 1 \leq z \leq n \land -z \cdot W \leq y \leq z \cdot W \bigg\}.
\end{align*}

\begin{definition}
	Let $(\Gamma,p,w)$ be a mean-payoff parity game.
We denote the set of vertices $v \in V$ such that 
$\val(\MPP{\cdot,p,w,\Gamma})(v) \circ \mu$ where $\circ \in \{<,\leq,=,\geq,>\}$
with $V^{\circ \mu}_\Gamma$.
\end{definition}

\smallskip\noindent\emph{Key Observation.}
\begin{sloppypar}
Let $((\Gamma = (V,E,\langle V_1, V_2 \rangle)),p,w)$ be a mean-payoff parity game.
Let $\mu \in [-W,W]$. The sets $V^{>\mu}_\Gamma,V^{=\mu}_\Gamma$ and $V^{<\mu}_\Gamma$ 
can be computed using any algorithm for threshold mean-payoff parity games 
twice (for example using Theorem~\ref{mfcs:thm:tmpp}). 
To calculate $V^{\geq \mu}_\Gamma$ and $V^{< \mu}_\Gamma$ use the algorithm on $(\Gamma,p,w)$ with the
mean-payoff parity objective $\phi = \parity{p,\Gamma} \cap \mpayoff{\mu,w,\Gamma}$.
Consider $((\Gamma' = (V,E,\langle V_2, V_1 \rangle)),p,w')$, where $w'(e) = -w(e)$ for all edges $e \in E$ 
and player-1 and player-2 vertices are swapped.
To calculate $V^{\leq \mu}_\Gamma$ and $V^{> \mu}_\Gamma$ use the algorithm on $(\Gamma',p,w)$ with 
mean-payoff parity objective $\phi = \parity{p,\Gamma'} \cap \mpayoff{-\mu,w',\Gamma'}$. 
Given the sets $V^{\leq \mu}_\Gamma$, $V^{> \mu}_\Gamma, V^{\geq \mu}_\Gamma$ and $V^{< \mu}_\Gamma$
we can extract the sets $V^{>\mu}_\Gamma,V^{=\mu}_\Gamma$ and $V^{<\mu}_\Gamma$.\\
All values $\mu'$ in  $S^{(\Gamma,p,w)}$ are of the form $\frac{y}{z}$. For those values
we can determine whether $v \in V_{\Gamma}^{\geq \mu'}$ by applying 
the algorithm for threshold mean-payoff parity games on $((\Gamma' = (V,E,\langle V_2, V_1 \rangle)),p,w')$ where 
$w'(e) = w(e) \cdot z$ for all $e \in E$ with the mean-payoff parity objectives $\phi =
\parity{p,\Gamma} \cap
\mpayoff{y,w',\Gamma}$. Note that in the worst case, the weight function $w'$ of
$\Gamma'$ is in $O(nW)$.
\end{sloppypar}

\smallskip\noindent\emph{Dichotomic Search.}
Let $(\Gamma,p,w)$ be a mean-payoff parity game. The dichotomic search algorithm is recursive algorithm
initialized with $\Gamma_0 = \Gamma$ and $S_0 = S^{(\Gamma,p,w)}$. In recursive call $i$ the following steps are
executed:
\begin{enumerate}
\item Let $r_i =  \min(S_i)$ and $s_i = \max(S_i)$.
\item Determine $a_1$, the largest element in $S_i$ less than or equal to $\frac{r_i + s_i}{2}$ and  
$a_2$, the smallest element in $S_i$ greater than or equal to $\frac{r_i + s_i}{2}$.
\item Determine the partitions $V_{\Gamma_i}^{<a_1}$, $V_{\Gamma_i}^{=a_1}$, 
$V_{\Gamma_i}^{=a_2}$, $V_{\Gamma_i}^{>a_2}$ using the key observation.
\item For all $v \in V_{\Gamma_i}^{=a_1}$ set the value to $a_1$,
for all $v \in V_{\Gamma_i}^{=a_2}$ set the value to $a_2$ and set the value to  $-\infty$
for all vertices $v$ which are not in any set calculated in step 3.
\item Recurse upon $\Gamma_i \restr V_{\Gamma_i}^{<a_1}$ and $\Gamma_i \restr V_{\Gamma_i}^{>a_2}$.
\end{enumerate}

\smallskip\noindent\emph{Correctness.}
Let $(\Gamma,p,w)$ be a mean-payoff parity game. We prove that the dichotomic search algorithm correctly calculates
$\val(\MPP{\cdot,p,w,\Gamma})(v)$ for all $v \in V$.
The algorithm is initialized with $\Gamma$ and $S^{(\Gamma,p,w)}$. 
By Lemma~\ref{mfcs:lem:optval} the values of the vertices 
$v \in V$ are in the set $S^{(\Gamma,p,w)}$. Because we perform a binary search over the set $S^{(\Gamma,p,w)}$ 
we can guarantee the
termination of the algorithm. Notice that we need to show that the values calculated in the 
subgames constructed in step~4 are identical to the values in the original game. Then correctness follows 
immediately by our key observation and because we perform a binary search over the set $S^{(\Gamma,p,w)}$.

\begin{lemma}\label{mfcs:lem:subgameval}
	Given a mean-payoff parity game $(\Gamma,p,w)$ and $\mu \in \Q$, 
let $\Gamma' = \Gamma \restr V^{>\mu}_\Gamma$ and 
$\Gamma'' = \Gamma \restr V^{<\mu}_\Gamma$. 
For all $v \in V^{> \mu}_\Gamma$, we have $\val(\MPP{\cdot, p,w,\Gamma'})(v) = \val(\MPP{\cdot,p,w,\Gamma})(v)$ and 
for all $v \in V^{< \mu}_\Gamma$, we have $\val(\MPP{\cdot,p,w,\Gamma''})(v) = \val(\MPP{\cdot,p,w,\Gamma})(v)$.
\end{lemma}

\begin{proof}
	\begin{sloppypar}
Let  $v \in V^{> \mu}_\Gamma$ be arbitrary.
We will prove $\val(\MPP{\cdot, p,w,\Gamma'})(v) = \val(\MPP{\cdot,p,w,\Gamma})(v)$ by showing the following two cases: 
\begin{itemize}
	\item $\val(\MPP{\cdot, p,w,\Gamma'})(v) \leq \val(\MPP{\cdot,p,w,\Gamma})(v)$: 
	Note that there can be no player-2 vertex in $V^{>\mu}_\Gamma$ with an edge to $V^{\leq \mu}_\Gamma$.
	Thus we cut away only edges of player-1 vertices in $\Gamma'$. 
	Consequently player-1 has less choices in $\Gamma'$ than in $\Gamma$ at each of her vertices.
	Thus $\val(\MPP{\cdot, p,w,\Gamma'})(v) \leq \val(\MPP{\cdot,p,w,\Gamma})(v)$ holds.

	\item $\val(\MPP{\cdot, p,w,\Gamma'})(v) \geq \val(\MPP{\cdot,p,w,\Gamma})(v)$:
	Let $\sigma$ be an optimal strategy for player~1 and let
	$\pi$ be an optimal strategy for player~2 which both exist
	by~\cite{CJH05}.
	We will show that $\sigma$ produces plays with vertices in $V^{> \mu}_\Gamma$ only, 
	if it starts from $v$.
	For the sake of contradiction assume that a play $\rho = \omega(v,\sigma,\pi)$ 
	contains a vertex $v^* \in  V^{\leq \mu}_\Gamma$.
	Notice that there are no player-2 vertices in $V^{>\mu}_\Gamma$ with edges to $V^{\leq \mu}_\Gamma$. 
	Thus $\sigma$ chose a successor vertex in $V^{\leq \mu}_\Gamma$. 
	But when $\rho$ ends up in $V^{\leq \mu}_\Gamma$ 
	the optimal player-2 strategy $\pi$ can guarantee that $\MPP{\rho,p,w,\Gamma} \leq \mu$
	by the definition of $V^{\leq \mu}_\Gamma$. 
	There is a strategy to keep the value of the play starting at
	$v$ greater than $\mu$ by the definition of $V^{>\mu}_\Gamma$. Thus any play $\rho$ leading 
	to $V^{\leq \mu}_\Gamma$ using $\sigma$ is not optimal which is a contradiction to our assumption. 
	Consequently $\val(\MPP{\cdot, p,w,\Gamma'})(v) \geq \val(\MPP{\cdot,p,w,\Gamma})(v)$ follows.
\end{itemize}
The fact that for all $v \in V^{< \mu}_\Gamma$, we have
$\val(\MPP{\cdot, p,w,\Gamma''})(v) = \val(\MPP{\cdot,p,w,\Gamma})(v)$ follows by a symmetric argument.
\end{sloppypar}
%
%
%

\end{proof}

\smallskip\noindent\emph{Running Time.}
The running time of the dichotomic search is $O(n\log(nW)\mathsf{TH})$
where $\mathsf{TH}$ is the running time of an algorithm for the threshold mean-payoff parity problem.
The additional factor $n$ comes from rescaling the weights of the mean-payoff parity game $\Gamma$ which is
described in the key observation.
The factor $O(\log(nW))$ is from using binary search on $S$ as $|S| = O(n^2W)$.

\begin{theorem}
	\begin{sloppypar}
		Given a mean-payoff parity game $(\Gamma,p,w)$ and an algorithm
		that solves the threshold mean-payoff parity problem in $O(\mathsf{TH})$, the
		value function of $(\Gamma,p,w)$ can be computed in time
		$O(n\log(nW)\mathsf{TH} )$.
	\end{sloppypar}
\end{theorem}

\noindent As a corollary of the above theorem and Theorem~\ref{mfcs:thm:tmpp}, the value function 
for mean-payoff parity games can be computed in $O(n^d \cdot m \cdot W \cdot\log(nW))$ 
time. 

\section{Conclusion}
In this chapter, we present faster algorithms for mean-payoff parity games. 
Our most interesting results are for mean-payoff Büchi and mean-payoff coBüchi games, 
which are the base cases. For threshold mean-payoff Büchi and mean-payoff coBüchi games, 
our bound $O(nmW)$ matches the current best-known bound for mean-payoff games. 
For the value problem, we show the dichotomic search approach of~\cite{brim2011} 
for mean-payoff games can be generalized to mean-payoff parity games. 
This gives an additional multiplicative factor of $n\log(nW)$ as compared to the threshold problem. 
A recent work by Comin et al.~\cite{Comin2017} shows that the value problem for mean-payoff objective can be 
solved with a multiplicative factor $n$ compared to the threshold objective 
(i.e., it shaves of the $\log$ factor). 
An interesting question is whether the approach of Comin et al.\ can be generalized to mean-payoff parity games.

	\chapter[Near-Linear Time Algorithms for Streett Objectives in Graphs and MDPs][Near-Linear Time
	Algs.\@ f.\@ Streett Obj.\@ in Graphs \& MDPs]{Near-Linear Time Algorithms for Streett Objectives in Graphs and MDPs}\label{cha:concur}
	In this chapter, we present randomized near-linear time algorithms for Streett objectives in graphs and MDPs.

\section{Introduction}\label{concur:sec:intro}
In this work, we present near-linear (hence near-optimal) randomized algorithms 
for the strong fairness verification in graphs and Markov Decision Processes (MDPs).
In the fundamental model-checking problem, the input is a \emph{model} and a 
\emph{specification}, and the algorithmic verification problem is to check whether
the model \emph{satisfies} the specification. 
We first describe the models and the specifications we consider, then the notion 
of satisfaction, and then previous results followed by our contributions.

\para{Models: Graphs and MDPs.}
Graphs and Markov decision processes (MDPs) are two classical models of reactive 
systems.
The states of a reactive system are represented by the vertices of a graph, the transitions
of the system are represented by the edges and non-terminating trajectories of the system
are represented as infinite paths of the graph.
Graphs are a classical model for reactive systems with nondeterminism, and 
MDPs extend graphs with probabilistic transitions that represent reactive systems with 
both nondeterminism and uncertainty. 
Thus, graphs and MDPs are the standard models of reactive systems with nondeterminism,
and nondeterminism with stochastic aspects, respectively~\cite{ClarkeBook,baierbook}.
Moreover, MDPs are used as models for concurrent finite-state processes~\cite{CY95,Vardi85}
as well as probabilistic systems in open environments~\cite{Segala95,prism,STORM,baierbook}.

\para{Specification: Strong fairness (aka Streett) objectives.}
A fundamental specification formalism in the analysis of reactive systems 
is the {\em strong fairness condition}.
The strong fairness conditions (aka Streett objectives) consist of $k$ types 
of requests and corresponding grants, and the requirement is that for each type if the 
request happens infinitely often, then the corresponding grant must also happen
infinitely often. 
Beyond safety, reachability, and liveness objectives, the most standard properties
that arise in the analysis of reactive systems are Streett objectives,  
and chapters of standard textbooks in verification are devoted to it (e.g., 
\cite[Chapter~3.3]{ClarkeBook},~\cite[Chapter~3]{MPProgress},~\cite[Chapters~8,~10]{AH04}).
Besides, $\omega$-regular objectives can be specified as Streett objectives, e.g., 
LTL formulas and non-deterministic $\omega$-automata can be translated to
deterministic Streett automata~\cite{Safra88} and efficient translations
have been an active research area~\cite{ChatterjeeGK13,EsparzaK14,KomarkovaK14}. 
Consequently, Streett objectives are a canonical class of objectives that arise in verification.

\para{Satisfaction.} 
The notions of satisfaction for graphs and MDPs are as follows:
For graphs, the notion of satisfaction requires that there is a trajectory (infinite path) 
that belongs to the set of paths specified by the Streett objective.
For MDPs, the satisfaction requires that there is a strategy to resolve the nondeterminism 
such that the Streett objective is ensured almost-surely (with probability~1).
Thus the algorithmic model-checking problem of graphs and MDPs with Streett objectives 
is a central problem in verification, and is 
at the heart of many state-of-the-art tools such as SPIN, NuSMV for graphs~\cite{SPIN,Cimatti2000}
and
PRISM, LiQuor, Storm for MDPs~\cite{prism,LIQUOR,STORM}.

Our contributions are related to the algorithmic complexity of graphs and MDPs with 
Streett objectives. 
We first present previous results and then our contributions.

\subparagraph*{Previous results.}
The most basic algorithm for the problem for graphs is based on repeated SCC (strongly
connected component) computation, and informally can be described as follows:
for a given SCC, (a)~if for every request type that is present in the SCC 
the corresponding grant type is also present in the SCC, then the SCC is identified 
as ``good'', (b)~else vertices of each request type that have no corresponding
grant type in the SCC are removed, and the algorithm recursively proceeds
on the remaining graph.
Finally, reachability to good SCCs is computed. 
The algorithm for MDPs is similar where the SCC computation is replaced with 
maximal end-component (MEC) computation and reachability to good SCCs is replaced
with probability~1 reachability to good MECs. 
The basic algorithms for graphs and MDPs with Streett objective have been improved in several
works, such as for graphs in~\cite{HT96,ChatterjeeHL15}, for MEC computation 
in~\cite{CH11,ChatterjeeH12,CH14}, and MDPs with Streett objectives in~\cite{CDHL16}.
For graphs/MDPs with $n$ vertices, $m$ edges, and $k$ request-grant pairs with $b$ denoting 
the size to describe the request grant pairs, the current best-known bound is $O(\min(n^2, m \sqrt{m
\log n}) + b \log n)$.

\subparagraph*{Our contributions.} 
In this work, our main contributions are randomized near-linear time (i.e.\ linear times a polylogarithmic factor) algorithms 
for graphs and MDPs with Streett objectives. In detail, our contributions are as follows:

\begin{itemize}

\item First, we present a near-linear time randomized algorithm for graphs with Streett objectives
	where the expected running time is $\widetilde{O}(m + b)$, where the $\widetilde{O}$ notation
	hides poly-log factors. 
	Our algorithm is based on a recent randomized algorithm for maintaining the SCC decomposition of graphs 
	under edge deletions, where the expected total running time is near linear~\cite{BPW19}.

\item Second, by exploiting the results of~\cite{BPW19} we present a randomized near-linear time
	algorithm for computing the MEC decomposition of an MDP where the expected running time is $\widetilde{O}(m)$.
	We extend the results of~\cite{BPW19} from graphs to MDPs and present a randomized algorithm to maintain 
	the MEC decomposition of an MDP under edge deletions, where the expected total running time is near
	linear~\cite{BPW19}.

\item Finally, we use the result of the above item to present a near-linear time randomized algorithm 
	for MDPs with Streett objectives where the expected running time is $\widetilde{O}(m + b)$.

\end{itemize}
All our algorithms are randomized and since they are near-linear in the size of the input, they are optimal 
up to poly-log factors. 
An important open question is whether there are deterministic algorithms that can improve the 
existing running time bound for graphs and MDPs with Streett
objectives. Our algorithms are deterministic except for the invocation of the decremental SCC
algorithm presented in~\cite{BPW19}.

\begin{table}[h!]
	\small
	\caption{Summary of Results.}\label{concur:tab:results}
	\centering
	\renewcommand{\arraystretch}{1.10}
	\begin{tabular}{l l l}
	
		\toprule
		Problem & New Running T. & Old Running T. \\ 
		\midrule

		Streett Objectives on Graphs & $\O(m + b)$ & $\O(\min(n^2, m \sqrt{m}) +
		b)$~\cite{CHL17,HT96}\\
		Almost-Sure Reachability & $\O(m)$ & $O(m \cdot n^{2/3})$~\cite{CDHL16,CH14}\\
		MEC Decomposition & $\O(m)$ & $O(m\cdot n^{2/3})$~\cite{CH14}\\
		Decremental MEC Decomposition & $\O(m)$ & $O(nm)$~\cite{CH14}\\
		Streett Objectives on MDPs   & $\O(m+b)$ & $\O(\min(n^2, m \sqrt{m}) +
		b)$~\cite{CDHL16}\\
		\bottomrule

	\end{tabular}
\end{table}

\section{Decremental SCCs}\label{concur:subsec:decsccs}

We first recall the result about decremental strongly
connected components maintenance in~\cite{BPW19} (cf.\ Theorem~\ref{concur:thm:decscc} below) and then augment the result for our
purposes. 

\SetKwFunction{query}{query}
\SetKwFunction{delete}{delete}
\begin{theorem}[Theorem 1.1 in~\cite{BPW19}]\label{concur:thm:decscc}
	Given a graph $G=(V,E)$ with $m$ edges and $n$ vertices, we can maintain a
	data structure $\A$ that supports the operations
	\begin{itemize}
	\item \delete{$u,v$}: Deletes the edge $(u,v)$ from the graph
		$G$.
	\item \query{$u,v$}: Returns whether $u$ and $v$ are in the same
		SCC in $G$,
	\end{itemize}
	in total expected update time $O(m \log^4 n)$ and with worst-case
	constant query time. The bound holds against an oblivious 
	adversary.
\end{theorem}

The preprocessing time of the algorithm is $O(m + n)$ using~\cite{Tarjan72}.
To use this algorithm we extend the query and update operations 
with three new operations described in Corollary~\ref{concur:cor:returnnewsccs}.

Intuitively, the first function is available in the algorithm described in~\cite{BPW19}.
The second function can be implemented directly from the construction of the data structure
maintained in~\cite{BPW19}.
The key idea for the third function is that when an SCC splits, we
consider the new SCCs. We distinguish between the largest of them and the others which we call small
SCCs. We then consider all edges incident to the small SCCs:
Note that as the new outgoing edges in the large SCC are also incident to a small 
SCC we can also determine the outgoing edges of the large SCC.
Observe that whenever an SCC splits all the small SCCs are at most half the size of the original SCC\@. 
That is, each vertex can appear only $O(\log n)$ times in small SCCs during the whole algorithm.
As an edge is only considered if one of the incident vertices is in a small SCC each 
edge is considered $O(\log n)$ times and the additional running time is bounded by $O(m \log n)$.
Furthermore, we define $\decrSCC$ as the running time of
the best decremental SCC algorithm which supports the operations in 
Corollary~\ref{concur:cor:returnnewsccs}. Currently, $\decrSCC = O(m \log^4 n)$.

\SetKwFunction{deleteannouncenooutgoing}{delete-announce-no-outgoing}
\SetKwFunction{deleteannounce}{delete-announce}
\SetKwFunction{rep}{rep}
\begin{corollary}\label{concur:cor:returnnewsccs}
	Given a graph $G = (V,E)$ with $m$ edges and $n$ vertices, we can maintain a
	data structure $\A$ that supports the operations
	\begin{itemize}
	\item \rep{$u$} (query-operation): 
		Returns a reference to the SCC containing the vertex $u$.
	\item \deleteannounce{$E$} (update-operation): 
		Deletes the set $E$ of edges from the graph $G$. 
		If the edge deletion creates new SCCs $C_1, \dots, C_k$ the
		operation returns a list $Q = \{ C_1, \dots, C_k \}$ 
		of references to the new SCCs.
	\item \deleteannouncenooutgoing{$E$} (update-operation): 
		Deletes the set $E$ of edges from the graph $G$. 
		The operation returns a list $Q = \{C_1, \dots, C_k\}$ of
		references to all new SCCs with no outgoing edges.
	\end{itemize}
	in total expected update time $O(m \log^4 n)$ and worst-case
	constant query time for the first operation. 
	The bound holds against an oblivious adaptive adversary.
\end{corollary}
\begin{proof}
	We first recall the data structure $\A$ maintained by the algorithm of~\cite{BPW19}. We need the
	following detail about $\A$ to prove the second and third point.
	The data structure $\A$ maintains a hierarchy of graphs $\hat{G} = \{ \hat{G}_0,
	\dots, \hat{G}_{\lfloor \log n \rfloor + 1}\}$:
	$$ \hat{G}_i = \condense{(V,\bigcup_{j<i} E_j)} \cup E_i $$
	where the edge sets $E_i$ form a partition of $E$.
	Additionally, the top graph $\hat{G}_{\lfloor \log n \rfloor + 1}$ contains all the edges of $G$
	and the SCCs in $\hat{G}_{\lfloor \log n \rfloor + 1}$ are thus the same as
	in $G$~\cite[P.\@ 5]{BPW19}. The top level graph thus corresponds to
	$\condense{G}$. 
	\begin{itemize}
		\item The data structure $\A$ maintained in~\cite{BPW19} can be extended
			to support the function $\rep{$v$}$ for some vertex $u$ because 
			$\A$ can in constant time identify the node representing $v$ in the graph $\hat{G}_{\lfloor \log n \rfloor
			+ 1}$~\cite[P.23]{BPW19}.

		\item The data structure $\A$ maintained in~\cite{BPW19} can be extended to
			support the modified delete-operation $\deleteannounce{$E$}$:
			At the beginning of the operation we initialize a new List $\L$.
			For each edge $e \in E$ we do the following: We compute $\A.\delete{$e$}$
			and when nodes in $\condense{G}$ are split due to an edge deletion, 
			$\A$ creates new nodes $\{ s_1,\dots,s_k \}$ in $\condense{G}$ by the fact 
			that $\A$ maintains the hierarchy $\hat{G}$. For each $s_i \in
			\{s_1, \dots, s_k\}$, store the pointer to the
			nodes in $\L$. In the end we return $\L$. This can be done with an additional constant effort
			when maintaining the graph $\condense{G}$. If an SCC $C$ splits which is
			already in $\L$ due to an edge deletion in $C$, we create a new node for each new SCCs
			(and add it to $\L$). 			

		\item We prove that the data structure $\A$ maintained in~\cite{BPW19} can be extended to
			support the modified delete operation
			$\deleteannouncenooutgoing{$E$}$ in total expected
			update time $O(m \log^4 n)$ time: Note that the top level
			maintained in the hierarchy corresponds to $\condense{G}$.
			Initially, we store the number of
			outgoing edges in a counter for all nodes in $\condense{G}$. 
			Clearly, this is in $O(m+n)$ time, the preprocessing
			time of $\A$. These counters will be maintained throughout the
			entire sequence of edge-deletions, i.e., for each edge deleted from
			$G$. Whenever an edge $e= (u,v) \in E$ is deleted we need to distinguish between the case
			(i) $e$ is in the graph $\condense{G}$ and (ii) $e$ is in one of the SCCs.
			In case (i) the counter of the node representing $u$ must be decremented. 
			This is in constant time.
			For case (ii), when $e$ is inside an SCC we also distinguish between two cases.
			Either the SCC decomposition does not change ($\condense{G}$ does not change) or the SCC
			decomposition changes. In the first case, we leave the counters unchanged and are
			done because $\condense{G}$ does not change.
			When the SCC decomposition changes, new nodes are created in $\condense{G}$ because by
			deleting edges we create more SCCs.
			Thus there is an SCC $C_o$ which is split into new SCCs $C = \{C_1, \dots, C_k\}$ when we delete the edge $e$.
			We initialize the counters for all new SCCs to zero (nodes created in
			$\condense{G}$) except for the largest one (i.e., $C_\ell = \arg\max_{C_i \in C} |C_i|$) which we set to the value of the original SCC $C_o$.
			We determine the number of outgoing edges of the SCCs
			in $C \setminus \{ C_\ell \}$ by looking at the incident edges.
			During this process we modify the counter of $C_\ell$, which we initialized to the counter of
			the original SCC $C_o$ in the following two cases:
			\begin{itemize}
				\item If an edge in $C \setminus \{ C_\ell \}$ goes to an SCC
					not in $C$ we decrement the counter of
					$C_\ell$. 
				\item If there is an edge going to an SCC in $C$ we increment the
					counter of $C_\ell$. 			
			\end{itemize}
			\begin{sloppypar}
				Assume the counters were correct before
				the edge deletion, that is, in $\condense{G}$. 
				Let $G'$ be the graph without $e$. We prove that we count each
				outgoing edge in $\condense{G'}$ correctly. Let $e'=(u,v)$ be an arbitrary edge in
				$\condense{G'}$. 
				Either $e'$ is present before the deletion of $e$ or $e'$ is new
				after the deletion of $e$. 
			\end{sloppypar}
			\begin{itemize}
				\item Edges present before the deletion of $e$ are of two types, going out of $C_o$ and not
					going out of $C_o$:
					When the edge $e'$ is not going out of $C_o$ the counter of the outgoing node is still
					correct because the number of outgoing edges did not change. 
					When the edge $e'$ is going out of $C_o$ there are two cases. Either it is now going out of
					$C_\ell$ or of some other SCC in $C \setminus \{C_\ell\}$. If $u = C_\ell$, notice that $C_\ell$ is set to the counter of
					$C_o$ and we count $e'$.
					On the other hand, if $u$ is in $C \setminus \{C_\ell\}$ and $v$ is a node outside
					of $C$ we count $e'$ when we look at all edges incident to $u$. Note that we also count
					it at $C_\ell$ because 
					this edge is originally incident to $C_o$ and but not to $C_\ell$.
					Thus we decrement the counter of $C_\ell$ for each such edge. 			%
				\item Edges new after the deletion of $e$ are also of two types: Going out of $C_\ell$ and
					going out of an SCC in $C \setminus \{C_\ell\}$. Note that all edges have their
					endpoints in $C$.
					If $u$ is in $C \setminus \{C_\ell\}$ we count $e'$ correctly because we looked at
					all edges incident to $u$.  
					If $u=C_\ell$ we need to increment the counter of $u$. This is done when
					we look at all the incident edges of $C \setminus \{ C_\ell \}$. 
			\end{itemize}			%
			For each edge, we have a constant amount of work: An edge can be
			in the smaller part of the partition at most $O(\log n)$ times
			because it is only considered when the corresponding original SCC
			halves. 
			When we execute the operation $\deleteannouncenooutgoing{$E$}$ while deleting edges 
			and we find that an SCC has no outgoing edges, i.e., the stored
			counter is zero, we add it to a list and return this list.\qedhere
	\end{itemize}
\end{proof}

\section{Graphs with Streett Objectives}\label{concur:sec:streettgraph}
In this section, we present an algorithm which computes the winning regions for graphs with Streett objectives.
The input is a directed graph $G=(V,E)$ and $k$ Streett pairs $(L_j,U_j)$ for $j
= 1,\dots, k$. The size of the input is measured in terms of $m = |E|$, $n =
|V|$, $k$ and $b = \sum_{j=1}^k(|L_j| + |U_j|) \leq 2nk$. 

\subparagraph*{Algorithm \textsf{Streett} and good component detection.}
Let $C$ be an SCC of $G$. In the good component detection problem, we compute (a) a
non-trivial SCS $G[X] \subseteq C$ induced by the set of vertices $X$, such that
for all $1 \leq j \leq k$ either $L_j \cap X = \emptyset$ or $U_j \cap X
\neq \emptyset$ or (b) that no such SCS exists. In the first case, there exists
an infinite path that eventually stays in $X$ and satisfies the Streett
objective, while in the latter case, there exists no path which satisfies the
Streett objective in $C$. From the results of~\cite[Chapter 9, Proposition
9.4]{AH04} the following algorithm,
called Algorithm \textsf{Streett}, suffices for the winning set computation:
\begin{enumerate}
	\item Compute the SCC decomposition of the graph;
	\item For each SCC $C$ for which the good component detection returns an
		SCS, label the SCC $C$ as satisfying.
	\item Output the set of vertices that can reach a satisfying SCC as the
		winning set.
\end{enumerate}
Since the first and last step are computable in linear time, the running time of
Algorithm \textsf{Streett} is dominated by the detection of good components in SCCs.
In the following, we assume that the input graph is strongly connected and focus
on good component detection.

\para{Bad vertices.}
A vertex is $\emph{bad}$ if there is some $1\leq j \leq k$ such that the vertex is in
$L_j$ but it is not strongly connected to any vertex of $U_j$. All other
vertices are good. Note that a good vertex might become bad if a vertex
deletion disconnects an SCS or a vertex of a set $U_j$. A good component is a
non-trivial SCS that contains only good vertices.

\para{Decremental strongly connected components.} 
Throughout the algorithm, we use the algorithm described in
Section~\ref{concur:subsec:decsccs} to maintain the SCCs of a graph when deleting edges. In particular, we use
Corollary~\ref{concur:cor:returnnewsccs} to obtain a list of the new SCCs which are
created by removing bad vertices. Note that we can `remove' a vertex
by deleting all its incident edges. Because the decremental SCC algorithm assumes an oblivious
adversary we sort the list of the new SCCs as, otherwise, the edge deletions performed by our algorithm would depend on the random choices of the decremental SCC algorithm.

\para{Data structure.}
During the algorithm, we maintain a decomposition of the vertices
in $G = (V,E)$: We maintain a list $Q$ of certain sets $S \subseteq V$ such that every SCC of $G$ is contained in some $S$
stored in $Q$.

\SetKwFunction{add}{add}
\SetKwFunction{pull}{pull}
The list $Q$ provides two operations: $Q.\add{$X$}$ enqueues $X$ to $Q$; and $Q.\pull{}$ dequeues an arbitrary element $X$ from $Q$.
For each set $S$ in the decomposition, we store a data structure $D(S)$ in the list $Q$.
This data structure $D(S)$ supports the following operations 

\SetKwFunction{Construct}{construct}
\SetKwFunction{Remove}{remove}
\SetKwFunction{Bad}{bad}
\SetKwFunction{sccs}{d-sccs}
\begin{enumerate}
	\item \Construct{$S$}: initializes the data structure for the set $S$
	\item \Remove{$S,B$} updates $S$ to $S \setminus B$ for a set $B \subseteq V$ and returns
		$D(S)$ for the new set $S$.
	\item \Bad{$S$} returns a reference to the set $\{ v \in S \mid \exists j \text{ with } v
		\in L_j \text{ and } U_j \cap S = \emptyset \}$
\item \sccs{$S$} returns the set of SCCs currently in $G[S]$. We implement \sccs{$S$} as a balanced binary search tree which allows logarithmic and updates and deletions.
\end{enumerate}
In~\cite{HT96} an implementation of this data structure with functions (1)-(3) is described that achieves
the following running times. For a set of vertices $S \subseteq V$, let
$\bits(S)$ be defined as $\sum_{j=1}^k (|S \cap L_j| + |S \cap U_j|)$.

\begin{lemma}[Lemma~2.1 in~\cite{HT96}]\label{concur:lem:Streett_ds}
	After a one-time preprocessing of time $O(k)$, the data structure $D(S)$ can
	be implemented in time $O(bits(S) + |S|)$ for \Construct{$S$}, time
	$O(\bits(B) + |B|)$ for \Remove{$S,B$} and constant running time
	for \Bad{$S$}.
\end{lemma}

We augment the data structure with the function \sccs{$S$} which runs in total
time of a decremental SCC algorithm supporting the first function in
Corollary~\ref{concur:cor:returnnewsccs}.

\para{Algorithm Description.}
The key idea is that the algorithm maintains the list $Q$ of data structures $D(S)$ as described
above when deleting bad vertices.
Initially, we enqueue the data structure
returned by \Construct{$V$} to $Q$. As long as $Q$ is non-empty, the algorithm repeatedly pulls a set $S$ from $Q$ and
identifies and deletes bad vertices from $G[S]$. If no edge is contained in
$G[S]$, the set $S$ is removed as it can only induce trivial SCCs. Otherwise, the
subgraph $G[S]$ is either determined to be strongly connected and output as a
good component or we identify and remove an SCC with at most
half of the vertices in $G[S]$. Consider Figure~\ref{concur:fig:goodcomp} for an illustration of an example
run of Algorithm~\ref{concur:alg:goodcomp}.

\begin{figure}[ht]
	\centering
	\includegraphics[width=\textwidth]{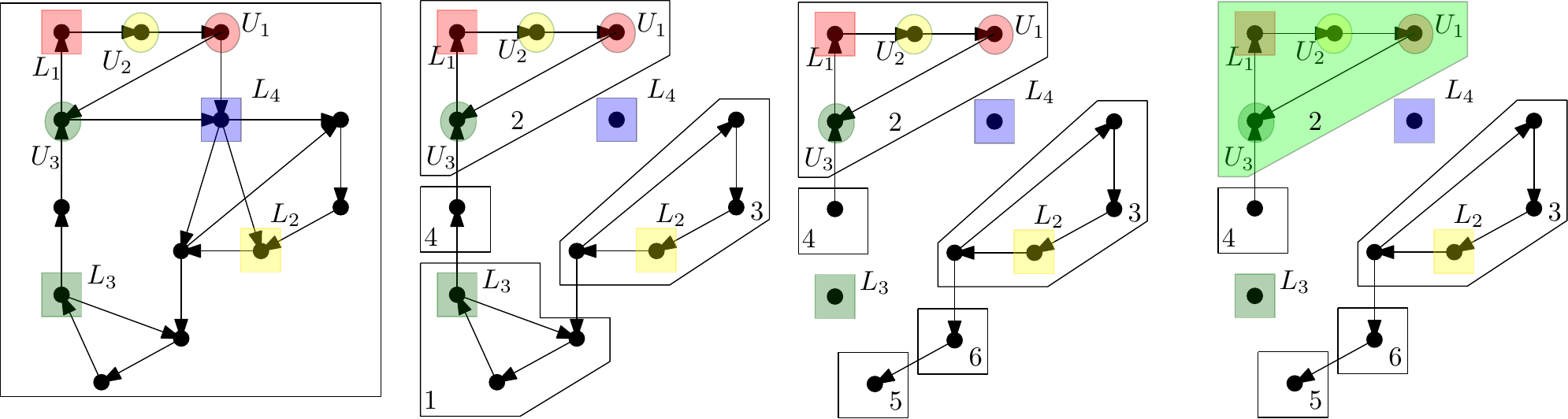}
	\caption{Illustration of one run of Algorithm~\ref{concur:alg:goodcomp}: The vertex in the set $L_4$
		is a bad vertex and we remove it from the SCC yielding four
		new SCCs. First, we look in the SCC containing
		$L_3$. The vertex in $L_3$ is a bad vertex because there is no vertex in
		$U_3$ in this SCC\@.
		Again two SCCs are created after its removal. The next SCC we process is the SCC containing
		$L_1$. It is a good component
		because the vertex in $L_1$ has a vertex in $U_1$ in the same SCC\@. No bad vertices are
		removed and the whole SCC is identified as a good component.
}\label{concur:fig:goodcomp}
\end{figure}

\para{Outline correctness and running time.}
In the following, 
when we talk about the \emph{input graph} $\hat{G}$ we mean the unmodified, strongly connected graph which we
use to initialize Algorithm~\ref{concur:alg:goodcomp}. 
In contrast, with the \emph{current} graph $G$ we refer to the graph where we already deleted vertices and their incident edges in the course of finding a good component.
For the correctness of Algorithm~\ref{concur:alg:goodcomp}, we show that if a good component exists, then
there is a set $S$ stored in list $Q$ which contains all vertices of this good component.

To obtain the running time bound of Algorithm~\ref{concur:alg:goodcomp}, we use the fact that we can maintain the SCC decomposition
under deletions in $O(\decrSCC)$ total time. With the properties of the data
structure described in Lemma~\ref{concur:lem:Streett_ds} we get a
running time of $\O(n+b)$ for the maintenance of the data structure and
identification of bad vertices over the whole algorithm. Combined, these ideas
lead to a total running time of $\O(\decrSCC + n + b)$ which is $\O(m + b)$ using
Corollary~\ref{concur:cor:returnnewsccs}.

\begin{algorithm}[t]
	\small
	\caption{Algorithm~\textsc{GoodComp}}\label{concur:alg:goodcomp}
	\KwIn{Strongly connected graph $G = (V,E)$ and Streett pairs $(L_j,U_j)$ for $j = 1, \dots, k$}
	\KwOut{a good component in $G$ if one exists}
	Invoke an instance $\A$ of the decremental SCC algorithm; Initialize $Q$ as a new list.\;
	$D(V) \gets$ \Construct{$V$}; $D(V).\sccs{$V$} \gets \{\A.\rep{$x$}\}$ for some $x \in V$\;
	$Q.\add{$D(V)$}$\; 
	\While{$Q$ is not empty}{\label{concur:alg:goodcomp:outerwhile}
		$D(S) \gets Q.\pull{}$\;\label{concur:alg:goodcomp:pull}
		\While{$D(S).\Bad{$S$}$ is not empty}{\label{concur:alg:goodcomp:removebad1}
			$B \gets D(S).\Bad{$S$};$ $D(S) \gets D(S).\Remove{$S,B$}$\;
			\tcp{obtain SCCs after deleting bad vertices from $S$}
			$D(S).\sccs{$S$} \gets D(S).\sccs{$S$} \setminus \left(\bigcup_{b \in
					B} \set{\A.\rep{$b$}}\right)$\;  \label{concur:alg:goodcomp:removebfromscc}
			$D(S).\sccs{$S$} \gets D(S).\sccs{$S$} \cup \A.\deleteannounce{$\edgeset{B}$}$\;\label{concur:alg:goodcomp:dannounce}
		}\label{concur:alg:goodcomp:removebad2}

		\If{$G[S]$ contains at least one edge}{\label{concur:alg:goodcomp:removetrivial}
			Initialize $K$ as a new list\;
			\For{ $X \gets D(S).\sccs{$S$}$}{\label{concur:alg:goodcomp:for1}
				\lIf{$X = S$} {
					\textbf{output} $G[S]$; \tcp*[f]{good component found}\DontPrintSemicolon\label{concur:alg:goodcomp:outputgoodcomp} 
				}
				\lIf{$|X| \leq {|S|\over 2}$}{
					$K.\add{$X$}$;
				}
			}
			Sort the SCCs in $K$ by vertex id (look at all the vertices in each SCC of
			$K$)\;\label{concur:alg:goodcomp:sort}
			$R \gets \emptyset$\tcp*{{\footnotesize Build $D(X)$ for SCCs $X$ in $K$
			and remove $X$ from $S$,$D(S)$ and \sccs{$S$}}}
			
			\For{$X \gets K.\pull{}$}{\label{concur:alg:goodcomp:for2}
				$R \gets R \cup X$; $D(X) \gets$ \Construct{$X$}\;
				$D(X).\sccs{$X$} \gets \set{\A.\rep{$x$}}$ for some $x \in
				X$\;\label{concur:alg:goodcomp:newSmallSCC}
				$D(S).\sccs{$S$} \gets D(S).\sccs{$S$} \setminus \{\A.\rep{$x$}\}$ for some $x \in X$\;\label{concur:alg:goodcomp:remSmallSCC}
				$Q.\add{$D(X)$}$\;\label{concur:alg:goodcomp:ds1}
			}
			\lIf{$D(S).\sccs{$S$} \not= \emptyset$}{
				$Q.\add{$D(S).\Remove{$S,R$}$}$\label{concur:alg:goodcomp:ds2} 
			}
		}
	}

	\Return No good component exists. 
\end{algorithm}

\begin{lemma}\label{concur:lem:goodcomp:runningtime}
	Algorithm~\ref{concur:alg:goodcomp} runs in expected time  $\O(m + b)$.
\end{lemma}
\begin{proof}
	The preprocessing and initialization of the data structure $D$ and the removal
	of bad vertices in the whole algorithm takes time $O(m + k + b)$ using
	Lemma~\ref{concur:lem:Streett_ds}. Since each vertex is deleted at most once, the data structure 
	can be constructed and maintained in total time $O(m)$. 
	Announcing the new SCCs after deleting the bad vertices at
	Line~\ref{concur:alg:goodcomp:dannounce} is in $O(\decrSCC) = \O(m)$ total time by 
	Corollary~\ref{concur:cor:returnnewsccs}.
	Consider an iteration of the while loop at Line~\ref{concur:alg:goodcomp:outerwhile}:
	A set $S$ is removed from $Q$. Let us denote by $n'$ the number of vertices of $S$. 
	If $G[S]$ does not contain any edge	after the removal of bad vertices, 
	then $S$ is not considered further by the algorithm. 
	Otherwise, the for-loop at Line~\ref{concur:alg:goodcomp:for1} considers all new SCCs.
	We can implement the for-loop in a lockstep fashion:
	In each step for each SCC we access the $i$-th vertex and as soon as all of the vertices 
	of an SCC are accessed we add it to the list $K$.
	When only one SCC is left we compute its size using the original set $S$ and 
	the sizes of the other SCCs. If its size is at most $|S|/2$ we add it to $K$.
	Note that this can be done in time proportional to the number of vertices in the SCCs in $S$
	of size at most $|S|/2$.
	The sorting operation at Line~\ref{concur:alg:goodcomp:sort} takes time $O(|K| \log |K|)$ plus the size of all
	the SCCs in $K$, that is $\sum_{K_i \in K} |K_i|$. Note that $O(|K| \log |K|) = O((\sum_{K_i
		\in K} |K_i|) \log (\sum_{K_i \in K} |K_i|))$.  Let $K_i \in K$ be an SCC stored in $K$. 
	Note that during the algorithm each vertex can appear at most $\log(n)$ times in the list $K$.
	This is by the fact that $K$ only contains SCCs that are at most half the size of the original set $S$.
	We obtain a running time bound
	of $O(n (\log n)^2)$ for Lines~\ref{concur:alg:goodcomp:for1}-\ref{concur:alg:goodcomp:sort}.\\
	Consider the second for-loop at Line~\ref{concur:alg:goodcomp:for2}:
	Let $|X| = n_1$. The operations \Remove{$\cdot$} and \Construct{$\cdot$} are called once per
	found SCC $G[X]$ with $X \neq S$ and take by
	Lemma~\ref{concur:lem:Streett_ds} $O(|X| + \bits(X))$ time. Whenever a vertex is in
	$X$, the size of the set in $Q$ containing $v$ originally is reduced by at least a factor of two due to the fact
	that $|X| = n_1 \leq n'/2$. This happens at most $\lceil \log n \rceil$
	times. By charging $O(1)$ to the vertices in $X$ and, respectively,
	to $\bits(X)$, the total running time for
	Lines~\ref{concur:alg:goodcomp:ds1} \&~\ref{concur:alg:goodcomp:ds2} can be bounded by $O((n +
	b) \log n)$ as each vertex and bit is only charged $O(\log n)$ times.
	Combining all parts yields the claimed running time bound of $O(\decrSCC + b\log n + n \log^2n) =
	\O(m +b)$. 
\end{proof}

The correctness of the algorithm is similar to the analysis given
in~\cite[Lemmas~3.6 \& 3.7]{CHL17} except that we additionally have to prove that
\sccs{$S$} holds the SCCs of $G[S]$. Lemma~\ref{concur:lem:goodcomp:ds} shows that we
maintain \sccs{$S$} properly for all the data structures in $Q$.

\begin{lemma}\label{concur:lem:goodcomp:ds}
	After each iteration of the outer while-loop every non-trivial SCC of the current
	graph is contained in one of the subgraphs $G[S]$ for which the data structure $D(S)$ is
	maintained in $Q$ and \sccs{$S$} stores a list of all SCCs contained in $S$.
\end{lemma}
\begin{proof}
	Initially, \sccs{$V$} stores the whole input graph as one SCC\@. 
	Thus, by the assumption that the input is a strongly connected graph the claim
	is true before the first iteration of the while loop.
	Thus let us assume the statement is true before an iteration of the loop. 
	We will show that then it also holds after the iteration.
	First, the SCCs of a set $S$ stored in the data structure are only modified when we
	remove bad vertices at Lines~\ref{concur:alg:goodcomp:removebad1}-\ref{concur:alg:goodcomp:removebfromscc}.
	Consider such bad vertex $b \in B$: If new SCCs are created due to the deletion of $b$, 
	the representative SCC of $b$ splits into multiple SCCs and 
	\sccs{$S$} is correctly updated in Lines~\ref{concur:alg:goodcomp:removebfromscc} and~\ref{concur:alg:goodcomp:dannounce}.
	Moreover, for each new non-trivial SCC $X$ that we remove from $S$, the algorithm adds a new set 
	to $Q$ in Lines~\ref{concur:alg:goodcomp:ds1} \&~\ref{concur:alg:goodcomp:ds2}. Also, we update the
	corresponding set \sccs{$X$} at Line~\ref{concur:alg:goodcomp:newSmallSCC} and the old set \sccs{$S$} at
	Line~\ref{concur:alg:goodcomp:remSmallSCC}. 
\end{proof}

We prove the next Lemma by showing that we never remove edges of vertices of
good components.
\begin{lemma}\label{concur:lem:goodcomp:maintaingood}
	After each iteration of the outer while-loop every good component
	of the input graph is contained in one of the subgraphs $G[S]$ for which the
	data structure $D(S)$ is maintained in the list $Q$.
\end{lemma}
\begin{proof}
	We first show that Algorithm~\ref{concur:alg:goodcomp} never removes edges or vertices that belong to
	a good component.   
	Consider an arbitrary iteration of the outer while loop at Line~\ref{concur:alg:goodcomp:outerwhile}.
	Let $D(S)$ be the data structure pulled from the list $Q$ in that iteration.
	Edges are only removed in Line~\ref{concur:alg:goodcomp:dannounce} of the algorithm 
	where bad vertices and incident edges are removed from $S$.
	As bad vertices cannot be part of any good component in $S$
	the algorithm does not delete any edge in a good component.
	That is, every good component of the initial graph stays strongly connected 
	in the modified graph during the whole algorithm and the claim follows from 
	Lemma~\ref{concur:lem:goodcomp:ds}.
\end{proof}

\begin{proposition}\label{concur:lem:goodcomp:correctness}
	Algorithm~\ref{concur:alg:goodcomp} outputs a good component if one exists,
	otherwise the algorithm reports that no such component exists.
\end{proposition}
\begin{proof}
	First consider the case where Algorithm~\ref{concur:alg:goodcomp} outputs a subgraph $G[S]$.
	We show that $G[S]$ is a good component: 
	Line~\ref{concur:alg:goodcomp:removetrivial} ensures only non-trivial SCSs are considered. 
	After the removal of bad vertices from $S$ in Lines~\ref{concur:alg:goodcomp:removebad1}-\ref{concur:alg:goodcomp:removebad2}, we
	know that for all $1 \leq j \leq k$ that $U_j \cap S \neq \emptyset$ if
	$S \cap L_j \neq \emptyset$. Due to Line~\ref{concur:alg:goodcomp:outputgoodcomp} there
	is only one SCC in $G[S]$ and thus $G[S]$ is a good component.
	Second, if Algorithm~\ref{concur:alg:goodcomp} terminates without a good component, by
	Lemma~\ref{concur:lem:goodcomp:maintaingood}, we have that the initial graph has 
	no good component and thus the result is correct as well.
\end{proof}

The running time bounds for the decremental SCC algorithm of~\cite{BPW19} (cf. Corollary~\ref{concur:cor:returnnewsccs}) only
hold against an oblivious adversary. Thus we have to show that in our algorithm the sequence of edge deletions does not depend on the random choices of the decremental SCC algorithm. 
The key observation is that only the order of the computed SCCs depends on the random choices
of the decremental SCC and we eliminate this effect by sorting the SCCs.
\begin{proposition}\label{concur:prop:goodcomp:obadv}
	The sequence of deleted edges does not depend on the random choices of the decremental SCC
	Algorithm but only on the given instance.
\end{proposition}
\begin{proof}
	Note that we only delete edges at Line~\ref{concur:alg:goodcomp:dannounce}.
	We prove that the claim holds for every iteration of the while loop:
	Initially, there is only one element in $Q$ at
	Line~\ref{concur:alg:goodcomp:outerwhile}. 
	Thus the claim holds initially.
	Assume that the claim holds before the while-loop.
	Thus, the $S$ we pull from $Q$ does not depend on the random choices of the
	decremental SCC algorithm $\A$. 
	Note, that when we remove vertices at Line~\ref{concur:alg:goodcomp:dannounce}
	the newly created SCCs are returned to $K$ are determined by the input instance and do not
	depend on the random choices of the algorithm.
	Only \emph{the order} in which the SCCs are returned 
	depends on the random choices of the decremental SCC algorithm $\A$. 
	When we look at the set $S$ of vertices which contains all of the new SCCs at 
	Line~\ref{concur:alg:goodcomp:for1}-\ref{concur:alg:goodcomp:ds1} we add all SCCs with
	at most $|S|/2$ vertices to a list. We sort this list by the minimum
	vertex id of each SCC and add them according to this order to our list data structure
	$Q$ (if there is an SCC with more than $|S|/2$ vertices it is at the
	end of $Q$). Thus, again the order in which SCCs are dequeued from $Q$ is
	fixed and does not depend on the random choices of $\A$. 
\end{proof}

Due to Lemma~\ref{concur:lem:goodcomp:runningtime}, Lemma~\ref{concur:lem:goodcomp:correctness} and
Proposition~\ref{concur:prop:goodcomp:obadv} we obtain the following result.
\begin{theorem}
	In a graph, the winning set for a $k$-pair Streett objective can be computed in $\O(m + b)$ expected time.  
\end{theorem}

\section{Algorithms for MDPs}\label{concur:sec:algformdps}

In this section, we present expected near-linear time algorithms for
computing a MEC decomposition, 
deciding almost-sure reachability and
maintaining a MEC decomposition in a decremental setting. 
In the last section, we present an algorithm for \emph{MDPs} with Streett objectives by using the new algorithm
for the decremental MEC decomposition.


\subsection{Maximal End-Component Decomposition}\label{concur:sec:mec}
In this section, we present an expected near linear time algorithm for MEC
decomposition. 
Our algorithm is an efficient implementation of the static algorithm presented in~\cite[p.
29]{CH14}: The difference is that the bottom SCCs are computed with a dynamic SCC algorithm instead
of recomputing the static SCC algorithm.
A similar algorithm was independently proposed in an unpublished extended version
of~\cite{CHILP16}.

\para{Algorithm Description.}
The MEC algorithm described in Algorithm~\ref{concur:alg:mec_algo} repeatedly
removes bottom SCCs and the corresponding random attractor. 
After removing bottom SCCs the new SCC decomposition with its bottom SCCs is computed using a dynamic SCC algorithm.

\begin{algorithm}[ht]
	\caption{MEC Algorithm}\label{concur:alg:mec_algo}
	\small
	\KwIn{MDP $P = (V, E, \langle V_1, V_R \rangle, \delta)$, decremental SCC algorithm $\A$}
	Invoke an instance $\A$ of the decremental SCC algorithm\;
	Compute the SCC-decomposition of $G = (V,E)$: $C = \{C_1, \dots, C_\ell\}$ \;
	Let $M = \emptyset$; $Q \gets \{ C_i \in C \mid \text{ $C_i$ has no outgoing edges} \}$\;
	\While{$Q$ is not empty}{
		$C \gets \emptyset$\;
		\lFor{$C_k \in Q$}{
			$C \gets C \cup C_k$; $M \gets M \cup \{C_k\}$ 
		}		
		$A \gets \attr{R}{C}{P}$\;\label{concur:alg:mec:attrcomp}
	$Q \gets \A.\deleteannouncenooutgoing{$\edgeset{A}$}$\tcp{remove $A$ from $P$}\;\label{concur:alg:mec:removeattr}}
	\Return{$M$}\;
\end{algorithm}

Correctness follows because our algorithm just removes attractors of bottom SCCs
and marks bottom SCCs as MECs\@. This is
precisely the second static algorithm presented in~\cite[p. 29]{CH14} except that the bottom SCCs
are computed using a dynamic data structure. 
By using the decremental SCC algorithm described in
Subsection~\ref{concur:subsec:decsccs} we obtain the following lemma.
\begin{lemma}\label{concur:lem:mec:running}
	Algorithm~\ref{concur:alg:mec_algo} returns the MEC-decomposition of an MDP $P$ in expected time $\O(m)$. 
\end{lemma}
\begin{proof}
	The running time of algorithm $\A$ is in total time $O(\decrSCC) = \O(m)$ by
	Theorem~\ref{concur:thm:decscc} and Corollary~\ref{concur:cor:returnnewsccs}.
	Initially, computing the SCC decomposition and determining the SCCs with no
	outgoing edges takes time $O(m+n)$ by using~\cite{Tarjan72}. 
	Each time we compute the attractor of a bottom
	SCC $C_k$ at Line~\ref{concur:alg:mec:attrcomp} we
	remove it from the graph by deleting all its edges and never
	process these edges and vertices again. Since we can compute the attractor $A$ at Line~\ref{concur:alg:mec:attrcomp} in 
	time $O(\sum_{v \in A} \In{A})$, we need $O(m+n)$ total time for computing the
	attractors of all bottom SCCs.
	Hence, the running time is dominated by the decremental SCC algorithm
	$\A$, which is $O(\decrSCC) = \O(m)$.
\end{proof}

The algorithm uses $O(m+n)$ space because the decremental
SCC algorithm $\A$ uses $O(m+n)$ space and $Q$ only contains 
vertices.
\begin{theorem}\label{concur:thm:mec}
	Given an MDP the MEC-decomposition can
	be computed in $\O(m)$ expected time. The algorithm uses $O(m + n)$ space.
\end{theorem}

Note that we can use the decremental SCC Algorithm $\A$ of~\cite{BPW19} even though this algorithm only
works against an oblivious adversary as the sequence of deleted edges does not depend on the random choices 
of the decremental SCC Algorithm.

\subsection{Almost-Sure Reachability}
In this section, we present an expected near linear-time algorithm for the
almost-sure reachability problem. 
In the almost-sure reachability problem, we are given an MDP $P$ and a target set $T$ and
we ask for which vertices player~1 has a strategy to reach $T$ almost surely, i.e., $\ASW{\reach{T}}$.
Due to~\cite[Theorem 4.1]{CDHL16} we can determine the set $\ASW{\reach{T}}$ in time $O(m + \MEC)$ where $\MEC$ is the running time of the fastest MEC algorithm. 
We use Theorem~\ref{concur:thm:mec} to compute the MEC decomposition and obtain the following theorem.

\begin{theorem}\label{concur:thm:asreach}
	We can compute $\ASW{\reach{T}}$ in	$\O(m)$ expected time.
\end{theorem}

\subsection{Decremental Maximal End-Component Decomposition}\label{concur:sec:decmec}
We present an expected near-linear time
algorithm for the MEC-decomposition which supports player-1 edge deletions and a query 
that answers if two vertices are in the same MEC\@.
We need the following lemma from~\cite{CH11} to prove the correctness of our algorithm.
Given an SCC $C$ we consider the set U of the random vertices in $C$ with edges leaving 
$C$. The lemma states that for all non-trivial MECs $X$ in $P$ the intersection with
$U$ is empty, i.e., $\attr{R}{U}{P} \cap X = \emptyset$.
\begin{lemma}[Lemma 2.1(1),~\cite{CH11}]\label{concur:lem:MECcorr}
	Let $C$ be an SCC in $P$. Let $U = \set{v \in C\cap V_R \mid E(v) \cap (V \setminus C) \neq
	\emptyset}$ be the random vertices in $C$ with edges leaving $C$. Let $Z = \attr{R}{U}{P} \cap
	C$. Then, for all non-trivial MECs $X$ in $P$ we have $Z \cap X = \emptyset$ and for any edge
	$(u,v)$ with $u \in X$ and $v \in Z$, $u$ must belong to $V_1$.
\end{lemma}

The \emph{pure MDP graph} $P^P$ of an MDP $P = (V, E, \langle V_1, V_R \rangle, \delta)$ is
the graph which contains only edges in non-trivial
MECs of $P$. More formally, the pure MDP graph $P^P$ is defined as follows:
Let $M_1, \dots M_k$ be the set of MECs of $P$.
Then we define the MDP $P^P = (V^P, E^P, \langle V_1^P, V_R^P \rangle, \delta^P)$ where
$V^P = V, V_1^P = V_1, V_R^P = V_R$, 
$E^P = \bigcup_{i= 1}^k \{(u,v) \in E \cap (M_i \times M_i)\}$ and
for each $v \in V_R$: $\delta^P(v)$ the uniform distribution over vertices $u$ with $(v,u) \in E^P$.

Throughout the algorithm, we maintain the pure MDP graph $P^P$ for an input MDP
$P$. Note that every non-trivial SCC in $P^P$ is also a MEC due to the fact that there are only edges inside
of MECs. Moreover, a trivial SCC $\{v\}$ is a MEC iff $v \in V_1$.
Note furthermore that when a player-1 edge of an MDP $P$ is
deleted, existing MECs might split up into several MECs 
but no new vertices are added to existing MECs.

Initially, we compute the MEC-decomposition in $\O(m)$ expected time using the algorithm
described in Section~\ref{concur:sec:mec}. Then we remove every edge that is not in a
MEC\@. 
The resulting graph is the pure MDP graph $P^P$. 
Additionally, we invoke a decremental SCC algorithm $\A$ which can (1) announce new SCCs
under edge deletions and return a list of their vertices and (2) can 
answer queries that ask whether two vertices $v,u$ belong to the same SCC\@.
When an edge $(u,v)$ is deleted, we know that (i) the MEC-decomposition stays
the same or (ii) one MEC splits up into new MECs and the rest of the
decomposition stays the same.
We first check if $u$ and $v$ are in the same
MEC, i.e., if it exists in $P^P$. If not, we are done. Otherwise, $u$ and $v$ are
in the same MEC $C$ and either (1) the MEC $C$ does not split or (2) the MEC
$C$ splits. In the case of (1) the SCCs of the pure MDP graph $P^P$ remain intact and nothing needs to be done. In the case of (2) we need to identify the new SCCs $C_1, \dots, C_k$ in $P^P$ using the decremental SCC algorithm $\A$. Let, w.l.o.g., $C_1$ be the SCC with the most vertices. We iterate through every edge
of the vertices in the SCCs $C_2, \dots, C_k$. By considering all the edges, we identify 
all SCCs (including $C_1$) which are also MECs. We remove all edges $(y,z)$ where $y$ and
$z$ are not in the same SCC to maintain the pure MDP graph $P^P$. 
For the SCCs that are not MECs let $U$ be the set of random vertices with edges
leaving its SCC\@. We compute and remove $A= \attr{R}{U}{P^P}$ (these vertices belong to no MEC due to
Lemma~\ref{concur:lem:MECcorr}) and recursively start the procedure on the new SCCs
generated by the deletion of the attractor. The algorithm is illustrated in Figure~\ref{concur:fig:decr_mec}.
\begin{figure}
	\centering
	\includegraphics[width=0.9\textwidth]{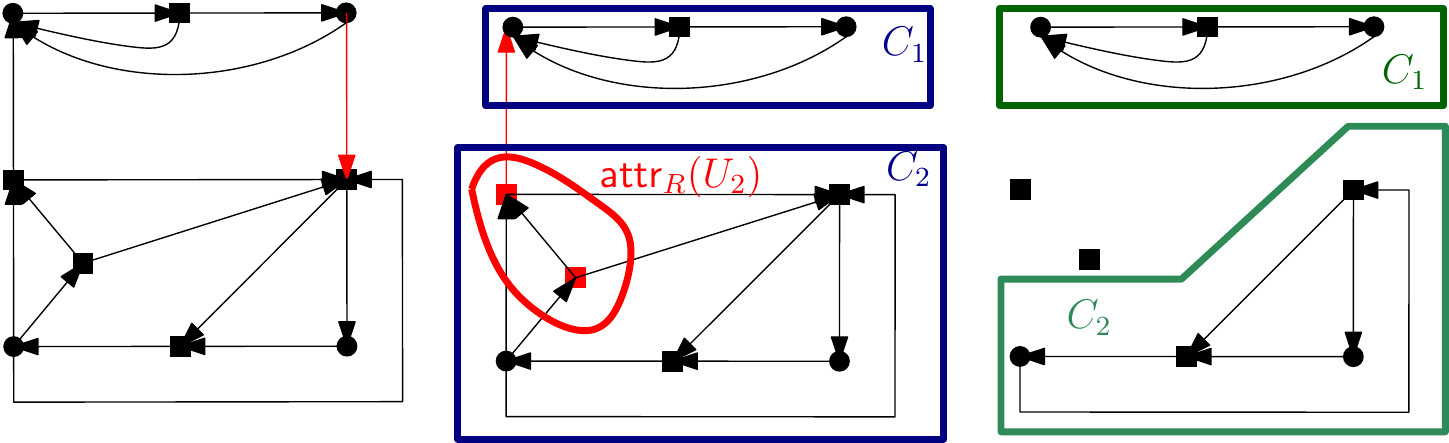}
	\caption{We delete an edge which splits the MEC into two new SCCs $C_1$ and $C_2$. The SCC $C_2$ 
	is not a MEC\@. We thus compute and remove the attractor of $U_2$ and the resulting SCC is a
MEC.}\label{concur:fig:decr_mec}
\end{figure}
\begin{algorithm}[t]
	\caption{Decremental MEC-update}\label{concur:alg:decr_mec}
	\small
	\KwIn{Player-1 Edge $e=(u,v)$}
	\If{$\A.\query{$u,v$}= \True$}{\label{concur:alg:decr_mec:checkid} 
		List $K \gets \{\A.\deleteannounce{$(u,v)$}$\}\;\label{concur:alg:decr_mec:removeedge}
		\While{$K \neq \emptyset$}{\label{concur:alg:decr_mec:newsccs} 
			pull a list $J$ of SCCs from $K$\ and let $C_1$ be the largest SCC\;
			$\{C_1, \dots C_k\} \gets$ Sort all SCCs in $J$ except $C_1$ by the smallest vertex id.\;\label{concur:alg:decr_mec:sort}
			$\MEC^{C_1} = \True$, $U_1 \gets \emptyset$\;\label{concur:alg:decr_mec:labelc1true}

			\For{$i = 2;\ i \leq k;\ i\!+\!+$}{\label{concur:alg:decr_mec:iteratesmaller}
				$\MEC^{C_i} = \True$, $U_i \gets \emptyset$\;\label{concur:alg:decr_mec:labelcitrue}

				\For{$e=(s,t)$ where $e \in
					\edgeset{C_i}$}{\label{concur:alg:decr_mec:iterateedges2k}
					\lIf{$(s \notin C_i) \lor (t \notin	C_i)$}{\label{concur:alg:decr_mec:putInD}
						$\A.\delete{$e$}$
					}

					\lIf{$(s \in V_R \land t \notin
						C_i)$}{\label{concur:alg:decr_mec:labelcifalse}
						$\MEC^{C_i} = \False$; $U_i \gets U_i \cup \set{s}$
					}

					\lIf{$(s \in V_R \land s \in
						C_1)$}{\label{concur:alg:decr_mec:labelc1false}
						$\MEC^{C_1} = \False$; $U_1 \gets U_1 \cup \set{s}$
					}
				}\label{concur:alg:decr_mec:iteratesmaller_end}
				\uIf{$\MEC^{C_i} = \False$}{\label{concur:alg:decr_mec:setMECi}
					$A \gets \attr{R}{U_i}{P^P}\cap
					C_i$\;\label{concur:alg:decr_mec:attrci}
					$J \gets \A.\deleteannounce{$\edgeset{A}$}\setminus \left(\bigcup_{a \in A}\A.\rep{$a$}\right)$\;\label{concur:alg:decr_mec:removeattrci}
					\lIf{$ J \neq \emptyset$}{
						$K \gets K \cup \{J\}$\label{concur:alg:decr_mec:addKi}
					}

				}
			}
			\If{$\MEC^{C_1} = \False$}{\label{concur:alg:decr_mec:setMEC1}
				$A \gets \attr{R}{U_1}{P^P} \cap C_1$\; \label{concur:alg:decr_mec:attrc1}
				$J \gets \A.\deleteannounce{$\edgeset{A}$}\setminus \left(\bigcup_{a \in A}\A.\rep{$a$}\right)$\;\label{concur:alg:decr_mec:removeattrc1}
				\lIf{$ J \neq \emptyset$}{
					$K \gets K \cup \{J\}$\label{concur:alg:decr_mec:addK1}
				}
			}		
		} 
	}
\end{algorithm}

\noindent Lemma~\ref{concur:lem:decr_mec:inv2} describes the key invariants of the while-loop at
Line~\ref{concur:alg:decr_mec:newsccs}. 
We prove it with a straightforward induction on the number
of iterations of the while-loop and apply Lemma~\ref{concur:lem:MECcorr}.
\begin{lemma}\label{concur:lem:decr_mec:inv2}
	Assume that $\A$ maintains the pure MDP graph $P^P$
	before the deletion of $e = (u,v)$ then the while-loop at Line~\ref{concur:alg:decr_mec:newsccs} maintains the
	following invariants: 
	\begin{enumerate}
	\item For the graph stored in $\A$  and all lists of SCCs $\{C_1, \dots, C_k\}$ in $K$ there are only edges inside the
		SCCs or between the SCCs in the list, i.e., 
		for each $(x,y) \in \bigcup_{j=0}^k E[C_j]$ we have $x,y \in \bigcup_{j=0}^k C_j$.
	\item If a non-trivial SCC of the graph in $\A$ is not a MEC of the current MDP it is in $K$.
	\item If $M$ is a MEC of the current MDP then we do not delete an edge of $M$ in the
		while-loop. 
	\end{enumerate}
\end{lemma}
\begin{proof}
	Initially, $\A$ contains the SCCs of the pure MDP graph $P^P$.
	If the deletion of $e$ does not change the SCCs the while loop is not executed and the three invariants hold.
	Thus in the remainder of the proof we consider the case where deletion of $e$ splits a SCC 
	and the newly created SCCs $C_1, \dots C_k$ are added to $K$.
%
%
	\begin{enumerate}
		\item Initially, the invariant holds because we assume that we maintain the pure MDP graph
			in $\A$. Assume the invariant holds before an iteration of the the while-loop.
			Let $C_i$ be an arbitrary SCC in the list $\{C_1, \dots, C_k\}$.
			Due to the for-loop at
			Lines~\ref{concur:alg:decr_mec:iterateedges2k}-\ref{concur:alg:decr_mec:iteratesmaller_end} we remove
			all edges going into or out of $C_i$ (Line~\ref{concur:alg:decr_mec:putInD}. Thus
			all outgoing edges of $C_i$ are removed. Note that because we
			delete edges between SCCs no new SCCs are created due to this removal.
			If $i \neq 1$ then, we remove edges inside of $C_i$ at Line~\ref{concur:alg:decr_mec:removeattrci}
			which might split up $C_i$ into new SCCs which are added to $K$. But
			notice that $C_i$ does not
			have any outgoing edges and thus the newly created SCCs have only edges inside the SCCs or
			between them which proves our claim.
			The same argument holds for $i = 1$ but now edges are removed at Line~\ref{concur:alg:decr_mec:removeattrc1}.

		\item Initially the invariant holds since we put any potential new SCCs after
			deleting the input edge $e$ into $K$ at Line~\ref{concur:alg:decr_mec:removeedge}.
			Assume the invariant holds before an iteration of the while-loop. 
			We prove that the invariant also holds after this iteration of the while-loop.  
			We only create potential new SCCs when we remove edges in $E(\attr{R}{U_i} \cap C_i)$ where $U_i = \{v \in C_i \cap V_R \mid E(v)
			\cap (V\setminus C_i) \}$  for all $1 \leq i \leq k$. The other edge deletions between
			SCCs cannot create new SCCs\@.
			Consider the removal of any $E(\attr{R}{U_i} \cap C_i)$:
			If the removal creates SCCs we put them into the set $J$ and add them to $K$
			(Lines~\ref{concur:alg:decr_mec:setMECi}-\ref{concur:alg:decr_mec:addKi} and
			Lines~\ref{concur:alg:decr_mec:setMEC1}-\ref{concur:alg:decr_mec:addK1}). 
			On the other hand, if the removal of $E(\attr{R}{U_i} \cap C_i)$ does not create new SCCs, $U_i$
			is empty and thus $C_i$ has no outgoing random edge. Thus $C_i$ is
			a MEC\@. If $K$ is now empty each
			SCC is a MEC due to the fact that the invariant is true at the end of the last iteration
			of the while-loop.
			As we do not remove elements from $K$, the condition holds for the remaining SCCs in $K$ due to the
			assumption that the invariant was true before the while-loop. 

		\item The  while loop deletes edges between SCCs at Line~\ref{concur:alg:decr_mec:putInD} and if we delete edges between
			SCCs we do not remove edges inside any MEC because a MEC is contained in the SCCs
			of the graph. If we delete attractors of random vertices with edges leaving their
			corresponding SCCs (Line~\ref{concur:alg:decr_mec:removeattrci}
			and~\ref{concur:alg:decr_mec:removeattrc1}), we do not remove edges of any nontrivial MEC due to
			Lemma~\ref{concur:lem:MECcorr}. 
	\end{enumerate}
\end{proof}

\begin{proposition}\label{concur:prop:decrMECcorrectness}
	Algorithm~\ref{concur:alg:decr_mec} maintains the pure MDP graph $P^P$ in the data structure $\A$ under
	player-1 edge deletions.
\end{proposition}
\begin{proof}
	We show that after deleting an edge using Algorithm~\ref{concur:alg:decr_mec}
	(i) every non-trivial SCC is a MEC and vice-versa, and (ii) there are no edges going from one MEC to another.
	Initially, we compute the pure MDP graph and both conditions are fulfilled.

	When we delete an edge and the while-loop at Line~\ref{concur:alg:decr_mec:newsccs}
	terminates (i) is true due to Lemma~\ref{concur:lem:decr_mec:inv2}(2,3). 
	That is, as we never delete edges within MECs they are still strongly
	connected  and
	when the while-loop terminates, $K = \emptyset$ which means that all SCCs
	are MECs. 	

	For (ii) notice that each SCC is once processed as a List $J$. 
	Consider an arbitrary SCC $C_i$ and the corresponding list of SCCs $J = \{C_1, \dots, C_k\}$ of the
	iteration in which $C_i$ was identified as a MEC\@. By
	Lemma~\ref{concur:lem:decr_mec:inv2}(1) there are no edges to SCCs not in the list. 
	Additionally, due to Line~\ref{concur:alg:decr_mec:putInD} we remove all edges from $C_i$ to other SCCs
	in $J$.
\end{proof}

Now that we maintain the pure MDP graph $P^P$ in $\A$,
we can answer MEC queries of the form:
$\query{$u,v$}$: Returns whether $u$ and $v$ are in the same MEC in $P$, by an SCC query
$\A.\query{$u,v$}$ on the pure MDP graph $P^P$. 

The key idea for the running time of Algorithm~\ref{concur:alg:decr_mec} is that
we do not look at edges of the largest SCCs but the new SCC decomposition by inspecting the edges of
the smaller SCCs. Note that we identify the largest SCC by processing the SCCs in a lockstep manner. This can only happen $\lceil \log n
\rceil$ times for each edge. Additionally, when we sort the SCCs, we only
look at the vertex ids of the smaller SCCs and when we charge this cost to the
vertices we need $O(n \log^2 n)$ additional time.

\begin{proposition}\label{concur:prop:decrMECrunning}
	\begin{sloppypar}
		Algorithm~\ref{concur:alg:decr_mec} maintains the MEC-decomposition of	$P$ under
		player-1 edge deletions in expected total time $\O(m)$. Algorithm~\ref{concur:alg:decr_mec} answers queries
		that ask whether two vertices $v,u$ belong to the same MEC in $O(1)$. The
		algorithm uses $O(m + n)$ space.
	\end{sloppypar}
\end{proposition}
\begin{proof}
	Initialization, i.e., the initial MEC-decomposition can be done in $\O(m)$ expected time as shown
	in Section~\ref{concur:sec:mec}. 
	We will prove that the while loop at Line~\ref{concur:alg:decr_mec:newsccs} can
	also be computed in total time $O(\decrSCC + m \log^2 n)$. 
	When the algorithm detects new SCCs, say $C_1,C_2, \dots, C_k$,  we consider all vertices of the new SCCs
	except the vertices in the largest SCC, i.e., w.l.o.g., $C_1$. Considering
	all the edges of the vertices in $C_i$ for $i \in [2,k]$ we identify all the
	SCCs which are also MECs and the set $U$, i.e., the set of random vertices
	with edges leaving $C_i$. This can be achieved with one pass through the
	edges in $C_i$. 
	When an SCC $C$ falls apart, an edge can be in the part with 
	at most $|C|/2$ vertices only $\lceil \log n \rceil$ times.
	As this event can happen only $\lceil \log n \rceil$ 
	times for each edge and we only perform a constant amount of work
	for each such edge, we obtain a running time of $O(m \log n)$. 
	Also, we sort the newly creates SCCs at Line~\ref{concur:alg:decr_mec:sort}. As we do this only for the smaller SCCs,
	we can again bound the total running time by $n \log^2 n$: Every vertex can
	only be in the smaller SCC $\lceil \log n \rceil$ times and sorting the SCCs costs
	$O(1 + \log n)$ (we need to go through the SCCs except for the largest and check for the
	smallest vertex id) for each vertex which yields a total running time of $O(n
	\log^2 n)$.

	Computing the attractor at Line~\ref{concur:alg:decr_mec:attrci} or
	Line~\ref{concur:alg:decr_mec:attrc1} only costs total $O(m)$ time as we either delete the edges of the vertices in
	the attractor or, for the edges we do not delete, we charge the constant amount of work to the edges in the
	smaller half of the $C/2$ vertices.	This results in a running time of $O(\decrSCC + m)$ or
	expected running time $\O(m)$.
	Finally, the query operation can be implemented by $\A.\query{$u,v$}$ and takes time $O(1)$. 
\end{proof}

Due to the fact that the decremental SCC algorithm we use in Corollary~\ref{concur:cor:returnnewsccs} only
works for an oblivious adversary, we prove the following proposition. The key
idea is that we sort SCCs returned by the decremental SCC Algorithm. Thus, 
the order in which new SCCs are returned does only depend on the given instance.
\begin{proposition}\label{concur:prop:mec_decr_oblivious}
	The sequence of deleted edges does not depend on the random choices of the
	decremental SCC Algorithm but only on the given instance.
\end{proposition}
\begin{proof}
	Note that we only delete edges at
	Lines~\ref{concur:alg:decr_mec:removeedge},~\ref{concur:alg:decr_mec:putInD},~\ref{concur:alg:decr_mec:removeattrci}
	and~\ref{concur:alg:decr_mec:removeattrc1}.
	The initial edge deletion at Line~\ref{concur:alg:decr_mec:removeedge} does not
	depend on the decremental SCC algorithm.

	We will prove that the claim holds for every iteration of the while-loop at
	Line~\ref{concur:alg:decr_mec:newsccs}:
	Initially, there is only one element in $K$. Thus the claim holds initially.

	Assume that the claim holds before the while-loop.
	Thus, the list of SCCs we pull from $K$ does not depend on the random choices of the
	decremental SCC algorithm $\A$. 
	The edges we remove at Line~\ref{concur:alg:decr_mec:putInD} does not create new
	SCCs because they are not in any SCCs. 
	Note, that when we remove vertices at Lines~\ref{concur:alg:decr_mec:removeattrci}
	and~\ref{concur:alg:decr_mec:removeattrc1}
	the newly created SCCs are returned to $J$ in a order depending on the
	random choices of the decremental SCC algorithm $\A$. 
	Thus, we sort the SCCs in $J$ by the minimum
	vertex id and add them according to this order to our list data structure
	$K$ (the largest is put at the end of the list). 
	Thus, again, the order in which SCCs are processed from $Q$ is
	fixed and does not depend on the random choices of $\A$. 
\end{proof}

The algorithm presented in~\cite{BPW19} fulfills all the conditions of
Proposition~\ref{concur:prop:decrMECrunning} due to Corollary~\ref{concur:cor:returnnewsccs}. Therefore we obtain the following theorem due to
Proposition~\ref{concur:prop:decrMECcorrectness} and
Proposition~\ref{concur:prop:decrMECrunning}.

\begin{theorem}\label{concur:thm:mec_decr}
	Given an MDP with $n$ vertices and $m$ edges, the MEC-decomposition can be maintained under the deletion of $O(m)$ player-1 edges
	in total expected time $\O(m)$ and we can answer queries that ask whether two vertices $v,u$ belong to the same MEC in $O(1)$ time. 
	The	algorithm uses $O(m + n)$ space. The bound holds against an oblivious
	adversary.
\end{theorem}

\subsection{MDPs with Streett Objectives}

Similar to graphs we compute the winning region of Streett objectives with $k$ pairs $(L_i,U_i)$ (for $1
\leq i \leq k$) for an MDP $P$ as follows: 
\begin{enumerate}
\item We compute the MEC-decomposition of $P$.
\item For each MEC, we find good end-components, i.e., end-components where
	$L_i \cap X = \emptyset$ or $U_i \cap X \neq \emptyset$ for all $1 \leq i \leq k$ and label the MEC as satisfying.
\item We output the set of vertices that can almost-surely reach a satisfying MECs.
\end{enumerate}
For 2., we find good \emph{end-components} similar to how we find good
components as in Section~\ref{concur:sec:streettgraph}.
The key idea is to use the decremental MEC-Algorithm described in
Section~\ref{concur:sec:decmec} instead of the decremental SCC Algorithm. 
We modify the Algorithm presented in Section~\ref{concur:sec:streettgraph} as follows
to detect good end-components:
First, we use the decremental MEC-algorithm instead of the decremental SCC Algorithm.
Towards this goal, we augment the decremental MEC-algorithm with a function to return a list
of references to the new MECs when we delete a set of edges. 
Second, the decremental MEC-algorithm does not allow the deletion of arbitrary edges, but only player-1 edges. 
To overcome this obstacle, we create an equivalent instance where we remove player-1 edges when we remove `bad' vertices. 
\begin{lemma}\label{concur:cor:returnnewmecs}
	Given an MDP $P = (V,E,\langle V_1, V_R \rangle, \delta)$ with $m$ edges and $n$ vertices, we can maintain a
	data structure that supports the operation
	\begin{itemize}
	\item $\deleteannounce{$E$}$: Deletes the set of $E$ of player-1 edges $(u,v)$ from
		the MDP $P$. If the edge deletion creates new MECs $C_1, \dots, C_k$ the
		operation returns a list $Q = \{ C_1, \dots, C_k \}$ of references
		to the new non-trivial MECs.
	\end{itemize}
	in total expected update time $\O(m)$. The bound holds against an oblivious adaptive adversary.
\end{lemma}
\begin{proof}
	The data structure $\A$ maintained in Section~\ref{concur:sec:decmec} can be extended to
	support the modified delete operation $\deleteannounce{$E$}$.
	To this end, we maintain a list $\L$ of references to the SCCs in $P^P$ that correspond to the
	new MECs in the current MDP\@.
	Whenever we identify a new MEC in Line~\ref{concur:alg:decr_mec:setMECi} or
	Line~\ref{concur:alg:decr_mec:setMEC1} in Algorithm~\ref{concur:alg:decr_mec}.
	we test whether it is non-trivial and add a reference to that MEC to the list $\L$. 
	Whenever a MEC in the List splits due to another edge update we remove the old identifier from the list
	and add the new non-trivial MECs (if any).
\end{proof}
\para{Deleting bad vertices.}
As the decremental MEC-algorithm only allows deletion of player-1 edges, we first modify the original instance $P = (V,E,\langle V_1,V_R \rangle, \delta)$ to a new instance $P' =(V',E',\langle V_1',V_R' \rangle, \delta')$ such that we can remove bad vertices by
deleting player-1 edges only.
In $P'$ each vertex $v \in V_x$ for $x \in \{1, R\}$
is split into two vertices $v_{in} \in V_1'$ and $v_{out} \in V_x'$ such that 
$E'=\{ (u_{out},v_{in}) \mid (u,v) \in E\} \cup \{(v_{in}, v_{out})\mid v \in V\}$
and $L_i' = \{ v_{in}  \in V' \mid v \in L_i \}$ and $U_i' = \{ v_{out} \in V' \mid v \in U_i \}$ for all $1 \leq i \leq k$.  
The new probability distribution is $\delta'(v_{out})[w_{in}] = \delta(v)[w]$ for $v \in V_R$ and
$w \in \Out{v}$. 
Note that for each $v \in V_R$ the corresponding vertex $v_{out} \in
V_R' $ has the same probabilities to reach the representation $v_{out}$ of a vertex as $v$.
The described reduction allows us to remove bad vertices from MECs by removing the
player-1 edge $(v_{in},v_{out})$. 

The key idea for the following lemma is that for each original vertex $v \in V$
either both $v_{in}$ and $v_{out}$ are part of a good end-component or none of
them. 
Note that the only way that $v_{in}$ and $v_{out}$ are strongly connected is when the other
vertex is also in the strongly connected component because $v_{in}$ ($v_{out}$) has only one
outgoing (incoming) edges to $v_{out}$ (from $v_{in}$).

\begin{lemma}\label{concur:lem:streetmec:modifiedinstance}
	There is a good end-component in the modified instance $P'$ iff there is a good
	component in the original instance $P$.
\end{lemma}
\begin{proof}
We show the "if" and the "only if" parts separately.
	\begin{itemize}
		\item (if) Assume there is a good end-component $X'$ in the modified instance $P'$. Thus, 
			$L_i' \cap X = \emptyset$ or $U_i' \cap X' \neq \emptyset$ for all $1
			\leq i \leq k$ and $X'$ is strongly connected. Note that for each vertex $v_{in}$ in
			$X'$, $v_{out}$ must also be in $X'$
			(and vice versa) as	otherwise $X'$ is not strongly connected.

			Let $X = \{v \in V \mid v_{in} \in X' \}$. We prove it is a good end-component in $P$.
			First we prove that $X$ is strongly connected.
			Let $v, u \in X$ be arbitrary. Note that $v_{in} \in X'$. There
			is a path from $v_{in}$ to $u_{in}$ because $X'$ is strongly connected. Thus there is a path
			from $v$ to $u$ because $v$ posesses all incoming and outgoing edges of $v_{in}$ and
			$v_{out}$ respectively (the same holds for $u$ and $u_{in},u_{out}$ and the respective vertices on the
			path). Moreover, there is no random edge out of $X$ as the random edges in $X'$ only go to
			vertices $v_{in}$ in $X'$ and we include each corresponding $v$ in $X$.
			Thus $X$ is a end-component in $P$.

			It remains to show that $X$ is also a \emph{good} end-component.
			Let $i \in [1,k]$. In the first case, $L_i' \cap X' = \emptyset$. But then, by the
			definition of $L_i'$, we know that $L_i \cap X = \emptyset$. In the second case, $U_i'
			\cap X' \neq \emptyset$, i.e., there is some $v_{out} \in U_i'$ in $X'$. But then the
			corresponding vertex $v \in U_i$ is in $X$ because
			$v_{in} \in X'$ due to the fact that $X'$ is otherwise not strongly connected ($v_{out}$
			has only one incoming edge: $v_{in}$).
			Hence, $X$ is a good end-component in $P$.

		\item (only if) Assume there is a good component $X$ in the original instance $P$. Thus, 
			$L_i \cap X = \emptyset$ or $U_i \cap X' \neq \emptyset$ for all $1
			\leq i \leq k$. 
			Let $X' = \{v_{in} \in V'  : v \in X \} \cup \{ v_{out} \in V' \mid v \in X \}$.
			We first prove that $X'$ is strongly connected: Let $v_j, u_\ell \in X'$ for $j,\ell \in
			\{\mathit{in},\mathit{out}\}$ be arbitrary.
			By the definition of $X'$, $v,u \in X$ and there is a path from $v$ to $u$ in $X$. 	
			But then there is a path from $v_{in}$ to $u_{out}$ due to the fact that for each $z$ on the
			path from $v$ to $u$ each $z_{in}$ ($z_{out}$) on this
			path contains all incoming edges of $z$ (outgoing edges of $z$) and $z_{in},z_{out} \in E'$ which proves the claim.
			Moreover, there is no random edge out of $X'$ the random edges in $X$ only go to vertices $v$ in $X$
			and we include $v_{in}$ in $X'$.
			Thus $X'$ is a end-component in $P'$.
			We prove it is a good end-component in $P'$.
			Let $i \in [1,k]$. In the first case, $L_i \cap X = \emptyset$. But then, by the
			definition of $L_i'$, we know that $L_i' \cap X = \emptyset$. In the second case, $U_i
			\cap X \neq \emptyset$, i.e., there is some $v \in U_i$ in $X$. Then the
			corresponding vertex $v_{out} \in U_i'$ is in $X'$ by the definition of $X'$.\qedhere
	\end{itemize}
\end{proof}

\noindent On the modified instance $P'$ the algorithm for MDPs is identical to Algorithm~\ref{concur:alg:goodcomp} except that
we use a dynamic MEC algorithm instead of a dynamic SCC algorithm.

\begin{theorem}\label{concur:thm:mdp_Streett}
	In an MDP the winning set for a $k$-pair Streett objectives can be computed in $\O(m + b)$ expected time. 
\end{theorem}

	\chapter[Faster Algorithms for Bounded Liveness in Graphs and Game Graphs][Faster Algs.\@ for
	B.\@ Liveness in Graphs \& Game Graphs]{Faster Algorithms for Bounded Liveness in Graphs and Game Graphs}\label{cha:icalp}
	In this chapter, we consider algorithms for computing the winning set of bounded B\"uchi objectives in graphs and MDPs.

\section{Introduction}

\para{Graphs and games on graphs.}
Graphs and two-player games played on graphs provide a general mathematical
framework for a wide range of problems in computer science:
in particular, for the analysis of reactive systems, where the 
vertices of the graph represent the states of a reactive system and 
the edges represent the transitions between the states.
The classical synthesis problem (the problem of Church) asks for the construction of a
winning strategy in a game played on the graph~\cite{Church62,Rab69-TAMS,PnueliR89} 
and the fundamental model-checking problem is an algorithmic graph problem~\cite{ModelCheckingBook}.

\para{Omega-regular specifications: strength and weakness.}
In the analysis of reactive systems, the desired temporal properties that the system should
satisfy constitute the specification.
The class of $\omega$-regular languages provides a robust specification formalism~\cite{MannaP92,PnueliR89}.
Every $\omega$-regular objective can be decomposed into a safety part and a liveness part~\cite{AlpernS85}. 
The safety part ensures that the system will not do anything ``bad'' 
(such as violating an invariant) within any finite number of transitions.
The liveness part ensures that the system will do something ``good'' 
(such as proceed or respond) in the long-run. 
Liveness can be violated only in the limit, by infinite sequences of transitions, 
as no bound is specified on when a ``good'' event must happen.
This infinitary formulation has several strengths, such as robustness and
simplicity~\cite{MannaP92,Thomas97}. 
However, there is also  a weakness of the classical definition of liveness: it
can be satisfied by systems that are unsatisfactory because no bound
can be put between the occurrence of desired events.

\para{Stronger notion of liveness.}
For the weakness of the infinitary formulation of liveness, alternative and stronger 
formulations of liveness have been proposed.
The first formulation is \emph{bounded liveness} which ensures, given a bound $d$, that eventually,
good events happen within $d$ transitions. The second formulation is 
\emph{finitary liveness} which requires the existence of a bound such that, eventually,
good events happen within the bound. 
Finitary liveness was proposed in~\cite{AlurH98} 
 and has been widely studied; e.g.,
games on graphs with finitary $\omega$-regular objectives~\cite{CHH09},
and logics such as PromptLTL based on finitary liveness~\cite{KPV09}.
The notion of bounded liveness has also been investigated in many contexts, such as
MSO with bounding quantifiers~\cite{BC06}, bounded model-checking~\cite{BCCSZ03}, and ``bounded 
until'' in logics such as RTCTL~\cite{EmersonMSS92}.

\para{Algorithmic questions for bounded liveness.}
In this work, we consider graphs and games on graphs with bounded liveness objectives.
Consider a graph with $n$ vertices, $m$ edges, and a bounded liveness objective
with bound $d$.
A basic algorithmic approach is to reduce the bounded liveness objective to a 
liveness objective on a larger graph (that we call the auxiliary graph) that explicitly keeps track of the
number of transitions since the last good event. 
This basic approach yields the following bounds: 
(a)~an $O(dm)$-time algorithm for graphs (applying the linear-time algorithm for liveness
objectives on graphs),
and 
(b)~an $O(n^2 d^2)$-time algorithm for games on graphs (applying the current best-known 
$O(n^2)$-time algorithm for games on graphs with liveness objectives~\cite{CH14}).
A fundamental algorithmic question is whether the above bounds can be improved.

\para{Our contributions.} 
In this work, our main contributions are improved algorithmic bounds for 
bounded liveness on graphs and games on graphs.
\begin{itemize}
\item In graphs, there are two relevant semantics: (a)~an existential semantic that asks whether
there exists a path to satisfy the objective, and (b)~a universal semantic that asks whether
all paths satisfy the objective. The answer to the universal semantics with bounded liveness is
``Yes''
if and only if the answer is ``No'' for existential semantics with the complementary bounded coliveness objective.
We consider graphs with the existential semantics and bounded liveness and bounded coliveness
objectives.
For bounded liveness objectives, all previous algorithmic approaches yield an
$O(n^3)$ worst-case time-bound (where $d=O(n)$) and we present a randomized algorithm with one-sided error
whose worst-case time-bound is $O(n^{2.5} \log n)$.
For bounded coliveness objectives, we present a deterministic linear-time algorithm.

\item For games on graphs with bounded liveness objectives, we present an $O(n^2 d)$-time 
algorithm that improves the previous $O(n^2 d^2)$-time algorithm.

\end{itemize}

\para{Significance of the contributions.}  On the technical front, it is
threefold.
\begin{enumerate}
  \item To break the $O(n^3)$-time barrier for graphs, we exploit randomization
  to estimate for all pairs of good events how far they are from each other.
  Using this information along with a suitably modified auxiliary graph results
  in the faster $O(n^{2.5} \log n)$-time algorithm.
\end{enumerate}
  To get the improved time bound of $O(n^2 d)$ for game graphs:
\begin{enumerate}
  \setcounter{enumi}{1}
  \item we construct an auxiliary game graph (similar to the graph case) and make a crucial observation that this game graph after
  each iteration has a lot of structure, a property we call \emph{induced
    symmetry};
  \item we strategically introduce as many ``layover'' vertices as there are
  good events; in combination with induced symmetry, this enables us to prove that a significant chunk of the auxiliary
  game graph is deleted after each iteration.
\end{enumerate}
Furthermore, there are several important implications of our contributions.
First, for graphs with bounded liveness objectives, the previous worst-case time-bound
is $O(n^3)$. 
In recent years, many such algorithmic problems with $O(n^3)$ bound are conditionally optimal with a reduction from classical problems such as 
BMM (boolean matrix multiplication)~\cite{AbboudW14,AbboudWY18,ChatterjeeDHL16b,CDHL16,CDHS18ICAPS,VW2018survey}.
Our new algorithm breaks the $O(n^3)$ barrier and shows that such conditional 
lower bound approaches do not apply for bounded liveness in graphs.
Second, for graphs with bounded coliveness objectives our linear-time bound shows
that there is a very efficient algorithm for the complement of the bounded liveness
objectives.
Finally, we show that the basic algorithmic approach for games on graphs 
can also be improved.
Given our results improve the bounds for graphs and games on graphs with bounded liveness
objectives, there are several interesting questions for future work.
Whether the bounds can be further improved or a deterministic sub-cubic time algorithm 
can be obtained for graphs with bounded liveness objectives are the most interesting 
algorithmic open questions.
\section{Algorithms for Graphs}
Graphs are a special case of game graphs with
$V_2=\emptyset$.  Hereon, we will call this ``the graph case'' as
opposed to ``the game graph case'' (where $V_1 \neq \emptyset$ and
$V_2 \neq \emptyset$).  The objectives we consider are \emph{prefix
  independent}, i.e., if $\omega \in \Omega$, then any play obtained by adding
or removing a finite prefix to or from $\omega$ is also in $\Omega$.  Hence,
with respect to computing winning vertices, it is enough to focus on strongly
connected graphs. The reasoning is as follows.

In the input graph, we call a strongly connected component (SCC) $S$ \emph{good} if the
graph restricted to $S$ has a winning vertex. Due to prefix independence, all
vertices in a good SCC and those from which you can reach a good SCC are
winning.  We will prove that such vertices are exactly the winning vertices, and that
this set can be computed by the following procedure:
\begin{itemize}
  \item Compute the SCCs of the input graph (can be done in linear time~\cite{Tarjan72}).
  \item Determine for each SCC if it is good (this step depends on the objective).
  \item Consider the set of all vertices belonging to a good SCC.\@ Perform
  reachability to this set.  (This can also be done in linear time.)
\end{itemize}

\begin{lemma}
  \label{lem:goodscc}
  A vertex $v$ is a winning vertex if and only if it there is path from $v$ to
  some vertex in a good SCC.
\end{lemma}
\begin{proof}
  As mentioned before, due to prefix independence, if $v$ has a path to some
  vertex in a good SCC, then it is winning.  Next, we show the converse.

  If $v$ is winning, then there is a winning play $\omega$ starting at $v$.
  Since SCCs themselves form a directed acyclic graph (DAG), $\omega$ must
  eventually enter an SCC $S$ and stay there.  Again, due to prefix
  independence, the vertices visited by $\omega$ in $S$ are also winning, i.e.,
  $S$ is a good SCC.
\end{proof}

By Lemma~\ref{lem:goodscc} and the procedure described above it, the problem of
computing the winning vertices is reduced to determining, given a
strongly-connected graph, whether there is a winning vertex or not.  More
formally, we get the following lemma.

\begin{lemma}
  \label{lem:reducetoscc}
  Let $S_1, S_2, \ldots$ be SCCs of the graph $G = (V, E)$.  When
  $V_2 = \emptyset$, i.e., in the graph case, for a prefix
  independent objective, the set of winning vertices can be computed in time
  $O(m + \sum_i t(S_i))$ time, where $m = |E|$ and $t(S_i)$ is the time required
  to compute whether $S_i$ is a good SCC or not.
\end{lemma}

In this chapter, we consider bounded B\"{u}chi and bounded coB\"{u}chi objectives.

\subsection{The Bounded B\"{u}chi Objective}
We are given a graph $G = (V, E)$, a set $B$ of B\"{u}chi vertices, and a positive
integer $d$.  A cyclic-walk in $G$ is a walk $(v_1, v_2, \ldots, v_\ell)$ such that $v_1 = v_\ell$.
We say that a cyclic-walk $C$ is \emph{feasible} if it has
at least one B{\"u}chi vertex and the number of edges in $C$
between any two consecutive B{\"u}chi vertices is at most $d$.  We assume that
$G$ is strongly connected, and our goal is to determine if there is a winning
vertex in $G$.  Then, using Lemma~\ref{lem:reducetoscc}, we generalize the
result to a graph that might not be strongly connected. 
The following lemma reduces this problem to finding a feasible cyclic-walk in $G$.

\begin{lemma}
  \label{lem:feasible}
  The strongly-connected input graph $G$ has a winning vertex with respect to
  the bounded B\"{u}chi objective if and only if it has a feasible cyclic-walk.
\end{lemma}
\begin{proof}
  If $G$ has a winning vertex, say $v$, then there is a winning play $\omega$
  that starts at $v$.  Let $\omega = \ls v_0 = v, v_1, v_2, \ldots \rs$; so by
  the definition of winning play, $\exists i \ge 1$ such that
  $\forall j \ge i: \{v_j,v_{j+1}, \dots, v_{j+d-1}\} \cap B \neq \emptyset$.
  Consider the set $\Inf{\omega}$ of vertices that appear infinitely often in
  $\omega$.  Since $\omega$ is winning, $\Inf{\omega} \cap B \neq \emptyset$.  Thus,
  we can choose a $j' \ge i$ such that $v_{j'} \in \Inf{\omega} \cap B$.  Since $v_{j'}$
  appears infinitely often, for some $j'' > j'$, we have that $v_{j''} = v_{j'}$.
  Thus $(v_{j'}, v_{j'+1}, \ldots, v_{j''} = v_{j'})$ is a feasible cyclic-walk
  because $j' \ge i$ and, as mentioned earlier,
  $\forall j \ge i: \{v_j,v_{j+1}, \dots, v_{j+d-1}\} \cap B \neq \emptyset$.

  In the other direction, if $G$ has a feasible cyclic-walk, then we can keep
  traversing it to construct a winning play, which means $G$ has a winning
  vertex.
\end{proof}

\subsubsection*{An $O(dm)$-time algorithm for bounded B\"{u}chi}

Next, we recall the basic $O(dm)$-time algorithm to determine if there is a feasible
cyclic-walk.  This algorithm tries to trace a feasible cycle by maintaining a
counter with each possible non-B{\"u}chi vertex denoting how far away we are
from the last visit to a B{\"u}chi vertex.  We construct a $(d{+}1)$-layered
auxiliary graph $G^*=(V^*, E^*)$, where
$V^* = (B \times \{0\}) \cup ((V \setminus B) \times \{1, \ldots, d\})$.
We define a more general graph here that we also use in Section~\ref{sec:games}.
We illustrate an example in
Figure~\ref{fig:layered}.  So, for $(v, \ell) \in V^*$, the integer $\ell$ corresponds
to the aforementioned counter.  We call the vertices in $B \times \{0\}$
B{\"u}chi vertices and the vertices in $(V \setminus B) \times \{1, \ldots, d\}$
non-B{\"u}chi vertices.  The edge set $E^*$ is constructed by
Algorithm~\ref{alg:lg}.  The last layer of the auxiliary graph is actually not
needed for the graph case but is needed for the game graph case later.  Observe
that the auxiliary graph is also a game graph.  (The ownership of the vertices
will be defined later in a natural way.)

\begin{algorithm}
	\caption{Construction of the auxiliary graph $G^*$ from $G$, $B$, and $d$.  It
		is easy to see that the running time of this algorithm is $O(dm)$ and $G^*$
		has at most $dm$ edges.}\label{alg:lg}
	\SetKwFunction{ConstructAuxiliaryGraph}{ConstructAuxiliaryGraph}
	\SetKwFunction{auxGraphDLayers}{AuxiliaryGraph-d-Layers}
	\Procedure{\ConstructAuxiliaryGraph{$G = (V, E), B \subseteq V$, $d$}}{
		$V^* \gets (B \times \{0\}) \cup ((V \setminus B) \times \{1, \ldots, d\})$
		and $E^* \gets \emptyset$\;
		\For{$(u, v) \in E$ such that $v \notin B$ (add counter-incrementing
			edges)}{
			\If{$u \notin B$}{
				\For{$i \in \{1, \ldots, d{-}1\}$}{
					Add $((u, i), (v, i{+}1))$ to $E^*$.\;
				}
				Add $((u, d), (v, d))$ to $E^*$ (edges in the last layer to $V\setminus B$ stay in the last layer).
			}\Else {
				Add $((u, 0), (v, 1))$ to $E^*$.
			}

		}
		\For{$(u, v) \in E$ such that $v \in B$ (add counter-resetting edges)}{
			\If{$u \notin B$}{
				\For{$i \in \{1, \ldots, d\}$}{
					Add $((u, i), (v, 0))$ to $E^*$.
				}
			}
			\Else{ 
				Add $((u, 0), (v, 0))$ to $E^*$.
			}
		}
		\Return {$G^* = (V^*, E^*)$}
	}
	\Procedure{\auxGraphDLayers{$G = (V, E), B \subseteq V$, $d$}}{

		$G^* \gets$ \ConstructAuxiliaryGraph{$G = (V, E), B \subseteq V, d$}\;
		Return the graph resulted by removing layer-$d$ from $G^*$,	called $G' = (V', E')$.\;
	}

\end{algorithm}
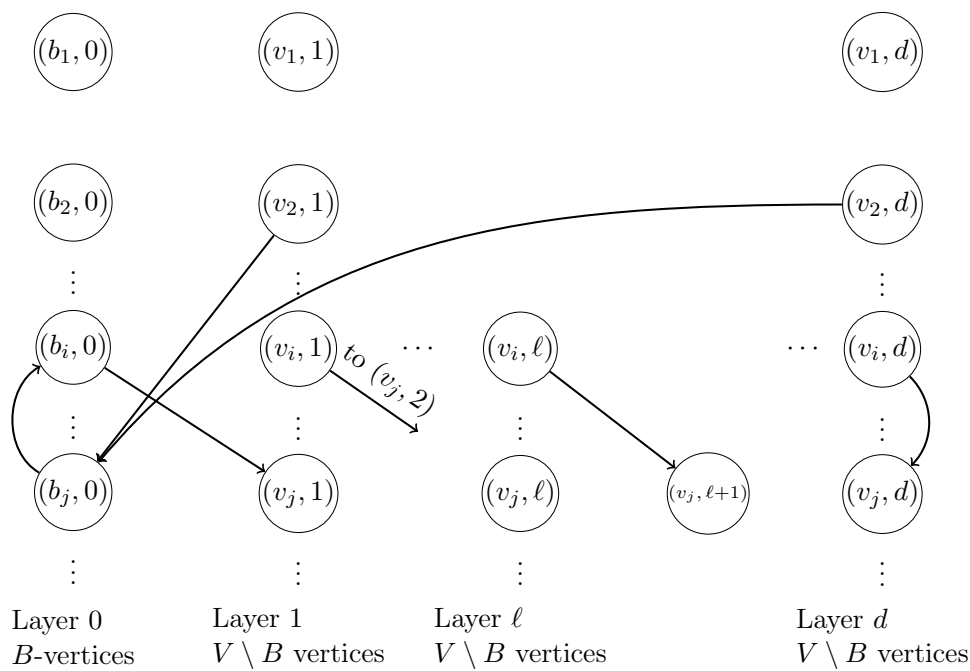
\begin{figure}[ht]
  \centering
  \resizebox{\textwidth}{!}{%
	  \begin{tikzpicture}[vert/.style={shape=circle, draw, inner sep=0pt}]
		  \node [vert] (b1) at (0,0) {$(b_1,0)$};
		  \node [vert, below=1cm of b1] (b2) {$(b_2,0)$};
		  \node [below of=b2] (vdots0) {$\vdots$};
		  \node [vert, below of=vdots0] (bi) {$(b_i,0)$};
		  \node [below of=bi] (vdots00) {$\vdots$};
		  \node [vert, below of=vdots00] (bj) {$(b_j,0)$};
		  \node [below of=bj] (vdots000) {$\vdots$};
		  \node [below of=vdots000, align=left] {Layer $0$\\$B$-vertices};

		  \node [vert,right=2cm of b1] (v11)  {$(v_1,1)$};
		  \node [vert, below=1cm of v11] (v12) {$(v_2,1)$};
		  \node [below of=v12] (vdots1) {$\vdots$};
		  \node [vert, below of=vdots1] (v1i) {$(v_i,1)$};
		  \node [below of=v1i] (vdots10) {$\vdots$};
		  \node [vert, below of=vdots10] (v1j) {$(v_j,1)$};
		  \node [below of=v1j] (vdots100) {$\vdots$};
		  \node [below of=vdots100, align=left] {Layer $1$\\$V\setminus B$ vertices};

		  \node [right=0.75cm of v1i] {$\cdots$};

		  \node [vert, right = 2cm of v1i] (vli) {$(v_i,\ell)$};
		  \node [below of=vli] (vdotsl0) {$\vdots$};
		  \node [vert, below of=vdotsl0] (vlj) {$(v_j,\ell)$};
		  \node [below of=vlj] (vdotsl00) {$\vdots$};
		  \node [below of=vdotsl00, align=left] {Layer $\ell$\\$V\setminus B$ vertices};

		  \node [vert, right = 4.5cm of v1j] (vlp1j) {\tiny $(v_j,\ell{+}1)$};

		  \node [right=3cm of vli] {$\cdots$};

		  \node [vert,right=10cm of b1] (vd1)  {$(v_1,d)$};
		  \node [vert, below=1cm of vd1] (vd2) {$(v_2,d)$};
		  \node [below of=vd2] (vdots1) {$\vdots$};
		  \node [vert, below of=vdots1] (vdi) {$(v_i,d)$};
		  \node [below of=vdi] (vdots10) {$\vdots$};
		  \node [vert, below of=vdots10] (vdj) {$(v_j,d)$};
		  \node [below of=vdj] (vdots100) {$\vdots$};
		  \node [below of=vdots100, align=left] {Layer $d$\\$V\setminus B$ vertices};

		  \draw[thick, ->] (bj) to [out=150, in=210] (bi);
		  \draw[thick, ->] (bi) -- (v1j);
		  \draw[thick, ->] (v1i) --  node[above,sloped] {to $(v_j,2)$} +(-35:2cm);
		  \draw[thick, ->] (vli) -- (vlp1j);
		  \draw[thick, ->] (v12) -- (bj);

		  \draw[thick, ->] (vdi) to [out=-45, in=45] (vdj);
		  \draw[thick, ->] (vd2) to [out=180, in=50] (bj);

	  \end{tikzpicture}
  }
  \caption{An illustration of how the auxiliary layered graph is constructed.
    If $G$ contains the edges $(b_j, b_i)$, $(b_i, v_j)$, $(v_2, b_j)$, and
    $(v_i, v_j)$, then the auxiliary layered graph $G^*$ will have shown edges.}
  \label{fig:layered}
\end{figure}


\begin{lemma}
\label{lem:rtlg}
The running time of the procedures \ConstructAuxiliaryGraph{$\cdot$} and
\auxGraphDLayers{$\cdot$} in Algorithm~\ref{alg:lg} is $O(dm)$.
\end{lemma}
\begin{proof}
  In \ConstructAuxiliaryGraph{$\cdot$}, each of the outer for loop runs for at
  most $m$ iterations, and each of the inner for loops runs for at most $d$
  iterations. \auxGraphDLayers{$\cdot$} just calls
  \ConstructAuxiliaryGraph{$\cdot$} and removes the last layer, which takes
  $O(dm)$ time.
\end{proof}

For the graph case, we are interested in $G^*$ induced on
layers-$\{0, 1, \ldots, d{-}1\}$.  Let $G'$ denote this graph.

\begin{lemma}
  \label{lem:redcyclegstar}
  The strongly-connected input graph $G$ has a feasible cyclic-walk if and only
  if $G'$ has a cycle.
\end{lemma}
\begin{proof}
  Let
  $C = (b_1,v_{1,1}, \ldots, v_{1,\ell_1}, b_2, v_{2,1}, \ldots, v_{2,\ell_2},
  b_3, \ldots, b_1)$, where each $b_i \in B$, each $v_{i,j} \in V \setminus B$,
  and each $\ell_i \le d - 1$, be a feasible cyclic-walk in $G$.  There is a
  corresponding cyclic-walk $C'$ in $G'$:
  \begin{itemize}
    \item for each $(b_i, v_{i,1}) \in C$, the edge $((b_i, 0), (v_{i,1}, 1)) \in E'$,
    \item for each $(v_{i,j}, v_{i, j{+}1}) \in C$, the edge
    $((v_{i,j}, j), (v_{i,j{+}1}, j{+}1)) \in E'$, 
    \item for each $(v_{i,\ell_j}, b_{i{+}1}) \in C$, the edge
    $((v_{i,\ell_j}, \ell_j), (b_{i{+}1}, 0)) \in E'$, and
    \item for the final edge $(v_{i,\ell_j}, b_1) \in C$, the edge
    $((v_{i,\ell_j}, \ell_j), (b_1, 0)) \in E'$.
  \end{itemize}
  If $C'$ consists of union of cycles can be short-cut to get a cycle in $G'$.

  In the other direction, consider a cycle in $G'$.  A projection of this
  cycle on the first coordinate of the vertices, by construction,
  gives a feasible cyclic-walk in $G$, because the number of edges between
  consecutive B{\"u}chi vertices is at most $d$.
\end{proof}

Thus, by Lemmas~\ref{lem:feasible}~and~\ref{lem:redcyclegstar}, we get
Algorithm~\ref{alg:bounded}.

\begin{algorithm}[H]
	\caption{This algorithm determines if the strongly-connected input graph
		has a winning vertex with respect to the bounded B\"{u}chi objective.}\label{alg:bounded}
	\SetKwFunction{boundedbuchi}{BoundedB\"{u}chi}
	\Procedure{\boundedbuchi{$G = (V, E), B \subseteq V$, $d$}}{
		$G' \gets \auxGraphDLayers{G, B, d}$\;
		Run depth-first search on $G'$ to determine if it has a cycle.\;
		\If{$G'$ has a cycle}{
			\Return{``$G$ has a winning vertex.''}
		}\Else{
			\Return ``$G$ does not have a winning vertex.''
		}
	}
\end{algorithm}

\begin{lemma}
  \label{lem:dmbounded}
  Algorithm~\ref{alg:bounded} determines if the strongly-connected input graph $G$
  has a winning vertex with respect to the bounded B\"{u}chi objective in $O(dm)$
  time.
\end{lemma}
\begin{proof}
  By Lemmas~\ref{lem:feasible}~and~\ref{lem:redcyclegstar}, $G$ has a winning
  vertex if and only if $G'$ has a cycle.  Since a depth-first search finds if
  there is a cycle in $G'$, the correctness of the algorithm is established.
  By Lemma~\ref{lem:rtlg}, \auxGraphDLayers{$\cdot$} takes $O(dm)$ time,
  and a depth-first search on $G'$ takes time $O(dm)$, because the number of
  edges in $G'$ is $O(dm)$.
\end{proof}

Thus, by Lemma~\ref{lem:reducetoscc}, we get the following theorem.

\begin{theorem}
  \label{thm:dmbounded}
  The set of winning vertices for the bounded B\"{u}chi objective in the graph
  case can be computed in time $O(dm)$.
\end{theorem}
\begin{proof}
  Let $S_1, S_2, \ldots$ be SCCs of the input graph $G = (V, E)$.  Let
  $m_1, m_2, \ldots$ be the number of edges in the SCCs $S_1, S_2, \ldots$.
  Then, by Lemma~\ref{lem:dmbounded}, for $i = 1, 2, \ldots$, we can determine in
  time $O(dm_i)$ whether $S_i$ is good.  Since $m \ge \sum_i m_i$, the proof is
  complete by Lemma~\ref{lem:reducetoscc}.
\end{proof}

\subsubsection*{An $O(|B|m)$-time algorithm for bounded B\"{u}chi}
Now, we briefly discuss an $O(|B|m)$-time algorithm for bounded B\"{u}chi.  Given
$G = (V, E)$ and $B$, consider the graph $G' = (B, E')$ such that
$(b, b') \in E'$ if the distance from $b$ to $b'$ in $G$ is at most $d$.  We
allow self loops in $G'$.  It is easy to see that $G$ has a feasible-cyclic walk
if and only if $G'$ has a cycle.  To construct $G'$, we perform $|B|$ breadth-first
searches, one starting from each vertex in $B$.  This takes time $O(|B|m)$.
Then, by a similar argument as in the proof of Theorem~\ref{thm:dmbounded}, we
get the following theorem.
\begin{theorem}
  \label{thm:bmbounded}
  The set of winning vertices for the bounded B\"{u}chi objective in the graph
  case can be computed in time $O(|B|m)$.
\end{theorem}

\begin{remark}
  Note that both algorithms that we have seen so far can take
  $\Theta(n^3)$ time if $m = \Theta(n^2)$ and $B$ and $d$ are $\Theta(n)$.  The
  next algorithm we see is combinatorial and has running time $O(n^{2.5}\log n)$ for
  the worst setting of the parameters and breaks the cubic barrier.  This also
  rules out any conditional lower bound approaches to get an $\Omega(n^3)$ lower
  bound for combinatorial algorithms.
\end{remark}
\subsubsection*{An $O((m+|B|^2)\sqrt{n} \log n)$-time algorithm for bounded
  B\"{u}chi}

In this section, we present an $O((m+|B|^2)\sqrt{n} \log n)$-time algorithm for
bounded B\"{u}chi in the graph case.  This is one of our main contributions.
Here, we give a procedure that computes distances between all pairs of B{\"u}chi
vertices if the distance is at least $\sqrt{N}$, where $N \ge |V|$ is a
parameter that we will fix later. This information can be used to reduce the
number of layers in the auxiliary graph to $\sqrt{N}$.  By $\dist$, we denote the distance
function for $G$.  For any $u, v \in V$, if $u \neq v$, then
$\dist(u,v)$ denotes the length of a shortest path from $u$ to $v$, and for any
$u \in V$, $\dist(u, u)$ denotes the length of a shortest cycle through $u$.

\begin{algorithm}[ht]
  \caption{This algorithm determines if the strongly-connected input graph
    has a winning vertex with respect to the bounded B\"{u}chi objective.}\label{alg:randbounded}
	\SetKwFunction{randboundedbuchi}{RandBoundedB\"{u}chi}
	\Procedure{\randboundedbuchi{$G = (V, E), B \subseteq V$, $d$, $N$}}{
		\If{$d < \sqrt{N}$}{
			\Return{\boundedbuchi{$G = (V, E), B \subseteq V$, $d$}}
		}
		Sample $4\sqrt{N}\ln N$ vertices uniformly at random, independently,
		and with replacement.\;
		$S \gets $ the set of sampled vertices.\;
		\For{$s \in S$}{
			Perform incoming and outgoing breadth-first search (BFS) to and from $s$.
			Compute distances $\dist(b, s)$ and $\dist(s, b)$ for each $b \in B$ during the BFSs.
		} 
		$G' \gets$ \auxGraphDLayers{$G, B, \sqrt{N} - 1$}
		\For{$b \in B$}{
			\For{$b' \in B$}{
				$\dist^S(b, b') \gets \infty$\;
				\For{$s \in S$}{
					$\dist^S(b, b') \gets \min\{\dist^S(b, b'), \dist(b, s) + \dist(s, b')\}$\;
				}

				\If{$\dist^S(b, b') \le d$}{
					Add $((b, 0), (b', 0))$ to $E'$ (this would be a self-loop if $b = b'$).\;
				}
			}
		}

		Run depth-first search on $G'$ to determine if it has a cycle.\;

		\If{$G'$ has a cycle}{
			\Return{``$G$ has a winning vertex.''}
		}\Else{
			\Return{``$G$ does not have a winning vertex.''}
		}
	}
\end{algorithm}

\begin{lemma}
  \label{lem:randbounded}
  Let $N \ge |V|$.  Algorithm~\ref{alg:randbounded} determines with probability
  at least $1 - 1/N^2$ if the strongly-connected input graph $G$ has a winning
  vertex with respect to the bounded B\"{u}chi objective in
  $O((m+|B|^2)\sqrt{N} \log N)$ time.  It never returns a \emph{false positive},
  i.e., if it outputs that $G$ has a winning vertex, then it is correct with
  probability $1$.  Its running time is $O((m+|B|^2)\sqrt{N} \log N)$.
\end{lemma}
\begin{proof}
  If $d < \sqrt{N}$, then we are done by Lemma~\ref{lem:dmbounded}.  Thus, we
  assume for the rest of the proof that $d \ge \sqrt{N}$.
  
  For any $b, b' \in B$, by $T(b,b')$, we denote a fixed shortest cycle through
  $b$ if $b=b'$ or a fixed shortest path from $b$ to $b'$ otherwise.  Let the
  event that a vertex $v(b,b') \in T(b,b')$ is sampled into $S$ be denoted by
  $\event(b,b')$.  Since $v(b,b') \in T(b,b')$, we have that
  $\dist(b,b') = \dist(b,v(b,b')) + \dist(v(b,b'),b')$.  This implies that if
  $\event(b,b')$ occurs, then $\dist(b,v(b,b'))$ and $\dist(v(b,b'),b')$ are
  computed by the algorithm using the incoming and outgoing BFS at $v(b,b')$,
  and hence $\dist^S(b,b') = \dist(b,b')$.  Let $\event^c(b,b')$ be the
  complement of $\event(b,b')$.  Now,
  $\Pr[\event^c(b,b')] = (1 - \dist(b,b')/|V|)^{4\sqrt{N}\ln N}$, because
  $1 - \dist(b,b')/|V|$ is the probability that a fixed sample does not contain a
  vertex of $T(b,b')$ and we draw $4\sqrt{N}\ln N$ independent samples.

  For any $b, b' \in B$, where $\dist(b,b') \geq \sqrt{N}$, we denote the event
  that $\dist^S(b,b') = \dist(b,b')$ by $\event'(b,b')$.  As noted earlier,
  $\dist^S(b,b') = \dist(b,b')$ if $\event(b,b')$ occurs, hence:
  \begin{align*}
    \Pr[\event'(b,b')] &\ge \Pr[\event(b,b')] &\text{$\event(b,b')$ is a subevent of $\event'(b,b')$,}\\
                       & = 1 - \Pr[\event^c(b,b')] \\
                       & = 1 - \left(1 - \frac{\dist(b,b')}{|V|}\right)^{4\sqrt{N}\ln N} &\text{by the argument earlier,}\\
                       & \ge 1 - \left(1 -
                         \frac{1}{\sqrt{N}}\right)^{4\sqrt{N}\ln N} &\text{because $\dist(b,b')/|V| \ge 1/\sqrt{N}$,}\\
                       & \ge 1 - \frac{1}{N^4} &\text{by well-known fact $(1-1/x)^x \leq 1/e$.}
  \end{align*}
  Since $N \ge |B|$, by the union bound, we have
  $\Pr[\forall (b, b') \in B \times B:\; \event'(b,b')] \ge 1 - 1/N^2$.  Let us
  condition on the event that for all $(b, b') \in B \times B:\; \event'(b,b')$,
  and let $G'$ be the auxiliary graph constructed by the algorithm.

  Suppose $G$ has a winning vertex.  By Lemma~\ref{lem:feasible}, there is a
  feasible cyclic-walk $C$ in $G$.  Then for any consecutive B\"{u}chi vertices
  $b$ and $b'$ in $C$, either $\dist(b,b') \ge \sqrt{N}$, in which case there is
  an edge $((b,0),(b',0))$ or $\dist(b,b') < \sqrt{N}$, in which case there
  exists a cycle $((b,0),(u_1,1),(u_2,2),\ldots,(u_\ell,\ell),(b',0))$ in $G'$,
  where $\ell < \sqrt{N} - 1$.  Thus, $C$ induces a cycle in $G'$.

  On the other hand, if there is a cycle $C'$ in $G'$, then a projection of $C'$
  on the first coordinate of the vertices, by construction of $G'$, gives a
  feasible cyclic-walk in $G$ after replacing all edges in $C'$ of the form
  $((b, 0), (b', 0))$ by corresponding paths of length at most $d$ that certify
  $\dist^S(b, b')$.  By Lemma~\ref{lem:feasible}, $G$ has a winning vertex.

  Also, if the algorithm does return that $G$ has a winning vertex, then $G'$
  has a cycle, and the existence of a feasible cyclic-walk in $G$ can be shown in the same
  way as above. This shows that the algorithm never returns a false positive.

  \paragraph*{Running time}
  Incoming and outgoing BFSs from the vertices in $S$ take time
  $O(m\sqrt{N} \log N)$.  \auxGraphDLayers{$\cdot$} takes $O(m\sqrt{N})$
  time.  Computing $\dist^S$ takes time $O(|B|^2\sqrt{N} \log N)$.  DFS on $G'$
  takes time $O(|B|^2 + m\sqrt{N})$.  In total, Algorithm~\ref{alg:randbounded}
  has running time $O((m+|B|^2)\sqrt{N} \log N)$.
\end{proof}

Finally, we use Lemma~\ref{lem:reducetoscc} to generalize the above to a graph that may
not be strongly connected.  Fix $N$ to be $n$ in Algorithm~\ref{alg:randbounded}
when running it for each SCC\@.  Then, by a similar argument as in the proof of
Theorem~\ref{thm:dmbounded}, we get the following theorem.  

\begin{theorem}\label{thm:randbounded}
  The set of winning vertices for the bounded B\"{u}chi objective can be computed
  with probability at least $1 - 1/n$ in time $O((m+|B|^2)\sqrt{n} \log n)$
  which is $O(n^{2.5}\log n)$.  Moreover, the algorithm never returns a \emph{false
    positive}, i.e., each vertex in the set it outputs is a winning vertex with
  probability $1$.
\end{theorem}
\begin{proof}
  Let $S_1, S_2, \ldots$ be SCCs of the input graph $G = (V, E)$.  Let
  $m_1, m_2, \ldots$ be the number of edges and by $\beta_1, \beta_2, \ldots, $
  be the number of B\"{u}chi vertices in the SCCs $S_1, S_2, \ldots$,
  respectively.  Then, by Lemma~\ref{lem:randbounded}, for $i = 1, 2, \ldots$,
  the algorithm outputs in time $O((m_i+\beta_i^2)\sqrt{n} \log n)$ whether
  $S_i$ is good.  Since $m \ge \sum_i m_i$ and
  $|B|^2 = (\sum_i \beta_i) ^2 \ge \sum_i \beta_i^2$, the running time bound is
  proved.

  The probability bound is obtained by a union bound over at most $n$ SCCs.
  Moreover, the algorithm never returns a \emph{false positive} by
  Lemma~\ref{lem:randbounded}.
\end{proof}

\subsection{The Bounded coB\"{u}chi Objective}
Given a graph $G = (V, E)$, a set $C$ of vertices, and a positive integer $d$, a
walk $W$ is called a \emph{feasible} walk if $W \subseteq C$ and the number of
vertices in $W$ is at least $d$.

We assume that $G$ is strongly connected, and our goal is to determine if there
is a winning vertex in $G$ for the bounded coB\"{u}chi objective.
The following lemma reduces this problem to finding a feasible walk in $G$.
\begin{lemma}
  \label{lem:feasiblew}
  The strongly-connected input graph $G$ with a set $C$ of vertices has a
  winning vertex with for the bounded coB\"{u}chi objective if and only if
  it has a feasible walk.
\end{lemma}
\begin{proof}
  If $G$ has a winning vertex, say $v$, then there is a winning play $\omega$
  that starts at $v$.  Let $\omega = \ls v_0 = v, v_1, v_2, \ldots \rs$; so by
  the definition of winning play,
  $\forall i \ge 0, \exists j \ge i: \{v_j,v_{j+1}, \dots, v_{j+d-1}\} \subseteq
  C$.  Any such walk $v_j,v_{j+1}, \dots, v_{j+d-1}$ is a feasible walk.

  In the other direction, say $G$ has a feasible walk
  $v'_1,v'_2, \dots, v'_\ell$, where $\ell \ge d$ and $v'_i \in C$ for
  $i \in \{1,\ldots,\ell\}$.  Now, consider the cyclic walk
  $v'_1,v'_2, \dots, v'_{d-1},v'_d,v'_{d-1},\dots,v'_1$.  We keep traversing it to
  construct a winning play.  The existence of a winning play implies that $G$
  has a winning vertex.
\end{proof}

Now we give an $O(m)$-time algorithm to determine if the strongly-connected
input graph $G$ has a feasible walk.

\begin{algorithm}[ht]
  \caption{This algorithm determines if the strongly-connected input graph has a
    winning vertex with respect to the bounded coB\"{u}chi objective.}\label{alg:boundedco}
\SetKwFunction{boundedcobuchi}{BoundedcoB\"{u}chi}
\SetKwFunction{longestPath}{LongestPath}
\Procedure{\boundedcobuchi{$G = (V, E), C \subseteq V$, $d$}}{
    $G' \gets G$ induced on $C$\;
    Perform DFS on $G'$; if there is a cycle, return ``$G$ has a winning vertex.''\;
	\If{\longestPath{$G'$} $\ge d$}{
		\Return{``$G$ has a winning vertex.''}
	}\Else{
		\Return {``$G$ does not have a winning vertex.''}
	}
}
\end{algorithm}

\begin{algorithm}[ht]
  \caption{An algorithm to compute the length of a longest path in a directed
    acyclic graph.}\label{alg:llpdag}
	\Procedure{\longestPath{$G = (V, E)$}}{
		Compute the topological ordering $T$ of $G$.\;
		$L \gets$ integer array of size $|V|$, initialized to all zeros.\;
		$M \gets 0$\;
		\For{vertex $v$ in order of $T$}{
			\For{incoming edges $(u, v)$}{
				$L[v] \gets \max\{L[v], L[u] + 1\}$
			}
			$M \gets \max\{M, L[v]\}$\;
		}
		\Return{$M$}
	}
\end{algorithm}

\begin{lemma}
  \label{lem:boundedco}
  Algorithm~\ref{alg:boundedco} determines if the strongly-connected input graph
  $G$ has a winning vertex with respect to the bounded coB\"{u}chi objective in
  $O(m)$ time.
\end{lemma}
\begin{proof}
  If $G'$ has a cycle, then this cycle can be repeated to get a feasible walk in
  $G$, which, by Lemma~\ref{lem:feasiblew}, means that $G$ has a winning
  vertex.

  If $G'$ is acyclic and has the length of longest path at least $d$, then there
  is a feasible walk in $G$, which, by Lemma~\ref{lem:feasiblew}, means that
  $G$ has a winning vertex.  Conversely, if $G'$ is acyclic and the length of
  longest path in $G'$ is less than $d$ then $G$ cannot contain a feasible walk,
  which, by Lemma~\ref{lem:feasiblew}, means that $G$ does not have a winning
  vertex.

  The calls to DFS on $G'$ and \longestPath{$G'$} take time $O(m)$; hence, the
  running time of Algorithm~\ref{alg:boundedco} is $O(m)$.
\end{proof}

By a similar argument as in the proof of Theorem~\ref{thm:dmbounded}, we
get the following theorem.

\begin{theorem}
  \label{thm:boundedc}
  The set of winning vertices for the bounded coB\"{u}chi objective in the graph
  case can be computed in time $O(m)$.
\end{theorem}

\section{Algorithms for Game Graphs}
\label{sec:games}
In this section, we present algorithms for the bounded B\"uchi objective in game graphs. 
We first introduce the auxiliary \emph{game graph} similar to the auxiliary graph defined 
earlier. We then show that we can  
compute in $O(n^2d^2)$  time the winning set of a given bounded B\"uchi objective on game graphs
by computing the winning set of a coB\"uchi objective on the auxiliary game graph.
Finally, we show how to improve the running time to $O(n^2d)$ by using structural properties of the
auxiliary game graph and adapting a known technique for solving B\"uchi Games~\cite{CH14}.

\para{The Auxiliary Game Graph.}
Given a game graph $\Gamma=(V,E,\langle V_1, V_2 \rangle)$ with $n$ vertices, $m$ edges
and a bounded B\"uchi objective $\pbuchi{B,d}$, we first construct the auxiliary graph by calling
\ConstructAuxiliaryGraph{$(V,E),B,d$} in Algorithm~\ref{alg:lg} and additionally 
partition the vertices of the auxiliary graph $V^*$ into player-1 vertices $V_1^*$ and player-2
vertices $V_2^*$, i.e., for each $(v,\ell) \in V^*$ we get $(v,\ell) \in V_1^*$ if $v
\in V_1$ and $(v,\ell) \in V^*_2$ if $v \in V_2$. 
The auxiliary game graph has $O(nd) = O(n^2)$ vertices and $O(md) = O(mn)$ edges.
We say that a vertex $(v, \ell) \in V^*$ is a \emph{layer-$\ell$ vertex} and $v$ is its \emph{first component}.

For any play $\lambda$, we denote by $\lambda_k$ the $k$th vertex of the play.
If a play has a superscript, it denotes the starting vertex of the play, e.g $\lambda^v$ means
that the play $\lambda$ starts at $v$. By $\lambda^v_k$ we refer to the $k$th vertex of the play
$\lambda^v$ which starts at $v$.
Given a finite feasible play $\lambda^{(w, \ell)}$ in $\Gamma^*$ starting at
$(w, \ell)$, we define $\proj(\lambda^{(w, \ell)})$ to be the projection
of $\lambda^{(w, \ell)}$ on the first component of the vertices in it; by
definition, this finite play starts at $w$ and is feasible in $\Gamma$.  Analogously,
given a finite feasible play $\lambda^w$ in $\Gamma$, we define
$\lift(\lambda^w, \ell)$ to be the unique finite feasible play in $\Gamma^*$ starting at
$(w, \ell)$ such that the first component of $\lift(\lambda^w, \ell)_k$ is the same
as $\lambda^w_k$. For $(u, v)$ in $E$ such that $(u, j) \in V^*$ define (the appropriate
next layer number if you followed the copy of $(u, v)$ starting in layer~$j$)
\[
  \nextl(u, v, j) = \begin{cases}
  j+1 \text{ if } j<d \text{ and }  v \notin B\\
  d \text{ if } j=d \text{ and } v \notin B\\
  0 \text{ if } v \in B\,.
\end{cases}
\]
Now, define $\lift(\lambda^w, \ell)_1 = (w, \ell)$, and for $k > 1$, given
$\lift(\lambda^w, \ell)_{k-1} = (\lambda^w_{k-1}, j)$ define
$\lift(\lambda^w, \ell)_k = (\lambda^w_k, \nextl(\lambda^w_{k-1}, \lambda^w_k,
j))$.
Similarly, given the finite feasible play $\lambda^{(w, \ell)}$ in $\Gamma^*$, we define
$\move(\lambda^{(w, \ell)}, \ell')$ to be the finite play that starts at $(w, \ell')$ in $\Gamma^*$
such that, for any $k$, the first components of $\lambda^{(w, \ell)}_k$ and
$\move(\lambda^{(w, \ell)}, \ell')_k$ are the same. By construction of $\Gamma^*$ the finite play
$\move(\lambda^{(w, \ell)}, \ell')$ is well-defined because (1) edges going from layer-$i$ vertices
to layer-$(i+1)$ vertices $(1 \leq i \leq d-1)$ exist in all layers with the same respective first components except 
in layer-$d$ where these edges go again to layer-$d$, (2) edges going to layer-0 vertices exist in all
layers $(1 \leq i \leq d)$ and (3) because edges originating from layer-0 vertices implies that both
plays are currently visiting the same layer-0 vertex.

In comparison, 
the goal of the two operations $\proj(\cdot)$ and $\lift(\cdot)$ is to map finite plays between $\Gamma^*$ and $\Gamma$ such that the
finite play in
$\Gamma^*$ has, for all vertices, the same first component as the corresponding finite play in
$\Gamma$ and vice versa.
In contrast, $\move(\lambda^{(w,\ell)},\ell')$ maps a finite play in $\Gamma^*$ to a finite play also in 
$\Gamma^*$ which has the same first component but a ``shifted'' starting vertex.

\subsection{An $O(n^2d^2)$-time Algorithm for Bounded B\"uchi in Games}
\begin{sloppypar}
In this section, we show that we can compute the winning set of 
a given \emph{bounded B\"uchi objective} on game graphs
by computing the winning set of a \emph{coB\"uchi objective} on the auxiliary game graph.
Then we apply the best-known algorithm
for computing the winning set of a B\"uchi objective on the auxiliary game graph to get the desired result.
In the following lemma, we prove that computing $\W{1}{\pbuchi{B,d, \Gamma}}$ is the same as
computing $\W{1}{\cobuchi{C^*, \Gamma^*}}$ where $C^*$ are the vertices in layers-$\{0, 1, \ldots, d{-}1\}$.
Intuitively, when a play $\phi$ in $\cobuchi{C^*, \Gamma^*}$ stays in layers-$\{0, 1, \ldots, d{-}1\}$,
it reaches a vertex in layer~0 every at most $d$ steps by
construction of $\Gamma^*$.
The layer-$0$ vertices correspond to the vertices in $B$ which means that a play $\phi'$ in $\Gamma$ defined as
the projection on the first component of the vertices in $\phi$ 
visits a vertex in $B$ every at most $d$ steps which implies that $\phi' \in \pbuchi{B,d, \Gamma}$. 
On the other hand, when
player~1 has a strategy in $\Gamma$ to visit a vertex in $B$ every at most $d$ steps, a similar strategy
which visits the same vertices in the first component in $\Gamma^*$ allows player~1 to stay in the first $d$ layers of the auxiliary graph.
\end{sloppypar}

\begin{lemma}\label{lem:auxgg}
	\begin{sloppypar}
		Let $\GG = (V,E,\langle V_1, V_2 \rangle)$ be a game graph with bounded B\"uchi objective
		$\pbuchi{B,d}$, let $\GG^* = (V^*,E^*, \langle V_1^*, V_2^* \rangle)$ be the corresponding auxiliary
		game graph,
		and let $C^*$ be the vertices in the first $d$ layers of the auxiliary graph, i.e., $C^*= \{(v,i) \in V^* \mid 0 \leq i \leq d-1 \}$.
		Then $\{ w \mid (w,i) \in  \W{1}{\cobuchi{C^*, \Gamma^*}}, \text{ for some } 0 \leq i \leq d\}
		=  \W{1}{\pbuchi{B,d, \Gamma}}$.
	\end{sloppypar}
\end{lemma}
\begin{proof}

	\par We first prove that $\{ w \mid (w,i) \in  \W{1}{\cobuchi{C^*, \Gamma^*}}, \text{ for some } 0 \leq i \leq d\}  \subseteq
	\W{1}{\pbuchi{B,d, \Gamma}}$.
	Let $(w,i) \in \W{1}{\cobuchi{C^*, \Gamma^*}} $. Then player~1 has a winning strategy
	$\sigma^*$ in $\Gamma^*$ such that for all player-2 strategies $\pi^*$, we have that $\omega((w,i),
	\sigma^*,\pi^*) \in \cobuchi{C^*, \Gamma^*}$. 

	Whenever player~1 makes a move in $\Gamma^*$, we define the corresponding move in $\Gamma$ as
	follows:
	For any finite play $\lambda^w$ in $\Gamma$ that ends in a player-1 vertex, define
        $\sigma(\lambda^w)$ to be the first component of
        $\sigma^*(\lift(\lambda^w, i))$.  (It does not matter how we define
        $\sigma$ for plays that do not start at $w$.)

	Next, we argue why $\sigma$ is a winning player-1 strategy for $\pbuchi{B,d, \Gamma}$ starting at $w$.  Let $\pi$ be an
        arbitrary player-2 strategy in $\Gamma$.  We define a corresponding
        player-2 strategy $\pi^*$ in $\Gamma^*$: for $\lambda^{(w, i)}$ that
        ends in a player-2 vertex $(u, j)$, let
        $v = \pi(\proj(\lambda^{(w, i)}))$ and define
        $\pi^*(\lambda^{(w, i)}) = (v, \nextl(u, v, j))$.

	Now, it is straightforward to show that the first component of $\omega((w,i),\sigma^*,\pi^*)_k$ is equal to
	$\omega(w,\sigma,\pi)_k$ by induction on $k$.

	Since the play $\omega((w,i), \sigma^*,\pi^*) \in \cobuchi{C^*, \Gamma^*}$, it
        stays in $C^*$ after a finite number of steps. Note that to stay in
        $C^*$ means to visit a layer-$0$ vertex after every at most $d$ steps because there are only
		$d$ layers in $C^*$ and each step that does not go to a layer-$0$ vertex increases the
		layer counter.
        Since the first component of each layer-$0$ vertex is in $B$, the play
        $\omega(w,\sigma,\pi)$ visits a vertex in $B$ every at most $d$ steps
        after a finite number of steps and is in $\pbuchi{B,d, \Gamma}$.

		\begin{sloppypar}
			The other direction, i.e., $\W{1}{\pbuchi{B,d, \Gamma}} \subseteq \{ w \mid (w,i) \in
				\W{1}{\cobuchi{C^*, \Gamma^*}},$ for some $0 \leq i
				\leq d\}$ can be shown with a similar argument.

			Let $w \in \W{1}{\pbuchi{B,d, \Gamma}}$. Then player~1 has a winning strategy
			$\sigma$ in $\Gamma$ such that for all player-2 strategies $\pi$, we have that
			$\omega(w,\sigma,\pi) \in \pbuchi{B,d, \Gamma}$. 

			Whenever player~1 makes a move in $\Gamma$, we define the corresponding move in $\Gamma^*$ as
			follows: Let $0 \leq i \leq d$ be arbitrary such that $(w, i) \in V^*$.
			For any finite play $\lambda^{(w,i)}$ in $\Gamma$ that ends in a player-1 vertex $(u,j)$, let
			$v = \sigma(\proj(\lambda^{(w,i)}))$ and
			define $\sigma^*(\lambda^{(w,i)}) = (v, \nextl(u,v,j))$. 
		\end{sloppypar}

	Next, we argue why $\sigma^*$ is a winning player-1 strategy for $\cobuchi{C^*, \Gamma^*}$
	starting at $(w,i)$. Let $\pi^*$ be an arbitrary player-2 strategy in $\Gamma^*$. 
	We define a corresponding player-2 strategy $\pi$ in $\Gamma$: For the finite play $\lambda^{w}$ that ends in a
	player-2 vertex $u$, let $\pi(\lambda^w)$ be the first component of
	$\pi^*(\lift(\lambda^w,i))$ (it does not matter how we define $\pi$ for plays that do not start
	at $w$). 

	Now it is straightforward to show by induction on $k$ that the first component of $\omega((w,i),
	\sigma^*, \pi^*)_k$ is equal to $\omega(w,\sigma,\pi)_k$.

	Since the play $\omega(w, \sigma, \pi) \in \pbuchi{B,d, \Gamma}$, it visits a vertex in $B$
	every at most $d$ steps after a finite number of steps. Thus, the play $\omega((w,i), \sigma^*,
	\pi^*)$ visits a layer-0 vertex every at most $d$ steps after a finite number of steps.
	Hence, $\omega((w,i), \sigma^*, \pi^*)$ stays in $C^*$ as the play
	increments the layer counter to at most $d-1$ and is in $\cobuchi{C^*,\Gamma^*}$.
\end{proof}

To compute
$\W{1}{\cobuchi{C^*}}$ in $\Gamma^*$, we observe that, by the duality of B\"uchi objectives,
$\W{1}{\cobuchi{C^*}} = V^* \setminus \W{2}{\buchi{V^* \setminus C^*}} = V^* \setminus \W{2}{\buchi{\{(v,d) \in V^*\}}}$. 
Since, traditionally, 
we always compute the player-1 winning set of a given objective, we swap player-1 and player-2 vertices
in $\Gamma^*$. Then we compute $W = \W{1}{\buchi{\{(v,d) \in V^*\}}}$ using the algorithm of Chatterjee and Henzinger~\cite{CH14}, which is the fastest algorithm for B\"uchi games known, and project $V^* \setminus W$  on the first coordinate.
We illustrate the details in Algorithm~\ref{alg:games:pbuchi}.

\begin{algorithm}[ht]
	\caption{Determine $\W{1}{\pbuchi{B,d}}$, given a game graph$\Gamma$}\label{alg:games:pbuchi}
	\SetKwFunction{bbuchigames}{BoundedB\"uchiGames}
	\SetKwFunction{buchiGamesFast}{B\"{u}chiGamesFast}
	\Procedure{\bbuchigames{$\Gamma = (V,E, \langle V_1, V_2\rangle), B, d$}}{
		$(V^*,E^*) \gets$ \ConstructAuxiliaryGraph{$(V,E)$}\label{alg:games:pbuchi:auxgg}\; 
		$V^*_1 \gets \{(v,i) \in V^* \mid v \in V_1 \}, V^*_2 \gets \{(v,i) \in V^* \mid v \in V_2
		\}$\;
		$\Gamma^* \gets (V^*,E^*, V^*_1, V^*_2); B^* \gets \{(v,d) \in V^* \mid v \in V\setminus
			B\}$\;
		$W \gets$ \buchiGamesFast{$\Gamma^* = (V^*,E^*,\langle V^*_2,V^*_1 \rangle),B^*$} (\cite{CH14},
		Algorithm~\ref{alg:games:fastbuchi})\label{alg:games:pbuchi:fastbuchi}\;
		\Return{$\{ x \mid (x,i)\in  V^* \setminus W \text{ for some } 0 \le i \le d\}$}
	}
\end{algorithm}

The correctness of Algorithm~\ref{alg:games:pbuchi} is due to the correctness of the fast B\"uchi
games algorithm~\cite[Theorem 2.14]{CH14}, the
argument above, and Lemma~\ref{lem:auxgg}.
The argument for the running time of Algorithm~\ref{alg:games:pbuchi} is as follows.
We first construct $\Gamma^*$ in $O(md)$ time and then compute the winning set of $\cobuchi{C^*}$ in time $O(|V^*|^2)$~\cite[Theorem 2.14]{CH14}. 
As $|V^*| = O(nd)$ and $d = O(n)$, we get the following theorem.
\begin{theorem}
	The set of winning vertices for the bounded B\"uchi objectives in games can be computed in time
	$O(n^2d^2) = O(n^4)$.
\end{theorem}

\subsection{An $O(n^2d)$-time Algorithm for Bounded B\"uchi in Games}
In this section, we give a refined running time analysis of Algorithm~\ref{alg:games:pbuchi}
giving us an $O(n^2d)$-time algorithm for bounded B\"uchi games.
We first describe the fastest algorithm for B\"uchi Games~\cite{CH14} for completeness.
Then, we identify key ideas of the refined running time analysis when the input is an auxiliary
game graph and prove the improved running time formally.

\subsubsection{The B\"uchi Games Algorithm of~\cite{CH14}}
Given a game graph $\Gamma = (V, E, \langle V_1, V_2 \rangle)$ and a set $B$ of B\"uchi
vertices\footnote{not to be confused with the input for the bounded B\"uchi problem in the previous and later sections}, we fix an order on the edges. In this fixed order, the edges $(u,v)$ where $u$ is a non-B\"uchi player-2 vertex, i.e., $u \in (V_2 \setminus B)$,
come before all other edges. We call them priority-1 edges. All the other edges are priority-0 edges. 

\begin{definition}\label{def:fastbuchi:gammai}
	Given a game graph $\Gamma = (V,E, \langle V_1,V_2 \rangle)$, let $\Gamma_i = (V,E_i, \langle V_1,V_2 \rangle)$ for $1 \leq i \leq \log n$
	be a subgraph of $\Gamma$ which we define as follows: For all $u \in V$, the set $E_i$ contains the following edges: 
	\begin{enumerate}
		\item If the outdegree of $u$ in $E$ is at most $2^i$, $E_i$ contains all edges of the form $(u,v)$,
			i.e., if $|\Out{u}|  \leq 2^i$ then the set $\{(u,v) \mid v \in \Out{u}\} \subseteq E_i$.
		\item If the edge $(v,u)$ belongs to the first $2^i$ inedges of vertex $u$ in $E$, we have $(v,u) \in E_i$ (``first''
			means with respect to to the fixed order we specified above).
	\end{enumerate}
	Note that $E_{i-1} \subseteq E_i$ since the order of the edges is fixed.
	We form a partition of $V$ in $\Gamma_i$ by giving each vertex a color:
	\begin{itemize}
		\item \emph{Blue}: A player-1 vertex $v$ in $\Gamma_i$ is blue if the outdegree of $v$ is greater than $2^i$.
		\item \emph{Red}: A player-2 vertex $u$ in $\Gamma_i$ is red if it has no outedge in $E_i$.\footnote{In the algorithm of
				Chatterjee and Henzinger~\cite{CH14} red vertices are player-2 vertices where an edge of $E$ is
				missing. We change this definition slightly, i.e., without changing their algorithm or
				correctness argument, by saying that player-2 vertices are red if they do not have
				any outedges in $E_i$.}
		\item All other vertices are \emph{white}.
	\end{itemize}

\end{definition}
Thus, if a player-1 vertex is white then all its outedges are in $E_i$, and if a player-2 vertex is white then
it has at least one outgoing edge in $E_i$.

\para{Algorithm description.}
The input of Algorithm~\ref{alg:games:fastbuchi} is a game graph $\Gamma$ and a set of B\"uchi
vertices $B$. Recall that every vertex in a player-1 closed set $S$ without B\"uchi vertices cannot 
be in the player-1 winning set of the given B\"uchi objective $\W{1}{\buchi{B}}$ (Proposition~\ref{prop:closedsets}~(2)). 
We repeatedly find such a set $S$ by removing from $V$ the player-1 attractor of the set $B$ 
(Proposition~\ref{obs:attractorclosedset}) and forming $S$ from all the remaining vertices. Then we remove the player-2 attractor of $S$.
In the algorithm, we identify such a set $S_j$ at Line~\ref{alg:games:fastbuchi:sj} 
and remove the attractor at Line~\ref{alg:games:fastbuchi:dj}. 
Note that a naive algorithm would take $O(nm)$ time, as the attractor of $S$ could always be of size $1$ 
and computing the attractor is in $O(m)$ time.
To obtain a quadratic-time (in the number of vertices) algorithm, the improved algorithm of Chatterjee and
Henzinger constructs, for $i=1,\dots,\log n$,
the graph $\Gamma_i$ which has at most $2^i$ edges. 
Due to the properties of $\Gamma_i$, it can be shown that the set $S_j$ has size of at 
least $2^{i-1}$. In this way, the attractor computation takes time proportional to the removed vertices.
Since player-1 vertices with missing outgoing edges or player-2 vertices with no 
outgoing edge in $\Gamma^i$, i.e., non-white vertices might still be able to reach a 
vertex in $B$, we compute the player-1 attractor of the non-white vertices combined with the vertices 
in $B$. We illustrate the details in Algorithm~\ref{alg:games:fastbuchi}.

\begin{algorithm}[ht]
	\caption{Determine $\W{1}{\buchi{B}}$, given a game graph
		$\Gamma$~\cite{CH14}}\label{alg:games:fastbuchi}
	  \Procedure{\buchiGamesFast{$\Gamma = (V, E, \langle V_1, V_2 \rangle), B$}}{
	  Let $j \gets 0$; $U \gets \emptyset$; $Y_0 \gets \attr{1}{B}{\Gamma}$; $S_0 \gets V \setminus
	  Y_0;$ $D_0 \gets \attr{2}{S_0}{\Gamma}$; $\Gamma^j \gets \Gamma$\;
	  $j \gets j+1$;
	  \While{$D_{j-1} \neq \emptyset$}{\label{alg:games:fastbuchi:while}
		  Remove the vertices in $D_{j-1}$ from $\Gamma^{j-1}$ to obtain $\Gamma^{j}$; and $U \gets U \cup
		  D_{j-1}$\;
		  $i \gets 1$\;
		  \Repeat{$S_j$ is nonempty or $i \ge 1 + \log n$\label{alg:games:fastbuchi:repeat}}
		  {
			  Construct $\Gamma^{j}_i$ from $\Gamma^j$ as described in
			  Definition~\ref{def:fastbuchi:gammai}. \;
			  Let $Z^j_i$ be the vertices of $V^j$ that are either red or blue\;
			  $Y^j_i \gets \attr{1}{B^j \cup Z^j_i}{\Gamma^{j}_i}$\;\label{alg:games:fastbuchi:p1attractor}
			  $S_j \gets V^j \setminus Y^j_i$\;\label{alg:games:fastbuchi:sj}
			  $i \gets i +1$\;
		  }\label{alg:games:fastbuchi:until}
		  \If{$S_j \neq \emptyset$}{
			  $D_j \gets \attr{2}{S_j}{\Gamma^j}$\label{alg:games:fastbuchi:dj} \;
		  }\Else{
			  \Return{$V \setminus U$}
		  }
		  $j \gets j+1$\;
	  }
  }
\end{algorithm}

The removal of player-1 closed sets in $\Gamma_i$ now includes vertices which are blue, red, and white.
The definition of a \emph{separating cut} further refines the definition of the winning regions for player~2 in this regard.

\para{Separating cut.}
A set $S$ of vertices induces a separating cut in a game graph $\Gamma_i$ or $\Gamma^j_i$ in
Algorithm~\ref{alg:games:fastbuchi} if
\begin{enumerate}
	\item the only edges from $S$ to $V \setminus S$ come from player-2 vertices in $S$
	\item every player-2 vertex in $S$ has an edge to another vertex in $S$
	\item every player-1 vertex in $S$ is white and 
	\item $B \cap S = \emptyset$.
\end{enumerate}
Thus, a separating cut $S$ is a player-1 closed set where (i) player-1 vertices are white and which (ii) does not contain a
vertex in $B$.

The following lemmas are needed to establish the improved running time guarantees in the next
section. Detailed proofs can be found in the paper by Chatterjee and Henzinger~\cite{CH14}. 

Lemma~\ref{lem:games:sj_sepcutGj_attr1} below says that the set $S_j$ is indeed a separating cut in $\Gamma^j$ (not only in
$\Gamma^j_i$) and that due to the careful construction of $\Gamma^j_i$ from the game graph
$\Gamma^j$ in iteration $j$, $S_j$
does not include a vertex of the player-1 attractor of the B\"uchi vertices in $\Gamma^j$.
\begin{lemma}[\cite{CH14}, Lemma 2.9]\label{lem:games:sj_sepcutGj_attr1}
	\begin{sloppypar}
		Let $S_j$ be the non-empty set computed by Algorithm~\ref{alg:games:fastbuchi} in iteration $j$.
		Then, (1) $S_j$ is a separating cut in $\Gamma^j$; and (2) 
		$S_j \cap \attr{1}{B^j}{\Gamma^j} = \emptyset$.
	\end{sloppypar}
\end{lemma}

Lemma~\ref{lem:games:sepcut_in_sj} establishes that the separating cut found in $\Gamma^j_i$ is
indeed the maximum separating cut in $\Gamma^j_i$. Also, if $\Gamma^j_i$ contains a separating
cut, Algorithm~\ref{alg:games:fastbuchi} finds it.
\begin{lemma}[\cite{CH14}, Lemma 2.11]\label{lem:games:sepcut_in_sj}
	Let $\Gamma^j_i$ be the game graph in iteration $j$ of the outer loop and iteration $i$ of the inner
	loop. If $S$ induces a separating cut in $\Gamma^j_i$, then $S \subseteq S_j$.
\end{lemma}

Lemma~\ref{lem:games:sj_sepcutGji} says that the set $S_j$ is a separating cut in $\Gamma^j_i$. This
does not follow from Lemma~\ref{lem:games:sj_sepcutGj_attr1}(1) because $\Gamma^j_i$ might have less edges
than $\Gamma^j$ and separating cuts are not preserved if we only consider a subset of edges in
$\Gamma^j$ (property 2 might be violated).
\begin{lemma}[\cite{CH14}, Lemma 2.12]\label{lem:games:sj_sepcutGji}
	Consider an iteration $j$ of the outer loop of Algorithm~\ref{alg:games:fastbuchi} such that the
	algorithm stops the inner loop at value $i$ and identifies a non-empty set $S_j$. Then, $S_j$
	is a separating cut in $\Gamma^j_i$.
\end{lemma}

\subsubsection{Faster Algorithm for Bounded B\"uchi Games}
In this section, we give the refined running time analysis of
Algorithm~\ref{alg:games:pbuchi}.  We note that $\Gamma^*$ gets redefined to be
$(V^*,E^*,\langle V^*_2,V^*_1 \rangle)$ in Algorithm~\ref{alg:games:pbuchi} on
Line~\ref{alg:games:pbuchi:fastbuchi}.  Therefore, from here on, when we say
player~1 (respectively player~2), we mean the player controlling the vertices in
$V_2^*$ (respectively, those in $V_1^*$).

\para{Distinct vertices.}
We call a set of vertices $S$ in $\Gamma^*$ \emph{distinct} if, for each pair of vertices $(v,\ell),
(v',\ell') \in S$, we have $v \neq v'$. 

\para{Copies of a vertex.}
Let $\copies{v}$ denote the set of ``copies'' of a vertex $v \in V^*$, i.e.,
for a layer-$0$ vertex $(v,0)$ we have that $\copies{(v,0)}=\{(v,0)\}$ and for a vertex $(v,\ell)$,
where $\ell>0$, we have $\copies{(v,\ell)} = \{(v,1), \dots, (v,d)\}$.

The improved running time guarantee is due to two key ideas. 

\para{Key idea 1.}
When \emph{there is a} vertex
$(v,\ell)$ in $D_j$ then $\copies{(v,\ell)} \subseteq D_j$, i.e., \emph{all} its
copies are in $D_j$.

On a very high level, the argument is that if there is a player-2 strategy to go
from a vertex to $S_j$, then there exists a player-2 strategy from all copies of
that vertex to $S_j$.  While the idea is simple to state, complicated
machinery is needed to prove it formally.  We prove the key idea in  
Claim~\ref{claim:games:copiesdeleted} building on Definition~\ref{def:is} and
Claim~\ref{clm:transstrat}.

Now, if we follow the original running-time argument~\cite{CH14}, then we can only
claim that we remove $2^{i-1}$ vertices in \emph{total} if the inner loop at Line~\ref{alg:games:fastbuchi:repeat} 
stops at iteration $i$, but the second key idea states something stronger.

\para{Key idea 2.}
If the inner loop at Line~\ref{alg:games:fastbuchi:repeat} 
stops at iteration $i^*$, we remove $2^{i^*-1}$ \emph{distinct} vertices.

Combining the key ideas, we remove from the game graph in iteration $j$ all copies of those distinct
vertices.
The $i$th iteration of the loop at Lines~\ref{alg:games:fastbuchi:repeat}--\ref{alg:games:fastbuchi:until} 
takes time $O(2^i nd)$ for
constructing the auxiliary version of ${(\Gamma^*)}_i$ and performing the attractor computations. The
iterations of the loop in Lines~\ref{alg:games:fastbuchi:repeat}--\ref{alg:games:fastbuchi:until} before $i' < i$ amount 
to a total running time of $O(2^ind)$.
Thus, we charge the $2^{i-1}$ removed distinct vertices 
the cost of the iteration and the iterations before, 
i.e., each such removed original vertex is charged $O(nd)$.
As we can remove only $n$ distinct vertices since they correspond to the vertices in the game graph $\Gamma$, we have a total cost of
$O(n^2d)$.

For the second key idea to work, we must modify the
original bounded B\"uchi instance $(\Gamma, B, d)$ carefully.
For every vertex in $v \in B$ we add a player-2 vertex $v'$ which is not in $B$ and an edge
$(v',v)$. Then we redirect all edges which go to $v$ in the original instance and make them go to $v'$ instead, i.e.,
for all $v \in B$ we have $V_2 \gets V_2 \cup \{v'\}$ and $E \gets (E \cup \{(v',v)\} \cup \{(u,v') \mid
	(u,v) \in E \}) \setminus  \{(u,v) \in E\}$.  Also, we
increase $d$ by one, as we increase the distance to all vertices in $B$ by one. Note that this
simple modification allows us to assume, without loss of generality, that all vertices
in $B$ have incoming edges from player-2 vertices only. 
Since we swap the player-1 vertices with player-2 vertices in Algorithm~\ref{alg:games:pbuchi} we can assume
that all incoming edges to a layer-$0$ vertex are from
player-1 vertices. This adds at most $n$ vertices and edges to $\Gamma$.

\begin{observation}\label{obs:gg:layer0nop2incoming}
	We can assume, without loss of generality, that all layer-$0$ vertices $v \in V^*$ 
	of the auxiliary game graph $\Gamma^*$ created at Line~\ref{alg:games:pbuchi:auxgg} 
	in Algorithm~\ref{alg:games:pbuchi} have no incoming edges from player-2 vertices, i.e., 
	if $(v,0) \in V^*$ then $\In{(v,0)} \cap V^*_2 = \emptyset$.
\end{observation}

With the above observation, we can prove the following proposition which is the crux of
this section.

\begin{proposition}
	Algorithm~\ref{alg:games:pbuchi} runs in time $O(n^2d) = O(n^3)$.
\end{proposition}
\begin{proof}
	In this proof we denote by $(\Gamma^*, B^*)$ the input of Algorithm~\ref{alg:games:fastbuchi} at
	Line~\ref{alg:games:pbuchi:fastbuchi} of Algorithm~\ref{alg:games:pbuchi}. The input to
	Algorithm~\ref{alg:games:pbuchi} is $(\Gamma,B,d)$. If we can show that the running time of the
	call to Algorithm~\ref{alg:games:fastbuchi} at Line~\ref{alg:games:pbuchi:fastbuchi} 
	is in $O(n^2d) = O(n^3)$ we are done, as the rest of the
	operations of Algorithm~\ref{alg:games:pbuchi} are in $O(md)$. This entails constructing
	$(\Gamma^*, B^*)$ and going through $W$. We therefore prove the following lemma.

	\begin{lemma}
		The total time Algorithm~\ref{alg:games:pbuchi} spends in
		Algorithm~\ref{alg:games:fastbuchi} is $O(n^2d) = O(n^3)$.
	\end{lemma}
	Every vertex $v$ in $\Gamma^*$ has only $O(n)$ out-edges by the definition of the auxiliary game
	graph. 
	Thus, when we consider the graphs $(\Gamma^*)_i$ of Definition~\ref{def:fastbuchi:gammai} for $1 \leq
	i \leq \log n$, 
	we have $(\Gamma^*)_{\log{n}} = \Gamma^*$. 
	The construction of $(\Gamma^*)_i$ ($1 \leq i \leq \log n)$ takes time $O(nd \cdot 2^{i})$. 

	We split the running time argument into two parts. In the first part, we bound the 
	running time of all except the last iteration of the while loop at Line~\ref{alg:games:fastbuchi:while}. 
	In the second part of the analysis, we bound the running time of the last iteration of the same
	loop.

	\para{Running time bound for all iterations of the while loop except the last.}
	Consider iteration $j$, and assume that Algorithm~\ref{alg:games:fastbuchi} stops the repeat-until
	loop at Line~\ref{alg:games:fastbuchi:until} with value $i^*$ and it is not the last iteration 
	of the while loop at Line~\ref{alg:games:fastbuchi:while}. Thus, $S_j$ is not
	empty. By Lemma~\ref{lem:games:sj_sepcutGji}, the set $S_j$ is a separating
	cut in $(\Gamma^*)^j_{i^*}$.  We make a detour to set up some claims.

        We need the following definition because it helps us translate plays and
        strategies from a vertex to its copies.
        \begin{definition}
          \label{def:is}
          If $\Gamma^*_s$ is an induced subgraph of $\Gamma^*$ such that for all
          $(u,\ell_s)$ in $\Gamma^*_s$ we have that $\copies{(u,\ell_s)}$ are
          also in $\Gamma^*_s$, then we say that $\Gamma^*_s$ has the
          \emph{induced-symmetry} property or that it is \emph{symmetrically
            induced}.
        \end{definition}
        The following claim is about the translation of a strategy from a vertex to
        its copy.
        \begin{claim}
          \label{clm:transstrat}
          Suppose $\Gamma^*_s$ is symmetrically induced.  Then, in $\Gamma^*_s$,
          if a player has a strategy to reach a copy of $w$ from a copy of $u$,
          then from all copies of $u$, she has a strategy to reach some copy of
          $w$.  More formally, in $\Gamma^*_s$, if player~$\rho$ has a strategy
          $\pi$ to reach $(w,\ell_d)$ from $(u,\ell_s)$, then for all copies
          $(u,\ell'_s)$, she also has a strategy $\pi'$ to reach $(w,\ell'_d)$
          for some $\ell'_d$.
        \end{claim}
        \begin{proof}
          We define $\pi'$.  Consider a finite feasible play
          $\lambda^{(u,\ell'_s)}$ that ends in a player-$\rho$ vertex $(v,j)$.
          Let $\pi(\move(\lambda^{(u,\ell'_s)},\ell_s)) = (y, p)$.  Define
          $\pi'(\lambda^{(v,\ell'_s)}) = (y, \nextl(v,y,j))$.  Now, the play
          $\move(\lambda^{(u,\ell'_s)},\ell_s)$ is feasible and the strategy
          $\pi'$ is well defined because $\Gamma^*_s$ is symmetrically induced.

          We argue why player~$\rho$ can reach a copy of $w$ using $\pi'$.  Let
          $\sigma'$ be an arbitrary strategy for the other player, i.e.,
          player~$(3-\rho)$.  For any finite feasible play
          $\lambda^{(u,\ell_s)}$ that ends in a player-$(3-\rho)$ vertex
          $(v,j)$, let $\sigma'(\move(\lambda^{(u,\ell_s)},\ell'_s)) = (y,p)$.
          Define $\sigma(\lambda^{(u,\ell_s)}) = (y, \nextl(v,y,j))$.  Again,
          $\move(\lambda^{(u,\ell_s)},\ell'_s)$ is feasible and $\sigma$ is well
          defined because $\Gamma^*_s$ is symmetrically induced.

		  \begin{sloppypar}
			  Now, it is straightforward to show by induction on $k$ that the first
			  components of $\omega((u,\ell_s),\sigma,\pi)_k$ and
			  $\omega((u,\ell'_s),\sigma',\pi')_k$ are the same.  This means that if
			  $\omega((u,\ell_s),\sigma,\pi)_k$ reaches $(w,\ell_d)$, then
			  $\omega((u,\ell'_s),\sigma',\pi')_k$ reaches $(w,\ell'_d)$ for some
			  $\ell'_d$.
		  \end{sloppypar}
        \end{proof}

        The following claim is a formal version of the first key idea.
	\begin{claim}\label{claim:games:copiesdeleted}
          If a vertex $(v,\ell)$ is in $D_j$, then
          $\copies{(v,\ell)} \subseteq D_j$; and, $(\Gamma^*)^j$ has induced
          symmetry.
	\end{claim}
	\begin{proof}
          We prove the claim by induction on $j$.

          \para{Base case, $j = 0$.}  If $(v,\ell) \in D_0$, then there is a
          player-2 strategy $\pi_1$ to reach $(w,p) \in S_0$.  The set
		  $S_0 = V\setminus \attr{1}{B^*}{\Gamma^*}$ is a player-1 closed set by Observation~\ref{obs:attractorclosedset}:
          This means that
          there is a player-2 strategy $\pi_2$ to stay inside $S_0$.  By
          construction of $\Gamma^*$, any edge from a non-layer-$d$ vertex goes
          to the next layer or to layer-$0$.  Then, since
          $S_0 \cap B^* = \emptyset$, that is, since $S_0$ does not contain any
          layer-$d$ vertices, any (infinite) play that stays inside $S_0$ must
          eventually return to layer-$0$.  Thus, player~2 can first use $\pi_1$
          to reach $(w,p) \in S_0$ from $(v,\ell)$, then use $\pi_2$ to reach
          $(x,0)\in S_0$ from $(w,p)$; effectively, this gives a player-2
          strategy to go to $(x,0)\in S_0$ from $(v,\ell)$.  Then, by
          Claim~\ref{clm:transstrat}, player~2 has a strategy to reach a copy of
          $(x,0)$ from $(v,\ell')$ for any $\ell'$ because $\Gamma^*$ itself has
          induced symmetry.  Now, $(x,0)$ does not have any other copy, this
          means player~2 has a strategy to reach $(x,0)\in S_0$ from
          $(v,\ell')$.  By induced symmetry of $\Gamma^*$ again, we have that
          all copies of $(v,\ell)$, i.e., $\copies{(v,\ell)}$ are in $\Gamma^*$;
          moreover, by the above argument, for each of these copies, there is a
          player-2 strategy to reach $S_0$, which implies that
          $\copies{(v,\ell)} \subseteq D_0$.  Noting that
          $(\Gamma^*)^0 = \Gamma^*$ has induced symmetry finishes the base case.

          \para{Induction step, $j \ge 1$.} By induction hypothesis,
          $(\Gamma^*)^{j-1}$ has induced symmetry, and if a vertex $(v,\ell)$ is
          in $D_{j-1}$, then $\copies{(v,\ell)} \subseteq D_{j-1}$.  This
          implies that deleting $D_{j-1}$ from $(\Gamma^*)^{j-1}$ to get
          $(\Gamma^*)^j$ means deleting all copies of a vertex being deleted.
          Therefore, since $(\Gamma^*)^{j-1}$ has induced symmetry,
          $(\Gamma^*)^j$ also has induced symmetry.

          Since $S_j$ is a separating cut (by
          Lemma~\ref{lem:games:sj_sepcutGj_attr1}), it is a player-1 closed set.
          Thus, by the same argument as in the base case that uses the induced
          symmetry of $(\Gamma^*)^j$, if $(v,\ell)$ is in $D_j$, then
          $\copies{(v,\ell)} \subseteq D_j$.  This completes the induction step
          and the proof.
	\end{proof}

The following claim is the formal proof of the second key idea.	
	\begin{claim}\label{claim:games:sizeofsj}
		The set $S_j$ contains at least $2^{i^*-1}$ \emph{distinct} vertices.
	\end{claim}
	\begin{proof}
		The proof is similar to the proof of~\cite[Lemma 2.13]{CH14} except that we must now argue that all of the
		$2^{i^*-1}$ vertices are distinct.
		Consider the set $S_j$ in the game graph of the iteration before, i.e., we argue about $S_j$ in $(\Gamma^*)^j_{{i^*}-1}$. 
		Note that we have the following two cases.
		\begin{itemize}
			\item In the first case, $S_j$ contains a player-1 vertex $(x,\ell)$ for $1 \leq \ell \leq d$ that is blue in $(\Gamma^*)^j_{{i^*}-1}$.
				Thus, $(x,\ell)$ has outdegree at least $2^{{i^*}-1}$ in $(\Gamma^*)^j_{{i^*}}$ 
				and none of these edges go to vertices in $V^j \setminus S_j$ in $(\Gamma^*)^j_{i^*}$.
				Thus, $S_j$ contains at least $2^{{i^*}-1}$ vertices. Note that vertex $(x,\ell)$ can only have edges to
				vertices which are distinct to $(x,\ell)$, i.e., for all $((x,\ell), (y,\ell')) \in E^*$ we have 
				$x \neq y$ because the game graph $\Gamma$ does not have self loops.
			\item In the second case, all player-1 vertices in $S_j$ are white in $(\Gamma^*)^j_{{i^*}-1}$. 
				Thus, their outedges in	$(\Gamma^*)^j_{i^*}$ and $(\Gamma^*)^j_{{i^*}-1}$ are identical.
				We now argue, why a player-2 vertex in $S_j$ exists: 
				Assume for contradiction that no player-2 vertex in $S_j$ exists. Hence, $S_j$ is a
				separating cut only consisting of player-1 vertices. As $S_j$ is a separating cut in
				$(\Gamma^*)^j_{{i^*}}$ we have $S_j \cap B = \emptyset$. Thus, $S_j$ is also a separating
				cut in $(\Gamma^*)^j_{{i^*}-1}$. But then, by Lemma~\ref{lem:games:sepcut_in_sj}, the
				algorithm would have terminated in iteration ${i^*}-1$ which is a contradiction because
				it terminated in iteration ${i^*}$.
				
				Note that repeat-until loop at
				Lines~\ref{alg:games:fastbuchi:repeat}--\ref{alg:games:fastbuchi:until}  would have stopped in 
				iteration ${i^*}-1$ in $(\Gamma^*)^j_{{i^*}-1}$ as all player-1 vertices in
				$S_j$ are white. 

				Consider a player-2 vertex $u$ in $S_j$.
				Note that $u$ must have an edge $(u,v) \in (E^*)^j_i$
				with $v \in S_j$ because $S_j$ is a separating cut in
				$(\Gamma^*)^j_{i^*}$ (Lemma~\ref{lem:games:sj_sepcutGji}).
				Again, there are two possibilities:
				\begin{itemize}
					\item For all player-2 vertices $u \in S_j$ there exists a vertex $v \in S_j$ with
						$(u,v) \in (E^*)^j_{i^*-1}$. 
						But then $S_j$ would be a separating cut in $(\Gamma^*)^j_{{i^*}-1}$ as the
						outedges of player~1 are identical in $(\Gamma^*)^j_{i^*}$ and $(\Gamma^*)^j_{{i^*}-1}$. By
						Lemma~\ref{lem:games:sepcut_in_sj}, the separating cut would have been found in
						iteration ${i^*}-1$ of the repeat-until loop at
						Line~\ref{alg:games:fastbuchi:repeat}, which is a contradiction.
					\item Therefore, 
						there exists a player-2 vertex $u \in S_j$ that has an edge $(u,v) \in (E^*)^j_{i^*}$ to
						a vertex $v \in S_j$ but this edge is not contained in $(E^*)^j_{{i^*}-1}$. This can only
						happen if $v$ has at least $2^{{i^*}-1}$ other inedges in $(E^*)^j_{{i^*}-1}$.
						Note that $u$ is a player-2 vertex not in $(B^*)^j$ (because all vertices of
						$(B^*)^j$ belong to $Y^j$), and hence the edge $(u,v)$ has
						priority~1 and recall that by the fixed inorder of edges priority-1 edges come
						before all priority-0 edges. Thus, it follows that since the edge $(u,v)$ is not in
						$(\Gamma^*)^j_{{i^*}-1}$, all inedges of $v$ that are in $(\Gamma^*)^j_{{i^*}-1}$ must have
						priority~1 by the fixed order of inedges, that is, all the inedges of $v$ in
						$(\Gamma^*)^j_{{i^*}-1}$ are from non-B\"uchi player-2 vertices. 
						Note that $v \in S_j$ and since $S_j$ is a separating cut and, thus, a closed
						set, all player-2 vertices
						which are not in $B^*$ with an edge to $v$ are also in $S_j$.
						Since $v$ has at least $2^{{i^*}-1}$ inedges from player-2 vertices which are not in $B^*$,
						the	set $S_j$ must contain at least $2^{{i^*}-1}$ vertices.

						Furthermore, all incoming edges are from distinct vertices: Note that
						$v$	cannot be a layer~$0$ vertex of $\Gamma^*$, 
						because by Observation~\ref{obs:gg:layer0nop2incoming} all
						vertices in $B$ of the given bounded B\"uchi objective have no incoming edges from a player-2 vertex.
						Also, layer-$d$ vertices cannot be in $S_j$ as they are in $B^*$ and would be in the 
						player-1 attractor $Y^j_{i^*}$ computed at Line~\ref{alg:games:fastbuchi:p1attractor}.
						All other vertices in $\Gamma^*$ have incoming edges only from distinct
						vertices. Thus, all $2^{{i^*}-1}$ such vertices are distinct.\qedhere
				\end{itemize}
		\end{itemize}

	\end{proof}

Due to Claim~\ref{claim:games:sizeofsj}, $S_j$ contains at least $2^{{i^*}-1}$ distinct vertices, and since $S_j\subseteq D_j$, the set $D_j$  
also contains all copies of all vertices in $S_j$ due to Claim~\ref{claim:games:copiesdeleted}. 
All of $D_j$ is deleted. 
We resume from the detour.  The time spent in all graphs $(\Gamma^*)^j_1, \dots, (\Gamma^*)^j_{{i^*}}$, i.e., the time spent in the
repeat-until loop at Line~\ref{alg:games:fastbuchi:repeat} for the graph construction and the attractor computations, 
sums up to $O(2^{{i^*}} \cdot nd)$. 
We charge $O(nd)$ work to each distinct vertex. 
This accounts for all the running time except for the last iteration of the outer loop. 
Since we always remove all copies of a vertex $v \in S_j$, 
the algorithm deletes at most $n$ distinct vertices throughout a run of the algorithm.
Thus, the total time spent over the whole algorithm other than the last iteration is $O(n^2d)$.

\para{The last iteration of the outer loop.}
In the last iteration $j^*$ of the outer loop, when no vertex is deleted, the algorithm works on all
$\log n$ game graphs, spending time $O(n \cdot 2^i)$ on game graph $(\Gamma^*)^{j^*}_i$.
Since each graph $(\Gamma^*)^{j^*}_i$ has at most $nd \cdot 2^{i+1}$ edges and there are $\log n$
graphs, the total number of edges worked in the last iteration is

$$ \sum^{\log n}_{i = 1} nd \cdot 2^{i+1} = 4 nd \sum_{i=1}^{\log n}2^{i-1} = 4nd(2^{\log n} - 1) =
4nd (n-1) = O(n^2d).$$

\end{proof}

\begin{theorem}
	The set of winning vertices for the bounded B\"uchi objective and bounded coB\"uchi objectives in game graphs can be computed in time
	$O(n^2d) = O(n^3)$.
\end{theorem}

	\chapter[Algorithms and Conditional Lower Bounds for Planning Problems][Algs.\@ and
	CLBs for Planning Problems]{Algorithms and Conditional Lower Bounds for Planning Problems}\label{cha:icaps}
	In this chapter, we consider queries of reachability objectives and sequential reachability objectives in graphs, MDPs, and game graphs.
\section{Introduction}
One of the basic and fundamental algorithmic problems in 
artificial intelligence is the \emph{planning problem}~\cite{LaValle,AIBook}. The most
basic planning problem is the \emph{discrete feasible planning problem}~\cite{LaValle}.
The problem has a finite \emph{state space} and a finite amount of \emph{actions} for each state. 
Starting from an \emph{initial state}, the planner repeatedly chooses an available action at the current state
which, as a result, produces a new current state as described by a \emph{state transition function}. 
The question is if the planner can produce a state which is in a certain subset of the state space called
\emph{goal} or \emph{target}\footnote{This chapter includes results originally intended for the planning community and thus we
	use planning-specific language and motivate the problems from a ``planning perspective''.}.

\smallskip\noindent{\textbf{Planning models.}} 
We study this problem in the following classical models: 
\begin{itemize}
	\item \emph{Graphs.} Discrete Feasible Planning can be directly translated into a graph
		search problem: The vertices in the graph describe the state space and for every action
		in a state, there is an edge to the vertex which corresponds to the new state given by the state transition
		function~\cite{LaValle,AIBook}. 

	\item \emph{MDPs.} In the presence of interaction with nature, the graph model is 
		extended with probabilities or stochastic transitions, which gives rise to 
		Markov Decision Processes (MDPs)~\cite{Howard,Puterman,FV97,PT87, GuestrinKPV03}.

	\item \emph{Games on graphs.} In the presence of 
		interaction with an adversarial 
		environment, the graph model is extended to game graphs 
		(or AND-OR graphs)~\cite{JACM85,hansen98andor}.
\end{itemize}

\smallskip\noindent{\textbf{Planning problems.}} 
The planner tries to solve a planning problem given one of the above described planning models. 
The starting position is not restricted to the vertices controlled by the planner but can be any kind of vertex in the considered model. 
We consider the following basic planning problems: 
\begin{itemize}
	\item \emph{Reachability.} Given a set $T$ of target vertices 
		the goal is to determine if some target vertex from the starting position is reachable.

	\item \emph{Coverage.} 
		In the coverage problem we are given $k$ different target sets, namely, 
		$T_1,\ldots,T_k$, and a starting vertex. 
		The coverage problem asks whether we can 
		achieve reachability for all target sets 
		$T_i$ where $1 \leq i \leq k$. 
		Coverage models the following scenarios: Consider 
		a robot stationed in an outpost and $k$ different locations of interest.
		If an event or an attack happens in one of the locations, 
		then that location must be reached.
		However, the location of the event or the attack is not 
		known in advance and the robot must be prepared that the 
		target set could be any of the $k$ target sets.
	\item \emph{AllCoverage}.
		In the AllCoverage problem there are again $k$ different 
		target sets $T_1, \dots, T_k$ but in contrast to
		Coverage we want to determine \emph{all starting positions}
		where Coverage with $T_1, \dots, T_k$ holds.
		This corresponds to finding a viable outpost for the 
		above described robot.

	\item \emph{Sequential reachability.} 
		In the \emph{sequential reachability} problem we are given $k$ different target sets, 
		namely, $T_1,T_2,\ldots,T_k$ and a starting position. 
		The goal is to output whether we can first reach $T_1$, 
		then $T_2$ and so on up to $T_k$ from the starting position. 
		This represents the scenario 
		that the tasks must be achieved in a sequence by the planner.

\end{itemize}
The above are natural planning problems and have been studied widely in 
the literature, e.g., in robot planning~\cite{KGFP09,kaelbling1998planning,choset2005principles}.

\smallskip\noindent{\textbf{Basic Planning Questions.}}
For the above problems the \emph{basic} planning questions are as follows:
(a)~for graphs, the question is whether there exists a plan (or a path) such 
that the planning problem is solved;
(b)~for MDPs, the basic question is whether there exists a strategy such that the 
planning problems is satisfied almost-surely (i.e., with probability~1);
and (c)~for games on graphs, the basic question is whether there exists a strategy 
that solves the planning problem irrespective of the choices of the adversary.
The almost-sure satisfaction for MDPs is also known as the strong cyclic planning 
in the planning literature~\cite{CPRT03}, and games on graphs question represent 
planning in the presence of a worst-case adversary~\cite{JACM85,hansen98andor} (aka 
adversarial planning, strong planning~\cite{MBKS14}, or conformant/contingent
planning~\cite{bonet2000planning,hoffmann2005contingent,palacios2007conformant}).

\smallskip\noindent{\textbf{Algorithmic study.}}
In this chapter, we study the planning problems 
for graphs, MDPs, and games on graphs algorithmically.
For all the above questions, polynomial-time algorithms exist.
When polynomial-time algorithms exist, proving an unconditional lower bound is 
extremely rare. 
A new approach in complexity theory aims to establish a conditional lower bound (CLB) 
based on a well-known conjecture.
Two standard conjectures for CLBs are as follows: The 
(a)~\emph{Boolean matrix multiplication (BMM)} conjecture states that there is 
no sub-cubic combinatorial algorithm for boolean matrix multiplication; and 
the (b)~\emph{Strong exponential-time hypothesis (SETH)} states that there is 
no sub-exponential time algorithm for the k-SAT problem when $k$ grows to infinity.
Many CLBs have been established based on the above conjectures, e.g., 
for dynamic graph algorithms and string matching~\cite{abboud2014popular,bringmann2015quadratic}.

{
	\renewcommand{\arraystretch}{1.3}
	\begin{table*}
		\resizebox{\textwidth}{!}{
			\begin{tabular}{@{}l l l l l l l l l@{}}\toprule
				& \multicolumn{2}{c}{\textbf{Graphs}} && \multicolumn{2}{c}{\textbf{MDPs}} && \multicolumn{2}{c}{\textbf{Games}}\\
				\cmidrule{2-3}\cmidrule{5-6}\cmidrule{8-9}
				\textbf{Objectives}	& Upper B.& Lower B. && Upper B. & Lower B.&& Upper B.& Lower B.\\ \midrule

				Reachability  & $O(m)$ &&& $\widetilde{O}(m)$ &&& $O(m)$ \\
				Coverage & $O(m + \sum_{i=1}^k |T_i|)$ & &&  $\O(k\cdot m)$ &$\tilde{\Omega}(k \cdot
					m)$ && $O(k\cdot m)$ &$\tilde{\Omega}(k\cdot
				m)$  \\
				&& &&& (Thm.~\ref{icaps:thm:cover:lb:mdps})&&& (Thm.~\ref{icaps:thm:cover:lb:gg})\\
				AllCoverage & $O(k \cdot m)$ &$\tilde{\Omega}(k \cdot m)$	&&  $O(k\cdot m)$ &$\tilde{\Omega}(k \cdot m)$ && $O(k\cdot m)$ & 
				$\tilde{\Omega}(k\cdot m)$ \\
				&& (Thm.~\ref{icaps:thm:allcover:lb:graphs}) &&& (Thm.~\ref{icaps:thm:allcover:lb:graphs}) &&& (Thm.~\ref{icaps:thm:allcover:lb:graphs})\\

				Sequential & $O(m + \sum_{i=1}^k
				|T_i|)$   &&& $\widetilde{O}(m + \sum_{i=1}^k |T_i|)$  &&& $O(k\cdot m)$ & $\tilde{\Omega}(k\cdot m)$
				\\
				& (Thm.~\ref{icaps:thm:seq:graphs:upper}) &&& (Thm.~\ref{icaps:thm:seq:mdp:upper}) &&&& (Thm.~\ref{icaps:thm:seq:lb:gg})\\
				\bottomrule
			\end{tabular}
		}
		\caption{Algorithmic bounds where $n$ and $m$ are the number of vertices and edges of the underlying model,	and $k$ denotes the number of different target sets.
			The $\tilde{\Omega}(\cdot)$ bounds are conditional lower bounds (CLBs) under the BMM conjecture	and SETH.\@ 
			They establish that polynomial improvements over the given bound are not possible, however,	polylogarithmic improvements are not excluded. 
	Note that CLBs are quadratic for $k=\Theta(n)$. 
	The new results have a corresponding theorem statement.}\label{icaps:tab:complexity}

	\end{table*}
}

\smallskip\noindent{\textbf{Previous results and our contributions.}}
We denote by $n$ and $m$ the number of vertices and edges of the underlying model, 
and $k$ denotes the number of different target sets.
The $\widetilde{O}$ notation hides poly-log factors, e.g. $O(m{(\log n)}^4) = \widetilde{O}(m)$. 
We call a running time \emph{near-linear} if it is linear in the input but has some additional polylogarithmic factor, e.g. $O(m{(\log n)}^4)$. 
For the reachability problem, while the graphs and games on graphs problem 
can be solved in linear time~\cite{B80,I81}, 
the current best-known bound for MDPs is 
$\widetilde{O}(m)$~\cite[Theorem 12]{CDHS19CONCUR}.
For the coverage and sequential reachability, an $O(k \cdot m)$ upper 
bound follows for graphs and games on graphs, and an $\O(k \cdot m)$ upper 
bound follows for MDPs.
Our contributions are as follows:
\begin{enumerate}
	\item \emph{Coverage problem:} 
		First, we present an $O(m + \sum_{i=1}^k |T_i|)$ time algorithm for graphs; 
		second, we present an $\Omega(k\cdot m)$ lower bound for MDPs and games on graphs, 
		both under the BMM conjecture and the SETH.\@
		Note that for graphs our upper bound is in linear time, however, if each $|T_i|$ is 
		constant and $k=\theta(n)$, for MDPs and games on graphs the CLB is quadratic.

	\item \emph{Sequential reachability problem:}
		First, we present an $O(m + \sum_{i=1}^k |T_i|)$ time algorithm for graphs;
		second, we present  an $\O( m  +  \sum_{i=1}^k |T_i|)$ time algorithm for MDPs;
		and third, we present an $\Omega(k\cdot m)$ lower bound for games on graphs, 
		both under the BMM conjecture and the SETH.\@

\end{enumerate}
The summary of the results is presented in Table~\ref{icaps:tab:complexity}. The  
most interesting results are the conditional lower bounds for MDPs and game graphs
for the coverage problem, the sub-quadratic algorithm for MDPs with sequential
reachability objectives, and the conditional lower bound for game graphs with sequential reachability objectives.

\smallskip\noindent{\textbf{Practical Significance.}}
The sequential reachability and coverage problems we consider are the tasks
defined in~\cite{KGFP09}, where the problems have been studied for games on graphs
and mentioned as future work for MDPs. The
applications of these problems have been demonstrated in robotics applications.
We present a complete algorithmic picture for games on graphs and MDPs, settling
open questions related to games and future work mentioned in~\cite{KGFP09}.

\smallskip\noindent{\textbf{Theoretical Significance.}}
Our results present a very interesting algorithmic picture for the natural planning 
questions in the fundamental models.
\begin{enumerate}
	\item 
		First, we establish results showing that some models are harder than others.
		More precisely,
		\begin{itemize}
			\item for the reachability problem, the MDP model seems harder than graphs and games on 
				graphs (linear-time algorithm for graphs and games on graphs, and only near-linear time algorithms are known for MDPs);
			\item for the coverage problem, MDPs, and games on graphs are harder than graphs
				(linear-time algorithm for graphs and quadratic CLBs for MDPs and games on graphs);
			\item for the sequential reachability problem, games on graphs are harder than MDPs and graphs 
				(linear-time upper bound for graphs and sub-quadratic upper bound for MDPs, whereas 
				quadratic CLB for games on graphs).
		\end{itemize}
		In summary, we establish model-separation results with CLBs: 
		For the coverage problem, MDPs and games on graphs are algorithmically harder than graphs; and 
		for the sequential reachability problem, games on graphs are algorithmically harder than MDPs and graphs.

	\item Second, we also establish problem-separation results. 
		For the model of MDPs consider the different problems: 
		Both for reachability and sequential reachability the upper bound is 
		sub-quadratic and in contrast to the coverage problem we establish a 
		quadratic CLB.\@ 

\end{enumerate}

\paragraph{Further Related Work}
In this chapter, our focus lies on the algorithmic complexity of fundamental
planning problems and we consider 
\emph{explicit state-space} graphs, MDPs, and game graphs, where
the complexities are polynomial. The explicit model and
algorithms for it are widely considered: For example, in LTL Synthesis~\cite{KGFP09,camacho2018finite,camacho2018ltl,CamachoBM18}, 
Probabilistic Planning~\cite{KolobovMWG11,Teichteil-Konigsbuch12,KellerE12, CamachoMM16},
Nondeterministic Planning~\cite{MattmullerOHB10,FuNBY11,MuiseMB12,AlfordKNG14,CamachoTMBM17}, Contingent Planning~\cite{MuiseBM14, BonetG11} and
Verification~\cite{CH14}. 
In factored models such as STRIPS and SAS+ the complexities are higher (PSPACE-complete and NP-complete~\cite{Bylander94,BackstromN95}), and then heuristics are the focus (e.g.,~\cite{hansen98andor}) rather than the exact algorithmic complexity. 
Notable exceptions are
\begin{enumerate}
	\item the work on parameterized complexity of planning problems (e.g.,~\cite{KroneggerPP13}),
	\item conditional lower bounds based on the ETH~\cite{impagliazzo1999complexity} showing that
		certain general propositional planning problems (e.g., propositional STRIPS with negative
		goals (PSN)) do not
		admit algorithms with {running times} of the form $2^{{|P|}^{c}}$ for instance size $|P|$
		and concrete constants $c>0$~\cite{Aghighi2016, backstrom2017time}, 
	\item conditional lower bounds based on the SETH of the form $2^{(1+\varepsilon)v}\cdot poly(|P|)$
		where $v$ is the number of variables and $\varepsilon>0$ for very large subclasses PSN~\cite{backstrom2017time},
	\item conditional lower bounds based on the graph colourability problem of the form
		$2^{v/2}\cdot poly(v)$,
	\item conditional lower bounds based on the ETH showing that the minimum constraint removal problem, a well-studied problem in both robotic motion planning, does not admit algorithms with {running times} of the form $2^{o(n)}$~\cite{EibenGKY18}. 

\end{enumerate}

\section{Coverage Problem}\label{icaps:sec:utarget}
In this section, we consider the coverage query problem in graphs, MDPs and game graphs. 
We are given a starting vertex $v$ and a coverage query. Our goal is to
check if a set of player-1 strategies exist such that the resulting plays achieve 
the given coverage query when starting at $v$.

First, we present a linear-time algorithm for graphs and quadratic algorithms for MDPs and game
graphs. Then we focus on the conditional lower bounds for MDPs and game graphs, which establish
that there is no subquadratic algorithm for the coverage problem when
one assumes the STC and OV conjectures.

\subsection{Algorithms}
The results below present the upper bound for graphs, MDPs and game graphs of the second row of 
Table~\ref{icaps:tab:complexity}.

\smallskip\noindent\emph{Coverage Problem in Graphs.}
For the coverage problem in graphs we are given a graph $G = (V,E)$, 
a coverage query $\Coverage{T_1, \dots, T_k}$ and a start vertex $s \in V$.
The algorithmic problem is to find out if starting from an initial 
vertex $v$ the reachability, i.e., $\reach{T_i}$, can be achieved for 
all $1 \leq i \leq k$. 
The algorithmic solution is as follows:
Initially, mark each $v \in T_i$ for $1\leq i \leq k$ with $i$. 
Compute the BFS tree starting from $s$ and check if all the targets 
are contained in the resulting BFS tree. This instantly gives an algorithm with running time
$O(m + \sum_{i=0}^k |T_i|)$. Note that the running time is linear and thus we cannot hope to 
find any quadratic lower bounds.

\smallskip\noindent\emph{Coverage Problem in MDPs and game graphs.}
We determine in MDPs and game graphs whether there exists a set of strategies for a given coverage query with $k$ reachability objectives and start vertex $v$, by applying the reachability algorithm 
of the respective model $k$ times, i.e., once for each of the target sets. 
This yields a solution in $\O(k m)$ time for MDPs and $O(km)$ time for game graphs respectively. 
Notice that for  $k \in \Theta(n)$ the running time is quadratic in the input size.

\subsection{Conditional Lower Bounds}
We present conditional lower bounds for the coverage problem in MDPs and game
graphs (i.e., the CLBs of the second row of Table~\ref{icaps:tab:complexity}).  
For MDPs and game graphs the conditional lower bounds complement the 
quadratic algorithms from the previous subsection. 
Notice that we cannot provide a quadratic lower bound for graphs as a linear-time algorithm exists.
The conditional lower bounds are due to reductions from \emph{OV} and \emph{triangle detection}.

\subsubsection{MDPs.}

We present the following conditional lower bounds for MDPs:

\begin{theorem}\label{icaps:thm:cover:lb:mdps}
	For all $\epsilon > 0$, checking if a vertex has a set of a.s.\ winning strategies for the coverage problem in MDPs does not admit: 
	\begin{enumerate}
		\item an $O(m^{2-\epsilon})$ algorithm under Conjecture~\ref{icaps:conj:ovc},\label{icaps:thm:cover:lb:mdps1}
		\item an $O({(k\cdot m)}^{1-\epsilon})$ algorithm under Conjecture~\ref{icaps:conj:ovc},\label{icaps:thm:cover:lb:mdps2}
		\item a combinatorial $O(n^{3-\epsilon})$ algorithm under Conjecture~\ref{icaps:conj:stc} and\label{icaps:thm:cover:lb:mdps3}
		\item a combinatorial $O({(k\cdot n^2)}^{1-\epsilon})$ algorithm under Conjecture~\ref{icaps:conj:stc}.\label{icaps:thm:cover:lb:mdps4} 
	\end{enumerate}
\end{theorem}

\smallskip\noindent\emph{Using the OV-Conjecture.}
Below we prove the results~\ref{icaps:thm:cover:lb:mdps1}--\ref{icaps:thm:cover:lb:mdps2} of Theorem~\ref{icaps:thm:cover:lb:mdps}. We reduce the OV problem to Coverage in MDPs. By applying Conjecture~\ref{icaps:conj:ovc} we infer the result.

\begin{reduction}\label{icaps:red:query_mdp_ov}
	Given two sets $S_1, S_2$ of $d$-dimensional vectors, we build the MDP
	$P$ as follows.
	\begin{itemize}
		\item 
			The vertices $V$ of the MDP are given by a start vertex $s$, 
			sets of vertices $S_1$ and $S_2$ representing the sets of vectors 
			and vertices $C = \{c_i \mid 1 \leq i \leq d\}$ representing the coordinates of the vectors in the OVC instance. 

		\item 
			The edges $E$ of $P$ are defined as follows: The start vertex $s$ has an edge to every vertex of
			$S_1$. Furthermore for each $x_i \in S_1$ there is an edge to $c_j \in C$ iff 
			$x_i[j] = 1$ and for each $y_i \in S_2$ there is an edge from $c_j \in S_2$ to
			$y_i$ iff $y_i[j] = 1$. Also, the $y_i$ have self-loops so that every vertex has an outgoing edge.

		\item 
			The set of vertices is partitioned into player-1 vertices $V_1  = S_1
			\cup C \cup S_2$ and random vertices $V_R  = \{s\}$.
	\end{itemize}
\end{reduction}

\begin{example}[Example: Reduction from OV to Coverage] 
	Let the OV instance be $S_1 = \{(1,1,0), (1,0,1), (0,1,1)\}, S_2 = \{(1,0,1), (1,1,0),
	(0,1,0) \}$. Notice that the second vector in $S_1$ and the third vector in $S_2$
	are orthogonal. Due to the fact that $s$ is a random vertex, there is a nonzero probability
	that $x_2$ is the successor. There is no path from $x_2$ to $y_3$. As $T_3 = \{y_3\}$, there
	is no a.s.\ winning strategy from $s$ for the given instance of coverage.
	We illustrate the example of the reduction in Figure~\ref{icaps:fig:cov_mdps_ov}. 
	\begin{figure}[t]
		\centering
		\includegraphics{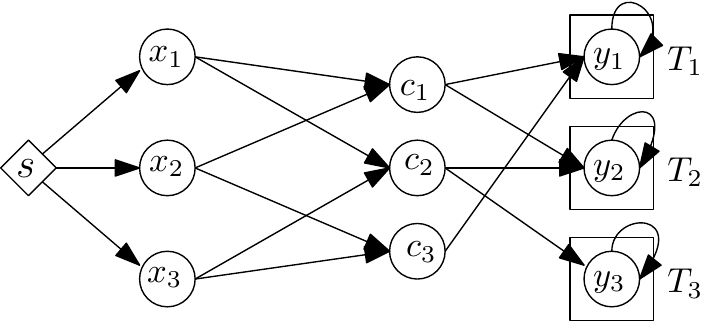}
	\caption{An example reduction from OV to Coverage in MDPs.}\label{icaps:fig:cov_mdps_ov}
\end{figure}
\end{example}

Notice that for orthogonal vectors $x_i$ and $y_j$ we have that
for each $c_\ell \in C$ either $x_i$ is not connected to $c_\ell$ or $y_j$ is not connected to $c_\ell$. 
Thus there is no path from $x_i$ to $y_i$. Starting from $s$ there is a non-zero probability to end in $x_i$
and thus also a non-zero probability to fail reaching the target set $T_j=\{y_j\}$ no matter of
player 1's strategy $\sigma_j$.

\begin{lemma}
	Let $P = (V, E, \langle V_1, V_R \rangle, \delta)$ 
	be the MDP given by Reduction~\ref{icaps:red:query_mdp_ov} and $T_i =\{y_i\}$ for $1\leq i\leq N$.
	There exist orthogonal vectors $x \in S_1$, $y \in S_2$ iff $s$ is not winning for $\Coverage{\{T_i \mid 1\leq i\leq N\}}$.
\end{lemma}

\begin{proof}
	The MDP $P$ is constructed in such a way that there is no
	path between vertex $x_i$ and $y_j$ iff the corresponding vectors are
	orthogonal in the OV instance: If $x_i$ is orthogonal to $y_j$, the outgoing
	edges lead to no vertex which has an incoming edge to $y_j$ as either $x_i[k]
	= 0$ or $y_j[k] = 0$. On the other hand, if there is no path from $x_i$ to
	$y_j$ we again have by the construction of the underlying graph that for all
	$1 \leq k \leq d: x_i[k] = 0$ or $y_j[k] =0$. This is the definition of
	orthogonality for $x_i$ and $y_j$.
    When starting from $s$ the token is randomly moved to one of the vertices $x_i$ and thus 
	player~1 can reach each $y_j$ almost surely from $s$ iff it can reach each $y_j$ from each $x_i$.	
    Thus, we have that there is an a.s.\@ winning player~1 strategy for $\reach{T_i}$ 
	iff $y_i$ has.
	Hence, $S_2$ has no orthogonal vector in $S_1$ iff 
	each $\reach{T_i}$ has an a.s.\@ winning player~1 strategy.
\end{proof}

The MDP $P$ has only $O(N)$ many vertices and Reduction~\ref{icaps:red:query_mdp_ov} can
be performed in 
$O(N \cdot d)$ time (recall that $d = \omega(\log N)$). The number of edges $m$ is $O(N \cdot d)$ 
and the number of target sets $k \in
\theta(N)$. Thus the points~\ref{icaps:thm:cover:lb:mdps1}--\ref{icaps:thm:cover:lb:mdps2} of Theorem~\ref{icaps:thm:cover:lb:mdps} follow.

\smallskip\noindent\emph{Using the ST-conjecture.}
Towards the results~\ref{icaps:thm:cover:lb:mdps3}--\ref{icaps:thm:cover:lb:mdps4} in Theorem~\ref{icaps:thm:cover:lb:mdps} we reduce the triangle detection problem to Coverage problem in MDPs. By applying Conjecture~\ref{icaps:conj:stc} we infer the result.

\begin{reduction}\label{icaps:red:query_mdp_triangle}
	Given an instance of triangle detection, i.e., a graph $G = (V,E)$, we build the
	following MDP $P = (V',E',\ls V'_1, V'_R \rs, \delta)$.
	\begin{itemize}
		\item 
			The vertices $V'$ are given as four copies $V_1,V_2,V_3,V_4$ of $V$
			and a start vertex $s$. 
		\item 
			The edges $E'$ of $P$ are defined as follows: There
			is an edge from $s$ to every $v_{1i} \in V_1$ for $i = 1 \dots n$. In
			addition for $1 \leq j \leq 4$ there is an edge from $v_{ji}$ to $v_{(j+1)k}$ iff
			$(v_i,v_k) \in E$. Finally, $v_{4i}$ for $i = 1 \dots n$ has a self-loop.

		\item 
			The set of vertices $V'$ is partitioned into player-1 vertices $V'_1 =
			\emptyset$ and random vertices $V'_R = \{s\} \cup V_1 \cup V_2 \cup V_3 \cup
			V_4$.
	\end{itemize}
\end{reduction}
Notice that all the vertices of the constructed MDP are random vertices.

\begin{example}[Reducing triangle detection to Coverage.]
	Let $G$ be the graph given in Figure~\ref{icaps:fig:cov_triangle_mdps}. 
	We construct the MDP $P$ as in Reduction~\ref{icaps:red:query_mdp_triangle}.
	Notice that $G$ has the triangle $(v_1,v_2,v_3)$ and the constructed MDP $P$ 
	has a nonzero chance to take the path marked by the fat edges that correspond to this triangle,
	i.e., player-1 does not have a winning strategy from $s$ for the coverage objective given in the 
	reduction because he cannot satisfy $T_1$. The example is illustrated in Figure~\ref{icaps:fig:cov_triangle_mdps}.
\begin{figure}[ht]
	\centering
	\includegraphics{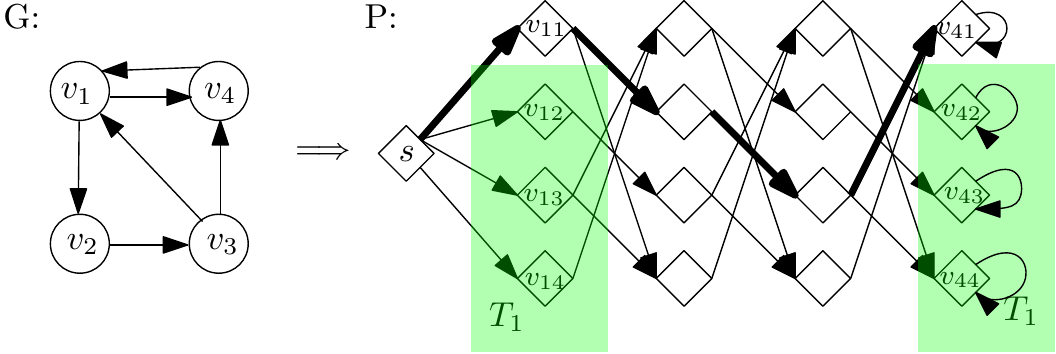}
	\caption{Reduction from Triangle to Coverage}\label{icaps:fig:cov_triangle_mdps}
\end{figure}
\end{example}

\begin{lemma}
	Let $P$ be the MDP given by Reduction~\ref{icaps:red:query_mdp_triangle} when applied to a graph $G$ and let $T_i = V_1 \setminus \{v_{1i}\} \cup V_4 \setminus \{v_{4i}\}$ for $i = 1 \dots n$ be target sets.
	The graph $G$ has a triangle iff $s$ is not winning for $\Coverage{\{T_i \mid 1\leq i\leq N\}}$ in $P$.
\end{lemma}

\begin{proof}
    First, $s$ is not winning for $\Coverage{T_1, \dots, T_N}$
    iff there is a $T_i$ such that player-1 has no of a.s.\ winning strategies from $s$ for $\reach{T_i}$.
    Second, there is a triangle in the graph $G$ iff there is a path from some vertex
	$v_{1i}$ in the first copy of $G$ to the same vertex in the fourth copy of $G$,
	$v_{4i}$. 
	Finally, notice that player~1 does not control any vertex and thus the strategy of player~1 does not matter and each
	possible path is played with non-zero probability.	
	If $G$ has a triangle containing vertex $v_i$ then the corresponding play from
	$v_{1i}$ to $v_{4i}$ has non-zero probability and is not in $\reach{T_i}$. 
	That is, $s$ is not winning for  and thus not winning for $\Coverage{\{T_i \mid 1\leq i\leq N\}}$.
	Now assume that s is not winning for  $\Coverage{\{T_i \mid 1\leq i\leq N\}}$
	and thus not winning for $\reach{T_i}$. Then there is a path from $v_{1i}$ to $v_{4i}$ and thus a triangle in $G$.
\end{proof}

Moreover, the size and the construction time of the MDP $P$
are linear in the size of the
original graph $G$ and we have $k = \theta(n)$ target sets. Thus~\ref{icaps:thm:cover:lb:mdps3}--\ref{icaps:thm:cover:lb:mdps4} of Theorem~\ref{icaps:thm:cover:lb:mdps} follow.

\subsubsection{Game Graphs.}
Next, we describe how the results for MDPs can be extended to game graphs. We prove the following theorem which states multiple specific lower bounds for checking if a vertex has a set of winning strategies for a coverage query.

\begin{theorem}\label{icaps:thm:cover:lb:gg}
	For all $\epsilon > 0$, checking if a vertex has a set of winning strategies for a coverage query in game graphs does not admit: 
	\begin{enumerate}
		\item an $O(m^{2-\epsilon})$ algorithm under Conjecture~\ref{icaps:conj:ovc},\label{icaps:thm:cover:lb:gg1}
		\item an $O({(k\cdot m)}^{1-\epsilon})$ algorithm under Conjecture~\ref{icaps:conj:ovc},\label{icaps:thm:cover:lb:gg2}
		\item a combinatorial $O(n^{3-\epsilon})$ algorithm under Conjecture~\ref{icaps:conj:stc} and\label{icaps:thm:cover:lb:gg3}
		\item a combinatorial $O({(k\cdot n^2)}^{1-\epsilon})$ algorithm under Conjecture~\ref{icaps:conj:stc}.\label{icaps:thm:cover:lb:gg4} 
	\end{enumerate}

\end{theorem}
\smallskip\noindent\emph{Using the OV-Conjecture.}
Below we prove the results~\ref{icaps:thm:cover:lb:gg1}--\ref{icaps:thm:cover:lb:gg2} of Theorem~\ref{icaps:thm:cover:lb:gg}. We reduce the OV problem to Coverage in game graphs.
By applying Conjecture~\ref{icaps:conj:ovc} we infer the result. In Reduction~\ref{icaps:red:query_games_ov} we change the random starting vertex of Reduction~\ref{icaps:red:query_mdp_ov} to a player-2 vertex. 
The rest of the reduction stays the same. 
The proof then proceeds as before with the adversary now overtaking the role of the random choices.

\begin{reduction}\label{icaps:red:query_games_ov}
	Given two sets $S_1, S_2$ of $d$-dimensional vectors, we build the following
	game graph $\Gamma = (V,E,\ls V_1, V_2 \rs)$.
	\begin{itemize}
		\item The vertices $V$ and edges $E$ are defined as before in Reduction~\ref{icaps:red:query_mdp_ov}



		\item The set of vertices is now partitioned into
			player-1 vertices $V_1  = S_1
			\cup C \cup S_2$ and player-2 vertices $V_2
			= \{s\}$.
	\end{itemize}
\end{reduction}

\begin{lemma}
	Let $\Gamma$ be the game graph given by Reduction~\ref{icaps:red:query_games_ov} 
	with a coverage query $\Coverage{\{T_i \mid 1 \leq i \leq n\}}$ where $T_i = \{y_i\}$ for $i = 1\dots N$.
	There exist orthogonal vectors $x \in S_1$, $y \in S_2$ iff there is no
	set of winning strategies from start vertex $s$ for the coverage query.
\end{lemma}

The game graph $\Gamma$ has only $O(N)$ many vertices and Reduction~\ref{icaps:red:query_games_ov} can
be performed in 
$O(N \cdot d)$ time (recall that $d = \omega(\log N)$). The number of edges $m$ is $O(N \cdot d)$ 
and the number of target sets $k \in
\theta(N)$. Thus the points~\ref{icaps:thm:cover:lb:gg1}--\ref{icaps:thm:cover:lb:gg2} in Theorem~\ref{icaps:thm:cover:lb:gg} follow.

\smallskip\noindent\emph{Using the STC conjecture.}
Below we prove the results~\ref{icaps:thm:cover:lb:gg3}--\ref{icaps:thm:cover:lb:gg4} 
in Theorem~\ref{icaps:thm:cover:lb:gg}. 
We reduce the triangle detection problem to Coverage in game graphs. By applying Conjecture~\ref{icaps:conj:stc} we infer the result. In Reduction~\ref{icaps:red:query_games_triangle} we change the random vertices of Reduction~\ref{icaps:red:query_mdp_triangle} to player-2 vertices. Notice that the
resulting game graph consists of only player-2 vertices. 
Again, if there is a path starting from $s$ which violates a reachability objectives 
in the given coverage query then player~2 wins.
As the reachability objectives are defined such that they rule out the triangles of the reduction,
player~1 only wins iff there is no such path, i.e., there is no triangle in the original graph.

\begin{reduction}\label{icaps:red:query_games_triangle}
	Given an instance of triangle detection, i.e., a graph $G = (V,E)$, we build the
	following game graph $\Gamma = (V',E',\ls V'_1, V'_2 \rs)$.
	\begin{itemize}
		\item The vertices $V'$ and Edges $E'$ are the same as in Reduction~\ref{icaps:red:query_mdp_triangle}.

		\item The set of vertices $V'$ is partitioned into player-1 vertices 
			$V'_1 = \emptyset$ and player-2 vertices $V'_2 = \{s\} \cup
			V_1 \cup V_2 \cup V_3 \cup V_4$.
	\end{itemize}
\end{reduction}

\begin{lemma}
	Let $\Gamma$ be the game graph given by Reduction~\ref{icaps:red:query_games_triangle} when applied to a graph $G$
	and let $T_i = V_1 \setminus \{v_{1i}\} \cup
	V_4 \setminus \{v_{4i}\}$ for $i = 1 \dots n$.
	The graph $G$ has a triangle iff $s$ is winning for $\Coverage{\{T_i \mid 1 \leq i \leq n\}}$ in $\Gamma$.
\end{lemma}


Moreover, the size and the construction time of game graph $\Gamma$
are linear in the size of the
original graph $G$ and we have $k = \theta(n)$ target sets.
Thus~\ref{icaps:thm:cover:lb:gg3}--\ref{icaps:thm:cover:lb:gg4} in Theorem~\ref{icaps:thm:cover:lb:gg} follow.


\section{AllCoverage Problem}
In this section, we consider the AllCoverage problem. First, we present simple algorithms for all models
based on the standard reachability problems of the models. 
Notice that the respective algorithms can also be used to
solve the corresponding Coverage Problem. Then we present a conditional lower bound for
graphs which establishes that the existing algorithm cannot be polynomially
improved under the STC and OV conjectures.

\subsection{Algorithms}
We present quadratic algorithms for MDPs, games and graphs. The results present the upper bounds for graphs, MDPs and Games third row of
Table~\ref{icaps:tab:complexity}. 

Given the query $\Coverage{\{T_i \mid 1 \leq i \leq k\}}$
for graphs, MDPs and game graphs, we propose an algorithm 
which first solves the $k$ reachability objectives using the basic results detailed in
Section~\ref{sec:basic:algorithms}.
Notice that the result of the algorithms for solving the reachability objective
is a set of vertices that have a strategy to achieve the objective.
Then we take the intersection of the resulting sets.
(1)~For graphs using BFS which is in $O(m)$ time we obtain an $O(k \cdot m)$ time algorithm.
(2)~For game graphs, using the $O(m)$-time attractor computation, 
we have an $O(k\cdot m)$ time algorithm.
(3)~For MDPs, the MEC-decomposition followed by $k$ many $O(m)$-time almost-sure 
reachability computation, 
gives an  $O(k \cdot m + \mectime)$ time algorithm.

\subsection{Conditional Lower Bounds}
In this section, we present conditional lower bounds for the AllCoverage
problem in graphs (i.e., the CLBs of the third row of Table~\ref{icaps:tab:complexity}).  
For MDPs and game graphs the conditional lower bounds follow from
Section~\ref{icaps:sec:utarget} because the Coverage problem can be trivially reduced to the
AllCoverage problem, i.e., once we have computed all the vertices that can reach a target it is easy to check 
whether a specific vertex can reach that target.
The conditional lower bounds are due to reductions from OV and the triangle detection problem.

\begin{theorem}\label{icaps:thm:allcover:lb:graphs}
For all $\epsilon > 0$, computing the solution of the AllCoverage problem in graphs does not admit
	\begin{enumerate}
		\item an $O(m^{2-\epsilon})$ algorithm under Conjecture~\ref{icaps:conj:ovc},\label{icaps:thm:allcover:lb:graphs1}
		\item an $O({(k\cdot m)}^{1-\epsilon})$ algorithm under Conjecture~\ref{icaps:conj:ovc},\label{icaps:thm:allcover:lb:graphs2}
		\item a combinatorial $O(n^{3-\epsilon})$ algorithm under Conjecture~\ref{icaps:conj:stc} and\label{icaps:thm:allcover:lb:graphs3}
		\item a combinatorial $O({(k\cdot n^2)}^{1-\epsilon})$ algorithm under Conjecture~\ref{icaps:conj:stc}.\label{icaps:thm:allcover:lb:graphs4}
	\end{enumerate}
\end{theorem}

\smallskip\noindent\emph{Using the OV-Conjecture.}
In this section we prove the results~\ref{icaps:thm:allcover:lb:graphs1}--\ref{icaps:thm:allcover:lb:graphs2} 
in Theorem~\ref{icaps:thm:allcover:lb:graphs}.
We reduce the OV problem to the AllCoverage problem in graphs. 
By applying Conjecture~\ref{icaps:conj:ovc} we infer the result.

\begin{reduction}\label{icaps:red:allcoverage_ov}
	Given two sets $S_1, S_2$ of $d$-dimensional vectors, we build the graph 
	$G$ as follows.
	\begin{itemize}
		\item 
			The construction of the graph is the same as in Reduction~\ref{icaps:red:query_mdp_ov} except the we do not have
			a vertex $s$.
	\end{itemize}
\end{reduction}

\begin{example}[Reducing AllCoverage to OV]
	Let the instance of OV be given by $S_1 = \{(1,1,0), (1,0,1), (0,1,1)\}, S_2 = \{(1,0,1), (1,1,0),
	(0,1,0) \} $. Notice that the second vector in $S_1$ and the third vector in $S_2$
	are orthogonal. We construct $G$ with Reduction~\ref{icaps:red:allcoverage_ov}. There is no path from $x_2$ to $y_3$. As $T_3 = \{y_3\}$,
	$x_2$ is not in the winning set of $\Coverage{\{T_1,T_2,T_3\}}$
	The reduction is illustrated in Figure~\ref{icaps:fig:allcoverage_ov}.	

\begin{figure}[ht]
	\centering
	\includegraphics{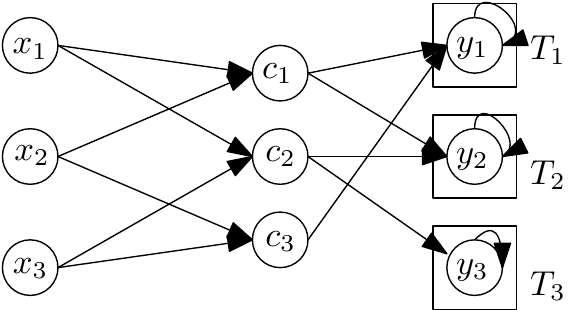}
	\caption{Reduction from OV to AllCoverage}\label{icaps:fig:allcoverage_ov}
\end{figure}
\end{example}

\begin{lemma}
	Let $G = (V, E)$
	be the graph given by Reduction~\ref{icaps:red:allcoverage_ov} with target
	sets $\Targets = \{T_i \mid T_i = \{y_i\}$ for $i = 1\dots N$ \}.
	A vector $x_i \in S_1$ is orthogonal to some vector in $S_2$ iff 
	the vertex $x_i$ is not in the winning set of $\Coverage{\Targets}$.
\end{lemma}

\begin{proof}
	The graph $P$ is constructed in such a way that there is no
	path between vertex $x_i$ and $y_j$ iff the corresponding vectors are
	orthogonal in the OV instance: If $x_i$ is orthogonal to $y_j$, the outgoing
	edges lead to no vertex which has an incoming edge to $y_j$ as either $x_i[k]
	= 0$ or $y_j[k] = 0$. On the other hand, if there is no path from $x_i$ to
	$y_j$ we again have by the construction of the underlying graph that for all
	$1 \leq k \leq d: x_i[k] = 0$ or $y_j[k] =0$. This is the definition of
	orthogonality for $x_i$ and $y_j$.
	Thus, $x_i$ is in the winning set of $\Coverage{\Targets}$ iff $x_i$ is orthogonal to some vector in $S_2$.
\end{proof}

Notice that we solve the given instance of \emph{OV} with our reduction as we compute all vectors in $S_1$ which are orthogonal to some vector in $S_2$.
The Graph $G$ has only $O(N)$ many vertices and Reduction~\ref{icaps:red:query_mdp_ov} can
be performed in 
$O(N \cdot d)$ time (recall that $d = \omega(\log N)$). The number of edges $m$ is $O(N \cdot d)$ 
and the number of target sets $k \in
\theta(N)$. Thus the points~\ref{icaps:thm:allcover:lb:graphs1}--\ref{icaps:thm:allcover:lb:graphs2} 
in Theorem~\ref{icaps:thm:allcover:lb:graphs} follow.

\smallskip\noindent\emph{Using the ST-Conjecture.}
Below we prove~\ref{icaps:thm:allcover:lb:graphs3}--\ref{icaps:thm:allcover:lb:graphs4} in Theorem~\ref{icaps:thm:allcover:lb:graphs}. We reduce the triangle detection problem to the AllCoverage problem in graphs. By applying Conjecture~\ref{icaps:conj:stc} we infer the result.

\begin{reduction}\label{icaps:red:allcoverage_bmm}
	Given an instance of triangle detection, i.e., a graph $G = (V,E)$, we build the
	following graph $G = (V',E')$.
	The vertices and edges are the same as in Reduction~\ref{icaps:red:query_mdp_triangle} except that
			we have player-1 vertices instead of random vertices and there is no start vertex $s$.
\end{reduction}

\begin{example}[Reducing Triangle to AllCoverage]
	Consider $G$ in Figure~\ref{icaps:fig:allcoverage_triangle}.
	Notice that $G$ has the triangle $(v_1 ,v_2,v_3)$ and there is a path from $v_{11}$ to
	$v_{41}$ in the constructed MDP $P$ illustrated by the strong edges corresponding to this triangle.
	The winning set of $\Coverage{T_1,T_2,T_3,T_4}$\\ contains 
	$v_{11}$
	by using the strategy 
	which uses the path from $v_{11}$ to $v_{41}$:
	First we achieve trivially, $\reach{T_2},\reach{T_3},\reach{T_4}$ with this strategy 
	by starting from $v_{11}$. Then $\reach{T_1}$ is achieved by arriving at $v_{41}$.
	\begin{figure}[ht]
		\centering
		\includegraphics{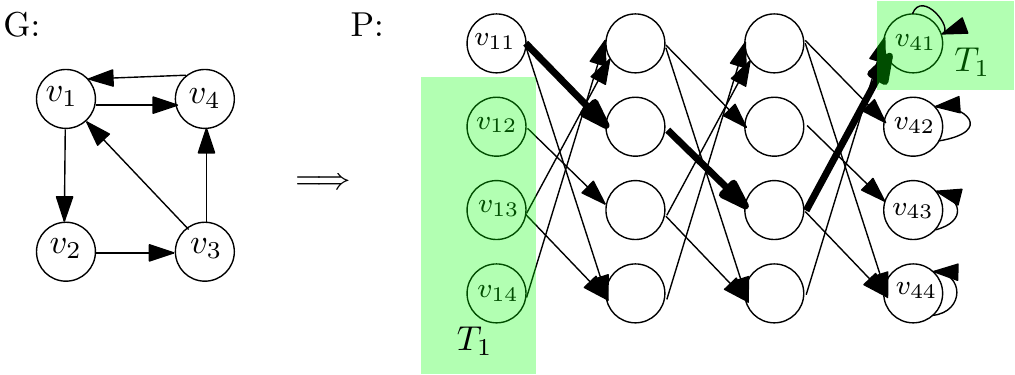}
		\caption{Reduction from Triangle to AllCoverage}\label{icaps:fig:allcoverage_triangle}
	\end{figure}
\end{example}

\begin{lemma}
	Let $G$ be the graph given by Reduction~\ref{icaps:red:allcoverage_bmm}
	with $n$ target set $T_i = V_1 \setminus \{v_{1i}\} \cup
	\{v_{4i}\}$ for $i = 1 \dots n$.  
	A graph $G$ has a triangle with vertex $v_i$ iff the vertex $v_{1i}$ is in the winning set
	of the query $\Coverage{T_1,\dots,T_n}$.
\end{lemma}

\begin{proof}
	Notice that there is a triangle in the graph $G$ iff there is a path from some vertex
	$v_{1i}$ in the first copy of $G$ to the same vertex in the fourth copy of $G$,
	$v_{4i}$. Also, a path $\sigma$ starting in $v_{1i}$ for $1 \leq i \leq n$ is a viable strategy in the 
	Coverage query for all $\reach{T_i}$ $1 \leq i \leq n$ objectives iff it is able to visit $v_{4i}$: 
	By definition, $T_j$ includes $v_{1i}$ for $1 \leq j \leq n$ and $j \neq i$. 
	Thus $\sigma$ achieves all $\reach{T_j}$ for $j \neq i$.
	To achieve $\reach{T_i}$, there must be a path from $v_{1i}$ to $v_{4i}$.
	Thus, there is a triangle with vertex $v_i$ iff
	$v_{1i}$ is in the winning set of the query $\Coverage{\Targets}$.
\end{proof}
Note that we solve the given instance of the \emph{triangle detection} problem if we know all vertices which are in triangles.
Moreover, the size and the construction time of the MDP $P$
are linear in the size of the
original graph $G$ and we have $k = \theta(n)$ target sets.
Thus~\ref{icaps:thm:allcover:lb:graphs3}--\ref{icaps:thm:allcover:lb:graphs4} 
in Theorem~\ref{icaps:thm:allcover:lb:graphs} follow.

\section{Sequential Reachability Problem}
We consider the sequential reachability problem in all models.
In contrast to the quadratic CLB for the coverage problem, 
quite surprisingly there is a subquadratic algorithm for MDPs.
We first present an algorithm for graphs and then build upon that to present the algorithm for MDPs.
For games, we present a quadratic algorithm and a quadratic CLB.\@

\subsection{Algorithms}
The results below present the upper bounds of the fourth row of
Table~\ref{icaps:tab:complexity}.

\subsubsection{Algorithm for Graphs.}
\noindent Given a graph $G=(V,E)$ and the sequential reachability objective
$\Seq{T_1, \dots, T_k}$, we compute the strongly connected components, 
contract each strongly connected component to a single vertex and 
remove multi edges. 
This results in a \emph{directed acyclic graph (DAG)}.
Additionally, the vertex $v'$ which represents an SCC $C$
in the resulting DAG $D$, is in all target sets of its members, i.e.,
$v' \in T_i$ if there exists a vertex $u \in C$ such that $u \in T_i$ for all 
$1 \leq i \leq k$.
Notice that this step does not change the reachability conditions of 
the resulting acyclic graph: 
Every vertex in an SCC can be reached starting from every other vertex in the same SCC.\@ 
Thus, it suffices to give an algorithm for DAGs.
Given a DAG $D = (V,E)$, we maintain 
(a) a set of unprocessed vertices $S$, which is initialized with $V$ and
(b) a queue $Q$ containing the vertices which are not processed but where all successors are processed, initialized with all vertices with no outgoing edges. 
Notice that the queue is initially non-empty because the bottom SCCs of $G$ are now vertices without outgoing edges in $D$.
Additionally, for each vertex $v$, we maintain the values $\mcount_v$, $\ell_v$ and $\best_v$. 
The variable $\mcount_v$ counts the number of vertices in $\Out{v}$ which are not processed yet. 
The label $\ell_v$ is such that vertex $v$ has a winning strategy for the objective
$\Seq{\Targets_{\ell_v}}$ where $\Targets_{\ell_v} = (T_{\ell_v}, \dots, T_k)$. 
In other words, there is a strategy to win from $v$ if we already visited the targets sets
$\Targets_{1}, \dots \Targets_{\ell_v-1}$. 
The variable $\best_v$ is used to store the minimum label of the already processed successors of $v$.
The algorithm proceeds as follows. While the queue $Q$ is not empty, 
we take a vertex $v$ from the queue and call \ProcessVertex{\cdot}.
The function computes the label $\ell_v$ of the vertex $v$ using $\best_v$ and the target sets where $v$ is in.
Then, it removes $v$ from $S$ and updates the variables $\best_w$ and $\mcount_w$ of its predecessors $w$.
In particular, we set $\best_w = \min(\best_w, \ell_v)$ and decrement $\mcount_w$ by one. When the queue is empty,
all vertices are processed and the algorithm terminates. 
We show that the described algorithm for DAGs 
has a linear running time, i.e., $O(m + \sum^n_{i=1} |T_i|)$ and the details are presented in Algorithm~\ref{icaps:alg:seqt_lprop}. 

\begin{algorithm}[b!]
	\KwIn{DAG $D = (V,E)$, targets $\Targets= (T_1,\dots, T_k)$}
    $S \gets V$\;
	$L_v \gets \{ i \mid v \in T_i : i = 1 \dots k $\} \quad  for $v \in V$\;
    $\mcount_v \gets |\Out{v}|$ \quad for $v \in V$\;
	$\best_v \gets \begin{cases} k+1 & \text{if $\Out{v} = \emptyset$}\\ \nill & \text{otherwise}\end{cases}$ \quad for $v \in V$\;\label{icaps:alg:seqt_lprop:bottomvs}
	$\ell_v \gets \nill$ \quad for $v \in V$\;
    $Q \gets \{v \mid \Out{v} = \emptyset\}$\;\smallskip
	\While{$S \neq \emptyset$}{\label{icaps:alg:seqt_lprop:while}
		$v = Q.\mathsf{pop}()$\;
		\textsc{\ProcessVertex{v}}\;\label{icaps:alg:seqt_callprocessv}
	}
	\Return{$\{v \in V \mid \ell_v = 1\}$}\;
	\BlankLine{}
	\Function{\ProcessVertex{$\text{Vertex } v$}}{
		$\ell_v \gets \best_v$\label{icaps:alg:seqt_lprop:l_v}\;
		\While{$\ell_v-1 \in L_v$}{\label{icaps:alg:seqt_lprop:while2}
			$\ell_v \gets \ell_v-1$\;\label{icaps:alg:seqt_lprop:l_v_end}
		}
		$S \gets S \setminus \{v \}$\;\label{icaps:alg:seqt_lprop:updateS}
		\For{$w \in \In{v}$}{\label{icaps:alg:seqt_lprop:for}
			$\best_w \gets
			\min(\best_w,\ell_v)$\;\label{icaps:alg:seqt_lprop:updatebest}
			$\mcount_w \gets
			\mcount_w-1$\;\label{icaps:alg:seqt_lprop:decrcounter}
			\If{$\mcount_w = 0$}{
				$Q.\mathsf{push}(w)$\;\label{icaps:alg:seqt_lprop:pushr}
			}
		}
	}
	\caption{Sequential Reachability in Graphs}\label{icaps:alg:seqt_lprop}
\end{algorithm}

\begin{proposition}[Correctness]~\label{icaps:lem:lprop_correct} Given a DAG $D = (V,E)$ and
	a sequential reachability objective $\Seq{\Targets}$ with target sets $\Targets = \{T_1, \dots,
	T_k\}$, Algorithm~\ref{icaps:alg:seqt_lprop} returns the set of all start
	vertices with a path for the objective $\Seq{\Targets}$.
\end{proposition}

\begin{observation}\label{icaps:obs:graph_seqt}
	The input graph has one or more vertices $v$ with $\Out{v} = \emptyset$ and thus $Q$ is
	non-empty after the initialization. 
\end{observation}

\begin{proof}
	Note that there is always a vertex $v \in V$ where $\Out{v} = \emptyset$
	because we assumed that $D$ is a DAG.\@ 
\end{proof}

The invariants below state that (a) the variables ($\best_v,\mcount_v,\ell_v$) have the intended meaning, (b)
$Q$ contains all the unprocessed vertices whose successors are already
processed and (c) that the queue contains vertices as long as $S$ is not empty.

\begin{lemma}\label{icaps:lem:seqt_graph_inv1}
	The following statements are invariants of the while loop at
	Line~\ref{icaps:alg:seqt_lprop:while}.

	\begin{enumerate}
		\item $\mcount_v = |\Out{v} \cap S|$
		\item $v \in Q$ iff $v \in S$ and all $\Out{v} \cap S = \emptyset$.
		\item If $S$ is not empty then the queue $Q$ is not empty.
		\item $\best_v = k+1$ for all $v \in V$ with $\Out{v} = \emptyset$.
		\item If $v \in Q$ then $\best_v \neq \nill$.
		\item If $v \in V \setminus S$ then $\ell_v \neq \nill$.
		\item $\best_v = \min_{w \in \Out{v}\setminus S} \ell_w$, for all $v \in V$
			with $\Out{v}\setminus S \neq \emptyset$. 
	\end{enumerate}
\end{lemma}

\begin{proof}
	\begin{enumerate}
		\item The counters $\mcount_v$ are initialized as  $|Out(v)|$ and
			$S$ is initialized as $V$. Thus the claim holds when first entering the while loop.

			Assume the claim holds at the beginning of the iteration where vertex $u$ is processed. 
			The set $S$ is only changed in Line~\ref{icaps:alg:seqt_lprop:updateS}.
			There $u$ is removed from the set. 
			The counters are only changed in Line~\ref{icaps:alg:seqt_lprop:decrcounter}:
			All counters of vertices $w$ with $u \in \Out{w}$ are decreased by one.
			Consequently $\mcount_v= |\Out{v} \cap S|$ holds for all $v \in V$ 
			also after this iteration of the loop and the claim follows.

		\item In the initial phase $S$ is set to $V$ and $Q$ is set to $\{v \in V \mid \Out{v}=\emptyset\}$. 
			Thus the claim holds when first entering the while loop.

			Assume the claim holds at the beginning of the iteration where vertex $v$ is processed. 
			The set $S$ is only changed in Line~\ref{icaps:alg:seqt_lprop:updateS}
			where $v$ is removed.

			First consider a vertex $w \in Q\setminus \{v\}$. 
			As $w$ is not removed from the set $S$ and no vertex is added to $S$
			the claim is still true for $w$.
			Now consider a vertex $w$ that might be added during the iteration of the loop.
			This can only happen in Line~\ref{icaps:alg:seqt_lprop:pushr} and the if conditions 
			ensure that $w \in S$ and $\Out{v} \cap S = \emptyset$ (by the previous invariant) and thus 
			the claim also holds for the newly added vertices.

		\item Due to Observation~\ref{icaps:obs:graph_seqt} the claim holds when first
			entering the while loop.

			Assume the claim holds at the beginning of the iteration, where vertex $v$
			is processed. The vertex $v$ is removed from $S$ in
			Line~\ref{icaps:alg:seqt_lprop:updateS} and if the set $S$ is empty now, the claim
			follows trivially. On the other hand, if $S$ is non-empty and $Q$ is also
			non-empty the claim follows again. In the third case $S$ is non-empty and $Q$
			is empty. Assume for contradiction that no vertex is added at
			line~\ref{icaps:alg:seqt_lprop:pushr}. By invariant (2), every vertex $v \in S$ has a
			successor in $S$ as otherwise, $v$ would be in $Q$. That implies that there
			exists a cycle which is a contradiction with $D$ being a DAG.\@

		\item For $v \in V$ with $\Out{v}= \emptyset$ the variables $best_v$ are initialized with $k+1$ (Line~\ref{icaps:alg:seqt_lprop:bottomvs}) and
			$best_v$ is only changed in Line~\ref{icaps:alg:seqt_lprop:updatebest}
			when a successor of the vertex is processed.
			As $v$ has no successor, $best_v$ is not changed during the algorithm.  

		\item 
			If $v \in Q$ initially, it must be due to the initialization and 
			we have $\best_v = k+1$. 
			The claim holds when first entering the while loop.
			Assume the claim holds at the beginning of the iteration where $v$
			is processed. The only time we add a vertex $w$ to $Q$ is at Line~\ref{icaps:alg:seqt_lprop:pushr}. Notice that we set $\best_w$ before at Line~\ref{icaps:alg:seqt_lprop:updatebest}.
		\item 
			Initially, every vertex is in $S$, thus the claim holds before the
			first iteration of the while loop.
			Assume the claim holds at the beginning of the iteration where
			$v$ is processed. The only time we remove a vertex from $S$ is at
			Line~\ref{icaps:alg:seqt_lprop:updateS}, i.e., 
			in $\ProcessVertex{v}$. 
			Notice that we set $\ell_v$ in Line~\ref{icaps:alg:seqt_lprop:l_v} 
			to $\best_v$ which cannot be $\nill$ 
			due to Lemma~\ref{icaps:lem:seqt_graph_inv1}~(5).

		\item 
            Initially, $V$ is $S$, and 
			for all $v \in V$ we set $\ell_v, \best_v$ to $\nill$. Notice 
			that $best_v$ with $\Out{v} \neq \emptyset$ are not changed at Line~\ref{icaps:alg:seqt_lprop:bottomvs}. 
			Thus the claim holds when the algorithm enters the loop.

			Now consider the iteration of vertex $v$ and assume the claim is true at the beginning.
			The set $S$ is only changed in Line~\ref{icaps:alg:seqt_lprop:updateS} where $v$ is removed.
			Let $S_{old}$ be the set at the beginning of the iteration and 
			$S_{new} = S_{old} \setminus \{v\}$ the updated set.
			Due to Lemma~\ref{icaps:lem:seqt_graph_inv1}~(6) $\ell_v \not= \nill$.
			For a vertex $w \in \In{v}$, the value $best_w$ is updated to $\min(\best_w,\ell_v)$
			(Line~\ref{icaps:alg:seqt_lprop:updatebest})
			which by assumption is equal to $\min_{x \in (\Out{w} \setminus S_{old}) \cup \{v\}} \ell_x = \min_{x \in (\Out{w} \setminus S_{new})} \ell_x$, i.e., the equation holds.
			For vertices $w \notin \In{v}$ 
			both $best_w$ as well as the right hand side of the equation are unchanged.
			Hence, the claim holds also after the iteration.
	\end{enumerate}
\end{proof}

From the following invariants, we obtain the correctness of our algorithm.
\begin{lemma}
	The following statements are invariants of the while loop at
	Line~\ref{icaps:alg:seqt_lprop:while} for all $v \in V \setminus S$:

	\begin{enumerate}
		\item There exists a path $p_v \in \Seq{\Targets_{\ell_v}}$.
		\item There exists no path $p_v \in \Seq{\Targets_{\ell_v-1}}$.
	\end{enumerate}
	where $\Targets_{\ell_v} = \{T_{\ell_v},\dots, T_k\}$ or $\ell_v > k$.

\end{lemma}

\begin{proof}
	As $S$ is initialized with the set of vertices $V$ the two statements
	trivially hold after the initialization.

	\smallskip\noindent Now consider the iteration where vertex $v$ is processed and assume the invariants hold at the beginning of the iteration.
	Let $be(v)= \min_{w \in \Out{v}} \ell_w $.
	By Lemma~\ref{icaps:lem:seqt_graph_inv1}~(2) we have  $\Out{v} \cap S = \emptyset$
	and by Lemma~\ref{icaps:lem:seqt_graph_inv1}~(6) also $\ell_w \not= \nill$ for all $w \in \Out{v}$.
	Thus by Lemma~\ref{icaps:lem:seqt_graph_inv1}~(7) we have $be(v)= best_v$ 
	and $\ell_v$ can be computed.
	The while loop in Line~\ref{icaps:alg:seqt_lprop:while2} decrements $\ell_v$ which is initialized to 
	$\best_{v}-1$ as long has $\best_v-1 \in L_v$.
 	$L_v$ contains $\ell_v, \dots, \best_v-1$ but does not contain $\ell_{v}-1$.

	\begin{enumerate}
		\item We next show that there is a path $p_v$ in
			$\Seq{\Targets_{\ell_v}}$: Let $w = be(v)$. 
			A path for vertex $w$ where $p_w \in \Seq{\Targets_{\ell_w}}$, exists by induction hypothesis. 
			The targets $\{T_{\ell_v}, \dots, T_{\ell_{w-1}} \}$ are visited by
			starting from $v$. The path is obtained as follows: $p_v = v, p_w$,
			which proves the claim.

		\item We next show that there is no path $p_v$ in $\Seq{\Targets_{\ell_v -1}}$.
			The current vertex $v$ is not in the set $T_{\ell_{v}-1}$ and 
			no successor $w$ has a path $p_w$ with $p_w \in \Seq{\Targets_{\ell_v
			-1}}$ because $\ell_{v}-1 <\ell_{v} \le \ell_{w}$ (Lines~\ref{icaps:alg:seqt_lprop:l_v}--\ref{icaps:alg:seqt_lprop:l_v_end}).
			Thus there is also no path $p_v \in \Seq{\Targets_{\ell_v -1}}$ which
			concludes the proof.
	\end{enumerate}
\end{proof}

\begin{proposition}~\label{icaps:lem:lprop_rtime}
	Algorithm~\ref{icaps:alg:seqt_lprop} has running time  $O( m + \sum^k_{i=1} |T_i|)$.
\end{proposition}

\begin{proof}
    We first argue that the initialization takes $O(m + \sum^k_{i=1} |T_i|)$ time. 
    Initializing the sets $L_v$ can be done by first initializing the 
	set $L_v$ as $\emptyset$ (in $O(n)$) and then iterate over all set $T_i$ and 
	for each $v \in T_i$ add $i$ to $L_v$ (in $O(\sum^k_{i=1} |T_i|)$).
	The other variables can be initialized by iterating over all vertices and for each vertex 
	consider all outgoing edges. That is in $O(n+m)=O(m)$ time.
	Now consider the main part of the algorithm.
	In the while loop we process each vertex $v \in V$ once (recall that $D$ is a DAG)
	and call the function \ProcessVertex{v} at Line~\ref{icaps:alg:seqt_callprocessv}
	In the function call \ProcessVertex{v}, 
	we iterate over the set $L_v$ (Lines~\ref{icaps:alg:seqt_lprop:while2}--\ref{icaps:alg:seqt_lprop:l_v_end})
	and all incoming edges of $v$ (Lines~\ref{icaps:alg:seqt_lprop:for}--\ref{icaps:alg:seqt_lprop:pushr}).
	If we sum over all the vertices we obtain a running time of 
	$O(m + \sum_{v \in V} |L_v|) = O(m + \sum^n_{i=1} |T_i|)$.
\end{proof}

\begin{theorem}\label{icaps:thm:seq:graphs:upper}
	Given a graph $G = (V,E)$, a sequential reachability objective  and a vertex $s \in V$
	we can decide whether $s$ is winning for $\Seq{\Targets}$ in $O(m + \sum_{i=1}^k |T_i|)$ time.
\end{theorem}

\subsubsection{Algorithm for MDPs.}
\noindent The algorithm for MDPs builds on to of the algorithm for graphs.
The first difference is that instead of computing an SCC decomposition and contracting SCCs for MDPs ws compute a MEC-decomposition and contract MECs into player-1 vertices:
Given an MDP $P =  (V,E, \langle V_1, V_R \rangle, \delta)$ 
with the sequential reachability objective $\Seq{T_1,\dots, T_k}$,
we compute the MEC-decomposition of the MDP.\@ 
Then, each MEC $M$ is contracted into a player-1 vertex $v'$
without self-loops. The resulting MDP is $P'$. 
The target sets of $P'$ are as follows: The vertex $v'$ is in all the target sets of the corresponding 
vertices in $M$, 
i.e., $v' \in T_i$ if there exists
a vertex $u \in M$ such that $u \in T_i$ for all $1 \leq i \leq k$.
Notice that this step does not change the reachability conditions of the resulting MDP:\@ Every vertex in the MEC can be
reached almost surely starting from every
other vertex in the same MEC, regardless of their type (player-1, random).
Thus it suffices to give an algorithm for MEC-free MDPs.

\smallskip\noindent\emph{Key Challenge.}
When computing a MEC-decomposition and contracting the MECs 
we get an MDP that may still contain cycles. 
Thus our MDP algorithm has to deal with cycles which are in contrast to the graph setting where we only had to deal with DAGs, i.e., in Algorithm~\ref{icaps:alg:seqt_lprop} we maintained a Queue $Q$ which contained 
all unprocessed vertices where all successors are processed. In each iteration, 
we process one such vertex.
Notice that when the queue is the only mechanism to
process the vertices, we need the fact (which we also show in Lemma~\ref{icaps:lem:seqt_graph_inv1}~(3)) 
that there either exists a vertex where all successors have been processed or all
vertices have been processed and the algorithm can terminate. When running the 
algorithm on MEC-free MDPs there might be a situation where the Queue $Q$
is empty and some vertices have not been processed yet. We illustrate such 
a situation in Example~\ref{icaps:fig:seq_mdps_challenge}.
Thus, we need an additional mechanism to process vertices in this situation.

\begin{example}[Queue empty but Graph not processed]
	Consider the MDP $P$ given in Figure~\ref{icaps:fig:seq_mdps_challenge}:
	Contracting the MECs of P into player-1 vertices, we obtain $P'$. 
	Notice that $P'$ still contains a cycle.
	Using only the queue to obtain the next vertex to process we have the following problem: After $v'_3$ is processed, the queue $Q$ is empty because $v'_4$ has still has an unprocessed succesor, namely $v'_1$. Notice that $v'_2$ and $v'_1$ have not been processed yet.

	\begin{figure}
		\centering
		\includegraphics{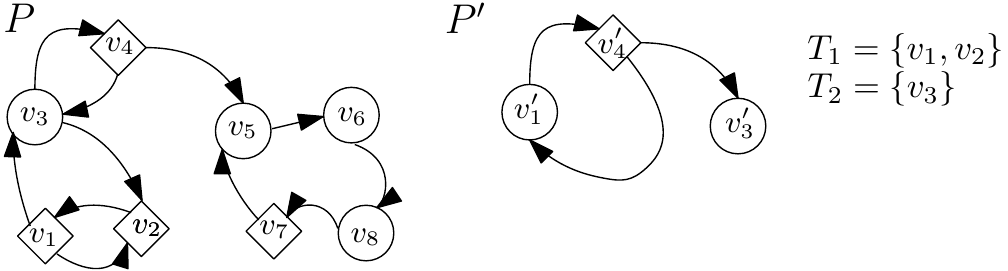}
		\caption{Key difficulty when computing Sequential Reachability in MDPs.}\label{icaps:fig:seq_mdps_challenge}
	\end{figure}
\end{example}

\smallskip\noindent\emph{Algorithm Description.}
The algorithm for MEC-free MDPs maintains the set of unprocessed vertices $S$ and a queue $Q$, the values $\mcount_v$, $\ell_v$, and $\best_v$ for each vertex $v$.
The value $\mcount_v$, as in Algorithm~\ref{icaps:alg:seqt_lprop}, stores the number of vertices in $\Out{v}$ which are not processed yet.
The label $\ell_v$ for $v$ is now such that $v$ has an \emph{almost-sure 
winning strategy} for the objective $\Seq{\Targets_{\ell_v}}$ 
where $\Targets_{\ell_v} = (T_{\ell_v}, \dots, T_k)$.
Random vertices might choose the worst possible successor with nonzero probability, i.e., the vertex with highest $\ell_v$, whereas player-1 vertices
always choose the vertex with the lowest $\ell_v$. We reflect this fact as follows in the variable $\best_v$:
The variable $\best_v$ stores the 
maximum (for $v \in V_R$) / minimum (for $v \in V_1$) 
label of the already processed successors of $v$. 
The set $S$ is initialized with $V$, and, initially, all vertices with no outgoing edges are added to the queue $Q$.
Notice that the bottom MECs of $P$ are now vertices without 
outgoing edges in $P'$ and thus $Q$ is initially non-empty.
If the queue $Q$ is non-empty, a vertex from the queue is processed as in Algorithm~\ref{icaps:alg:seqt_lprop}. 
When $Q$ is empty, the algorithm has to process a vertex where some successors
are not processed yet. In that case, we consider all the random vertices for which at least
one successor is processed and choose the random vertex with maximum $best_v$ to process next. 
We show that, as the graph has no MECs, whenever $Q$ is empty (and $S$ is not) there exists such a random vertex. 
Moreover, whenever $Q$ is empty, 
all vertices in the set of unprocessed vertices $S$ 
have a strategy that satisfies $\Seq{\Targets_{m}}$ for 
$m = \max_{v \in V_R \cap S} \mathit{best}_v$:
Intuitively, this is due to the fact that all vertices in $S$ can reach a vertex $v'$ (which is possibly different to  $v= \argmax_{v \in V_R \cap S} \mathit{best}_v$) in the set of already
processed vertices and in the worst case $\ell_{v'}=m$, i.e., they can satisfy $\Seq{\Targets_{m}}$.
For the vertex $v = \arg \max_{v \in V_R \cap S} \best_v$ all successors $w$ without a label are
in $S$ and are going to obtain a label $\ell_w$ 
of at most $m$. The current value of $\best_v$ is $m$, i.e., $v$ has a successor $w$ with $\ell_w=m$. 
As $v$ is in $V_R$ the final value of $best_v$ must be $m$. 
Hence, one can process $v$ without knowing the exact label of all the successors. 
We present the details in Algorithm~\ref{icaps:alg:seq_mdps} and prove a running time which is in $O(m \log n + \sum_{i=0}^k |T_i|)$. Notice that the running is near-linear time, linear up to a $\log n$ factor.

\begin{algorithm}[p]
	\KwIn{MEC-free MDP $P = (V,E, \ls V_1, V_R \rs, \delta)$, targets $\Targets = (T_1,\dots,T_k)$  }
	\KwOut{All vertices with a strategy for $\Seq{\Targets}$.} 
    $S \gets V$\;
	$L_v \gets \{ i \mid v \in T_i : i = 1 \dots k $\} \quad  for $v \in V$\;
    $\mcount_v \gets |\Out{v}|$ \quad for $v \in V$\;
	$\best_v \gets \begin{cases} k+1 & \text{if $\Out{v} = \emptyset$}\\ \nill & \text{otherwise}\end{cases}$ \quad for $v \in V$\label{icaps:alg:seq_mdps:bottomvs}\;
	$\ell_v \gets \nill$ \quad for $v \in V$\;
    $Q \gets \{v \mid \Out{v} = \emptyset\}$\;\smallskip
	\BlankLine{}
	\While{$S \not= \emptyset$\label{icaps:alg:seq_mdps:whileloop}}
	{
		\If{$Q \neq \emptyset$} {
			v = $Q.pop()$\;
			\ProcessVertex{v}\;
		} 
		\Else{
			$v \gets \argmax_{v \in V_R \cap S}\ best_v$\label{icaps:alg:seq_mdps:argmax}\;
			\ProcessVertex{v}\;
		}

	}
	\Return{$\{v \in V \mid \ell_v=1\}$};
	\BlankLine{}
	\BlankLine{}
	\Function{\ProcessVertex{\text{Vertex }v}}{
		$\ell_v \gets best_v $\label{icaps:alg:seq_mdps:l_v}\;
		\While{$\ell_v-1 \in L_v$\label{icaps:alg:seq_mdps:computeEllv}}
		{
			$\ell_v \gets \ell_v-1$\;
		}
		$S \gets S \setminus \{v\}$\label{icaps:alg:seq_mdps:updateS}\; 
		\For{$w \in \{w:(w,v) \in E\}$}
		{
			\uIf{$w \in V_1$\label{icaps:alg:seq_mdps:updatebest}}
			{
				$\best_w \gets \min(\best_w,\ell_v)$\label{icaps:alg:seq_mdps:updatebest1}
			}
			\Else{
				$\best_w \gets \max(\best_w,\ell_v)$\label{icaps:alg:seq_mdps:updatebest2}
			}
			$\mcount_w \gets \mcount_w -1$\label{icaps:alg:seq_mdps:decrcounter}\;
			\uIf{$\mcount_w =0 \land w \in S$}
			{
				$Q.push(w)$\label{icaps:alg:seq_mdps:pushr}
			}
		}
	}
	\caption{Sequential Reachability for MEC-free MDPs.}\label{icaps:alg:seq_mdps}
\end{algorithm}

\begin{proposition}[Correctness]
	Given an MDP $P$ and a sequential reachability objective $\Seq{\Targets}$ with targets $\Targets =
	(T_1, \dots, T_k)$, Algorithm~\ref{icaps:alg:seq_mdps} returns the set of
	all	start vertices with a player-1 strategy for the objective $\Seq{\Targets}$.
\end{proposition}

We next state invariants of the while loop (see Line~\ref{icaps:alg:seq_mdps:whileloop}) that will
enable us to show the correctness of the algorithm.
The invariants state that (a) the variables $best_v$ and $\mcount_v$ have the meaning as described in the algorithm description for all $v \in V$,
(b) $Q$ contains all the unprocessed vertices whose successors are already processed, and 
(c) that the function $\argmax$ is well-defined whenever called, i.e., there is a random vertex where $best_v$ is not $\nill$. 

\begin{lemma}\label{icaps:lem:seq_mdps:invariants1}
	The following statements are invariants of the while loop in
	Line~\ref{icaps:alg:seq_mdps:whileloop}.
	\begin{enumerate}
		\item $\mcount_v= |Out(v) \cap S|$;
		\item $v \in Q$ iff $v \in S$ and $\Out{v} \cap S = \emptyset$;
		\item $best_v = k+1$, for all $v \in V$ with $\Out{v}= \emptyset$.
		\item If $v \in Q$ we have $\best_v\not=\nill$.
		\item If $v \in V \setminus S$ we have $\ell_v\not=\nill$.
		\item For all $v \in V$ with $\Out{v}\setminus S \not= \emptyset$: $best_v = 
			\begin{cases}
				\min_{w \in \Out{v} \setminus S} \ell_w & v \in V_1 \\
				\max_{w \in \Out{v} \setminus S} \ell_w & v \in V_R
			\end{cases}$ 

		\item If $S\not= \emptyset$ and $Q=\emptyset$ there is a $v\in S \cap V_R$ such that $best_v \not= \nill$.
	\end{enumerate}
\end{lemma}

\begin{proof}
 The proofs of (1) --- (5) proceed as the proofs of the corresponding statements in the proof of Lemma~\ref{icaps:lem:seqt_graph_inv1}. 
	\begin{enumerate}	
		\item[6.] 
			
			Initially $S=V$ and for all $v \in V$ we set $\ell_v, \best_v$ to $\nill$. 
			Also, $\best_v$ with $\Out{v} \not= \emptyset$ are not changed at Line~\ref{icaps:alg:seq_mdps:bottomvs} and the claim holds when the algorithm enters the loop.

			\begin{sloppypar}
				Now consider the iteration of vertex $v$ and assume the claim is true at the beginning.
				The set $S$ is only changed in Line~\ref{icaps:alg:seq_mdps:updateS} where $v$ is removed.
				Let $S_{old}$ be the set at the beginning of the iteration and 
				$S_{new} = S_{old} \setminus \{v\}$ the updated set.
				First notice that $best_v \not= \nill$ as $v$ is either chosen by
				(a) as element of $Q$ or
				(b) by $\argmax$.
				In the former case we apply Lemma~\ref{icaps:lem:seq_mdps:invariants1}~(4)
				and in the latter case $best_v \not= \nill$ by the definition of $\argmax$.
				For a vertex $w \in \In{v} \cap V_1$ the value $best_w$ is updated to $\min(\best_w,\ell_v)$ (Line~\ref{icaps:alg:seq_mdps:updatebest1})
				which by assumption is equal to $\min_{x \in (\Out{w} \setminus S_{old}) \cup \{v\}} \ell_x = 
				\min_{x \in (\Out{w} \setminus S_{new})} \ell_x$, i.e., the equation holds.
				For a vertex $w \in \In{v} \cap V_R$  the value $best_w$ is 
				updated to $\max(\best_w,\ell_v)$ (Line~\ref{icaps:alg:seq_mdps:updatebest2})
				which by assumption is equal to 
				$\max_{x \in (\Out{w} \setminus S_{old}) \cup \{v\}} \ell_x = 
				\max_{x \in (\Out{w} \setminus S_{new})} \ell_x$, i.e., the equation holds.
				For vertices $w \notin \In{v}$ 
				$best_w$ remains unchanged. Hence, the claim holds for $w$ by the assumption
				that the invariant is true before the iteration.
			\end{sloppypar}
		\item[7.]
            Initially the statement is  true as each  MEC-free 
			MDP has a vertex $v$ with $\Out{v}= \emptyset$ and thus $Q$ is non-empty 
			(otherwise there would be an SCC with no outgoing edge which thus would be a MEC).

			Now consider the iteration processing vertex $v$ and assume the claim is true at the beginning
			and $Q=\emptyset$.
			Notice that $best_w$ is set for vertices as soon as one vertex in $\Out{w}$ was processed.
			Towards a contradiction assume that all vertices $w \in S \cap V_R$ have $best_w=\nill$, 
			i.e., no vertex $w \in S \cap V_R$ has a successor in $V \setminus S$.
			Note that $S$ contains only the vertices which are not processed yet.
			Each $w \in S$ has at least one successor in $S$ as otherwise, $w$ would be in $Q$. 
			Thus $S$ is either empty which would make the statement trivially true 
			or has again a bottom SCC (on the induced subgraph $G_P$) with more than one vertex that has no random outgoing edges. 
			Again such an SCC would be a MEC and we obtain our desired contradiction.\qedhere
	\end{enumerate}
\end{proof}

From the following invariant, we obtain the correctness of our algorithm.

\begin{lemma}
	The following statements are invariants of the while loop in Line~\ref{icaps:alg:seq_mdps:whileloop} for all $v \in V \setminus S$:
	\begin{enumerate}
		\item there exists a player~1 strategy $\sigma$ s.t. $\Pr_v^\sigma(\Seq{\Targets_{\ell_v}}) = 1$; and
		\item there is no player~1 strategy $\sigma$ s.t. $\Pr_v^\sigma(\Seq{\Targets_{\ell_v-1}}) = 1$. 
	\end{enumerate}
	where $\Targets_{\ell_v} = \{ T_{\ell_v}, \dots, T_k\}$ or $\ell_v > k$.
\end{lemma}
\begin{proof}
	As $S$ is initialized as set $V$ the two statements hold after the initialization.

	Now consider the iteration where vertex $v$ is processed and assume the invariants hold at the beginning of the iteration. 
	Notice that we do not change $\ell_v'$ for any other vertex 
	$v' \neq v$ and the invariant holds trivially for $v'$.
	We first introduce the following notation 
	\[
		be(v)=
		\begin{cases}
			\min_{w \in \Out{v}} \ell_w & v \in V_1 \\
			\max_{w \in \Out{v}} \ell_w & v \in V_R
		\end{cases}
	\]
	We distinguish the case where $Q$ is non-empty and the case where $Q$ is empty.
	\begin{itemize}
		\item \emph{Case $Q  \not= \emptyset$}:
			By Lemma~\ref{icaps:lem:seq_mdps:invariants1}~(2) we have $\Out{v} \cap S = \emptyset$. Because we only remove vertices from $S$ if we process them, 
			all $w \in \Out{v}$ are processed and thus $\ell_w \not= \nill$. 
			Thus by Lemma~\ref{icaps:lem:seq_mdps:invariants1}~(6) 
			we have $be(v)= best_v$.
			By the while-loop in Line~\ref{icaps:alg:seq_mdps:computeEllv} 
			we have $L_v \supseteq \{\ell_v, \dots, best_v-1\}$ but does not contain $\ell_{v}-1$, i.e., $\ell_{v}-1 \notin L_v$.

			(1) Thus we can easily obtain a strategy $\sigma$ 
			with $\Pr_v^\sigma(\Seq{\Targets_{\ell_v}}) = 1$ as follows.

			If $v \in V_1$ pick the vertex $w$ that corresponds to $be(v)$ and
			then player~1 can follow the existing strategy $\sigma'$ 
			for vertex $w$. Because the invariant holds for
			$w$, there exists a strategy $\sigma'$ such that $\Pr_w^{\sigma'}(\Seq{\Targets_{be(v)}}) = 1$. 

			\begin{sloppypar}
				If $v \in V_R$ 
				let vertex $w \in \Out{v}$ be the randomly chosen vertex. 
				By the invariant which holds during the iteration, 
				$w$ has a strategy $\sigma$ such that
				$\Pr_w^\sigma(\Seq{\Targets_{be(v)}}) = 1$. 
				Combined with $L_v$, this is the desired strategy, 
				i.e., $\Pr_v^\sigma(\Seq{\Targets_{\ell_v}}) = 1$

				(2) We next show that there is no strategy for $\Seq{\Targets_{\ell_v -1}}$.
				By Line~\ref{icaps:alg:seq_mdps:computeEllv} we have 
				$v \notin T_{\ell_{v}-1}$.
				If $v \in V_1$ no successor $w \in \Out{v}$ 
				has a strategy $\sigma$ with
				$\Pr_w^\sigma(\Seq{\Targets_{\ell_v -1}}) = 1$ as the invariant 
				holds also for $w$. 
				Thus there is also no strategy $\sigma$ for $v$ such that
				$\Pr_v^\sigma(\Seq{\Targets_{\ell_v -1}}) = 1$.
				If $v \in V_R$ there is at least one successor $w$ (because the invariant holds also for $w$) which has no strategy $\sigma$ such that
				$\Pr_w^\sigma(\Seq{\Targets_{\ell_v -1}}) = 1$.
				Consequently there is no strategy $\sigma$ for $v$ with
				$\Pr_v^\sigma(\Seq{\Targets_{\ell_v -1}}) = 1$ as there is a non-zero
				chance that a vertex $w$ is picked that, by the fact that the invariant holds at the current iteration, cannot reach a node in $T_{l_v-1}$.
			\end{sloppypar}

		\item \emph{Case $Q  = \emptyset$}: 
			Due to Lemma~\ref{icaps:lem:seq_mdps:invariants1}~(7) 
			there is at least one vertex in $V_R \cap S$ such that $\best_v \neq \nill$.
			Let $\best_{\max} = \max_{v \in V_R \cap S}\ best_v$.

			(1) As we have no MEC (in $S$), 
			there is a strategy $\sigma$, so that the play almost surely leaves 
			$S$ by using one of the outgoing edges of a random node: 
			Note that for all 
			random nodes between $S$ and $V\setminus S$ we have a strategy which achieves at least
			$\Seq{\Targets_{\best_{\max}}}$.
			The strategy $\sigma$ can be arbitrary, except that for a player-1
			vertex $x \in S$ with an edge $(x,y)$ where $y \in V\setminus S$ we choose
			$\sigma(x) \in S$ (which must exist as $x$ would be in $Q$
			otherwise). As there are no MECs (in $S$) the
			strategy $\sigma_1$ will eventually lead to a vertex in $V \setminus S$ using a
			random node. This implies that from 
			each vertex in $S$ player~1 has a strategy to reach a vertex in
			$V \setminus S$ coming from a random vertex. 
			Because the invariant holds at the current iteration 
			each successor of a random vertex $v'$ where $\best_{v'} \neq \nill$  has a strategy 
			to satisfy $\Seq{\Targets_{\best_{\max}}}$. Thus it follows that
			from each vertex in $S$ player~1 has a strategy to satisfy
			$\Seq{\Targets_{\best_{\max}}}$.
			Now consider the random vertex $v$ that was chosen by the 
			algorithm as $\argmax_{v \in V_R \cap S}\ best_v$. 
			Because $v'$ is a random vertex, 
			all successors have a strategy to satisfy $\Seq{\Targets_{\best_{\max}}}$ almost-surely. 
			As $L_v$ contains $\ell_v, \dots, best_v-1$ but does not contain $\ell_{v}-1$
			we obtain a strategy $\sigma$ with $\Pr_v^\sigma(\Seq{\Targets_{\ell_v}}) = 1$. 

			(2) By the choice of $v$ there is also a successor (that is chosen with non-zero probability)
			that, by assumption, has no strategy for $\Seq{\Targets_{\best_{\max}-1}}$ and,
			moreover, $L_v$ does not contain $\ell_{v}-1$.
			Thus, when starting in $v$ each strategy will fail to satisfy
			$\Seq{\Targets_{\best_{\max}-1}}$ with non-zero probability, i.e., there is no strategy
			$\sigma$ for $\Pr_v^\sigma(\Seq{\Targets_{\ell_v -1}}) = 1$.
	\end{itemize}		
\end{proof}

\begin{proposition}[Running Time]
	Algorithm~\ref{icaps:alg:seq_mdps} runs in $O(m \log n + \sum_{i=0}^k |T_i|)$ time.
\end{proposition}
\begin{proof}
	Initializing the algorithm takes $O(m + \sum_{i=0}^k |T_i|)$ time and 
	calling the function $\ProcessVertex{v}$ takes time $O(|\In{v}| + |L_v|)$ 
	(cf.\ proof of Proposition~\ref{icaps:lem:lprop_rtime}).
	Consider the main while-loop at Line~\ref{icaps:alg:seq_mdps:whileloop} where 
	every vertex is processed once 
	(recall that either $Q$ is nonempty or there exists a $v\in S \cap V_R$ such that $\best_v \neq \nill$ due to Lemma~\ref{icaps:lem:seq_mdps:invariants1}).
	The costly operations are the calls to the \ProcessVertex{\cdot} function and the evaluation of 
	the $\argmax$ function.
	Summing up over all vertices we obtain a $O(m + \sum_{i=0}^k |T_i|)$ bound for the calls to $\ProcessVertex{\cdot}$.
	To compute $\argmax$ efficiently we have to maintain a priority queue containing all not yet
	processed random vertices.
	As we have $O(m)$ updates this costs only $O(m \log n)$ for one of the standard
	implementations of priority queues. 
	Summing up this yields a $O(m \log n + \sum_{i=0}^k |T_i|)$ running time for Algorithm~\ref{icaps:alg:seq_mdps}.
\end{proof}

Recall that the running time for computing the MEC decomposition, i.e., $\mectime$ is $\O(m)$ and
thus we obtain the desired bound and the following theorem.

\begin{theorem}\label{icaps:thm:seq:mdp:upper}
	Given an MDP $P$, a start vertex $s$ and a sequential reachability objective $\Seq{\Targets}$, 
	we can compute whether there is a player-1 strategy $\sigma_1$ at $s$ for
	$\Seq{\Targets}$ in $O(\mectime + m \log n + \sum_{i=0}^k |T_i|)$ time.
\end{theorem}

\subsubsection{Algorithm for Games.}
Given a game graph $\Gamma$ with sequential reachability objectives $\Seq{\Targets}$ where $\Targets = (T_1,\dots,T_k)$, the basic algorithm (stated as Algorithm~\ref{icaps:alg:seqt_games})
performs $k$ player-1 attractor computations. 
It starts with computing the attractor $S_{k}=\attr{1}{T_k}{\Gamma}$ of $T_k$, and then iteratively computes
the sets $S_{\ell}=\attr{1}{S_{\ell+1} \cap T_{\ell}}{\Gamma}$ for $1 \leq \ell < k$,	
and finally returns the set $S_1$ as the start vertices from which player~1 can reach all the target sets in the given order.
This gives an $O(k\cdot m)$-time algorithm. 
Note that for $k = \Theta(n)$ the running time is quadratic in the input size. 

\begin{algorithm}[H]
	\KwIn{Game graph $\GG = ((V,E), \ls V_1, V_2 \rs)$ and\\ \hspace{32pt} target sets $T= (T_1,\dots, T_k)$}
	$S_{k+1} \gets V$\;
	$\ell \gets k$\;
	\While{$\ell > 0$}{
		$S_{\ell} \gets \attr{1}{T_{\ell} \cap S_{\ell+1}}{\Gamma}$\;
		$\ell \gets \ell-1$\;
	}
	\Return{$S_1$}\;
	\caption{Sequential Reachability for Games.}\label{icaps:alg:seqt_games}
\end{algorithm}

\subsection{Conditional Lower Bounds}
We present CLBs for game graphs based on the conjectures STC and OVC
which establish the CLBs for the fourth row of Table~\ref{icaps:tab:complexity}.
Notice that we cannot provide conditional lower bounds for graphs and MDPs 
as linear time algorithms for these two models exist.

\begin{theorem}\label{icaps:thm:seq:lb:gg}
	For all $\epsilon > 0$, checking if a vertex has a winning strategy for the sequential
	reachability problem in game graphs does not admit
	\begin{enumerate}
		\item an $O(m^{2-\epsilon})$ algorithm under Conjecture~\ref{icaps:conj:ovc},\label{icaps:thm:seq:lb:gg1}
		\item an $O({(k\cdot m)}^{1-\epsilon})$ algorithm under Conjecture~\ref{icaps:conj:ovc},\label{icaps:thm:seq:lb:gg2}
		\item a combinatorial $O(n^{3-\epsilon})$ algorithm under Conjecture~\ref{icaps:conj:stc} and\label{icaps:thm:seq:lb:gg3}
		\item a combinatorial $O({(k\cdot n^2)}^{1-\epsilon})$ algorithm under Conjecture~\ref{icaps:conj:stc}.\label{icaps:thm:seq:lb:gg4}
	\end{enumerate}
\end{theorem}

\smallskip\noindent\emph{Using the OV-Conjecture.}
Below we prove (\ref{icaps:thm:seq:lb:gg1}--\ref{icaps:thm:seq:lb:gg2}) of
Theorem~\ref{icaps:thm:seq:lb:gg} by reducing the OV problem to the sequential reachability problem in game graphs.
The reduction is an extension of Reduction~\ref{icaps:red:query_mdp_ov}, where we 
(a) produce a player-2 vertex instead of a random vertex and 
(b) also every vertex of $S_2$ has an edge back to $s$.

\begin{example}
	Let the OV instange be given by $S_1 = \{(1,1,0), (1,0,1), (0,1,1)\},\\ S_2 = \{(1,0,1), (1,1,0),
	(0,1,0) \} $. Notice that the second vector in $S_1$ and the third vector in $S_2$
	are orthogonal. Due to the fact that $s$ is a player-2 vertex, it can choose 
	$x_2$ as the successor. There is no path from $x_2$ to $y_3$. As $T_3 = \{y_3\}$, there
is no winning strategy for player 1 from $s$ for the given sequential reachability objective.
	We illustrate the reduction in Figure~\ref{icaps:fig:seq_games_ov}.
\begin{figure}[b]
	\centering
	\includegraphics{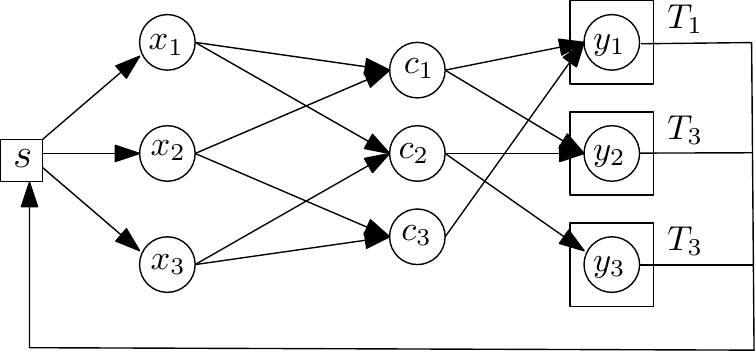}
	\caption{Reduction from OV to Sequential Reachability.}\label{icaps:fig:seq_games_ov}
\end{figure}
\end{example}

\begin{reduction}\label{icaps:red:seq_games_ov}
	Given two sets $S_1, S_2$ of $d$-dimensional vectors, we build the following
	game graph $\Gamma$.

	\begin{itemize}
		\item The vertices $V$ of the game graph are given by a start vertex $s$, 
			sets of vertices $S_1$ and $S_2$ representing the sets of vectors 
			and vertices $C = \{c_i \mid 1 \leq i \leq d\}$ representing the coordinates of the vectors in the OVC instance. 
		\item The edges $E$ of $\Gamma$ are defined as follows: the start vertex $s$ has an edge to every vertex of
			$S_1$ and every vertex of $S_2$ has an edge back to $s$; 
			furthermore for each $x_i \in S_1$ there is an edge to $c_j \in C$ iff 
			$x_i[j] = 1$ and for each $y_i \in S_2$ there is an edge from $c_j \in S_2$ to
			$y$ iff $y_i[j] = 1$.

		\item The set of vertices is partitioned into player-1 vertices $V_1  = S_1
			\cup C \cup S_2$ and player-2 vertices $V_2  = \{s\}$.
	\end{itemize}
\end{reduction}

\begin{lemma}
	Let $\Gamma$ be the game graph given by Reduction~\ref{icaps:red:seq_games_ov} with a
	sequential objective $\Seq{\Targets}$ where $\Targets = (T_1,\dots, T_k)$ and $T_i = \{y_i\}$ for $i = 1
	\dots N$.
	There exist orthogonal vectors 
	$x_i \in S_1$, $y_j \in S_2$ iff $s$ has no player-1 strategy $\sigma_1$ to ensure
	winning for the objective $\Seq{\Targets}$. 
\end{lemma}

\begin{proof}
	Notice that the game graph $\Gamma$ is constructed in
	such a way that there is no path
	between $x_i$ and $y_j$ iff they are orthogonal in the OV instance.
	Notice that each play starting at $s$ revisits $s$ every four steps and 
	if there is no path between $x_i$ and $y_j$ then player~2 can disrupt player~1 from visiting a target $T_j$
	by moving the token to $x_i$ whenever the token is in $s$.
	However, if there is no such $x_i$ and $y_j$, player~2 cannot disrupt player~1
	from $s$ because no matter which vertex $x_i$ player~2 chooses,
	player~1 has a strategy to reach the next target set. 
	If $s$ has no player-1 strategy $\sigma_1$ to ensure winning for the objective
	$\Seq{\Targets}$ there must be a target player~1 cannot reach. This must be due to the
	fact that there is no path between some $x_i$ and $y_j$ and player~2 always
	chooses $x_i$.
\end{proof}

The number of vertices in $\Gamma$, constructed by Reduction~\ref{icaps:red:query_mdp_ov}
is $O(N)$ and the construction can be performed in $O(N \log N)$ time (recall
that $(d = \omega(\log N))$). The number of edges $m$ is $O(N \log N)$ 
and the number of target sets $k \in \theta(N) = \theta(m/\log N)$. 
Thus (\ref{icaps:thm:seq:lb:gg1}--\ref{icaps:thm:seq:lb:gg2}) in Theorem~\ref{icaps:thm:seq:lb:gg} follow.

\smallskip\noindent\emph{Using the ST-Conjecture.}
In this section we prove the results~\ref{icaps:thm:seq:lb:gg3}--\ref{icaps:thm:seq:lb:gg4} in Theorem~\ref{icaps:thm:seq:lb:gg}.
We reduce the triangle detection problem to the sequential reachability problem in game graphs.
The reduction extends Reduction~\ref{icaps:red:query_mdp_triangle}, where we 
(a) produce player-2 vertices instead of random vertices and 
(b) every vertex in the fourth copy has an edge back to $s$.

	\begin{reduction}\label{icaps:red:seq_games_triangle}
		Given an instance of triangle detection, i.e., a graph $G = (V,E)$, we build the
		following game graph $\Gamma = (V',E',\ls V'_1, V'_2 \rs)$.
		\begin{itemize}
			\item The vertices $V'$ are given as four copies $V_1,V_2,V_3,V_4$ of $V$
				and a start vertex $s$. 
			\item The edges $E'$ are defined as follows: There
				is an edge from $s$ to every $v_{1i} \in V_1$ where $i=1\dots n$. In
				addition for $1 \leq j \leq 3$ there is an edge from $v_{ji}$ to $v_{(j+1)k}$ iff
				$(v_i,v_k) \in E$. Furthermore there are edges from every $v_{4i} \in V_4$ to
				the start vertex $s$.

			\item The set of vertices $V'$ is partitioned into player-1 vertices $V'_1 =
				\emptyset$ and player-2 vertices $V'_2 = \{s\} \cup V_1 \cup V_2 \cup V_3 \cup
				V_4$.
		\end{itemize}
	\end{reduction}

	\begin{example}[Reduction triangle detection to sequential reach.\ in games]
		Consider the graph $G$ given in Figure~\ref{icaps:fig:seq_triangle_games}.
		The vertices of $G$ are player-2 vertices in $\Gamma$ and the
		graph is copied four times. The edges of $\Gamma$ go to the same target but to next copy of
		the graph. Notice that $G$ has the triangle $(v_1,v_2,v_3)$ and the constructed game
		graph $\Gamma$ enables player-2 to take the path marked by the fat edges, i.e., player-1 does not
		have a winning strategy from $s$ for the sequential reachability objective given in the reduction because
		he cannot satisfy $T_1$.
		We illustrate the example of the reduction in Figure~\ref{icaps:fig:seq_triangle_games}. 
		\begin{figure}[ht]
			\centering
			\includegraphics{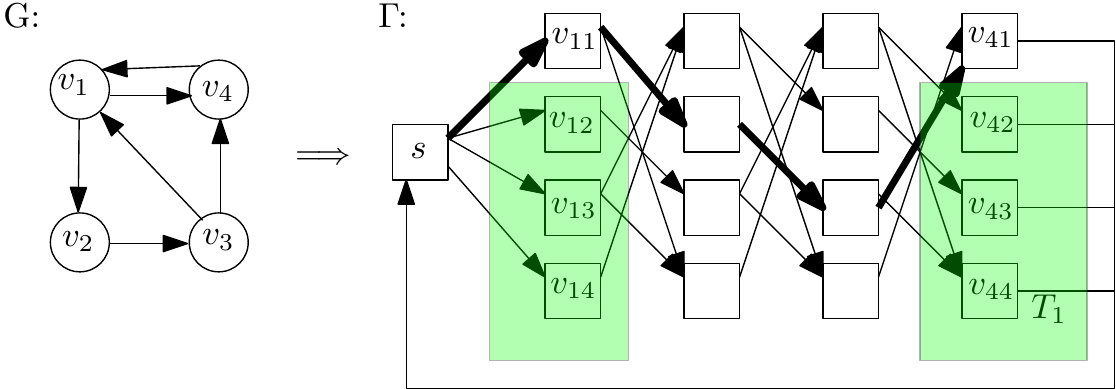}
			\caption{Reduction from Triangle to Sequential Reachability}\label{icaps:fig:seq_triangle_games}
		\end{figure}
	\end{example}

\begin{lemma}
	Let $\Gamma'$ be the game graphs given by
	Reduction~\ref{icaps:red:seq_games_triangle} with
	$\Seq{\Targets}$ as follows: $\Targets = (T_1, T_2, \dots, T_k)$ where $T_i = V_1 \setminus\{v_{1i}\} \cup
	V_4\setminus \{v_{4i}\}$ for $i = 1 \dots k$.
	The graph $G$ has a triangle iff there is no strategy $\sigma_1$ to ensure winning for the objective 
	$\Seq{\Targets}$ from start vertex s.
\end{lemma}

\begin{proof}
	For the correctness of the reduction notice that there is a triangle in the graph $G$ iff 
	there is a path from some vertex $v_{1i}$ in the first copy of $G$ to the same
	vertex in the fourth copy of $G$, $v_{4i}$ in $P$. 
	Player~2 then has a strategy to always visit only 
	$v_{1i}$ from the first copy and only $v_{4i}$ from the fourth copy 
	which prevents player~1 from visiting target $T_i$.
\end{proof}

The size and the construction time of graph $\Gamma$, given by
Reduction~\ref{icaps:red:query_mdp_triangle}, are linear in the size of the
original graph $G$ and we have $k = \Theta(n)$ target sets. 
Thus (\ref{icaps:thm:seq:lb:gg3}--\ref{icaps:thm:seq:lb:gg4}) in Theorem~\ref{icaps:thm:seq:lb:gg} follow.

\section{Discussion and Conclusion}

In this chapter, we presented lower bound results for planning objectives in an explicit state space.
We next discuss implications from these results for the same planning objectives in factored models and then
end this chapter with concluding remarks.

\subsection{Implications for factored models}
\begin{sloppypar}
Here we relate our results for the explicit state space to factored models, like STRIPS~\cite{AIBook}. 
For the reachability problem different conditional lower bounds were established
in~\cite{Aghighi2016,backstrom2017time}. In the following, we discuss to which extent our conditional lower
bounds provide lower bounds for the corresponding problems in the factored models.
We use the AllCoverage problem on graphs as an example but 
similar arguments apply to the other planning problems as well.
A planning instance of the AllCoverage problem for graphs in a factored model is given by 
variables $V$,
the domain of the variables $D$, 
the actions $A$ and
a set of conditions defining the target sets $s_{G_1}, \dots, s_{G_k}$.
A state is a function that specifies a value in $D$ to every variable in $V$. 
The planner can go from one state to another by applying the actions defined in the 
function $A$. 
$A$ is a mapping from an input state to an output state, i.e., 
the possible transitions between states are defined by the actions.
The goal of the planner is to output all states which can reach all sets $s_{G_i}$ $(1 \leq i \leq
k)$ using the actions described in $A$. 
We next investigate how our graph based lower bounds can be interpreted in the factored model.
To obtain lower bounds similar to Theorem~\ref{icaps:thm:allcover:lb:graphs} we aim to encode
the graphs computed by our reductions (see e.g. Reduction~\ref{icaps:red:allcoverage_ov} and Reduction~\ref{icaps:red:allcoverage_bmm}) in the factored model.  
A naive encoding that simply numbers the vertices and then uses a binary encoding of these numbers 
can represent $n$ states with $v = \log_2 n$ variables. 
This encoding would give lower bounds w.r.t.\ the number of variables $v$ (and thus the state space)
that exclude  $O(k \cdot 2^{(1-\epsilon)v} \cdot poly(v) )$ algorithms for the AllCoverage problem
in the factored planning model. However, these lower bounds
are only w.r.t.\ the number of variables and not w.r.t.\ the total size of the problem instance which, in particular, also includes the number of actions. 
The naive encoding requires as many actions as there are edges in the graph, i.e., the size of the
problem instance is dominated by the number of actions. Thus, using this naive encoding we do not get interesting lower bounds w.r.t.\ the instance size.
To obtain lower bounds w.r.t.\ the instance size we have to encode the vertices of the graph in a way 
that also allows us to encode the edges of the graphs compactly.
For our reductions from triangle detection, i.e., Reduction~\ref{icaps:red:allcoverage_bmm}, the edge relation can be 
rather arbitrary and, thus, the reduction is unlikely to allow for a compact representation in the factored model 
(without making additional assumptions).
In contrast, our reductions from the OV-problem, i.e., Reduction~\ref{icaps:red:allcoverage_ov}, are
well-suited for such a compact encoding as the resulting graphs are sparse
and the edge relation is defined systematically based on the content of the vertices. 
Thus, if one uses the binary vectors as a basis for the 
encoding of the vertices in the factored model then the transitions can be represented with a relatively low number of actions.
However, the details of such an encoding and the corresponding lower bounds depend on the actual factored model. 
Investigating such encodings for concrete factored models is beyond the scope of this chapter and we thus leave it open as an interesting direction for future research.
\end{sloppypar}

\subsection{Concluding Remarks}
In this chapter, we study several natural planning problems in graphs, MDPs, and game graphs, 
which are basic algorithmic problems in artificial intelligence. 
Our main contributions are a sub-quadratic algorithm for sequential reachability in
MDPs, and quadratic conditional lower bounds.
Note that graphs are a special case of both MDPs and game graphs, and the 
algorithmic problems are simplest for graphs, and in all cases except for AllCoverage, we have
linear-time upper bounds. 
The key highlight of our results is an interesting separation of MDPs and game graphs:
for reachability, MDPs are harder than game graphs;
for the coverage problem, both MDPs and game graphs are hard 
(quadratic CLBs);
for sequential reachability, game graphs are harder than MDPs.


In this chapter, we clarified the algorithmic landscape of basic planning problems with CLBs and better algorithms.
An interesting direction of future work would be to consider CLBs for 
other polynomial-time problems in planning and AI in general. For MDPs with
sequential reachability objectives, we establish sub-quadratic upper bounds, and hence the
techniques of the chapter that establish quadratic CLBs are not applicable.
Other CLB techniques for this problem are an interesting topic to
investigate as future work.

	\chapter[Quasipolynomial Set-Based Symbolic Algorithms for Parity Games][QP.\@ Set-Based Symbolic
	Algs.\@ f.\@ Parity Games]{Quasipolynomial Set-Based Symbolic Algorithms for Parity Games}\label{cha:lpar}
	In this chapter, we present the first quasi-polynomial symbolic algorithm for parity games.

\section{Introduction}\label{lpar:sec:intro}

We present new contributions related to algorithms for parity
games in the set-based symbolic model of computation.

\para{Parity games.} Games on graphs are central in many
applications in computer science, especially, in the formal analysis of reactive
systems. The vertices of the graph represent states of the system, the edges
represent transitions of the system, the infinite paths of the graph represent
traces of the system and the players represent the interacting agents.  The
reactive synthesis problem (Church's problem~\cite{Church62}) is equivalent to
constructing a \emph{winning strategy} in a graph
game~\cite{BuchiL69,RamadgeW87,PnueliR89}. Besides reactive synthesis, the
game graph problem has been used in many other applications, such as
(1)~verification of branching-time properties~\cite{EmersonJ91},
(2)~verification of open systems~\cite{AlurHK02}, 
(3)~simulation and refinement between reactive systems~\cite{Milner71,HenzingerKR02,AHKV98};
(4)~compatibility checking~\cite{InterfaceTheories}, 
(5)~program repair~\cite{JobstmannGB05},
(6)~synthesis of programs~\cite{CCHRS11}; to
name a few. Game graphs with \emph{parity winning} conditions are particularly
important since all $\omega$-regular winning conditions (such as safety,
reachability, liveness, fairness) as well as all Linear-time Temporal Logic
(LTL) winning conditions can be translated into parity
conditions~\cite{Safra88,Safra89}.  In a parity game, every vertex of the game
graph is assigned a non-negative integer priority from
$\set{0,1,\ldots,\numprio-1}$, and a play is winning if the minimum priority
visited infinitely often is even. Game graphs with parity conditions can model
all the applications mentioned above and are also equivalent to the modal
$\mu$-calculus~\cite{Kozen83} model-checking problem~\cite{EmersonJ91}.  Thus
the parity games problem is a core algorithmic problem in formal methods, and
has received wide attention over the
decades~\cite{EmersonJ91,BrowneCJLM97,Seidl96,Jurdzinski00,VogeJ00,JurdzinskiPZ08,Schewe17,calude2017stoc,lazic2017qp}.

\para{Models of computation: Explicit and symbolic
algorithms.} For the algorithmic analysis of parity games, two models of
computation are relevant.  First, the standard model of \emph{explicit}
algorithms, where the algorithms operate on the explicit representation of the
game graph.  Second, the model of \emph{implicit or symbolic} algorithms, where
the algorithms do not explicitly access the game graph but operate with a set
of predefined operations.  For parity games, the most relevant class of symbolic
algorithms are called \emph{set-based symbolic algorithms}, where the allowed
symbolic operations are: (a)~basic set operations such as union, intersection,
complement, and inclusion; and (b)~one step predecessor (Pre) operations
(see~\cite{ClarkeBook,deAlfaroHM01,HMR05}).

\para{Significance of set-based symbolic algorithms.} We
describe the two most significant aspects of set-based symbolic algorithms.
\begin{compactenum}
\item Consider large-scale finite-state systems, e.g., hardware circuits,
or programs with many Boolean variables or bounded-domain integer variables.
While the underlying game graph is described implicitly (such as program code),
the explicit game graph representation is huge (e.g., exponential in the number
of variables).  The implicit representation and symbolic algorithms often do not
incur the exponential blow-up, which is inevitable for algorithms that require
the explicit representation of the game graph.  Data-structures such as Binary
Decision Diagrams (BDDs)~\cite{bryant1986graph} (with well-established tools
e.g. CuDD~\cite{CuDD}) support symbolic algorithms that are used in verification
tools such as NuSMV~\cite{Cimatti2000}.

\item In several domains of formal analysis of infinite-state systems,
such as games of hybrid automata or timed automata, the underlying state space
is infinite, but there is a finite quotient. 
Symbolic algorithms provide a practical and scalable approach for the analysis of such systems:
For many applications, the winning set is characterized by $\mu$-calculus formulas with one-step predecessor 
operations which immediately give the desired set-based symbolic algorithms~\cite{deAlfaroHM01,deAlfaroFHMS03}. 
Thus, the set-based symbolic model of computation is an equally important theoretical model of computation to be
studied as the explicit model.
\end{compactenum}

\para{Symbolic resources.} In the explicit model of
computation, the two important resources are time and space. Similarly, in the
symbolic model of computation, the two important resources are the number of
symbolic operations and the symbolic space. 
\begin{compactitem}
\item \emph{Symbolic operations:} Since a symbolic algorithm uses a set of
predefined operations, instead of time complexity, the first efficiency measure
for a symbolic algorithm is the number of symbolic operations required.  Note
that basic set operations (that only involve variables of the current state) are less resource
intensive compared to the predecessor operations (that involve both variables of
the current and of the next state).  Thus, in our analysis, we will distinguish
between the number of basic set operations and the number of predecessor
operations. 

\item \emph{Symbolic space:} We refer to the number of sets stored by a set-based symbolic algorithm as the
symbolic space for the following reason: A set that contains all vertices, or a set that contains all vertices
where the first variable is true, represent each $\Theta(n)$ vertices, but can be
represented as BDD of constant size.  While the size of a set and the size of its symbolic
representation is notoriously hard to characterize (e.g., for BDDs it can
depend on the variable reordering), the symbolic model of computation counts every
set as unit symbolic space, and the symbolic space requirement is
thus the maximum number of sets required by a symbolic algorithm. 
\end{compactitem}
\begin{sloppypar}
The goal is to find algorithms that minimize the symbolic space (ideally
poly-logarithmic) and the symbolic operations.
\end{sloppypar}

\para{Previous results.} 
We summarize the main previous results for parity games on graphs 
with $n$ vertices, $m$ edges, and $\numprio$ priorities. 
To be concise in the following discussion, 
we ignore denominators in~$\numprio$ in the bounds.
\begin{compactitem}

\item \begin{sloppypar}\emph{Explicit algorithms.} The classical algorithm for parity games
requires $O(n^{\numprio-1}\cdot m)$ time and linear
space~\cite{Zielonka98,McNaughton93}, which was then improved by the
small progress measure algorithm that requires $O(n^{\numprio/2}\cdot m)$ time
and $O(\numprio \cdot n)$ space~\cite{Jurdzinski00}.  Many improvements
have been achieved since then, such as the big-step algorithm~\cite{Schewe17},
the sub-exponential time algorithm~\cite{JurdzinskiPZ08}, an improved algorithm
for dense graphs~\cite{ChatterjeeHL15}, and the strategy-improvement
algorithm~\cite{VogeJ00}, but the most important breakthrough was achieved in 2017 
where a quasi-polynomial time $O(n^{\lceil \log\numprio \rceil+6})$ algorithm
was obtained~\cite{calude2017stoc}.  While the original algorithm
of~\cite{calude2017stoc} required quasi-polynomial time and space, a succinct
small progress measure based algorithm~\cite{lazic2017qp} and value-iteration
based approach~\cite{FearnleyJS0W17} achieve the quasi-polynomial time bound
with quasi-linear space. However, all of the above algorithms are inherently
explicit algorithms.  
\end{sloppypar}

\item \emph{Set-based symbolic algorithms.} The basic set-based symbolic
algorithm (based on the direct evaluation of the nested fixed point of the
$\mu$-calculus formula) requires $O(n^\numprio)$ symbolic
operations and $O(\numprio)$ space~\cite{EmersonL86}.  In a breakthrough
result~\cite{BrowneCJLM97} presented a set-based symbolic algorithm that
requires $O(n^{\numprio/2+1})$ symbolic operations and $O(n^{\numprio/2+1})$
symbolic space (for a simplified exposition see~\cite{Seidl96}). In recent
work~\cite{chatterjee2017symbolic}, a new set-based symbolic algorithm was
presented that requires $O(n^{\numprio/3+1})$ symbolic operations and $O(n)$
symbolic space, where the symbolic space requirement is $O(n)$ even with a constant
number of priorities. 

\end{compactitem}

\para{Open questions.}
Despite the wealth of results for parity games, many fundamental algorithmic
questions are still open. Besides the major and long-standing open question of
the existence of a polynomial-time algorithm for parity games, two important open
questions about set-based symbolic algorithms are as follows:
\begin{compactitem} 
\item \emph{Question~1.} Does there exist a set-based symbolic algorithm that
requires only quasi-polynomially many symbolic operations?

\item \emph{Question~2.} Given the $O(\numprio)$ symbolic space requirement of
the basic algorithm, whereas all other algorithms require at least $O(n)$
space (even for a constant number of priorities) an important question is:
Does there exist a set-based symbolic algorithm that requires
$\widetilde{O}(\numprio)$ symbolic space (note that $\widetilde{O}$ hides poly-logarithmic factors), but beats the number of symbolic
operations of the basic algorithm?
This question is especially relevant since in
many applications the number of priorities is small, e.g., in
determinization of $\omega$-automata, the number of priorities is
logarithmic in the size of the automata~\cite{Safra88}.
\end{compactitem}

\begin{table}[!t]
	\centering
	\renewcommand{\arraystretch}{1.3}
	\caption{Set-Based Symbolic Algorithms for Parity
		Games.}\label{lpar:tab:comparison}
	\begin{tabular}{@{}lcc@{}}
		\toprule
		reference & symbolic operations & symbolic space\\
		\midrule
		\cite{EmersonL86,Zielonka98} & $O(n^\numprio)$ & $O(\numprio)$ \\
		\cite{BrowneCJLM97,Seidl96} & $O(n^{\numprio/2+1})$ &
		$O(n^{\numprio/2+1})$\\
		\cite{chatterjee2017symbolic} & $O(n^{\numprio/3+1})$ & $O(n)$ \\
		Thm.~\ref{lpar:thm:ordered_main},\!~\ref{lpar:thm:reducedspace_main} & $n^{O(\log
		\numprio)}$ & $O(\numprio \log n)$\\
		\bottomrule
	\end{tabular}
\end{table}

\para{Our contributions.} 
In this work, we not only answer the above open questions (Question~1 and
Question~2) in the affirmative but also show that both can be achieved by the same
algorithm:
\begin{compactitem}
\item First, we present a black-box set-based symbolic algorithm based on
explicit progress measure algorithm for parity games that uses $O(n)$
symbolic space and $n^{O(\log \numprio)}$ symbolic operations.  
There are two important consequences of our algorithm: (a)~First, given the
ordered progress measure algorithm (which is an explicit algorithm), as a consequence of our
black-box algorithm, we obtain a set-based symbolic algorithm for parity
games that requires quasi-polynomially many symbolic operations and $O(n)$
symbolic space.  (b)~Second, any future improvement in progress measure
based explicit algorithm (such as polynomial-time progress measure
algorithm) would immediately imply the same improvement for set-based
symbolic algorithms.  Thus we answer Question~1 in the affirmative and also
show that improvements in explicit progress measure algorithms carry over to
symbolic algorithms.

\item \begin{sloppypar} Second, we present a set-based symbolic algorithm that requires quasi-polynomially
many symbolic operations and $O(\numprio \cdot \log
n)=\widetilde{O}(\numprio)$ symbolic space.  Thus we not only answer
Question~2 in affirmative, we also match the number of
symbolic operations with the current best-known bounds for explicit
algorithms.  Moreover, for the important case of $\numprio \leq \log
n$, our algorithm requires polynomially many symbolic operations and
poly-logarithmic symbolic space.\end{sloppypar}
\end{compactitem}
We compare our main results with previous set-based symbolic algorithms in Table~\ref{lpar:tab:comparison}.

\para{Symbolic Implementations.}
Recently, symbolic algorithms for parity games received attention from a practical
perspective: First, three explicit algorithms (Zielonka's recursive algorithm,
Priority Promotion~\cite{BenerecettiDM16} and Fixpoint-Iteration~\cite{BFL14}) were converted to symbolic implementations~\cite{SWW18}. 
The symbolic solvers had a huge performance gain compared to the corresponding
explicit solvers on several of practical instances. 
Second, four symbolic algorithms to solve parity games were compared to their
explicit versions (Zielonka's recursive algorithm, small progress measure, and an automata-based algorithm~\cite{KV98,SMPV16})
\cite{SMV18}. For the symbolic versions of the small progress measure, two implementations were considered: (i) Symbolic
Small Progress Measure using Algebraic Decision Diagrams~\cite{BKV04} and (ii) the Set-Based Symbolic
Small Progress Measure~\cite{chatterjee2017symbolic}. The symbolic
algorithms were shown to perform better in several structured instances.

\para{Other related works.} 
Besides the discussed theoretical results on parity games,
there are several practical approaches for parity games, such as, 
(a)~accelerated progress measure~\cite{deAlfaroF07}, 
(b)~quasi-dominion~\cite{BenerecettiDM16}, 
(c)~identifying winning cores~\cite{Vester16},
(d)~BDD-based approaches~\cite{KP12,KP14}, and 
(e)~an extensive comparison of various solvers~\cite{dijk2018oink}.
A straightforward symbolic implementation (not set-based) of small progress measure was
done in~\cite{BKV04} using \emph{Algebraic Decision Diagrams (ADDs)} and BDDs. 
Unfortunately, the running time is not comparable with our results
as using ADDs breaks the boundaries of the Set-Based Symbolic Model: ADDs
can be seen as BDDs which allow the use of a finite domain at the leaves~\cite{BFGHMPS1997}.
Recently, a novel approach for solving parity games in quasi-polynomial time
which uses the register-index was introduced~\cite{Lehtinen18}.
Moreover,~\cite{Lehtinen18} presents a $\mu$-calculus formula describing the
winning regions of the parity game with alternation depth based on the
register-index. The existence of such a $\mu$-calculus formula does not immediately imply
a quasi-polynomial set-based symbolic algorithm due to constructing the formula
using the register-index. 

\smallskip\noindent
Our work considers the theoretical model of symbolic computation 
and presents a black-box algorithm as well as a quasi-polynomial algorithm,
matching the best-known bounds of explicit algorithms. 
Thus our work makes a significant contribution towards the theoretical understanding of symbolic computation for parity games.

\section{The Progress Measure Algorithm}\label{lpar:ss:pm}
\emph{High-level intuition.} Let $\PG = (V,E, \langle V_1,
V_2 \rangle, p)$ be a parity game and let $(\WOrder, \prec)$ be a finite total order
with a maximal element $\top$ and a minimal element $\min$.
Let $C = \{1,\dots,d-1\}$ be the set of all priorities.
A \emph{ranking function} is a function
$f$ which maps every vertex in $V$ to a value in $\WOrder$. The value of a
vertex $v$ concerning the ranking function $f$ is called \emph{rank} of
$v$. The rank $f(v)$ of a vertex $v$ determines how ``close'' the vertex is to
being in $W_z$, the winning set of a fixed player $z$.
Initially, the rank of every vertex is the minimal value of $\WOrder$. The
\emph{progress measure algorithm} iteratively
increases the rank of a vertex $v$ with an operator called $\Lift$ with
respect to the successors of $v$ and another function called ``$\slift$''. The algorithm
terminates when the rank of no vertex can be increased any further, i.e., the
least fixed point of $\Lift$ is reached. We call the least simultaneous fixed
point of all $\Lift$-operators \emph{progress measure}. 
%
When the rank of a vertex is the maximal element of the total order it is
declared winning for player $z$. The rest of the vertices are declared winning for the adversarial player
$\z$. 

\smallskip\noindent\emph{Ranking Function.}
Let $\WOrder$ be a total order with a minimal element $\mina$ and maximal element
$\top$. A ranking function is a function ${f: V \mapsto \WOrder}$.

\smallskip\noindent\emph{The $\best$ function.} The \emph{$\best$} function
represents the ability of player $z$,  
to choose the vertex in $\Out{v}$ with the
minimal ranking function. Analogously, it constitutes the ability of
player $\z$ to choose the vertex in $\Out{v}$ with the
maximal ranking function. Formally, the function $\best$ is defined for a vertex $v$ and a
ranking function $f$ as follows:
\begin{equation*}
\best(f,v)  = \begin{cases}
	\max \{ f(w) \mid w \in \Out {v} \} & \text{if } v \in V_{\z}\\
	\min \{ f(w) \mid w \in \Out{v} \} & \text{if } v \in V_z
\end{cases}
\end{equation*}

\para{The $\slift$-function.}
The function $\slift: \WOrder \times C \mapsto \WOrder$ defines how the rank of
a vertex $v$ is increased according to the rank $r$ of a successor vertex,
and the priority $p(v)$ of $v$. The $\slift$ function 
needs to be monotonic in the first argument. Notice that we do not need
information about the graph to compute the $\slift$ function. In all known
progress measures, the $\slift$ function is computable in constant time.

\smallskip\noindent\emph{The $\Lift$-operation.}
The $\Lift$-operation potentially increases the rank of a vertex $v$ according to
its priority $p(v)$ and the rank of all its successors in the graph.\footnote{Notice that in the original definition~\cite{Jurdzinski00} the $\slift$ function is applied to all successors and the best of them is chosen subsequently.
As the $\slift$ is monotone in the first argument the two definitions are equivalent.}
\begin{equation*}
\Lift(f,v)(u) = \begin{cases}
	\slift(\best(f,v), p(v)) & \text{if $u = v$}\\
	f(u)					  & \text{otherwise}
\end{cases}
\end{equation*}
A ranking function is
a \emph{progress measure} if it is the least simultaneous fixed point of all
$\Lift(\cdot,v)$-operators.

\smallskip\noindent\emph{The Progress Measure Algorithm.}
The progress measure algorithm initializes the ranking function $f$ with the
minimum element of $\WOrder$. Then, the $\Lift(\cdot,v)$-operator is computed in an
arbitrary order regarding the vertices. The winning set of player z can be obtained from a progress
measure by selecting those vertices whose rank is $\top$.
Notice that we need to define the total order $(\WOrder,\prec)$ and a function
$\slift$ to initialize the algorithm. 

For example, the following instantiations of the progress measure algorithm
determine the winning set of a parity game:
(i) Small Progress Measure~\cite{Jurdzinski00}, (ii) Succinct Progress
Measure~\cite{lazic2017qp} and the (iii) The Ordered
Approach~\cite{FearnleyJS0W17}. 
The running time is dominated by the size of $\WOrder$.
For a discussion on the state-of-the-art size of $\WOrder$ we refer
the reader to Remark~\ref{lpar:rem:bound}.

\section{Set-Based Symbolic Black Box Progress Measure Algorithm}\label{lpar:sec:bb3}
In the symbolic setting, a parity game $(\Gamma,p)$ is given by a game graph $\Gamma$ and $V_i =
\{v \in V \mid p(v) = i\}$ for
$i \in C = \{1,\dots,d-1\}$.
In this section, we briefly present a basic version of a set-based symbolic progress measure
algorithm. The key idea is to compute the Lift-operation with a
set-based symbolic algorithm. Then, we improve the basic 
version to obtain our black box set-based symbolic progress measure algorithm.
Finally, we prove its correctness and analyze the symbolic resources.

\smallskip\noindent\emph{Basic black box Algorithm.}
Throughout the algorithm, we maintain the family $\mathbf{S}$ of sets of vertices, which contains a set 
for every element in $\WOrder$, i.e.,  $\S = \{ S_r \mid r \in \WOrder\}$. 
Intuitively, a vertex $v \in S_r$ has rank $f(v) = r$.
Initially, we put each vertex into the set with the minimal value of $\WOrder$,
i.e., $S_\mina$. In each iteration, we consider all non-empty sets $S_r \in \S$: 
The algorithm checks if the ranking function of the predecessors of the vertices
in $S_r$ must be increased, i.e. $\Lift(f,v)(v) \succ f(v)$ where $v \in \Pre{}{S_r}$, 
and if so, performs the $\Lift$-operator for the predecessors.
We repeat this step until the algorithm arrives at a fixed point.

\smallskip\noindent\emph{Performing a Lift operation.} 
To compute the $\Lift$-operation in the set-based symbolic setting we need to
compute two functions for the predecessors of $S_r$: (1)~the $\slift$-function and (2)~the $\best$-function.
By definition, the $\slift$-function does not access the game graph or vertices thereof.
Thus, we can compute the $\slift$-function without the use of symbolic operations. 
To compute the $\best$-function we need access to the game graph.
It turns out it is simpler to compute the vertices with $\best(f,v) \succeq r$
rather than the vertices with $\best(f,v) = r$. 
%
%
%
Thus, we lift all vertices $v$ with $\best(f,v) \succeq r$ to the rank $\slift(r,p(v))$. 
To this end, we first compute the set $S_{\succeq r} = \bigcup_{l \succeq r} S_l$ of vertices with rank $\succeq r$.
Then, we compute $P = \CP_z(S_{\succeq r})$ and, hence, the set $P$ comprises the vertices $v \in P$ with
$\best(f,v) \succeq r$. Finally, to compute $\Lift(f,v)$, for each $c \in
C$, we consider the vertices of $P$ with priority $c$, i.e., the set  $(P \cap V_c)$, and add  them to the set $S_{\slift(r,c)}$. 
Notice that we lift each vertex $v$ to $\slift(r,p(v))$
where $r = \best(f,v)$ as we consider all non-empty sets $S_r \in \S$. No vertex
$v$ will be lifted to a set higher than $\slift(r,p(v))$ where  $r = \best(f,v)$ due to the monotonicity
of the lift function in the first argument.
If after an iteration of the algorithm a vertex appears in several sets of $\mathbf{S}$ 
we only keep it in the set corresponding to the largest rank and remove it from all the other sets.

\subsection{Improving the Basic Algorithm}\label{lpar:ss:algorithm}


In this section, we improve the basic Algorithm by (a) reducing the symbolic
space from $O(|\mathcal{W}|)$ to $O(n)$ 
and (b) by reducing the number of symbolic operations required to compute the fixed point.

\smallskip\noindent\emph{Key Idea.}
The naive algorithm considers each non-empty set $S_r$ in iteration
$i\+1$ again no matter if $S_r$ has been changed in iteration $i$ or not.
Notice that we only need to consider the predecessors of the set $S_r$ again when
the set $S_{\succeq r}$ in iteration $i\+1$ contains 
additional vertices compared to the set $S_{\succeq r}$ in iteration $i$.
To overcome this weakness, we propose Algorithm~\ref{lpar:alg:blackbox3}. In this
algorithm, we introduce a data structure called $D$. In the data structure $D$ we keep track
of the sets $S_{\succeq r}$ instead of the sets $S_{r}$. The set $S_{\succeq r}$
contains all vertices with a rank greater than or equal to $r$. Furthermore, we separately keep track of the elements $r \in \WOrder$ where the set
$S_{\succeq r}$ changed since the last time $S_{\succeq r}$ was selected to be processed. 
These elements of $\WOrder$ are called \emph{active}.
Moreover, if we have two sets $S_{\succeq r} = S_{\succeq r'}$ with $r \prec r'$
there is no need to process the set $S_{\succeq r}$
because $\slift(r',c) \succeq \slift(r,c)$ holds due to the monotonicity of the
$\slift$ function. To summarize, we precisely store a set $S_{\succeq r}$ if there is no $r'$
with $r \prec r'$ and $S_{\succeq r} = S_{\succeq r'}$.
Notice, that this instantly gives us a bound on the symbolic space of
$O(n)$. 

\smallskip\noindent\emph{Algorithm Description.} 
In Algorithm~\ref{lpar:alg:blackbox3} we use the data structure $D$ to manage the active $r \in \WOrder$ and in each iteration of 
the outer while-loop we process the corresponding set $S_{\succeq r}$ of such an
$r \in \WOrder$. 
We first compute $P = \CP_z(S_{\succeq r})$, then,
for each $c \in C$  we compute $r' = \slift(r,c)$
and update the set $S_{\succeq r'}$ by adding $P \cap V_c$. 
The inner while-loop ensures that $P \cap V_c$ is also added to all the sets $S_{\succeq r'}$ with $r' \prec r$ and 
the properties of the data structure are maintained, i.e.,
(a) all active elements are in the active list, and
(b) exactly those $r \in \WOrder$ with $S_{\succeq r} \supset S_{\succeq r'}$, for $r
\prec r'$ are stored in $D$.

\begin{algorithm}[t]
	\SetKwInOut{Input}{input}
	\SetKwInOut{Output}{output}
	\SetAlgoVlined
	\Input{Parity Game $\PG$}
	\caption{Black Box Set-Based Symbolic Progress Measure}\label{lpar:alg:blackbox3}
	Initialize data structure $D$\; \label{lpar:alg:bb3:dinit}
	D.activate$(min)$\;\label{lpar:alg:bb3:activeinit}
	\While{$r \gets D.popActiveSet()$\label{lpar:alg:bb3:while1}}{\label{lpar:alg:bb3:deactivate}
		$S_{\succeq r}  \gets D.getSet(r)$\;
		$P \gets \CP_z(S_{\succeq r})$\;\label{lpar:alg:bb3:cp}
		\For{$c \in C$} {\label{lpar:alg:bb3:for1}
			$r' \gets \slift(r,c)$\;\label{lpar:alg:bb3:lift}
			$S_{\succeq r'} \gets D.getSet(r')$\;
			\While{$P \cap V_c \not\subseteq S_{\succeq r'}\label{lpar:alg:bb3:while2}$}
			{
				$S_{\succeq r'} \gets S_{\succeq r'} \cup (P \cap V_c)$\label{lpar:alg:bb3:liftvertices}\;
				$S_{\succeq next(r')} \gets D.getSet(D.getNext(r'))$\;
				\If{$r' = \top$ or $S_{\succeq r'} \supset S_{\succeq next(r')}$\label{lpar:alg:bb3:if1}}
				{
					\tcp{$S_{\succeq r'}$ is a superset of $S_{\succeq
							next(r')}$ and we save it}
					$D.update(r', S_{\succeq r'})$;\label{lpar:alg:bb3:store}
					$D.activate(r')$\;\label{lpar:alg:bb3:activate}
				}				
				$S_{\succeq prev(r')} \gets D.getSet(D.getPrevious(r'))$\;
				\If{$S_{\succeq r'} = S_{\succeq prev(r')} \cup (P \cap V_c)$}{\label{lpar:alg:bb3:else1}
					\tcp{We only keep sets which are different}
					$D.removeSet(D.getPrevious(r'))$\;\label{lpar:alg:bb3:removed}
				}
				$r'\gets D.getPrevious(r')$\;\label{lpar:alg:bb3:while2end}
				$S_{\succeq r'}  \gets
				D.getSet(r')\;$\;\label{lpar:alg:bb3:while2end2}
			}    
		}
	}
	\Return $S_\top$
\end{algorithm}

\smallskip\noindent\emph{Active Elements.} Intuitively, an element $r \in \WOrder$ is \emph{active}
if $S_{\succeq r} \supset S_{\succeq r'}$, for all $r'$ where $r \prec r'$ 
and $S_{\succeq r}$ has been changed since the last time 
$S_{\succeq r}$ was selected at Line~\ref{lpar:alg:bb3:while1} of Algorithm~\ref{lpar:alg:blackbox3}.
We define active elements more formally later.

\begin{datastructure}\label{lpar:alg:bb3:ds}
Our algorithm relies on a data structure $D$, which supports the following operations:
\begin{compactitem}
\item $D.popActiveSet()$ returns an element $r \in \WOrder$ marked as active and makes it
	inactive. If all elements are inactive, returns false.

      \item $D.getSet(r)$ returns the set $S_{\succeq r}$. 		  
      
      \item $D.getNext(r)$ returns the smallest $r'$ with $S_{\succeq r} \supset  S_{\succeq r'}$.
      
      \item $D.getPrevious(r)$ returns the largest $r'$ where $S_{\succeq r} \subset S_{\succeq r'}$.
    
      \item $D.removeSet(r)$ 
			     marks $r$ as inactive. 

      \item $D.activate(r)$ marks $r$ as active.
      
      \item $D.update(r,S)$  updates the set $S_{\succeq r}$ to $S$, i.e., $D.getSet(r)$ returns $S$.
							 Moreover, all sets $S_{\succeq r'}$ with $r' \prec r$ and $S_{\succeq r'}=S_{\succeq r}$ beforehand
							 are updated to $S$ as well.
  \end{compactitem}
 We initialize $D$ with $D.update(\mina,V)$ and $D.update(\top, \emptyset)$.
\end{datastructure}

We can define \emph{active} elements formally now as the definition depends on
$D$.
\begin{definition}
	Let $S^0_{\succeq r}=D.getSet(r)$ be the set stored in $D$ for  $S_{\succeq r}$ 
	after the initialization of $D$, 
	and let $S^i_{\succeq r}=D.getSet(r)$ be the set
	stored in $D$ for $S_{\succeq r}$ after the $i$-th iteration of the
	while-loop at Line~\ref{lpar:alg:bb3:while1}. 
%
         An element $r \in \WOrder$ is \emph{active} after the $i$-th iteration
	of the while-loop 
	if 
	(i) for all $r' \in \WOrder$ where $r \prec r'$ we have 
	 $S^i_{\succeq r} \supset S^i_{\succeq r'}$ and
	(ii) there is a $j < i$ such that $S^i_{\succeq r} \supset S^j_{\succeq r}$ 
	and for all $j < j' \leq i$ the set $S_{\succeq r}$ 
	is not selected in Line~\ref{lpar:alg:bb3:while1} in the $j'$-th iteration.
	Additionally, we consider $r=\min$ as active before the first iteration. An
	element $r \in \WOrder$ is \emph{inactive} if it is not active. 
\end{definition}

Notice that, in Algorithm~\ref{lpar:alg:blackbox3} an $r \in \WOrder$ is active iff $r$ is marked as active in $D$. 
The algorithm ensures this in a very direct way.  
At the beginning, only $\min \in W$ is active, which is also set active in $D$
in the initial phase of the algorithm.
Whenever some vertices are added to a set $S_{\succeq r}$, it
is tested whether $S_{\succeq r}$ is larger than its successor and if so $r$
is activated (Lines~\ref{lpar:alg:bb3:if1}-\ref{lpar:alg:bb3:activate}).  
On the other hand, if something is added to the successor of
$S_{\succeq r}$ in the data structure $D$ then the algorithm tests whether the two sets are equal and
if so $r$ is rendered inactive (Lines~\ref{lpar:alg:bb3:else1}-\ref{lpar:alg:bb3:removed}). \smallskip

\noindent\emph{Implementation of the data structure $D$.}
The data structure uses an AVL-tree and 
a doubly linked list called ``active list'' that keeps track of the active elements.
The nodes of the tree contain a pointer to the corresponding set $S_{\succeq r}$
and to the corresponding element in the active list.

\begin{compactitem}
	\item Initialization of the data structure $D$: Create the AVL tree with the elements $min$ and $\top$.
	      The former points to the set of all vertices and the latter to the empty set.
	      Create the doubly linked list called ``active list'' as an empty list.
	      
	\item $D.popActiveSet()$: Return the first element from the active list and
		remove it from the active list. If the list is empty, return false.
		
	\item \begin{sloppypar} $D.getSet(r)$: Searches the AVL tree for $r$ or for the next greater
		element (w.r.t. $\succeq$). Then we return the set by using the pointer
		we stored at the node. \end{sloppypar}
	\item $D.getNext(r)$: 
	        First performs $D.getSet(r)$ and then computes the inorder successor in the 
	        AVL-tree. This corresponds to the next greater node w.r.t. $\succeq$.
	        	
	\item $D.getPrevious(r)$: 
	        First performs $D.getSet(r)$ and then computes the inorder predecessor in the 
	        AVL-tree. This corresponds to the next smaller node w.r.t. $\succeq$.
		
	\item $D.removeSet(r)$: This operation needs the element $r$ to be stored in the
		AVL tree. Search the AVL tree for $r$. Remove the corresponding element from
		the active list and the AVL Tree. 
		
	\item $D.activate(r)$: This operation needs the element $r$ to be stored in the
		AVL tree.
		Add $r$ to the active list and add pointers to the AVL-tree.
		The element in the active list contains a pointer to the tree element and vice versa.
	
	\item $D.update(r,S)$ :        
	        Perform $S_{\succeq r} \gets D.getSet(r)$: 
	        If $r$ is contained in the AVL tree then update $S_{\succeq r}$ to $S$.
	        Otherwise, insert $r$ as a new element and let the element point to $S$.		
\end{compactitem}

We initialize the data structure $D$ with $min \in \WOrder$ and
$\top \in \WOrder$. Thus, whenever we query $D$ for a value $r \in \WOrder$ we find
it or there exists an $r' \succ r$ which is in $D$.

\smallskip\noindent\emph{Analysis of the data structure $D$.}
\begin{sloppypar}
The data structure can be implemented with an AVL-tree and 
a doubly linked list called ``active list'' that keeps track of the active elements such that
all of the operations can be performed in $O(\log n)$:
when the algorithm computes $D.update(r,S)$ we store $r$ and a pointer
to the set $S$ as a node in the AVL tree.
By construction, the algorithm only stores pointers to different sets and when we additionally preserve anti-monotonicity among the sets 
we only store $\leq n$ sets. Therefore, the AVL tree has only $\leq n$
nodes with pointers to the corresponding sets and searching for a set with
the operation $D.getSet(r)$ only adds a factor of $\log n$ to the non-symbolic operations
when we store $r$ as key with a pointer to $S_{\succeq r}$ in the AVL-tree.
Moreover, we maintain pointers between the elements of the active list and the corresponding vertices in the AVL tree.
\end{sloppypar}
\begin{remark}\label{lpar:remark1}
	The described algorithm is based on a data structure $D$ 
	which keeps track of the sets that will be processed at some point later in time. 
	Note that this data structure
	does not access the game graph but only stores pointers to sets that the
	Algorithm~\ref{lpar:alg:blackbox3} maintains. 
	The size of the AVL tree implementing $D$ is proportional to the symbolic space of the algorithm.
\end{remark}

\subsection{Correctness}\label{lpar:ss:correctness}


To prove the correctness of Algorithm~\ref{lpar:alg:blackbox3} we tacitly assume that the algorithm terminates.
An upper bound on the running time is then shown in Proposition~\ref{lpar:prop:numberactivations}.

\begin{proposition}[Correctness.]\label{lpar:lem:bb3:correctness}
	Let $\PG$ be a parity game.
	Given a finite total order $(\WOrder,\prec)$ with minimum element $\mina$, a	maximum element $\top$ 
	and a monotonic function $\slift:\WOrder \times C \mapsto \WOrder$ 
	Algorithm~\ref{lpar:alg:blackbox3} computes the least simultaneous fixed point of
	all $\Lift(\cdot,v)$-operators.	
\end{proposition}

To prove the correctness of Algorithm~\ref{lpar:alg:blackbox3}, we prove that when
Algorithm~\ref{lpar:alg:blackbox3} terminates, the function $\rho(v) =
\max\{r \in \WOrder \mid v \in S_{\succeq r}\}$ is equal to the least simultaneous fixed point of
all $\Lift(\cdot,v)$-operators. We show that when the properties described in
Invariant~\ref{lpar:inv:correctness} hold, the function $\rho$ is equal to the 
least fixed point at the termination of the algorithm. Then, we prove that we maintain the properties of Invariant~\ref{lpar:inv:correctness}.

\begin{invariant}\label{lpar:inv:correctness}
	Let $\trho$ be the least simultaneous fixed point of
	$\Lift(\cdot,v)$ and $\rho(v) = \max\{r \in \WOrder \mid v \in S_{\succeq r}\}$
	be the ranking function w.r.t.\ the sets $S_{\succeq r}$ that are
	maintained by the algorithm.
	\begin{compactenum}
		\item Before each iteration of the while-loop at
			Line~\ref{lpar:alg:bb3:while1} we have $S_{\succeq r_2} \subseteq
			S_{\succeq r_1}$ for all $r_1 \preceq r_2$ (anti-monotonicity).
		\item Throughout Algorithm~\ref{lpar:alg:blackbox3} we have $\trho(v) \succeq
			\rho(v)$ for all $v \in V$.
		\item For all $r \in \WOrder$: (a) $r$ is active or (b) for all $v \in \CP_z(S_{\succeq r})\!:\\ 
					\text{ if } \best(\rho,v) = r$, then	$\rho(v) = \Lift(\rho,v)(v)$.
	\end{compactenum}
\end{invariant}

In the following paragraph, we describe the intuition of
Invariant~\ref{lpar:inv:correctness}. Then,
we show that the properties of Invariant~\ref{lpar:inv:correctness} are sufficient to obtain the correctness of
Algorithm~\ref{lpar:alg:blackbox3}. Finally, we prove that each property holds
during the while-loop at Line~\ref{lpar:alg:bb3:while1}.

\smallskip\noindent\emph{Intuitive Description.} The intuitive description is as follows:
\begin{compactenum}
	\item Ensures that the sets $S_{\succeq r}$ contain the correct elements. Having the
		sets $S_{\succeq r}$ allows computing $\best(f,v) \succeq r$ as discussed at the
		beginning of the section.
	\item Guarantees that $\rho$ is a lower
		bound on $\tilde{\rho}$ throughout the algorithm.
	\item When an $r \in \WOrder$ is not active, 
		the rank of no vertex can be increased by applying $\slift$ to the
		vertices which have $\best(\rho,v) = r$.
\end{compactenum}
When the algorithm terminates, all $r \in \WOrder$ are inactive and $\rho$ is 
a fixed point of all $\Lift(\rho,v)$ by condition (3b). The next lemma proves
that Algorithm~\ref{lpar:alg:blackbox3} computes the least simultaneous fixed point of
all $\Lift(\cdot, v)$ operators for a parity game.

\begin{lemma}[The Invariant is sufficient]
	Let the $\slift$ function be monotonic in the first argument and $(\WOrder, \prec)$ be a total order.
	The ranking function $\rho$ at termination of Algorithm~\ref{lpar:alg:blackbox3} is
	equal to the least simultaneous fixed point of all $\Lift(\cdot,v)$-operators
	for the given parity game $\PG$.
\end{lemma}
\begin{proof}
	Consider the ranking function $\rho(v) = \max\{r \in \WOrder \mid v \in S_{\succeq r}\}$
	computed by Algorithm~\ref{lpar:alg:blackbox3}.
	By Invariant~\ref{lpar:inv:correctness}(2) we have $\trho(v) \succeq \rho(v)$ for all $v \in V$. 
	We next show that $\rho(v)$ is a fixed point of $\Lift(\rho,v)$ for all $v \in V$.
	When the algorithm terminates, no $r \in \WOrder$ is active. 
	Consider an arbitrary $v$ and let $r=\best(\rho,v)$.
	Now, as the set $r$ is not active, by Invariant~\ref{lpar:inv:correctness}(3b),
	we have $\rho(v) = \Lift(\rho,v)(v)$.
	Thus $\rho(v)$ is a fixed point of $\Lift(\rho,v)$ for all vertices in
	$V$.
	Therefore, as $\rho$ is a simultaneous fixed point of  all $\Lift(\cdot,v)$-operators
	and $\trho$ is the least such fixed point, 
	we obtain $\rho(v) \succeq \trho(v)$ for all $v \in V$. 
	Hence we have $\rho(v) = \trho(v)$ for all $v \in V$.
\end{proof}

The following lemmas prove each part of the invariant separately. The first part of the 
invariant describes the anti-monotonicity property which is needed to compute the $\best$
function with the $\CP_z$ operator. 

\begin{lemma}\label{lpar:lem:antimonoton}
	Invariant~\ref{lpar:inv:correctness}(1) holds:
	Let $r_1,r_2 \in \WOrder$ and $r_1 \preceq r_2$.
	Before each iteration of the while-loop at Line~\ref{lpar:alg:bb3:while1} we have that if a vertex
	$v$ is in a set $S_{\succeq r_2}$ then it is also in $S_{\succeq r_1}$  (anti-monotonicity).
\end{lemma}
\begin{proof}
	We prove the claim by induction over the iterations of the while-loop.
	Initially, the claim is satisfied as the only non-empty set is $S_{min}$.
	It remains to show that when the claim is valid at the beginning of an
	iteration, then the claim also holds in the next iteration. 
	By the induction hypothesis, the claim holds for the sets at the beginning of the while-loop. 
	In the trivial case, the algorithm terminates and the claim holds by the induction hypothesis. 
	Otherwise, the sets are only modified at Line~\ref{lpar:alg:bb3:liftvertices} and
	stored at Line~\ref{lpar:alg:bb3:store}.
	First, the vertices $P \cap V_c$ are added into the set $S_{\succeq r'}$. 
	Let $r''=D.getPrevious(r')$.
	Notice that after activating $r'$ all $r$ with $r'' \prec r \prec r'$ refer to the same set as $r'$
	and thus we add $P \cap V_c$ implicitly to all $r$.
	In the next iteration the while-loop then adds $P \cap V_c$ also to the set $S_{\succeq r''}$.
	As this done iteratively until a set $S_{\succeq r^*}$ with $P \cap V_c \subseteq S_{\succeq r^*}$ is reached (Lines~\ref{lpar:alg:bb3:while2}-\ref{lpar:alg:bb3:while2end}),
	the algorithm ensures that $P \cap V_c$ is contained in all set $S_{\succeq r''}$ with $r \prec r'$.
	By induction hypothesis we know that the invariant holds for all $r_2 \succ r'$ ($S_{\succeq r_2}$ is unchanged),
	and as the algorithm added $P \cap V_c$ to all set $S_{\succeq r''}$ with $r'' \preceq r'$
	the claim holds for all $r_1, r_2 \preceq r'$.  
\end{proof}

The second part of the invariant shows that the fixed point 
Algorithm~\ref{lpar:alg:blackbox3} computes is always smaller or equal to the least
fixed point. In particular, the fixed point computed by the algorithm is
defined as $\rho(v) = \max\{r \in \WOrder \mid v \in S_{\succeq r}\}$ and we denote
the least fixed point with $\trho$. The proof is by induction: In the
beginning, every vertex is initialized with the minimum element which 
suffices for the claim. When we apply the $\slift$ function to vertices, we observe
that by the induction hypothesis the current value of a vertex is below or equal to
the fixed point. Additionally, we obtain a rank that is also smaller or equal to the lifted value of $\trho$ for every vertex 
as $\slift$ is a monotonic function.

\begin{lemma}\label{lpar:lem:underlfp}
	Invariant~\ref{lpar:inv:correctness}(2) holds: Throughout
	Algorithm~\ref{lpar:alg:blackbox3} we have $\trho(v) \geq \rho(v)$ for all $v \in V$.
\end{lemma}
\begin{proof}
	Before the while-loop at Line~\ref{lpar:alg:bb3:while1} the claim is obviously
	satisfied as $\trho(v) \succeq \min$ for all $v \in V$.
	We prove the claim by induction over the iterations of the while-loop:
	Assume we have $\rho(v) \preceq \trho(v)$ for all $v \in V$ before an iteration of
	the while-loop. 
	The function $\rho(\cdot)$ is only changed at
	Line~\ref{lpar:alg:bb3:liftvertices} and stored at Line~\ref{lpar:alg:bb3:store} where the set 
	$(P \cap V_c)$ is added to $S_{\succeq r'}$.
	For $v \in P \cap V_c$ we have that $v$ is a priority $c$ vertex and
	either $v$ is a player-$z$ vertex with a successor in $S_{\succeq r}$ or
	a player-$\z$ vertex with all successors in $S_{\succeq r}$.
	Thus, $r \preceq \best(\rho,v)$ for $v \in P$.
	At Line~\ref{lpar:alg:bb3:lift} we compute the $\slift$-operation
	for ranking $r$ with priority $c$ which results in the ranking $r'$ for the first iteration of the while-loop.
	By the monotonicity of the $\slift$ operation and the induction hypothesis we have that
	$r'=\slift(r,c)(v) \preceq \slift(\best(\rho,v),c)(v) \preceq \slift(\best(\trho,v),c)=\trho(v)$ for $v \in P \cap V_c$
	and thus adding $v$ to $S_{\geq r'}$ maintains the invariant
	(if $\rho(v) \succ r'$ beforehand it is not changed and otherwise it is lifted to $r' \prec \trho(v)$).
	In the later iterations of the while-loop $P \cap V_c$ is added to sets with smaller $r'$, which does not 
	affect $\rho$, as these vertices already appear in sets with larger rank.
\end{proof}

The following lemma proves the third part of Invariant~\ref{lpar:inv:correctness}:
Either there is an active $r \in \WOrder$, i.e., the set $S_{\geq r}$ needs to be
processed, or $\rho(v)$ is a fixed point. 
We prove the property again by induction:
Initially, the set $\mina \in \WOrder$ is active and every other set is empty
which trivially fulfills the property. Then, in every iteration
when we change a set with value $r$ we either activate it, or
there is a set with a value $r' \succeq r$ where $S_{\succeq r'}$ subsumes
$S_{\succeq r}$. In the former case, the condition is instantly
fulfilled. In the latter case, there is no vertex $v$ where $best(\rho,v)= r$
which renders $S_{\succeq r}$ irrelevant by definition of $\rho$.  

\begin{lemma}
	Invariant~\ref{lpar:inv:correctness}(3) holds: For all $r \in \WOrder$:
	\begin{compactenum}
	\item $S_{\succeq r}$ is active or,
	\item $\forall v \in \CP_z(S_{\succeq r})\!:
		\text{ if } \best(\rho,v) = r$, then	$\rho(v) = \Lift(\rho,v)(v)$
	\end{compactenum} 
\end{lemma}
\begin{proof}

	We prove this invariant by induction over the iterations of the
	while-loop: Before the while-loop at Line~\ref{lpar:alg:bb3:while1} the claim is
	obviously satisfied as we activate $\mina$ which contains all
	vertices; for all other $r \in \WOrder$ the set $S_{\succeq r}$ is empty and thus
	condition (2) is trivially satisfied.

	Assume the condition holds at the beginning of the loop.
	We can, therefore, assume by the induction hypothesis that the condition holds for all the sets.
	If there is no active $r \in \WOrder$, the algorithm terminates and the condition
	holds by the induction hypothesis.
	The condition for a set $S_{\succeq r}$ can be violated only if either the set $S_{\succeq r}$ is changed or
	the set $S_{\succeq r}$ is deactivated.
	That is either at Line~\ref{lpar:alg:bb3:deactivate}, Line~\ref{lpar:alg:bb3:liftvertices} 
	or Line~\ref{lpar:alg:bb3:removed} of the algorithm.

	Let us first consider the changes made in the while-loop.
	If a set $S_{\succeq r}$ is changed in Line~\ref{lpar:alg:bb3:liftvertices}, then the algorithm
	either activates $r$ (Line~\ref{lpar:alg:bb3:activate}) and thus satisfies (1)
	or $S_{\succeq next(r)} = S_{\succeq r}$ which implies that $\best(\rho,v)
	\neq r$ and thus (2) is fulfilled trivially.
	At Line~\ref{lpar:alg:bb3:removed} there is no vertex $v$ with $\best(\rho,v) = r$ 
	(as there is no vertex $w$ with $\rho(w)=r$) 
	and thus (2) is satisfied (and it
	is safe to remove/deactivate the set in  Line~\ref{lpar:alg:bb3:removed}).

	Now consider the case where we remove the set $S_{\succeq r}$ 
	and make $r$ inactive at Line~\ref{lpar:alg:bb3:deactivate}.
	If the set $S_{\succeq r}$ is unchanged during the iteration of the outer while-loop then
	$S_{\succeq r}$ satisfies condition (2) after the iteration. This is because 
	for all $v$ with $\best(\rho,v) = r$ and $p(v)=c$ we have that
	if $v$ is not already contained in $S_{\succeq \slift(r,c)}$ the algorithm adds it to the set $S_{\succeq \slift(r,c)}$
	in Line~\ref{lpar:alg:bb3:liftvertices} in the first iteration of the while-loop
	when processing $c$. This is equivalent to applying $\Lift(\rho,v)(v) =
	\slift(r,c)$. 
	If the set $S_{\succeq r}$ is changed during the iteration then this happens in the inner while-loop.
	As argued above, then either $r$ is activated and thus
	satisfies (1) or  $S_{\succeq next(r)} = S_{\succeq r}$ holds. Thus,
	there is no vertex $v$ with $\best(\rho,v) = r$, i.e., (2) is satisfied.
\end{proof}

\subsubsection{Symbolic Resources}

In the following, we discuss the symbolic resources
Algorithm~\ref{lpar:alg:blackbox3} needs. We determine the number of \os
operations, the number of basic set operations, and the
symbolic space consumption.

\begin{proposition}\label{lpar:prop:numberactivations}
	The number of \os operations in Algorithm~\ref{lpar:alg:blackbox3} is in $O(n \cdot |\WOrder|)$.
\end{proposition}
\begin{proof}
	Each iteration of the while-loop at Line~\ref{lpar:alg:bb3:while1} processes an
	active $r$. That means, that the set $S_{\succeq r}$ was changed in a prior
	iteration. We use a symbolic one-step operation at Line~\ref{lpar:alg:bb3:cp} for
	each active $S_{\succeq r}$. It, therefore, suffices to count the number of possibly
	active sets throughout the execution of the algorithm. 
%
	Initially only $S_{\succeq \min}$ is active. After extracting an active
	set out of the data structure $D$, it is deactivated at Line~\ref{lpar:alg:bb3:deactivate}.
	We only activate a set $S_{\succeq x}$ when a new vertex is added to it at
	Line~\ref{lpar:alg:bb3:activate}. Because there can only be $n$ vertices with
	ranking $\succeq x$ for all $x \in \WOrder$ the size of each set  $|S_{\succeq x}|$ is smaller or equal to $n$. 
	In the worst case, we eventually put every vertex into every set $S_{\succeq x}$ where $x \in
	\WOrder$. Thus we activate $n \cdot |\WOrder|$ sets which is equal to the number of 
	symbolic one-step operations.
\end{proof}

A similar argument works for analysing the number of basic set operations. 
\begin{proposition}\label{lpar:prop:bb3:setoperations}
	The number of basic set operations in Algorithm~\ref{lpar:alg:blackbox3} is in
	$O(\numprio \cdot n \cdot |\WOrder| )$.
\end{proposition}
\begin{proof}[Proof of Proposition~\ref{lpar:prop:bb3:setoperations}.]
	As proven in Proposition~\ref{lpar:prop:numberactivations}, there are $O(n \cdot|\WOrder|)$
	iterations of the outer while-loop and thus $O(n \cdot |\WOrder|)$ iterations of the for-loop. 
	Thus the inner while-loop is started $O(\numprio n \cdot |\WOrder|)$ times.
	The test whether the while-loop is started only requires two basic set operations and the overall costs 
	are bound by $O(\numprio n \cdot |\WOrder|)$.
	We bound the overall costs for the iterations of the inner while-loop by an amortized analysis.
	First, notice that each iteration just requires 8 basic set operations 
	(including testing the while condition afterward).	
	In each iteration for a value $r' \in \WOrder$ we charge the $r'$ for the
	involved basic set operations.
	Notice, that in each such iteration new vertices are added to the set $S_{\succeq r'}$
	and thus $r'$ is processed at most $n$ times.
	Thus each $r' \in \WOrder$ is charged for at most $8n$ basic set operations 
	Therefore, the number of basic set operations is $O(\numprio \cdot n \cdot |\WOrder|) + O(n \cdot
	|\WOrder|) = O(\numprio \cdot n \cdot |\WOrder|)$. 
\end{proof}

Due to Proposition~\ref{lpar:lem:bb3:correctness},
Proposition~\ref{lpar:prop:numberactivations},
Proposition~\ref{lpar:prop:bb3:setoperations} and the fact that we use $\leq n$ sets
in the data structure $D$, we obtain Theorem~\ref{lpar:thm:bb3}.

\begin{theorem}\label{lpar:thm:bb3}
	Given a parity game, a finite total order $(\WOrder,\succ)$ and a monotonic function $\slift$
	we can compute the least fixed point of all
	$\Lift(\cdot,v)$ operators with $O(n \cdot |\WOrder|)$ \os operations, $O(\numprio \cdot n \cdot |\WOrder|)$ basic set operations, 
	and $O(n)$ symbolic space. 
\end{theorem}

\section{Implementing the Ordered Progress Measure}

In this section, we plug the ordered approach to progress measure (OPM) described
by Fearnley et al.~\cite{FearnleyJS0W17} into Algorithm~\ref{lpar:alg:blackbox3}. To do this, 
we recall the witnesses they use in their algorithm and encode it with a specially-tailored
technique to obtain an algorithm with a sublinear amount of symbolic space. 
Finally, we argue that the function $\slift:\WOrder \times C \mapsto \WOrder$ and the total
order $(\WOrder,\preceq)$ described in~\cite{FearnleyJS0W17} can be used to
fully implement Algorithm~\ref{lpar:alg:blackbox3}.

\smallskip\noindent\emph{The Ordered Progress Measure.}
To implement the ordered progress measure algorithm we argue that the
$\slift$-operation is monotonic in the first argument and the order
$(\WOrder,\preceq)$ is a total finite order to fulfill the conditions of
Algorithm~\ref{lpar:alg:blackbox3}.
Let $\PG$ be a parity game and $C = \{0,\dots,\numprio-1\}$ be the set of priorities
in $\PG$. 
The set $\WOrder$ in the ordered progress measure consists of tuples of priorities of length $k$,
where $k \in O(\log n)$. 
Each element in the tuple is an element of $C\_ = C \cup \{\_\}$, i.e., it is either a priority or "$\_$". 
The set $C\_$ has a total order $(C\_, \preceq)$ such that 
$\_$ is the smallest element, 
odd priorities are ordered descending and are considered smaller than 
even priorities which are ordered ascending.
The order $(\WOrder,\preceq)$ is then obtained by extending the order $(C\_,\preceq)$
lexicographically to the tuples $r \in \WOrder$.

For the details of the $\slift$ function, we refer the reader to the work of
Fearnley et al.~\cite{FearnleyJS0W17}. An implementation of the lift
operation can be found at the GitHub repository
of the Oink system~\cite{dijk2018oink}.

By the results in \cite{FearnleyJS0W17} the order $(\WOrder,\preceq)$ and $\slift$ meet the requirements
of our algorithms.

\begin{lemma}\label{lpar:lem:lift_monotonic}
   The following holds:
	(1) The function $\slift:\WOrder \times C \mapsto \WOrder$ is monotonic in the first
	parameter~\cite[p.6]{FearnleyJS0W17}.
	(2) The order $(\WOrder,\preceq)$ is a total finite order~\cite[p.3]{FearnleyJS0W17}.
	(3) Let $\rho$ be the least simultaneous fixed point of all $\Lift(\cdot,v )$ operators. Then 
	$\rho(v) = \top$ iff player $\E$ has a strategy to win the parity game $\PG$
	when starting from $v$~\cite[Lemma 7.3, Lemma 7.4]{FearnleyJS0W17}.
\end{lemma}

Theorem~\ref{lpar:thm:bb3} together with
Lemma~\ref{lpar:lem:lift_monotonic} imply the following theorem.

\begin{theorem}\label{lpar:thm:ordered_main}
	Algorithm~\ref{lpar:alg:blackbox3} implemented with the OPM computes the winning set
	of a parity game with $O(n\cdot|\WOrder|)$ \os operations, $O(\numprio \cdot n \cdot |\WOrder|)$
	basic set operations, and $O(n)$ symbolic space.
\end{theorem}

\begin{remark}{(Bounds for $|\WOrder|$).}\label{lpar:rem:bound}
	We now discuss the bounds on $|\WOrder|$. The breakthrough result of~\cite{calude2017stoc} 
	shows that $|\WOrder|$ is quasi-polynomial ($n^{O(\log \numprio)}$) in general and polynomial
	when $\numprio \leq \log n$. Using the refined analysis of~\cite{FearnleyJS0W17}, 
	we obtain the following bound on $|\WOrder|$: 
	in general, $\min(n \cdot \log(n)^{\numprio-1}, h \cdot
	n^{c_{1.45}+\log_2(h)})$, where $c_{1.45} = \log_2(e) < 1.45$ and $h=\lceil 1+\numprio /\log(n) \rceil$;
	and if $\numprio \leq \log n$, then $|\WOrder|$ is polynomial due
	to~\cite[Theorem 2.8]{calude2017stoc} and \cite[Corollary 8.8]{FearnleyJS0W17}.
	Note that  $O\left(n^{2.45+\log_2(\numprio )}\right)$ gives a naive upper bound on $|\WOrder|$ in general. 
	Plugging the bounds in Theorem~\ref{lpar:thm:ordered_main} we obtain a set-based symbolic 
	algorithm that requires quasi-polynomially many \os and basic set operations and $O(n)$ symbolic space.
	The algorithm requires only polynomially many \os and basic set operations when $\numprio  \leq \log n$.
\end{remark}

\section{Reducing the Number of Sets for the OPM}\label{lpar:sec:reducesets}

In this section, we tailor a data structure for the OPM to only use
$O(\numprio \cdot \log n)$ sets. While each progress measure can be encoded by $\log(|\WOrder|)$
many sets, the challenge is to provide a representation that also allows to efficiently compute the sets $S_{\succeq r}$.
Such a representation has been provided for the small progress measure~\cite{chatterjee2017symbolic} 
and in the following, we adapt their techniques for the OPM.

\smallskip\noindent\emph{Key Idea.}
The key idea of the symbolic space reduction is that we \emph{encode} the
value of each coordinate of the rank $r$ separately. A set no longer just
stores the vertices with specific rank $r=b_1 \dots b_k$ but instead stores
all vertices where, say, the first coordinate $b_1$ is equal to a specific value in
$C\_$. This encoding enables us to use only a polylogarithmic amount of symbolic space
under the assumption that the number of priorities in the game graph is
polylogarithmic in the number of vertices.

\smallskip\noindent\emph{Symbolic Space Reduction.}
Let the rank of $v$ be $r=b_{1}\dots b_k$.
Vertex $v$ is in the set $C_x^i$ iff the $i$th coordinate of the rank of $v$ is $x$
and a vertex $v$ is in the set $C_\top$ iff the rank of $v$ is $\top$. 
Thus $O(\log(n) \cdot \numprio)$ sets suffice to encode all $r \in
\WOrder$. We demonstrate this encoding of the sets in
Example~\ref{lpar:ex:setex}.

\begin{example}\label{lpar:ex:setex}
	Let $\PG$ be a parity game containing the vertices $v_1,v_2,v_3$. Assume the
	following ranking function:
	$
		f(v_1)  = 65433, f(v_2) = 75422, f(v_3) =
		\mathunderscore\mathunderscore\mathunderscore 32.
	$
	Using the definition of our encoding, we have that:
	$
		\{v_3\} \subseteq C_{\_}^1, \{v_1\} \subseteq C_6^1,\{v_2\} \subseteq C_7^1, 
		\{v_3\} \subseteq C_{\_}^2,	\{v_1,v_2\} \subseteq C_5^2,\{v_3\} \subseteq C_{\_}^2, 
		\{v_1,v_2\} \subseteq C_4^3,\{v_2\} \subseteq C_2^4, \{v_1,v_3\} \subseteq C_3^4, \{v_1\} \subseteq C_3^5,\{v_2, v_3\} \subseteq C_2^5
	$.
\end{example}

\smallskip\noindent\emph{Computing the set $S_{\succeq r}$ from $C_x^i$.}
We obtain the set $S_r$ for rank  $r=b_{1}\dots b_k$ with an intersection of the sets 
$\bigcap_{i=1}^k C_{b_i}^i = S_r$.
To acquire the set $S_{\succeq r}$ we first consider
sets where the first $i$ elements are the equal to $b_1, \dots b_i$
but the $i\mathtt{+}1$th element $x$ is $\succ b_{i+1}$.	
\begin{equation}\label{lpar:eq:constructsgreaterr}
  S^i_{\succeq r} =\bigcap_{1 \leq j \leq i} C_{b_j}^j \cap \bigcup_{x\succ b_{i+1}} C_{x}^{i+1}
\end{equation}

To \emph{construct} the set $S_{\succeq r}$ we apply the following union
operations:
\begin{equation}\label{lpar:eq:constructsets}
	S_{\succeq r} = \bigcup_{i = 1}^{k-1} S^i_{\succeq r} \cup S_r \cup C_{\top}
\end{equation}

That is, we can compute the set  $S_{\succeq r}$ with $O(\numprio \cdot \log n)$ set operations and four additional sets.
Notice that there is no need to store all sets $S^i_{\succeq r}$ as we can immediately add them to the final set when we have computed them.
The number of $\cup$-operations is immediately bounded by $O(\numprio \cdot k)=O(\numprio \cdot \log n)$ by the above definitions.
In order to bound the number of $\cap$-operations by $O(\log n)$, we do the following.
To compute the sets $S^i_{\succeq r}$  we introduce an additional set 
$T^{i} = \bigcap_{1 \leq j \leq i} C_{b_j}^j$.
We have that $S^i_{\succeq r} = T^i \cap \bigcup_{x\succ b_{i+1}} C_{x}^{i+1}$
and $T^{i+1}= T^{i} \cap C_{b_{i+1}}^{i+1}$, i.e., we just need two $\cap$ operation to compute the next set $S^{i+1}_{\succeq r}$.
Moreover, we have that $S_r=T^k$ and thus can be computed with just one $\cap$ operation.
In total, this amounts to $2k-2 = O(\log n)$ many  $\cap$-operations.

\smallskip\noindent\emph{Updating the set $S_{\succeq r}$ to $S'_{\succeq
		r}$.}\label{lpar:par:updatesets}
Assume that the set $S_{\succeq r}$ is the old set that is saved within the sets
$C^i_x$. The new set, $S'_{\succeq r}$ is an updated set, which is also a
superset.
First, compute the difference $S_\Delta = S'_{\succeq r} \setminus S_{\succeq r}$.
Intuitively, the algorithm increased the rank of the vertices in $S_\Delta$.  We delete their old values by updating 
$C^i_{x} = C^i_{x} \setminus S_\Delta$ \ for all $i = 0\dots k$ and each $x \in
\{0,\dots \numprio-1\}$. Then we add the vertices to the set  $C^i_{r_i} = C^i_{r_i}
\cup S_\Delta$ for all $i = 0\dots k$.
In total there are $O(\numprio \log n)$ many $\setminus$-operations and $O(k) =
O(\log n)$ many $\cup$-operations.

\subsection{Algorithmic Details for the Reduced Symbolic Space Algorithm 1}
In this section we present the algorithmic details for the reduced symbolic space algorithm.
\begin{algorithm}[ht]
	\SetKwInOut{Input}{input}
	\SetKwInOut{Output}{output}
	\SetAlgoVlined
	\Input{Parity Game $\PG$}
	\caption{Parity Algorithm with reduced sets}\label{lpar:alg:bb4}
	Initialize $C^i_{\_} \gets  V$ for $0 \leq i \leq k$\; \label{lpar:alg:bb4Linit1}
	Initialize $C^i_{c} \gets \emptyset$ for $0 \leq i \leq k$, $c \in C$\; \label{lpar:alg:bb4Linit2}
	D.activate$(\min)$\;
	\While{$r \gets D.popActiveSet()$}{\label{lpar:alg:bb4:while1}
		$S_{\succeq r}  \gets D.getSet(r)$\;
		D.deactivate($r$)\;
		$P \gets \CP_z(S_{\succeq r})$\;\label{lpar:alg:bb4:getCP}
		\For{$c \in C$} {\label{lpar:alg:bb4:for1}
			$r' \gets \slift(r,c)$\;
			$S_{\succeq r'}  \gets D.getSet(r')\;$\;\label{lpar:alg:bb4:getSet1}
			$rold \gets r'$\;\label{lpar:alg:bb4:roldv}
			$S_{rold} \gets S_{\succeq r'} \cup (P \cap
			V_c)$\label{lpar:alg:bb4:rolds}\;
			\While{$P \cap V_c \not\subseteq S_{\succeq
					r'}$}{\label{lpar:alg:bb4:while2}
				$S_{\succeq r'} \gets S_{\succeq r'} \cup (P \cap V_c)$\;
				$S_{\succeq next(r')} \gets D.getSet(D.next(r'))\;$\;
				\If{$r' = \top $ or $S_{\succeq r'} \supset S_{\succeq next(r')}$}
				{
					$D.activate(r', S_{\succeq r'})$\;
				}
				$S_{\succeq prev(r')} \gets D.getSet(D.getPrevious(r'))$\;
				\If{$S_{\succeq r'} = S_{\succeq prev(r')} \cup (P \cap V_c)$}{
					$D.removeSet(D.getPrevious(r'))$\;
				}				
				$r'\gets D.getPrevious(r')$\;
				$S_{\succeq r'}  \gets D.getSet(r')\;$\;
			}   
			$D.update(rold, S_{rold})\;$\;
			\label{lpar:alg:bb4:update}
		}
	}
	\Return $S_\top$
\end{algorithm}

\smallskip\noindent\emph{Data Structure.}\label{lpar:alg:bb4:ds}
The new data structure $D$ supports all functions of Data
Structure~\ref{lpar:alg:bb3:ds} but the nodes of the AVL tree no longer store a
pointer to a set (but only the value $r$).
The data structure $D$ stores $\numprio \log n + 1 $ sets as described in Section~\ref{lpar:sec:reducesets}, and 
overrides the $D.update(r,S)$-operation.

\begin{itemize}
	\item D.update(r,S): Updates the set $S_{\succeq r}$ to be $S$.
		Every set $S_{\succeq r'}$ with $r' \prec r$  is updated to
		$S_{\succeq r'}\cup S$.
\end{itemize}

\smallskip\noindent\emph{Implementation of the data structure $D$.}
\begin{itemize}
	\item D.getSet(r): This function computes the set $S_{\succeq r}$ as
		described in Section~\ref{lpar:sec:reducesets} and returns it.
	\item D.update(r,S): Computes the update of the set $S_{\succeq r}$ with $S$ as
		described in Section~\ref{lpar:sec:reducesets}.
		Precondition: $S_{\succeq r} \subseteq S$.

\end{itemize}

Notice that Algorithm~\ref{lpar:alg:bb4} only differs from Algorithm~\ref{lpar:alg:blackbox3} in 
(a) the way the  \texttt{D.getSet(r)} method is implemented and 
(b) and how we update the sets, i.e., the overridden \texttt{D.update(r,S)} method.
Hence, to establish the correctness of Algorithm~\ref{lpar:alg:bb4}, it suffices to show that these two methods
do the same as the corresponding operations in Algorithm~\ref{lpar:alg:blackbox3}.
The correctness then directly follows from Proposition~\ref{lpar:lem:bb3:correctness}.

\begin{proposition}[Correctness]\label{lpar:prop:bb4:correctness}
	Given a parity game $\PG$ Algorithm~\ref{lpar:alg:bb4}
	computes the corresponding winning set.
\end{proposition}

\begin{proof}
	\begin{sloppypar}
		We show that the \texttt{D.getSet(r)} method  (in the interplay with the \texttt{D.update(r,S)} method)
		in Algorithm~\ref{lpar:alg:bb4} returns the sames set as the \texttt{D.getSet(r)} method in 
		Algorithm~\ref{lpar:alg:blackbox3}. The correctness then directly follows from Proposition~\ref{lpar:lem:bb3:correctness}.
	\end{sloppypar}

	The proof is by induction. Consider the base case after the initialization of the data structure $D$
	and its sets $C_c^i$.
	In Algorithm~\ref{lpar:alg:blackbox3} we have that $D.getSet(\min)$ would return $V$
	and $D.getSet(r)$ would return the empty set for $\min \prec r$.
	In Algorithm~\ref{lpar:alg:bb4}, by the initialization in Line~\ref{lpar:alg:bb4Linit1},
	we have that $D.getSet(\min)$ would return $V$ and 
	by the initialization in Line~\ref{lpar:alg:bb4Linit2} we have $D.getSet(r)$ would return the empty set for $\min \prec r$.
	Therefore, the base case is satisfied.

	Notice that the data structures both in Algorithm~\ref{lpar:alg:blackbox3} and
	Algorithm~\ref{lpar:alg:bb4} are only changed in the for-loop. 
	Assume the claim holds before an iteration of the for-loop. Let $\bar{r}$ be
	the element of $\WOrder$ currently processed by the outer while-loop,
	and let $\bar{r}'$ the $r'$ currently processed by the for-loop.
	Consider some $r \in \WOrder$. 
	By induction hypothesis, $D.getSet(r)$ coincides in both algorithms beforehand.
	If $r \succ \bar{r}'$ then $D.getSet(r)$ is not affected by the changes in both algorithms and 
	thus $D.getSet(r)$ coincides in both algorithms after the iteration of the for loop.
	If $r \preceq \bar{r}'$ then Algorithm~\ref{lpar:alg:blackbox3} updates the data structure such
	that $P \cup V_c$ is added to the set $S_{\succeq r}$ (the set returned by $D.getSet(r)$).
	Now consider Algorithm~\ref{lpar:alg:bb4}. Here the algorithm adds the set $P \cup
	V_c$ to the set $C_x^i$
	that correspond to $\bar{r}'$.
	As $r \preceq \bar{r}'$ there is an $i\geq 0$ such that $r$ and $\bar{r}'$ coincide on the first $i$ elements 
	and the set $P \cup V_c$ is then contained in the set $S^i_{\succeq r}$.
	That means that the set returned by $D.getSet(r)$ contains the set $P \cup V_c$. 
	Moreover, as only vertices in $P \cup V_c$ are affected by the update, all the vertices that were previously contained in
	$D.getSet(r)$ are still contained in the set. In other words, Algorithm~\ref{lpar:alg:bb4} adds $P \cup V_c$  to the set $S_{\succeq r}$.
	That is, the two $D.getSet(r)$ methods coincide also after the iteration of the for-loop for all $r \in \WOrder$.
	Thus, we have that $D.getSet(r)$ coincides in the two algorithms and the correctness of Algorithm~\ref{lpar:alg:blackbox3}
	extends to Algorithm~\ref{lpar:alg:bb4}.
\end{proof}

\begin{proposition}[Symbolic operations]\label{lpar:prop:bb4:runningtime}
	Algorithm~\ref{lpar:alg:bb4} uses
	$O(\numprio \log n)$ sets with $O(n |\WOrder|)$ symbolic one-step operations and
	$O( \numprio^2 n |\WOrder| \log n)$ basic
	set operations.
\end{proposition}
\begin{proof}
	There are $O(n \cdot |\WOrder|)$ iterations if the while-loop at Line~\ref{lpar:alg:bb4:while1} (cf.~Proposition~\ref{lpar:prop:numberactivations}). 
	Therefore, the number of \os operations is $O(n \cdot |\WOrder|)$.
	In each iteration of the while-loop, the
	for-loop at Line~\ref{lpar:alg:bb4:for1} has $\numprio$ iterations. 
	The basic set operations at Line~\ref{lpar:alg:bb4:getSet1},
	Line~\ref{lpar:alg:bb4:while2}, Line~\ref{lpar:alg:bb4:rolds} and at
	Line~\ref{lpar:alg:bb4:update} occur $O(\numprio \cdot n \cdot
	|\WOrder|)$ times. This sums up to a total of $O(\numprio^2 \cdot n \cdot |\WOrder| \log n)$
	basic set operations (as each \texttt{getSet(r)} requires $O(\numprio \log n)$
	basic set operations). 
	By the same amortized argument as in the proof of
	Proposition~\ref{lpar:prop:numberactivations},
	we obtain that the total number of basic set-operations in executions of the inner while-loop is in 
	$O(n |\WOrder| \numprio \log n)$.
	The number of basic set operations is, thus, in $O(\numprio^2 n |\WOrder| \log
	n)+ O(\numprio n |\WOrder| \log n) = O(\numprio^2  n \cdot |\WOrder|  \log n)$. 
\end{proof}

Due to Proposition~\ref{lpar:prop:bb4:correctness} and
Proposition~\ref{lpar:prop:bb4:runningtime} and the fact that we use only $O(\numprio \log
n)$ symbolic space in the modified data structure $D$, we obtain Theorem~\ref{lpar:thm:bb4}.

\begin{theorem}\label{lpar:thm:bb4}\label{lpar:thm:reducedspace_main}
	The winning set of a parity game can be computed in  $O(n \cdot |\WOrder|)$ \os operations,
	$O(\numprio^2 n \cdot |\WOrder| \cdot \log n)$ basic set operations, and $O(\numprio \cdot \log n)$ symbolic space. 
\end{theorem}

\begin{remark}\label{lpar:remark:polyspace}
	Note that Theorem~\ref{lpar:thm:bb4} achieves bounds similar to Theorem~\ref{lpar:thm:bb3}
	with a factor $\numprio \cdot \log n$ increase in basic set operations,
	however, the symbolic space requirement decreases from $O(n)$ to $O(\numprio \cdot \log n)$. 
	In particular, using the bounds as mentioned in Remark~\ref{lpar:rem:bound}, we obtain 
	a set-based symbolic algorithm that requires quasi-polynomially many \os and
	basic set operations, and $O(\numprio \cdot \log n)$ symbolic space,
	and moreover, when $\numprio \leq \log n$, then the algorithm requires polynomially
	many \os and basic set operations and 
	only poly-logarithmic $O(\log^2 n)$ symbolic space. 
\end{remark}

\begin{remark}\label{lpar:remark2}
	Recall that our AVL-tree data structure potentially requires $O(n)$ non-symbolic space 
	(cf.~Remark~\ref{lpar:remark1}) in the worst case. 
	The algorithm above reduces the symbolic space requirement to $O(\numprio \cdot \log n)$.
	We now argue how to reduce the non-symbolic space to the same bound.
	The main purpose of our AVL-tree data structure is to (a) avoid storing all sets explicitly and
	(b) efficiently maintain pointers to the active sets. 
	As we now have a succinct representation of all sets we are only left with (b).
	For the price of a higher number of basic set operations, we can rid our
	black box algorithm of the AVL-tree in the implementation of the data
	structure:
	We can explicitly compute the active sets in each iteration of the while
	loop. This increases the number of basic set operations to $O(\numprio \cdot n \cdot |\WOrder|^2 \log n)$ while
	the other symbolic resource consumption stays the same.
\end{remark}

We give the algorithmic details of Remark~\ref{lpar:remark2} in the next section.
\section[Algorithmic Details for the Reduced Symbolic Space Algorithm 2][Details for the
Reduced Symbolic Space Alg. 2]{Algorithmic Details for the Reduced Symbolic Space Algorithm 2}
\label{lpar:app:pureSymbolic}
In this section, we present a black box parity algorithm that only uses
$O(\numprio
\cdot \log n)$ symbolic
space. In return, the algorithm needs a factor $|\WOrder|$ more basic set
operations.

\begin{definition}
	An element $r \in \WOrder$ is \emph{active} if 
	(i) $S_{\succeq r} \supset S_{\succeq r'}$, for $r \prec r'$ and 
	(ii) 
	if the set $S_{\succeq r}$ has changed
	since the last time $r$ was selected at Line~\ref{lpar:alg:bb5:while1} of
	Algorithm~\ref{lpar:alg:bb5}.
\end{definition}

\smallskip\noindent\emph{Key Intuition.}
Algorithm~\ref{lpar:alg:bb4} uses the AVL tree to keep track of the $r \in \WOrder$ that need to
be potentially processed in the future, i.e., the \emph{active} r in $\WOrder$.
Algorithm~\ref{lpar:alg:bb5} initializes $\numprio \log n +1$ additional sets, to determine if the
set $S_{\succeq r}$ was changed since the last time we chose $r$ at Line~\ref{lpar:alg:bb5:while1}. 
We call the additional $\numprio \log n + 1$ sets \emph{the data structure $O$}. 
In each iteration of Line~\ref{lpar:alg:bb5:while1} we update $O$ by copying the
content the current set $S_{\succeq r}$ to the corresponding set $T_{\succeq r}$
in $O$. Then we do our usual lifting like in Algorithm~\ref{lpar:alg:bb4} with the set $S_{\succeq r}$.
To determine the next active $r$, we iterate from the minimum element in $\WOrder$ to the
maximum element. 
To check condition (ii) in the iteration $r$, we determine
the set $S_{\succeq r}$ and $T_{\succeq r}$ using the data structures $C$ for
$S_{\succeq r}$ and $O$ for $T_{\succeq r}$ respectively. 
Notice that when $S_{\succeq r} \neq T_{\succeq r}$ we changed the set
$S_{\succeq r}$ since the
last time $r$ was picked in the while-loop at
Line~\ref{lpar:alg:bb5:while1}. To check condition (i) we need to check whether $S_r
\neq \emptyset$ which would mean that there is no set with higher rank $r' > r$ which already
contains all elements of $S_{\succeq r}$, i.e., $S_{\succeq r} = S_{\succeq r'}$. 
If both conditions are satisfied we return the active set.

\begin{algorithm}
	\SetKwInOut{Input}{input}
	\SetKwInOut{Output}{output}
	\SetAlgoVlined
	\SetKwFunction{algo}{algo}
	\SetKwFunction{getactive}{getMinActiveSet}
	\SetKwProg{myalg}{Algorithm}{}{}
	\Input{Parity Game $\PG$, finite total order $(\WOrder,\preceq)$, monotone function
		$\slift(\cdot)$, a player $z \in \{ \E, \O\}$}
	\caption{Parity Algorithm with reduced sets}\label{lpar:alg:bb5}
	Initialize $C$ with $C^i_{\_} \gets V$ for $0 \leq i \leq k$\; 
	Initialize $C$ with $C^i_{c} \gets \emptyset$ for $0 \leq i \leq k$, $c \in C$\;
	Initialize $O$ with $O^i_{c} \gets \emptyset$ for $0 \leq i \leq k$, $c \in C \cup \{\_\}$\; 

	\While{$r \gets \getactive{}$}{\label{lpar:alg:bb5:while1}
		$S_{\succeq r} \gets D.getSet(r)\;$\;
		$O.update(r,S_{\succeq r})$\;\label{lpar:alg:bb5:copy}
		$P \gets \CP_z(S_{\succeq r})$\;\label{lpar:alg:bb5:getCP}
		\For{$c \in C$} {\label{lpar:alg:bb5:for1}
			$r' \gets \slift(r,c)$\;\label{lpar:alg:bb5:liftcomp}
			$S_{\succeq r'}  \gets C.getSet(r')\;$\;\label{lpar:alg:bb5:getSet1}
			$C.update(r', S_{\succeq r'} \cup (P \cap V_c))\;$\;
			\label{lpar:alg:bb5:update}
		}
	}
	\KwRet {$S_\top$}\;\,\,

	\SetKwProg{myfunc}{function}{}{}
	\myfunc{\getactive{}}{
		$S_{\succeq r} \gets V$\;
		$T_{\succeq r} \gets O.getSet(\min)$\;
		\For{$r = \min,\min+1, \dots, \top$}{\label{lpar:alg:bb5:for2}
			$S_{r}  \gets C.getSet2(r)\;$\;
			\If{ $S_{r} \not= \emptyset$ and $S_{\succeq r} \supset T_{\succeq r}$}{\label{lpar:alg:bb5:if1}
				\Return $r$;
			}		  
			$T_{r} \gets O.getSet2(r)\;$\;
			$S_{\succeq r} \gets S_{\succeq r} \setminus S_{r}$\;\label{lpar:alg:bb5:merge1}
			$T_{\succeq r} \gets T_{\succeq r} \setminus T_{r}$\;\label{lpar:alg:bb5:merge2}
		}
		\Return false\;
	}

\end{algorithm}

In Subsection~\ref{lpar:appb:ds} we detail the data structure which is used for $C$
and $O$ in Algorithm~\ref{lpar:alg:bb5}. The proofs of Subsection~\ref{lpar:appb:correctness} and
Subsection~\ref{lpar:appb:resources} yield the theorem detailed below.

\begin{theorem}
	The winning set of a parity game can be computed without the usage of
	nonsymbolic space with $O(\numprio \log n)$ symbolic
	space, $O(\numprio n|\WOrder|^2 \log n)$ basic set operations and $O(n |\WOrder|)$ \os
	operations.
\end{theorem}

\subsection{Data Structure.}\label{lpar:appb:ds}
The new data structure which is used by $D$ and $O$ supports the following functions
\begin{itemize}
	\item getSet(r): returns the set $S_{\succeq r}$. 
	\item getSet2(r): returns the set $S_{r}$.
	\item update(r,S): Updates the set $S_{\succeq r}$ to be $S$.
		Every set $S_{\succeq r'}$ with $r' \preceq r$  is updated to
		$S_{\succeq r'}\cup S$.
\end{itemize}

\smallskip\noindent\emph{Implementation of the data structure.}

\begin{itemize}
	\item getSet(r): This function computes the set $S_{\succeq r}$ as
		described in Section~\ref{lpar:sec:reducesets} and returns it.
	\item getSet2(r,S): This function computes the set $S_{r}$ as described in
		Section~\ref{lpar:sec:reducesets} and returns it.
	\item update(r,S): Computes the update of the set $S_{\succeq r}$ with $S \subseteq V$ as
		described in Section~\ref{lpar:sec:reducesets}.
		Precondition: $S_{\succeq r}$ is a subset of $S$.

\end{itemize}

\subsection{Correctness}\label{lpar:appb:correctness}
In this section, we prove the correctness of Algorithm~\ref{lpar:alg:bb5}. 

\begin{proposition}
	Given a parity game $\PG$, Algorithm~\ref{lpar:alg:bb5} computes the winning
	set.
\end{proposition}

The only difference in Algorithm~\ref{lpar:alg:bb4} and Algorithm~\ref{lpar:alg:bb5} is how active $r \in \WOrder$ are obtained.
That is, in order to prove the correctness of Algorithm~\ref{lpar:alg:bb5} it suffices to show that the 
function \getactive{$\cdot$} returns the set of an active $r \in \WOrder$ whenever there exist one, and false otherwise.
The function $\getactive{$\cdot$}$ returns the minimal\footnote{Notice
	that the fact that $\getactive{$\cdot$}$ returns the minimal active $r$ is only used to consistently maintain the data structure $O$.}
$r \in \WOrder$ that satisfies 
(i) $S_{\succeq r} \supset S_{\succeq r'}$, for $r \prec r'$ (by the anti-monotonicity of the set it suffices to test this for the direct successor) and (ii) $T_{r} \subset S_{r}$. The next lemma shows that these sets are actually the active $r \in \WOrder$.

\begin{lemma}
	An element $r \in \WOrder$ is \emph{active} iff
	(i) $S_{\succeq r} \supset S_{\succeq r'}$, for $r \prec r'$ and 
	(ii) $T_{\succeq r} \subset S_{\succeq r}$
\end{lemma}
\begin{proof}
	We show the claim by induction.
	$C$ is initialized  such that only $\min$ is active 
	and all the other $r \in \WOrder$ correspond to the empty set and thus all, but $\min$ and $\top$
	violate (i).
	Now as all the sets of $O$ are empty (ii) holds for $\min$ but does not hold for $\top$.
	Hence, the statement is true when first entering the loop.

	Now consider an iteration of the loop where $\bar{r}$ is processed and assume the claim is true beforehand.
	By induction hypothesis $\bar{r}$ is the minimal active $r \in \WOrder$.
	Notice that the claim is only affected by update operations on $O$ and $C$.	
	First, we argue that the claim is preserved in Line~\ref{lpar:alg:bb5:copy}.
	This operation only affects $r' \preceq \bar{r}$.
	Notice that $\bar{r}$ is no longer active:
	We process $\bar{r}$ and violate condition
	(ii) as we set $T_{\succeq \bar{r}} = S_{\succeq \bar{r}}$ at
	Line~\ref{lpar:alg:bb5:copy}. 
	Let $r' \prec \bar{r}$. As $\bar{r}$ is the minimal active $r \in \WOrder$ we have that $r'$ is inactive.
	Then, by the induction hypothesis, we have $T_{\succeq r'} = S_{\succeq r'}$ beforehand
	and by the anti-monotonicity of the sets in $C$ also after the execution of Line~\ref{lpar:alg:bb5:copy}.
	That is, $r'$ is inactive and (ii) is still violated.
	Second, we show that also the update in Line~\ref{lpar:alg:bb5:update} preserves the claim.
	Let $\bar{r} \in \WOrder$ be the value computed by the $\slift$ function at
	Line~\ref{lpar:alg:bb5:liftcomp}.
	Again, this operation only affects $r' \preceq \bar{r}$.
	If the update violates condition (i) for a $r'$, i.e. $S_{\succeq r'} =
	S_{\succeq \hat{r}}$ for $r' \prec \hat{r}$ then this $r'$ also becomes inactive by the definition of active.
	Let us assume (i) is not violated and $r'$ was active beforehand.
	Then, by the induction hypothesis, we have $T_{\succeq r'} \subset S_{\succeq r'}$ beforehand, and as only 
	vertices are added to $S_{\succeq r'}$ the condition (ii) still holds.
	Let us assume (i) is not violated and $r'$ was inactive beforehand.
	Then by the induction hypothesis $T_{\succeq r'} = S_{\succeq r'}$ beforehand, and as 
	(ii) is satisfied afterwards iff $(P \cap V_c) \not\subset S_{\succeq r'}$.
	Notice that if $(P \cap V_c) \not\subset S_{\succeq r'}$ then $r'$ is
	changed and $r'$ is active because condition (i) is satisfied by assumption. 
	If $(P \cap V_c) \subset S_{\succeq r'}$ then we do not change the set and
	it stays inactive. Hence, we have that (2) is satisfied iff $r'$ becomes active.
\end{proof}
Thus, we have that \getactive{$\cdot$} returns an active set whenever there exist one, and false otherwise
and consequently Algorithm~\ref{lpar:alg:bb5} computes the winning set of $\PG$.

\subsection{Symbolic Resource consumption}\label{lpar:appb:resources}
In this subsection, we present the symbolic resource consumption of
Algorithm~\ref{lpar:alg:bb5}. First, we present the bounds for basic set operations,
\os operations, and finally, symbolic space.
\begin{proposition}
	Algorithm~\ref{lpar:alg:bb5} needs $O(\numprio n|\WOrder|^2\log n)$ basic set operations,
	$O(n|\WOrder|)$ \os operations and $O(\numprio \log n)$ symbolic space.
\end{proposition}
\begin{proof}

	\smallskip\noindent
	For the $n |\WOrder|$ \os operations bound, notice that we increase the rank of
	at least one vertex in each iteration of Algorithm~\ref{lpar:alg:bb5} by adding the vertex to
	a set with a higher rank. We can only do this until every vertex has
	a rank $|\WOrder|$. Therefore the while-loop of Algorithm~\ref{lpar:alg:bb5} has
	$O(n\cdot |\WOrder|)$ iterations. In each iteration, we perform exactly one $\CP_z$
	operation.

	\smallskip\noindent 
	For the $O(\numprio n|\WOrder|^2\log n)$ basic set operations bound, recall that the while
	loop at Line~\ref{lpar:alg:bb5:while1} is called $O(n\WOrder)$ times. 
	First, the call to the \texttt{getActiveSet()} function requires  $O(|\WOrder|
	\log n)$ basic set operations.
	There is one \texttt{update} operation in Line~\ref{lpar:alg:bb5:copy}  
	and $\numprio$ \texttt{update} operations in the for-loop at Line~\ref{lpar:alg:bb5:update},
	i.e., each iteration requires $O(\numprio^2 \log n)$ basic set operations for updates.
	Therefore, we get the bound $O(n |\WOrder| \cdot ( |\WOrder| \log n + \numprio^2 \log n) =
	O(\numprio n|\WOrder|^2\log n)$ (using $|\WOrder| > \numprio$).

	\smallskip\noindent
	For the $O(\numprio \log n)$ symbolic space bound, notice that we use the data
	structures $O$ and $C$, both of them use $\numprio \log n +1$ symbolic space,
	and a constant number of additional sets.
\end{proof}

\section{Conclusion}

In this work, we present improved set-based symbolic algorithms for parity games.
There are several interesting directions for future work. 
On the practical side, practical implementation and experiments with case studies,
especially for the algorithm presented in Section~\ref{lpar:sec:bb3} instantiated with either the ordered approach or the succinct progress measure,
is an interesting direction.
On the theoretical side, recent work~\cite{ChatterjeeDHL18} has established lower bounds for 
symbolic algorithms for graphs, and whether lower bounds can be established for 
symbolic algorithms for parity games is another interesting direction for future work.

	\chapter[Symbolic Time and Space Tradeoffs for Probabilistic Verification][Symbolic Time \& 
	Space Tradeoffs f. Prob. Verification]{Symbolic Time and Space Tradeoffs for Probabilistic Verification}\label{cha:lics}
	In this chapter, we present improved symbolic algorithms for the MEC decomposition of an MDP and parity objectives in MDPs.

\section{Introduction}
The verification of probabilistic systems, e.g., randomized protocols, 
or agents in uncertain environments like robot planning is a fundamental
problem in formal methods. 
We study a classical graph algorithmic problem that arises in the verification of
probabilistic systems and present a faster symbolic algorithm for it.
We start with the description of the graph problem and its applications, then 
describe the symbolic model of computation, then previous results, and finally our 
contributions.


\para{Applications.} 
In verification of probabilistic systems, the classical model is called 
\emph{Markov decision processes (MDPs)}~\cite{Howard},
where there are two types of vertices.
The vertices in $V_1$ are the regular vertices in a graph algorithmic 
setting and the vertices in $V_R$ represent random vertices.
MDPs are used to model and solve control problems in systems such as stochastic systems~\cite{FV97},
concurrent probabilistic systems~\cite{CY95}, planning problems in artificial intelligence~\cite{Puterman},
and many problems in the verification of probabilistic systems~\cite{baierbook}. 
The MEC decomposition problem is a central algorithmic problem in the
verification of probabilistic systems~\cite{CY95,baierbook} and it is a core component
in all leading tools of probabilistic verification~\cite{prism,STORM}. 
Some key applications are as follows:
(a)~the almost-sure reachability problem can be solved in linear time given
the MEC decomposition~\cite{CDHL16};
(b)~verification of MDPs wrt $\omega$-regular properties requires 
MEC decomposition~\cite{CY95,CY90,baierbook,CH11};
(c)~algorithmic analysis of MDPs with quantitative objectives as well
as the combination of $\omega$-regular and quantitative objectives 
requires MEC decomposition~\cite{CH-ILC,BBCFK11,CHJS15};
and (d)~applying learning algorithms to verification requires MEC decomposition 
computation~\cite{KPR18,DHKP17}.

\para{Symbolic model and algorithms.}
In verification, a system consists of variables, and a state of the system corresponds to a set of valuations, one for each variable. 
This naturally induces a directed graph: vertices represent states and the directed edges represent
state transitions. 
However, as the transition systems are huge they are usually not explicitly represented during their analysis. 
Instead they are \emph{implicitly represented} using e.g., binary-decision diagrams (BDDs)~\cite{bryant1986graph,bry92}. 
An elegant theoretical model for algorithms that works on this implicit representation, without considering the specifics of the representation and implementation, 
has been developed, called \emph{symbolic algorithms} (see
e.g.~\cite{BurchCMDH90,ClarkeMCH96,Somenzi99,ClarkeBook,ClarkeGJLV03,GentiliniPP08,ChatterjeeHJS13}).
A symbolic algorithm is allowed to use the same mathematical, logical, and memory access operations as a regular RAM algorithm, except for the
access to the input graph: 
It is not given  access to the input graph through an adjacency list or adjacency matrix representation but instead 
\emph{only} through two types of \emph{symbolic operations}:
\begin{compactenum}
\item {\emph{One-step operations Pre and Post:}} Each \emph{predecessor \textsc{Pre}} 
	(resp., \emph{successor \textsc{Post}}) operation is given a set $X$ of vertices and
returns the set of vertices~$Y$ with an edge to (resp., edge from) 
some vertex of~$X$.
\item \emph{Basic set operations:} Each basic set operation is given one or two sets of vertices or edges 
and performs a union, intersection, or complement on these sets.
\end{compactenum}
Symbolic operations are more expensive than the non-symbolic operations and 
thus \emph{symbolic time} is defined as the \emph{number of symbolic operations} of a symbolic algorithm. 
One unit of space is defined as one set (not the size of the set) due to the implicit representation as a BDD\@.
We define \emph{symbolic space} of a symbolic algorithm as the maximal \emph{number of sets} stored simultaneously.
Moreover, as the symbolic model is motivated by the compact representation of huge graphs,
we aim for symbolic algorithms that require sub-linear space.  

\para{Previous results and main open question.} 
We summarize the previous results and the main open question.
We denote by $|V| = n$ and $|E| = m$ the number of vertices and edges, respectively.
\begin{compactitem}
\item\para{Standard RAM model algorithms.} 
The computation of the MEC (aka controllable recurrent set in early works) 
decomposition problem has been a central problem since the work of~\cite{CY90,CY95,deAlfaroThesis}.
The classical algorithm for this problem requires $O(n)$ SCC decomposition calls,
and the running time is $O(nm)$.
The above bound was improved to (a)~$O(m\sqrt{m})$ in~\cite{CH11} and (b)~$O(n^2)$ in~\cite{CH14}.
While the above algorithms are deterministic, a randomized algorithm with expected almost-linear
$\widetilde{O}(m)$ running time  has been presented in~\cite{CDHS19CONCUR}.

\item \emph{Symbolic algorithms.}
The symbolic version of the classical algorithm for MEC decomposition requires $O(n)$ 
symbolic SCC computation. 
Given the $O(n)$ symbolic operations SCC computation algorithm 
from~\cite{GPP03}, we obtain an $O(n^2)$ symbolic operations MEC decomposition 
algorithm, which requires $O(\log n)$ symbolic space.
A symbolic algorithm, based on the algorithm of~\cite{CH11}, was presented in~\cite{ChatterjeeHJS13},
which requires  $O(n \sqrt{m})$ symbolic operations and  $O(\sqrt{m})$ symbolic space.
\end{compactitem}
The classical algorithm from the 1990s with the linear symbolic-operations SCC 
decomposition algorithm from 2003 gives the $O(n^2)$ symbolic operations bound, 
and since then the main open question for the MEC decomposition problem has been 
whether the worst-case $O(n^2)$ symbolic operations bound can be beaten.

\para{Our contributions.}
\begin{enumerate}
\item In this work we answer the open question in the affirmative. 
Our main result presents a symbolic operation and symbolic space trade-off algorithm that
for any $0<\epsilon \leq 1/2$ requires $\widetilde{O}(n^{2-\epsilon})$ symbolic operations and 
$\widetilde{O}(n^{\epsilon})$ symbolic space. 
In particular, our algorithm for $\epsilon=1/2$ requires $\widetilde{O}(n^{1.5})$ symbolic operations 
and $\widetilde{O}(\sqrt{n})$ symbolic space, which improves both the symbolic operations and 
symbolic space of~\cite{ChatterjeeHJS13}.

\item We also show that our techniques extend beyond MEC computation and is also applicable to 
almost-sure winning (probability-1 winning) of $\omega$-regular objectives for MDPs. 
We consider parity objectives which are cannonical form to express $\omega$-regular objectives.
For parity objectives with $d$ priorities the previous symbolic algorithms require 
$O(d)$ calls to MEC decomposition; thus leading to bounds such as 
(a)~$O(n^2\cdot d)$ symbolic operations and $O(\log n)$ symbolic space;
or (b)~$O(n \sqrt{m} \cdot d)$ symbolic operations and $O(\sqrt{m})$ symbolic space.
In contrast we present an approach that requires $O(\log d)$ calls to MEC decomposition,
and thus our algorithm requires
$\widetilde{O}(n^{2-\epsilon})$ symbolic operations and
$\widetilde{O}(n^{\epsilon})$ symbolic space, for all $0 < \epsilon \leq 1/2$.
Thus we improve the time-space product from $\O(n^2d)$ to $\widetilde{O}(n^2)$.
\end{enumerate}

\para{Technical contribution.} 
Our main technical contributions are as follows:
(1)~We use a separator technique for the decremental SCC algorithm from~\cite{CHILP16}. 
However, while previous MEC decomposition algorithms for the standard RAM model 
(e.g.~\cite{CDHS19CONCUR}) use ideas from decremental SCC algorithms, data-structures used in 
decremental SCC algorithms of~\cite{CHILP16,BPW19} 
such as Even-Shiloach trees~\cite{EvenS81} have no symbolic representation.
A key novelty of our algorithm is that instead of basing our algorithm on a decremental algorithm we use 
the incremental MEC decomposition algorithm of~\cite{CH11} 
along with the separator technique. 
Moreover, the algorithms for decremental SCC of~\cite{CHILP16,BPW19} are randomized algorithms, 
in contrast, our symbolic algorithm is deterministic.
(2)~Since our algorithm is based on an incremental algorithm approach, we need to support an 
operation of collapsing ECs even though we do not have access to the graph directly (e.g.\ through an adjacency list representation), but only have access to the graph through symbolic operations.
(3)~All MEC algorithms in the classic model first decompose the graph into its SCCs and then run on each SCC\@. However,
to achieve sub-linear space we cannot store all SCCs, and we show that our algorithm
has a tail-recursive property that can be utilized to achieve sub-linear space.
With the combination of the above ideas, we beat the long-standing $O(n^2)$ symbolic operations barrier
for the MEC decomposition problem, along with sub-linear symbolic space.

\para{Implications.}
Given that MEC decomposition is a central algorithmic problem for MDPs, our result has 
several implications.
The two most notable examples in probabilistic verification are:
(a)~almost-sure reachability objectives in MDPs can be solved with 
$\widetilde{O}(n^{1.5})$ symbolic operations, improving the previous known
$O(n \sqrt{m})$ symbolic operations bound; and
(b)~almost-sure winning sets for canonical $\omega$-regular objectives such as
parity and Rabin objectives with $d$-colors can be solved with 
$O(d)$ calls to the MEC decomposition followed by a call to almost-sure 
reachability~\cite{deAlfaroThesis,ChaThesis}, and hence our result gives  
an $\widetilde{O}(d n^{1.5})$ symbolic operations bound 
improving the previous known $O(d n \sqrt{m})$ 
bound.

\section{Algorithmic Tools}
In this section, we present various algorithmic tools that we use in our algorithms. 

\subsection{Symbolic SCCs Algorithm}\label{lics:ss:symbolicscc}
In~\cite{GPP03} and~\cite{ChatterjeeDHL18} symbolic algorithms and lower bounds for computing the
SCC-decomposition are presented.
Let $D$ be the diameter of $G$ and let $D_C$ be the diameter of 
the SCC $C$. The set $\SCCs{G}$ is a family of sets, where each set contains one SCC of $G$.  
The upper bounds are summarized in Theorem~\ref{lics:thm:SCCs}.
\begin{theorem}[\cite{GPP03,ChatterjeeDHL18}]\label{lics:thm:SCCs}
	The SCCs can be computed in
	$\Theta(\min(n, D|\SCCs{G}|,\\ \sum_{C \in \SCCs{G}} (D_C + 1)))$ symbolic operations and $\O(1)$ symbolic space. 
\end{theorem}
Note that the SCC algorithm of~\cite{GPP03} can be easily adapted
to accept a starting vertex which specifies the 
SCC computed first.
Also, SCCs are output when they are detected by the algorithm
and can be processed before the remaining SCCs of the graph are computed. 

\SetKwFunction{SCCFind}{SCCFind}
We write $\SCCFind{$V',s,P$}$ when we refer to the above described algorithm for computing 
the SCCs of the subgraph with vertices $V' \subseteq V$ and using $s \in V$ as the starting vertex for the algorithm 
in the MDP $P = (V,E,\langle V_1, V_R \rangle, \delta)$.

\subsection{Symbolic Random Attractors}
Given a set of vertices $T$ in an MDP $P$, the random attractor $\Attr{R}{P}{T}$ is a set of vertices
consisting of (1) $T$, (2) random vertices with an edge to some vertex in $\Attr{R}{P}{T}$, (3) player-1 vertices
with all outgoing edges in $\Attr{R}{P}{T}$.  
Formally, given an MDP $P = (V,E, \langle V_1, V_R \rangle, \delta)$, let $T$ be a set of vertices.
The random attractor $\Attr{R}{P}{T}$ of $T$ is defined as: $\Attr{R}{P}{T} = \bigcup_{i\geq 0} A_i$ where $A_0 = T$ and $A_{i+1} = A_i \cup
(\Pre{E}{A_i} \setminus (V_1 \cap \Pre{E}{V \setminus A_i}))$ for all $i>0$. We sometimes refer to $A_i$ as the $i$-th level of the attractor.

\begin{lemma}[\cite{CHLOT18}]\label{lics:lem:attrRunning}
	The random attractor $\Attr{R}{P}{T}$ can be computed with at most $O(|\Attr{R}{P}{T}\setminus{T}|+1)$ many symbolic
	operations.
\end{lemma}
The lemma below establishes that the random attractor of random vertices with edges out of a strongly connected set
is not included in any end-component and that it can be removed without affecting the ECs of the remaining graph. Hence, we use the lemma to identify vertices that do not belong to any EC\@.
The proof is analogous to the proof in Lemma~2.1 of~\cite{CH11}.

\begin{lemma}[\cite{CH11}]\label{lics:lem:attr_remove}
	Let $P = (V,E,\ls V_1, V_R \rs, \delta)$ be an MDP\@.
	Let $C$ be a strongly connected subset of $V$.
	Let $U = \Set{v \in C \cap V_R \mid \Out{v} \cap (V \setminus C) \neq \emptyset}$ 
	be the random vertices in $C$ with edges out of $C$. 
	Let $Z = \Attr{R}{P}{U} \cap C$. 
	Then, for all non-trivial EC's, $X \subseteq C$ in $P$ we have $Z \cap X = \emptyset$.
\end{lemma}

\subsection{Separators}~\label{lics:ss:separator}
Given a strongly connected set of vertices $X$ with $|X| = n$, a separator is a non-empty set $T \subseteq X$ such that  
the size of the SCC in the graph induced by $X \setminus T$ is small. 
More formally, we call $T \subseteq X$ a $q$-separator if each SCC in the subgraph induced by $X \setminus T$ has at most $n
-q \cdot |T|$ vertices. For example, if we compute a $q = \sqrt{n}$-separator, the SCCs in the
subgraph induced by $X \setminus T$ have size $\leq \sqrt{n}$ as $|T|$ is non-empty. In
\cite{CHILP16}, the authors
present an algorithm that computes a $q$-separator when the diameter of $X$ is large. We briefly sketch the symbolic 
version of this algorithm.

The procedure $\Separator{$X,\gamma$}$ computes a $q = \lfloor \gamma/(2 \log n) \rfloor$-separator $T$
when $X$ has diameter at least $\gamma$:

\begin{enumerate}
	\item Let $q \gets \lfloor \gamma/(2 \log n) \rfloor$.
	\item Try to compute a BFS tree $K$ of either $X$ or the reversed graph of $X$ of depth at least $\gamma$ 
with an arbitrary vertex $v \in X$ as root. In the symbolic algorithm, we build the BFS trees
with $\Pre{}{\cdot}$ and $\Post{}{\cdot}$ operations.
	\item If the BFS trees of both $X$ and the reversed graph $X$ with root $v$ have less than $\gamma$ levels then return $\emptyset$. Note that the
		diameter of $X$ is then $\leq 2\gamma$.
	\item Let layer $L_i$ of $K$ be the set of vertices $L_i \subseteq X$ with distance $i$ from $v$.
		Due to~\cite[Lemma 6]{CHILP16}, there is a certain layer $L_i$ of the BFS tree $K$ which is a $q$-separator. 
		Intuitively, they argue that removing the layer $L_i$ of $K$ separates $X$ into two parts: (a)
		$\bigcup_{j<i} L_j$ and (b) $\bigcup_{j>i} L_j$
		which cannot be strongly connected anymore due to the fact that $K$ is a BFS tree.
		They show that one can always efficiently find a layer $L_i$ such that both part (a) and part (b) are
		small.
	\item We efficiently find the layer $L_i$ while building $K$.
\end{enumerate}

The detailed symbolic implementation of $\Separator{$X,\gamma$}$ is illustrated in Algorithm~\ref{lics:alg:separator}.
\begin{algorithm}\label{lics:alg:separator}
	\caption{Computing the separator}
	\Procedure{\Separator{$X,\gamma$}}{
		$q \gets \lfloor \gamma/(2 \log n) \rfloor,\ v \gets \Pick{X}, i\gets 0,\ K \gets \Set{v},\ c \gets 0, L \gets \emptyset, R \gets \emptyset $\;
		\While{$(\Post{E}{K} \cap X) \not\subseteq K$}{\label{lics:alg:separator:while1}
			\If{$q \leq i \leq \gamma/2$}{ 
				$Z \gets \Post{E}{K}\setminus K \cap X$\;
				\lIf{$L = \emptyset$ and $|Z| \leq 2^{i/q-1} $}{
					$L \gets Z$
				}
			}
			\If{$\gamma/2 \leq i \leq \gamma - q$ }{ 
				$Z \gets \Post{E}{K}\setminus K \cap X$\;
				\lIf{$|Z| \leq 2^{(\gamma-i)/q-1}$}{
					$R \gets Z$
				}
			}
			\lIf{$i \leq \gamma/2$}{
				$c\gets c+ |\Post{E}{K}\setminus K \cap X|$
			}
			$K \gets K \cup (\Post{E}{K} \cap X)$\;
			$i \gets i +1$\;
		}
		\If{$i < \gamma$}{
			$i\gets 0,\ K \gets \Set{v},\ c\gets 0, L \gets \emptyset, R \gets \emptyset$\;
			\While{$i \leq \gamma$ and $(\Pre{E}{K} \cap X) \not\subseteq K$}{\label{lics:alg:separator:while2}
				\If{$q \leq i \leq \gamma/2$}{ 
					$Z \gets \Pre{E}{K}\setminus K \cap X$\;
					\lIf{$L = \emptyset$ and $|Z| \leq 2^{i/q-1} $}{
						$L \gets Z$
					}
				}
				\If{$\gamma/2 \leq i \leq \gamma - q$ }{ 
					$Z \gets \Pre{E}{K}\setminus K \cap X$\;
					\lIf{$|Z| \leq 2^{(\gamma-i)/q-1} $}{
						$R \gets Z$
					}
				}
				\lIf{$i \leq \gamma/2$}{
					$c\gets c+ |\Pre{E}{K}\setminus K \cap X|$
				}
				$K \gets K \cup (\Pre{E}{K} \cap X)$\;
				$i \gets i +1$\;
			}

			\If{$i < \gamma$}
			{
				\Return{$\emptyset$}; \tcp{$\Diam{X} < 2\gamma$}
			}
		}
		\lIf{$c< |X|/2$}{
			\Return{$L$}
			}\lElse{
			\Return{$R$}
		}
	}
\end{algorithm}
The following lemmas summarize useful properties of $\Separator{$X,\gamma$}$.

\begin{observation}[\cite{CHILP16}]\label{lics:obs:SepSize}
	A $q$-separator $S$ of a graph $G$ with $|V|=n$ vertices contains at most $\frac{n}{q}$ vertices, 
	i.e., $|S| <  \frac{n}{q}$.
\end{observation}

\begin{lemma}[\cite{CHILP16}]\label{lics:lem:sepqual}
	Let $X$ be a strongly connected set of vertices with $|X| = k$, let $r \in
	X$, and	let $\gamma$ be an integer such that $q=\lfloor \gamma/(2\log k) \rfloor \geq 1$. Then $\Separator{$X,\gamma$}$ 
	computes a $\lfloor \gamma / (2\log k) \rfloor$-separator
	if there exists a vertex $v \in X$ where the distance
	between $r$ and $v$ is at least $\gamma$. If no
	such vertex $v$ exists, then $\Separator{$X,\gamma$}$ returns the empty set.
\end{lemma}

The following lemma bounds the symbolic resources of the algorithm. 
\begin{lemma}\label{lics:lem:runspacesep}
	Algorithm~\ref{lics:alg:separator} runs in $O(|X|)$ symbolic operations and uses $O(1)$ symbolic space.
\end{lemma}
\begin{proof}
	The bound on the symbolic operations of the while loop at Line~\ref{lics:alg:separator:while1} is clearly
	in $O(|X|)$ because we perform $\Post{E}{K} \cap X$ operations until the set $K$ fully contains
	$X$ or $\Post{E}{K} \cap X$ is fully contained in $K$. Note that we only perform a constant amount
	of symbolic work in the body of the while-loop.
	As each $\Post{E}{K} \cap X$ operation adds at least one vertex to $X$ until termination, we perform
	$O(|X|)$ many operations in total. 
	A similar argument holds for the while-loop at Line~\ref{lics:alg:separator:while2}.
	Note that we use a constant amount of sets in Algorithm~\ref{lics:alg:separator} which implies 
	$O(1)$ symbolic space.
\end{proof}

\section{Symbolic MEC decomposition}
In this section, we first define how we collapse end-components.
Then we present the algorithm for the symbolic MEC decomposition.

\subsection{Collapsing End-components}
A key concept in our algorithm is to collapse a detected EC $X \subseteq V$ of an MDP $P = (V,E,\langle V_1,V_R \rangle, \delta)$ to a single player-1 vertex $v \in X$ in order to speed up the computation of end-components that contain $X$.
Notice that we do not have access to the graph directly, but only have access to the graph through symbolic operations.

\begin{algorithm}[b]
	\Procedure{\CollapseEC{$X,P$}}{	
	\If{$X \cap V_1 \neq \emptyset$}{$v \gets \Pick{X \cap V_1}$\;
		
		$E \gets E \cup \left(((\Pre{E}{X}\setminus X) \times \Set{v}) \cup (\Set{v} \times (\Post{E}{X}\setminus X))\right)$; \tcp{$v$ gets edges of $X$}

		$E \gets E \setminus \left((V  \times (X\setminus v)) \cup ((X\setminus v) \times V)\right); $\tcp{Remove all other edges of $X$}
	\Return{$\{v\}$}
	}
	\lElse(\tcp*[h]{Remove all edges of $X$}){$E \gets E \setminus \left( X  \times X \right)$}
}
	\caption{Collapse an EC $X$ into single vertex}\label{lics:alg:CollapseEC}
\end{algorithm}

We define the collapsing (see Algorithm~\ref{lics:alg:CollapseEC}) of an EC $X$ as picking a player-1
vertex $v \in X$, directing all the incoming edges $X$ to $v$, directing the outgoing edges of $X$ from $v$ and removing all edges to and from vertices in $X\setminus v$. Also, we do not include edges going from $X$ to $v$.
For vertices not in $X$, we have that they are in a non-trivial MEC in the modified MDP 
iff they are in a non-trivial MEC in the original MDP\@.
We denote with $\EC{P}$ the set of sets with all end-components in $P$.
The following lemma summarizes the property as observed in~\cite{CH11}.

\begin{lemma}\label{lics:lem:collapse}
    For MDP $P = (V,E,\ls V_1, V_R \rs, \delta)$ with $X \in \EC{P}$
    and the MDP $P'$ that results from collapsing $X$ to $v\in X$ with $\CollapseEC{$X,P$}$ we have: 
    (a) for $D \subseteq V \setminus X$ we have $D \in \EC{P}$ iff $D \in \EC{P'}$; 
    (b) for $D \in \EC{P}$ with $D \cap X \not= \emptyset$ we have $(D \setminus X) \cup \{v\} \in \EC{P'}$; 
	and (c) for $D \in \EC{P'}$ with $v \in D$ we have $D \cup X \in \EC{P}$.
\end{lemma}


\SetKwFunction{SymMec}{SymMEC}
\SetKwFunction{SymbolicMec}{SymbolicMEC}
\begin{algorithm}
	\Procedure{\SymbolicMec{$(P' = (V',E',\langle V'_1, V'_R \rangle, \delta'),\gamma)$}}{
	$M \gets \emptyset, \textsc{MECs} \gets \emptyset, P \gets P'$\;
	
	\While(\tcp*[h]{compute vertices in non-trivial MECs of $C$}){$C \gets \SCCFind{$V, \emptyset,P'$}$}
          {\label{lics:alg:metasymmec:while1}
			  $M \gets M \cup \SymMec{$C, \gamma,P$}$\tcp{uses P instead of P'}\label{lics:alg:metasymmec:rec1}
	}
	
	\While(\tcp*[h]{compute non-trivial MECs}){$C \gets \SCCFind{$M, \emptyset,P'$}$}{\label{lics:alg:metasymmec:while2}
		$\textsc{MECs} \gets \textsc{MECs} \cup \{C\}$\label{lics:alg:metasymmec:addSCC}
	}
	\Return{$\textsc{MECs}$}\label{lics:alg:metasymmec:returnMECS}
}
	\caption{Computing the MEC decomposition of an MDP}\label{lics:alg:metasymmec}
\end{algorithm}

\begin{remark}
    Note that we modify the set of edges $E$ of $P$ using basic set operations only and that such operations 
    are supported by all standard symbolic tools like BDDs.
    After the update, we then use the new set $E$ for all subsequent $\Pre{}{.}$ and $\Post{}{.}$ operations.
\end{remark}

\subsection{Algorithm Description}
The input to our algorithm is an MDP $P' = (V',E', \langle V_1', V_R' \rangle, \delta')$.
In the first stage, we compute the set $M\subseteq V$ of all vertices that are in a non-trivial MEC,
and then, in the second stage, we use this set $M$
to compute the MECs with an SCC algorithm (see Algorithm~\ref{lics:alg:metasymmec}).

In the first stage, we iteratively compute the SCCs 
$C_1, \dots, C_\ell$ of $P'$ and immediately apply $\SymMec{$\cdot$}$ (cf. Algorithm~\ref{lics:alg:symmec}) to $C_i$ 
to compute the vertices $M_i$ of $C_i$ that are in a non-trivial MEC\@. 
The set $M$ is the union over these sets, i.e., $M = \bigcup_{1 \leq i \leq \ell} M_i$.
$\SymMec{$\cdot$}$ applies collapsing operations and thus modifies the edge set of the MDP\@.
We thus hand a copy $P$ of the original MDP $P'$ to $\SymMec{$\cdot$}$.
Finally, to obtain all non-trivial MECs in $P'$ we restrict the graph to the vertices in $M$ and compute the SCCs which correspond to the MECs. Note that trivial MECs are player-1 vertices that are not contained in any non-trivial MEC and thus can be simply computed
by iterating over the vertices of $V_1 \setminus M$.

In the following we focus on the function $\SymMec{$\cdot$}$ which is the core of our algorithm.
To this end, we introduce the operation $\ROut{S} = \Pre{E}{V \setminus S} \cap (S \cap V_{R})$, 
that computes the set of random vertices in $S$ with edges to $V \setminus S$. 

\para{The function $\SymMec{$\cdot$}$.}
$\SymMec{$\cdot$}$ works recursively: The input is a set $S$ of strongly connected
vertices and a parameter $\gamma$ for separator computations that is fixed over all recursive calls.
The main idea is to compute a separator to divide the original graph into smaller SCCs, recursively compute the vertices
which are in MECs of the smaller SCCs, and 
then compute the MECs of the original graph by incrementally adding the vertices of the separator
back into the recursively computed MEC decomposition.
In each recursive call, we first check whether the given set $S$ is larger than one 
(if not it cannot be a non-trivial MEC) and then 
if $S$ is a non-trivial EC by checking if $\ROut{S} = \emptyset$.
If $S$ is a non-trivial EC we collapse $S$ by calling Algorithm~\ref{lics:alg:CollapseEC}
and the algorithm then returns $M=S$. 
If $\ROut{S} \neq \emptyset$, $S$ is not a non-trivial EC\@ but it may contain nontrivial ECs of $C$\@. 
We then try to compute a balanced separator $T$ of $S$ which is nonempty
only if the diameter of $S$ is large enough ($\geq 2\gamma$). 
We further distinguish between the two cases: 
In the first case, we succeed to compute the balanced separator and
we recurse on the strongly connected components in $\SCCs{S\setminus\Attr{R}{P}{T}}$.\vspace{0.05cm}
After computing and collapsing the ECs in $\SCCs{S\setminus\Attr{R}{P}{T}}$
we incrementally add vertices of $T$ to $S$  and compute the ECs of $S \setminus T$ until $T$ is empty.
In each incremental step, we add one vertex $v$ to $S\setminus T$ and find the SCC of $v$ in $S\setminus T$. 
If the random attractor of $\ROut{S'}$ (note that this computation now considers all vertices in $P$) 
does not contain the whole SCC $S'$ we are able to prove that we can identify a new EC of $C$. 

Otherwise, the vertex $v$ does not create a new EC in $S \setminus T$.
In the second case where we fail to compute the balanced separator we know
that the diameter of $S$ is small ($<2\gamma$). We remove the
random attractor $X$ of $\ROut{S}$ (this set cannot contain ECs due to Lemma~\ref{lics:lem:attr_remove}), 
recompute the strongly connected components in the set $S \setminus X$ and recurse on 
each of them one after the other.

\begin{figure*}
	\begin{center}
		\includegraphics[width=0.8\linewidth]{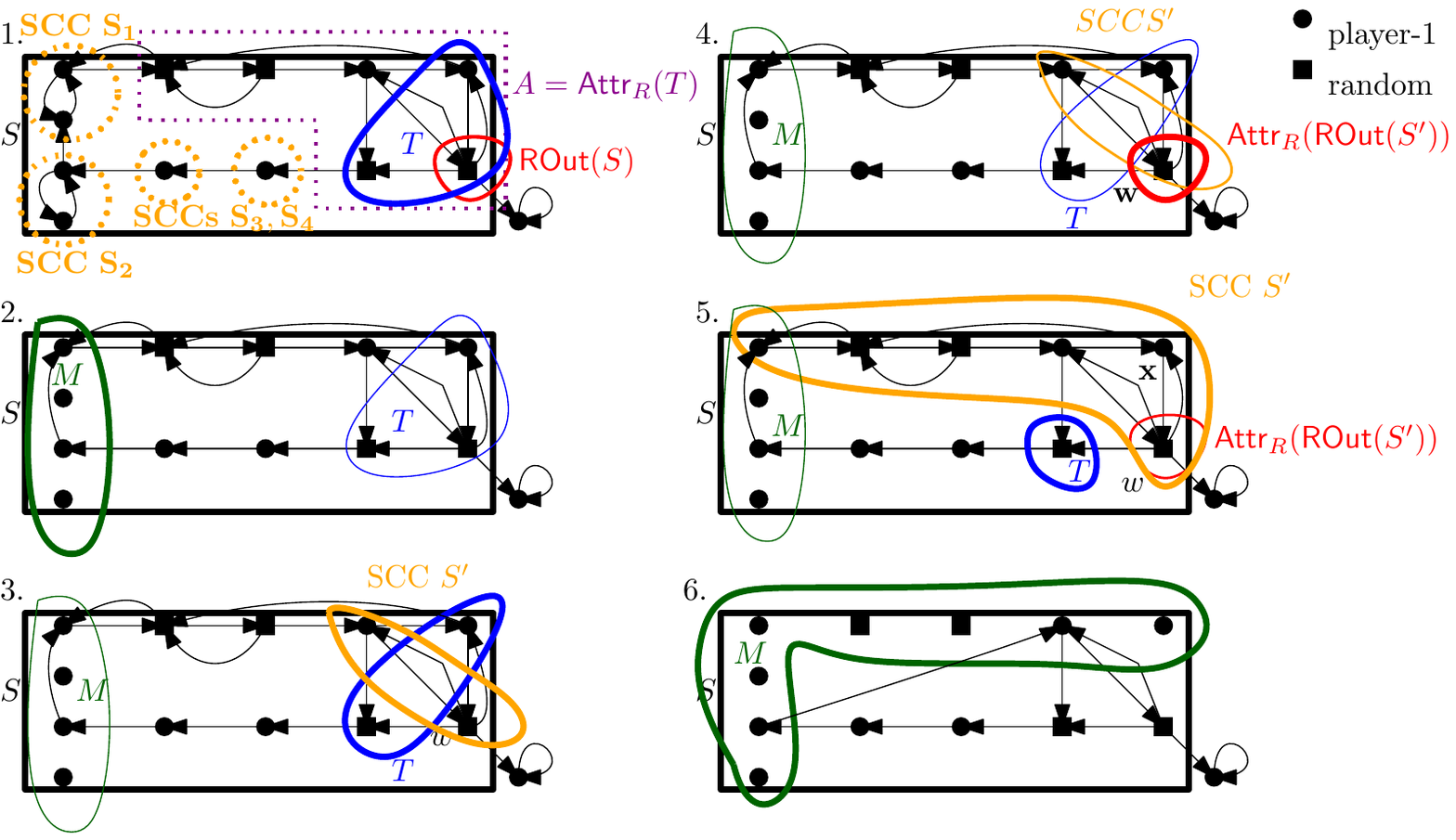}\label{lics:fig:symmec}
	\end{center}
\end{figure*}

Figure~\ref{lics:fig:symmec} illustrates one recursive call of the $\SymMec{$S,\gamma,P$}$ function:
In the first step, we compute that $\ROut{S}$ is nonempty, i.e., $S$ cannot be a MEC\@.
We then successfully compute a separator $T$ with parameter $\gamma$ of the strongly
connected graph $S$ and recurse on the SCCs of $S \setminus \Attr{R}{P}{T}$ (orange dottet circles).
In the second step, we recursively identify two nontrivial MECs and add them to $M$. Next, 
we perform the incremental iterations for all vertices in $T$
to identify the MECs containing the vertices in $\Attr{R}{P}{T}$, (i.e., also in $T$).
In step three, we pick $w \in T$, remove it
from $T$ and compute its SCC $S' = S\setminus T$. In step four, we compute the random attractor of
$\ROut{S'}$, i.e.,
$\Attr{R}{P}{\ROut{S'}} = \{w\}$. Because it contains $w$, there is no new MEC containing $w$. 
In step five, we remove $x$ from $T$ and compute its SCC $S'$. This time, the random attractor
of $\ROut{S'}$, i.e., $\Attr{R}{P}{\ROut{S'}}$, does not contain $x$ and the SCC $S'$ of
$x$ without $\Attr{R}{P}{\ROut{S'}}$ has size greater one. We add the vertices in $S' \setminus
\Attr{R}{P}{\ROut{S'}}$ to $M$. In step six, we do the last incremental
iteration, identify no new nontrivial MECs and return $M$.

Algorithm~\ref{lics:alg:symmec} illustrates the pseudocode of $\SymMec{$S,\gamma,P$}$.
\begin{remark}
Note that our symbolic algorithm does not require randomization which is in contrast to the best
known MEC decomposition algorithms for the standard RAM model~\cite{BPW19}.
The latter algorithms rely on decremental SCCs algorithms that are randomized as they maintain
ES-trees~\cite{EvenS81} from randomly chosen centers. Instead, our symbolic algorithm relies on a
deterministic incremental approach. 
\end{remark}


\begin{algorithm}[ht]
	\Procedure{\SymMec{$S,\gamma,P$}}{
	$M \gets \emptyset$\;\label{lics:alg:symmec:initM}
	\lIf{$|S| \leq 1$}{\Return{$M$}}\label{lics:alg:symmec:ret0}
	\If(\tcp*[h]{check if $S$ is an EC}){$\ROut{S} = \emptyset$}{\label{lics:alg:symmec:if1}
		$\CollapseEC{$S,P$}$\;\label{lics:alg:symmec:collapse1}
		\Return{$M \cup S$}\;\label{lics:alg:symmec:ret1}
	}
	$(T = \Set{v_1, \dots, v_t}) \gets \Separator{$S,\gamma$}$\;\label{lics:alg:symmec:sep}
	\If(\tcp*[h]{$T \neq \emptyset$ if $\Diam{S} \ge 2\gamma$}) {$T \neq \emptyset$}{
		$A \gets \Attr{R}{P}{T} \cap S$\;\label{lics:alg:symmec:attrSep}
		\While{$S_j \gets \SCCFind{$S \setminus A, \emptyset, P$}$}{\label{lics:alg:symmec:while1}
			$M \gets M \cup \SymMec{$S_j, \gamma,P$}$\label{lics:alg:symmec:rec1}
		}
		\While(\tcp*[h]{incremental EC detection for $v \in T$}){$T \neq \emptyset$}{\label{lics:alg:symmec:while2}
			$v \gets \Pick{T}$\;\label{lics:alg:symmec:pick}
			$T \gets T \setminus \{v\}$\;\label{lics:alg:symmec:removevfromT}
			$S' \gets \SCCFind{$S \setminus T,\Set{v},P$}$\;\label{lics:alg:symmec:computescc2}
			\lIf(\tcp*[h]{EC is trivial}){$|S'| = 1$}{\Continue}\label{lics:alg:symmec:vtrivial}
			$Z \gets \Attr{R}{P}{\ROut{S'}}$; \tcp{consider vertices in $P$}\label{lics:alg:symmec:attrIncr}
			$S' \gets S' \setminus Z$\;\label{lics:alg:symmec:snew}
			\If(\tcp*[h]{$S'$ contains a non-trivial EC }){$S' \neq \emptyset$}{
				$U \gets \SCCFind{$S',\Set{v},P$}$\label{lics:alg:symmec:bottomSCC}\;
				$\CollapseEC{$U,P$}$\;\label{lics:alg:symmec:collapse2}
				$M \gets M \cup U$\;\label{lics:alg:symmec:addMEC}
			}
		}
		\Return{$M$}\;\label{lics:alg:symmec:ret2}
	}
	\Else(\tcp*[h]{$T = \emptyset$ and thus  $\Diam{S} < 2\gamma$}){
		$X \gets \Attr{R}{P}{\ROut{S}}$\;\label{lics:alg:symmec:attrB}
		\While{$S_j \gets \SCCFind{$S\setminus X, \emptyset,P$}$}{\label{lics:alg:symmec:while3}
			$M \gets M \cup \SymMec{$S_j, \gamma,P$}$\label{lics:alg:symmec:rec2}
		}
		\Return{$M$}\;\label{lics:alg:symmec:ret3}
	}
}
	\caption{Compute the ECs of an SCS recursively}\label{lics:alg:symmec}
\end{algorithm}

\subsection{Correctness}

We first consider the correctness of Algorithm~\ref{lics:alg:symmec} which will then imply the correctness
of Algorithm~\ref{lics:alg:metasymmec}. 
To this end, consider an MDP $P = (V,E,\ls V_1, V_R \rs, \delta)$ and an SCC $C$ of $P$.
The non-trivial end-components of a subset $S$ of $V$ are given by $\EC{S}$ and 
the set of vertices that appear in non-trivial (maximal) ECs of $C$ are denoted by 
$M_C=\bigcup_{Q \in \EC{C}} Q$.

We first observe that in all calls to $\SymMec{$S,\gamma,P$}$ (see Algorithm~\ref{lics:alg:symmec}.) the set 
$S$ is strongly connected.
\begin{lemma}\label{lics:lem:corr:0}
	For a strongly connected set $S \subseteq C$ we have 
	that in the computation of $\SymMec{$S,\gamma,P$}$ 
	for all the calls $\SymMec{$S_j,\gamma,P$}$ the set 
	$S_j$ is strongly connected and $S_j \subseteq S$.
\end{lemma}
 \begin{proof}
     If $|S| =1$ or $\ROut{S} = \emptyset$ there are no recursive calls the statement is true.
     Now consider the case where $T \not= \emptyset$.
     At Line~\ref{lics:alg:symmec:rec1}, we consider $S_j$ which is an SCC in 
 	the graph $S \setminus (T \cup A)$, obviously a subset of $S$ and strongly connected.
 	Now consider the case where $T = \emptyset$.
 	At Line~\ref{lics:alg:symmec:rec2} we consider $S_j$ which is an SCC of 	$S \setminus X$, obviously a subset of $S$ and strongly connected. 
 \end{proof}

Let $M$ be the set returned by $\SymMec{$S,\gamma,P$}$.
Note that the ultimate goal of Algorithm~\ref{lics:alg:symmec} is to compute $M_C$ for a given SCC $C$ of $P'$.
The next lemma shows that every call of Algorithm~\ref{lics:alg:symmec} with a strongly connected set $S$ returns the 
set $M_S=\bigcup_{Q \in \EC{S}} Q$ and collapses all ECs $Q \in \EC{S}$ in $P$.

\begin{lemma}\label{lics:lem:allECInM}  
    Let  $S$ be a strongly connected set of vertices and $M$ be the set returned by $\SymMec{$S,\gamma,P$}$, then for all non-trivial ECs $Q \in \EC{S}$ in $P'$ 
	we have $Q \subseteq M$ and $Q$ is collapsed in $P$.
\end{lemma}
\begin{proof}
	Let  $Q \in \EC{C}$ with $Q \subseteq S$.
	We prove the statement by induction on the size of $S$.
%
	For the base case, i.e., if $S$ is empty or contains only one vertex, 
	Algorithm~\ref{lics:alg:symmec} returns the empty-set at Line~\ref{lics:alg:symmec:ret0} and thus the condition holds.

	For the inductive step, let $|S| > 1$.
	If $Q = S$, we detect it at Line~\ref{lics:alg:symmec:if1} 
	and return $S$ at Line~\ref{lics:alg:symmec:ret1}. 
	Then we collapse the set of vertices at Line~\ref{lics:alg:symmec:collapse1}.
	Thus, $M$ contains $Q$ and the claim holds.
%
	We distinguish two cases concerning $T$ at Line~\ref{lics:alg:symmec:sep}:
	First if we have non-empty $T$ also the random attractor $A$ of $T$ at Line~\ref{lics:alg:symmec:attrSep}
	is non-empty.
	Thus the SCCs of $S \setminus A$ are strictly smaller than $S$. 
	Thus, we can use the induction hypothesis to stipulate that if $Q$ is contained 
	in one of the SCCs computed at Line~\ref{lics:alg:symmec:while1} 
	we have $Q \subseteq M$ after the while-loop at Lines~\ref{lics:alg:symmec:while1}-\ref{lics:alg:symmec:rec1}. 
	Additionally, by the induction hypothesis, they are collapsed into a player-1 vertex.
	Also, if $Q$ contains a subset $Q' \subset Q$ that is a non-trivial EC
	and contained in one of the SCCs of $S \setminus A$,
	we have that $Q' \subseteq M$ after the while-loop at Lines~\ref{lics:alg:symmec:while1}-\ref{lics:alg:symmec:rec1}.
	Again, by induction hypothesis, $Q'$ is collapsed into a player-1 vertex.
	We proceed by proving that the separator $T$ must contain one of the vertices of $Q$.
	\begin{claim}\label{lics:claim:TcontainsQ}
		For any non-trivial EC $Q \in \EC{S}$ which is not fully contained
		in one of the SCCs computed at Line~\ref{lics:alg:symmec:while1} we have $T \cap Q \not= \emptyset$.
	\end{claim}
	\begin{proof}
		Let $Q$ be an arbitrary EC in $\EC{S}$.
		Observe that if $Q$ is not fully contained in an SCC of $S \setminus A$ 
		(computed at Line~\ref{lics:alg:symmec:while1}), we have $Q \cap A \neq \emptyset$.
		We now show that $Q \cap T \neq \emptyset$:
		Assume the contrary, i.e., $Q \cap T = \emptyset$. 
		Consequently, $Q$ is a strongly connected set in $S \setminus T$ and not 
		fully contained in an SCC of $S \setminus A$. 
		Also, $Q$ contains a vertex of $A \setminus T$. 
		Let $v \in Q \cap (A \setminus T)$ such that there is no $v' \in Q \cap (A \setminus T)$ which 
		is on a lower level of the attractor $A$.
		If $v \in V_1$, by the definition of the attractor and the above assumption, all its outing edges leave the set $Q$, which is in contradiction to $Q$ being strongly connected.
		Thus we have $v \in V_R$ and, by the definition of the attractor,
		it must have an edge to a smaller level and by the above assumption the target vertex is not in $Q$.
		That is $Q$ has an outgoing random edge which contradicts the assumption that $Q$ is an EC\@.		
	\end{proof}

	The following claim shows an invariant for the while-loop at Line~\ref{lics:alg:symmec:while2}.
	\begin{claim}\label{lics:claim:non-trivialEC}
		For iteration $j \geq 0$ of the while-loop at Line~\ref{lics:alg:symmec:while2} holds: There are no non-trivial ECs $E \subseteq S \setminus T$ in $P$.
	\end{claim}
	\begin{proof}
		The induction base $j=0$ holds 
		for the following reasons: $A \setminus T$ cannot include $Q$ due to Claim~\ref{lics:claim:TcontainsQ}.
		$S \setminus A$ cannot contain $Q$ 
		due to the fact that $Q$ must contain a vertex 
		in $T$ and that we collapsed the ECs contained in SCCs computed at 
		Line~\ref{lics:alg:symmec:while1}.
		For the induction step, assume that there are no non-trivial ECs in $S \setminus T$
		before iteration $j = \ell$.
		When we include one vertex $v$ of $T$ into $S$ 
		at Line~\ref{lics:alg:symmec:pick} there can only be one new non-trivial EC $Q_v$ in $S$, 
		i.e., the one including $v$, as the remaining graph is unchanged. 
		We argue that if such an EC $Q_v$ exists it is removed before iteration $\ell+1$ of the while-loop 
		at Line~\ref{lics:alg:symmec:while2}:
		The EC $Q_v$ must be strongly connected, and thus it is contained in the SCC $S'$ of $v$ 
		at Line~\ref{lics:alg:symmec:computescc2}. Note, that if $|S'| = 1$, the new EC is trivial, which 
		concludes the proof. Otherwise, we compute $Z$ at Line~\ref{lics:alg:symmec:attrIncr}, which does not
		contain a vertex of $Q_v$ by Lemma~\ref{lics:lem:attr_remove} and remove it from $S'$. 
		If $S'$ is non-empty, it contains exactly one non-trivial SCC as we argued above. 
		That is the SCC $Q_v$ and we thus add $Q_v$ to $M$ and collapse it. 
		As a consequence there is no non-trivial EC in $S \setminus T$ left.
		Thus the claim holds. 
	\end{proof}	
	
	It remains to show that Algorithm~\ref{lics:alg:symmec} finds the EC $Q$ 
	which contains vertices in $S \setminus M$.
	Consider the first iteration where $Q \subseteq S \setminus T$.
	We show that $Q \setminus M = \emptyset$ after this iteration:
	Because $Q$ might contain non-trivial ECs $Q'_i$ that were already collapsed due to 
	the recursive call, or due to prior iteration of the while-loop at Line~\ref{lics:alg:symmec:while2},
	it follows from Lemma~\ref{lics:lem:collapse} that there exists 	
	an EC $Q' \subseteq Q$  in $P$ such that $Q' \supset (Q\setminus M)$
	and $Q'$ contains one or more vertices corresponding to the collapsed sub-ECs of $Q$.
	Due to Claim~\ref{lics:claim:non-trivialEC} we have that $Q'$ is collapsed in the current iteration 
	and thus $Q'$ is added to $M$, i.e., we have that $Q \subseteq M$ and that $Q$ is collapsed.
%

	Now consider the case $T= \emptyset$ (i.e., $\Diam{S} < 2 \gamma$).
	As $\ROut{S} \not= \emptyset$ we have that the set $X$ computed at Line~\ref{lics:alg:symmec:attrB} is non-empty. 
	Note that $Q$ cannot contain a vertex in $X$ by Lemma~\ref{lics:lem:attr_remove} 
	and the fact that $S$ is strongly connected by Lemma~\ref{lics:lem:corr:0}. 
	Consequently, $Q$ is contained in one of the SCCs of $S \setminus X$. 
	Note that each such SCC has size strictly smaller than $S$ due to the fact that $X$ is non-empty.
	The claim then holds by the induction hypothesis.
	This completes the proof of Lemma~\ref{lics:lem:allECInM}.
\end{proof}

By the above lemma, we have that the algorithm finds all ECs\@. 
We next show that all vertices added to $M$ are contained
in some EC\@.

\begin{lemma}\label{lics:lem:MValid}
	The set $M$ is a subset of $M_S$, i.e., $M \subseteq \bigcup_{Q \in \EC{S}} Q$.
\end{lemma}
\begin{proof}
	We show the claim by induction over the size of $S$.
	For the base case, i.e., $|S| \leq 1$ the claim holds trivially because we return
	the empty set at Line~\ref{lics:alg:symmec:ret0} and any set of size less than two only contains
	trivial ECs.
	For the inductive step, let $|S| > 1$. 
%
	For the return statement at Line~\ref{lics:alg:symmec:ret1} we argue as follows:
	Because $S$ has no random vertices with edges out of $S$ 
	(considering $P'$, Line~\ref{lics:alg:symmec:if1}), $S$ is a non-trivial EC\@. 
	Thus, in Line~\ref{lics:alg:symmec:ret1} we correctly return $S$.

	For the case $T \not=\emptyset$ we argue as follows:
	Let $T$, $A$ be the separator and its attractor as computed in Line~\ref{lics:alg:symmec:sep} and
	Line~\ref{lics:alg:symmec:attrSep}. Note that for all SCCs $S_j$ in the graph $S \setminus A$
	we have $|S_j| < |S|$ because $|A| > 0$. Thus, by induction hypothesis, all vertices added 
	in the recursive call at Line~\ref{lics:alg:symmec:rec1} are in $M_S$.
	We claim that for each iteration of the while loop in Line~\ref{lics:alg:symmec:while2} 
	we add only non-trivial ECs to $M$:
	Note that by Claim~\ref{lics:claim:non-trivialEC} there is no non-trivial EC at the start of the
	while loop in Line~\ref{lics:alg:symmec:while2}.
	Let $v$ be the vertex we choose to remove from $T$ at Line~\ref{lics:alg:symmec:pick}.
	Let $S_v$ be the SCC of $v$ computed at Line~\ref{lics:alg:symmec:computescc2}.
	If $|S_v| = 1$ we continue to the next iteration without adding anything to $M$ and 
	the claim holds.
	Otherwise, there are two cases based on the 
	computation of the attractor $Z$ of random vertices 
	with edges out of $S_v$ at Line~\ref{lics:alg:symmec:attrIncr}:
	If $Z$ contains $S_v$, we do not add vertices to $M$ and the claim holds.
	If $S_v \setminus Z \neq \emptyset$ we add the SCC $S'_v$ of $v$ in $S_v \setminus Z$ to $M$.
	It remains to show that the non-trivial SCC $S'_v$ as computed at Line~\ref{lics:alg:symmec:snew} is a non-trivial EC\@:
	Note that for all $v \in S'_v \cap V_1$ we have $\Out{v} \cap S'_v \neq \emptyset$, otherwise
	$v \in Z$. Also, for all $u \in S'_v \cap V_R$ we must have $\Out{u} \subseteq S'_v$, otherwise
	$u \in Z$. 	Thus, $S_v$ is an EC\@.

    For the case $T =\emptyset$ note that each SCC found at Line~\ref{lics:alg:symmec:while3} must be of size strictly
	less than $|S|$ because $|X| \geq 1$. 
	Thus, by induction hypothesis, the vertices added at Line~\ref{lics:alg:symmec:rec2}
	are in $M_S$.
\end{proof}

Lemma~\ref{lics:lem:allECInM} and Lemma~\ref{lics:lem:MValid} imply the correctness of 
Algorithm~\ref{lics:alg:symmec} and Algorithm~\ref{lics:alg:metasymmec}.

\begin{proposition}[Correctness]\label{lics:prop:symmec:corr}
  Given an SCC $S$ of an MDP, Algorithm~\ref{lics:alg:symmec} returns the set $M_C=\bigcup_{Q \in \EC{C}} Q$, i.e., the set of vertices that are contained in a non-trivial MEC of $S$.
\end{proposition}
\begin{proposition}[Correctness]
  Given an MDP $P'$, Algorithm~\ref{lics:alg:metasymmec} returns the set of non-trivial MECs of $P'$.
\end{proposition}
\begin{proof}
	Due to Proposition~\ref{lics:prop:symmec:corr} $M$ contains all vertices in nontrivial ECs.
	It remains to show that the nontrivial MECs are the SCCs of $V \cap M$. 
	By definition, each nontrivial EC is strongly connected. Towards a contradiction assume
	that a nontrivial MEC $Q$ is not an SCC but part of some larger SCC $C$ of $V \cap M$ (there is
	an SCC which is not a MEC). 
	But then, by the definition of $M$, we have that $C$ is the union of several MECs, i.e., it is strongly connected and has no random outgoing edges and thus $C$ is an $EC$. This is in contradiction to $Q$ being a MEC\@.
	Thus, each nontrivial MEC is an SCC and, by the definition of $M$, the union of the MECs covers
	$M$. Thus there are no further SCCs.
%
\end{proof}

\subsection{Symbolic Operations Analysis}
We first bound the total number of symbolic operations for computing the separator $T$ at 
Line~\ref{lics:alg:symmec:sep}, recursing upon the SCCs $S \setminus T$ at Line~\ref{lics:alg:symmec:rec1} and 
adding the vertices in $T$ back to compute the rest of the ECs of $S$ at 
Lines~\ref{lics:alg:symmec:while2}--\ref{lics:alg:symmec:ret2} during all calls to Algorithm~\ref{lics:alg:symmec}.
\begin{lemma}\label{lics:lem:running_rec2}
	The total number of symbolic operations of Lines~\ref{lics:alg:symmec:sep}--\ref{lics:alg:symmec:ret2}
	in all calls to Algorithm~\ref{lics:alg:symmec} is in $O\left(\frac{n^2}{\lfloor \gamma /(2\log n)\rfloor}\right)$. 
\end{lemma}
\begin{proof}
	In Lemma~\ref{lics:lem:runspacesep} we proved that computing the separator at 
	Line~\ref{lics:alg:symmec:sep} takes $O(|S|)$ symbolic operations. 
	The same holds for computing the attractor at Line~\ref{lics:alg:symmec:attrSep} due to 
	Lemma~\ref{lics:lem:attrRunning} and computing the SCCs at Line~\ref{lics:alg:symmec:while1}. 
	Each iteration of the while-loop at Line~\ref{lics:alg:symmec:while2}
	takes $O(|S|)$ symbolic operations: Computing the SCC of $v$ twice can be done in $O(|S|)$ symbolic operations 
	due to Theorem~\ref{lics:thm:SCCs}. Similarly, computing the random attractor at 
	Line~\ref{lics:alg:symmec:attrIncr} is in $O(|S|)$ symbolic operations due to Lemma~\ref{lics:lem:attrRunning}.
	The remaining lines can be done in a constant amount of symbolic operations.
	It remains to bound the symbolic operations of the while-loop at Line~\ref{lics:alg:symmec:while1} where
	we call Algorithm~\ref{lics:alg:symmec} recursively at Line~\ref{lics:alg:symmec:rec1} for each 
	SCC in $S \setminus A$. Note that the size of $T$ determines the number of iterations the while-loop 
	at Line~\ref{lics:alg:symmec:while2} has, and how big the SCCs in $S \setminus A$ are. 
	We obtain that $|T|$ is of size at most
	$\frac{|S|}{\lfloor \gamma /(2 \log |S|)\rfloor}$ combining 
	Observation~\ref{lics:obs:SepSize} and Lemma~\ref{lics:lem:sepqual}.
	Due to the argument above the following equation bounds the 
	running time of Lines~\ref{lics:alg:symmec:sep}--\ref{lics:alg:symmec:ret2} for some constant $c$ which is greater than the number of constant symbolic operations in $\SymMec{$C,\gamma,P$}$ if $|S| \geq \gamma$.
	\begin{equation*}
		F(S)  \leq |T| \cdot |S| \cdot c + 
		\max_{|T| = 1 \dots \frac{|S|}{\lfloor(\gamma/2 \log{n}) \rfloor}} 
			\sum_{S_i \in \SCCs{S \setminus T}} F(S_i)
\end{equation*}
If $|S| < \gamma$ we only have the costs for computing the separator and thus $F(S) \leq c \cdot |S|$. 
 Next, we prove the bound in Claim~\ref{lics:claim:fs}.
 	\begin{claim}\label{lics:claim:fs}
 		$F(S) \in O\left(\frac{|S|^2}{\lfloor\gamma/(2 \log{|S|}) \rfloor} + |S|\right)$.
 	\end{claim}
 	\begin{proof}	
 		We prove the inequality by induction on the size of $S$.
 		That is we show $F(S) \leq \frac{|S|^2}{Z} c'$ where $Z = \lfloor\gamma/(2 \log |S|)\rfloor$ for some $c' > c$.
 		Obviously the inequality is true for $|S| < \gamma$, i.e., the base case is true.
 		For the inductive step, consider 
 		$F(S) \leq \frac{|S|^2}{Z}c'$ for $|S| \geq \gamma$.
  		Note that
 		\begin{align*}
 			\sum_{S_i \in \SCCs{S \setminus T}} F(S_i) 
 			\leq \sum_{S_i \in \SCCs{S \setminus T}} c' (|S_i|^2/(\lfloor\gamma/(2 \log |S_i|)\rfloor) + |S_i|)\leq&\\
 			c'|S| + \hspace{-5pt}\sum_{S_i \in \SCCs{S \setminus T}}\hspace{-5pt} c' |S_i|^2/Z \leq
 			c'|S| + \hspace{-5pt}\sum_{S_i \in \SCCs{S \setminus T}}\hspace{-5pt} c' |S_i| \cdot (|S|-|T|Z) /Z \leq&\\
 			c'|S| + c'(|S|-|T|)(|S|-|T|Z) /Z
 		\end{align*}
 		The first inequality is due to the induction hypothesis and the third inequality is 
 		due to the fact that we have a $Z$-separator and Lemma~\ref{lics:lem:sepqual}.
 		It remains to add $|S||T|c$:
 		\begin{align*}
 			F(S) \leq& c'|S| +  c'(|S|-|T|)(|S|-|T|Z) /Z + |S||T|c \leq\\
 				   &c'|S| + c'/Z (|S|^2 -|S||T|-|T||S|Z+|T||S|) + |S||T|c = \\
 				   &c'|S| + |S|^2/Zc' -|S||T|c' + |S||T|c \leq c'|S| + \frac{|S|^2}{Z}c'.
 		\end{align*}
 		The first inequality is due to the fact that $|T| \leq |S|/Z$ (Observation~\ref{lics:obs:SepSize}).
 		This concludes our proof by induction of $F(S) \leq c' \cdot (\frac{|S|^2}{Z} + |S|)$.
 	\end{proof}
	Now using that claim and $|S| \leq n$ we obtain a $O\left( n^2 / \lfloor\gamma/(2 \log{n}) \rfloor + n\right)$ bound which can further be simplified to $O\left(n^2/\lfloor\gamma/(2 \log{n}) \rfloor\right)$ by using $\gamma \leq n$.
\end{proof}

The second part of our analysis bounds the symbolic operations of the case when 
$\Diam{S} < 2\gamma$ at Lines~\ref{lics:alg:symmec:attrB}--\ref{lics:alg:symmec:rec2} and the work done from Line~\ref{lics:alg:symmec:initM} to Line~\ref{lics:alg:symmec:ret1}.

\begin{lemma}\label{lics:lem:running_rec3}
	The total number of symbolic operations of Lines~\ref{lics:alg:symmec:attrB}--\ref{lics:alg:symmec:rec2}
	in all calls to Algorithm~\ref{lics:alg:symmec} is in $O(n \cdot \gamma + \frac{n^2}{\lfloor \gamma /(2 \log n)\rfloor})$. 
\end{lemma}
\begin{proof}
	If a vertex is in the set $X$ at Line~\ref{lics:alg:symmec:attrB} it is not recursed upon or ever looked at again,
	thus we charge the symbolic operations of all attractor computations to the vertices in the attractor. 
	This adds up to a total of $O(n)$ symbolic operations by Lemma~\ref{lics:lem:attrRunning}.
	Additionally, note that $|X|$ is non-empty, as otherwise $|S|$ is declared as EC in the 
	if-condition at Line~\ref{lics:alg:symmec:if1}. Thus, Lines~\ref{lics:alg:symmec:attrB}--\ref{lics:alg:symmec:rec2}
	occur at most $n$ times.
	The number of symbolic operations for computing SCCs is in time $O(\sum_{C \in \SCCs{G}} (D_C +1))$ due to 
	Theorem~\ref{lics:thm:SCCs}. We can distribute the costs the SCCs according to the diameters of the SCCs.	
	We provide separate arguments for counting the costs for SCCs $S_j$ with  $\Diam{S_j} < 2 \gamma$ 
	and SCCs $S_j$ with  $\Diam{S_j} \geq 2 \gamma$.
    First, we consider the costs for SCCs $S_j$ with  $\Diam{S_j} < 2 \gamma$. 
    Computing the SCC $S_j$ costs $O(\gamma)$ operations and  
    the algorithm either terminates in the next step or at least one vertex is removed from the SCC via a separator or attractor.
    That is we have at most $n$ such SCCs computations and thus an overall cost of $O(n \gamma)$.
	Now we consider the costs for SCCs $S_j$ with  $\Diam{S_j} \geq 2 \gamma$. 
	Whenever computing such an SCC, we simply charge all its vertices for the costs of computing the SCC, 
	i.e., $O(1)$ for each vertex.
	For such an SCC the algorithm either terminates in the next step or a separator is computed and removed.
	We thus have that each vertex is charged again after at least $\lfloor \gamma /(2 \log |S|)\rfloor$
	many nodes are removed from its SCC\@. In total, each vertex is charged at most $\frac{|S|}{\lfloor \gamma /(2 \log |S|)\rfloor}$ many times.
	We get an $O(|S|\frac{|S|}{\lfloor \gamma /(2 \log |S|)\rfloor})$ upper bound for the SCCs with large diameter,
\end{proof}

Putting Lemma~\ref{lics:lem:running_rec2} and Lemma~\ref{lics:lem:running_rec3} together,
we obtain the $O(n \cdot \gamma + \frac{n^2}{\lfloor \gamma/(2\log n) \rfloor})$ bound for Algorithm~\ref{lics:alg:symmec}
which  also applies to Algorithm~\ref{lics:alg:metasymmec} as the SCC-computations only require $O(n)$ operations.
\begin{proposition}\label{lics:prop:symmec:time}\label{lics:prop:metasymmec:time}
	Algorithm~\ref{lics:alg:symmec} and Algorithm~\ref{lics:alg:metasymmec} both have $O(n \cdot \gamma + \frac{n^2}{\lfloor \gamma/(2\log n) \rfloor})$ symbolic operations.
\end{proposition}

\subsection{Symbolic Space} 
Symbolic space usage counted as the maximum number of sets (and not their size) at any point in time
is a crucial metric and limiting factor of symbolic computation in practice~\cite{ClarkeBook}.
In this section, we consider the symbolic space usage of Algorithm~\ref{lics:alg:symmec}, highlight a key issue, and present a solution to the issue. 

\para{Key Issue.}
Even though Algorithm~\ref{lics:alg:metasymmec} beats the current best \emph{symbolic} Algorithm for computing
the MEC decomposition in the number of symbolic operations (current best: $O(n \sqrt{m})$, space: $O(\sqrt{n})$~\cite{CHLOT18}),
without further improvements, Algorithm~\ref{lics:alg:symmec} requires $O(n)$ symbolic space as we discuss in the following.
First, note that each call of $\SymMec{$S,\gamma,P$}$, when excluding the sets stored by recursive calls, only stores a constant number of sets and requires a logarithmic number of sets to execute $\SCCFind{$\cdot$}$.
That is, the recursion depth is the crucial factor here. 
As we show below, by the $\lfloor \gamma/(2\log n)\rfloor $-separator property, the recursion depth due to the case $T\not=\emptyset$ is  $O\left( n / \lfloor \gamma /(2\log n)\rfloor\right)$.
However, the case $T=\emptyset$ might lead to a recursion depth of $O(n)$ when 
in each iteration only a constant number of vertices is removed
and the diameter of the resulting SCC is still smaller than $2\gamma$.
As Algorithm~\ref{lics:alg:symmec} uses a constant amount of sets for each recursive call, 
it uses $\O(n)$ space in total.

\para{Reducing the symbolic space.}
We resolve the above space issue by modifying Algorithm~\ref{lics:alg:symmec} for the case $T=\emptyset$ (see Algorithm~\ref{lics:alg:symmec_reduced_space}.):
At the while loop at Line~\ref{lics:alg:symmec:while3} we first consider the SCCs $S_j$ 
with less than $|S|/2$ vertices and recurse on them.
If there is an SCC $S'$ with more than $|S|/2$ vertices we process it at the end, i.e., we use one additional set to store that SCC until the SCC algorithm terminates.
As now all the computations of the current call to $\SymMec{$S,\gamma,P$}$ are done we can 
simply reuse the sets of the current calls to start the computation for $S'$. 
We do so by setting $S$ to $S'$ and continuing in the Line~1 of Algorithm~\ref{lics:alg:symmec}.
Using that we only recurse on sets which are of size $\leq |S|/2$
and thus get a recursion depth of $O(\log n)$ for this case.
Moreover, the modified algorithm has the same computation operations as the original one and thus 
the bounds for the number of symbolic operations apply as well. 

\begin{algorithm}[t]
\small
\footnotesize
	\setcounter{AlgoLine}{22}
	\Else(\tcp*[h]{$T = \emptyset$ and thus  $\Diam{S} < 2\gamma$}){
		$X \gets \Attr{R}{P}{\ROut{S}}$; $S_{tmp} = \emptyset$\;\label{lics:alg:symmec_reduced_space:attrB}
		\While{$S_j \gets \SCCFind{$S\setminus X, \emptyset,P$}$}{\label{lics:alg:symmec_reduced_space:while3}
			\lIf{$S_j \geq |S|/2$ and $|S_j|>1$}{$S_{tmp} \gets S_j$; \Continue} 
			$M \gets M \cup \SymMec{$S_j,M, \gamma,P$}$\label{lics:alg:symmec_reduced_space:rec2}
		}
		\lIf{$S_{tmp} \neq \emptyset$}{$S\gets S_{tmp}$; goto Line~3}
		\Return{$M$}\;\label{lics:alg:symmec_reduced_space:ret3}
	}
	\caption{Reduced Space Version of Algorithm~\ref{lics:alg:symmec}}\label{lics:alg:symmec_reduced_space}
\end{algorithm}

\begin{lemma}\label{lics:lem:recursiondepth}
	\begin{sloppypar}
		The maximum recursion depth of the modified algorithm is in  
		$O(\frac{n}{\lfloor \gamma /(2\log n)\rfloor} + \log n)$.
	\end{sloppypar}
\end{lemma}
 \begin{proof}
     Consider the recursion occurring due to Line~\ref{lics:alg:symmec:rec1}.
 	Because $T$ is a $\lfloor \gamma/(2\log n)\rfloor$-separator, an SCC in $S \setminus T$
 	contains at most $n-\lfloor \gamma/(2\log n)\rfloor \cdot |T| \leq n-\lfloor \gamma/(2\log n)\rfloor$
 	 vertices as $|T| \geq 1$. We determine how often we can remove
 	$\lfloor \gamma/(2\log n)\rfloor$ from $n$ until there are no vertices 
 	left which gives the recursion depth $k$.  It follows that $k=
 	\frac{n}{\lfloor \gamma /(2\log n)\rfloor}$.
 	Now consider the recursion occurring due to Line~\ref{lics:alg:symmec:rec2}. In the modified version we have that $|S_j| < |S|/2$ 
 	and thus this kind of recursion is bounded by $O(\log n$).
 \end{proof}

Algorithm~\ref{lics:alg:symmec} has $O(n\gamma + n^2/\lfloor \gamma/(2\log n) \rfloor)$ many symbolic
operations due to Lemma~\ref{lics:lem:running_rec2} and Lemma~\ref{lics:lem:running_rec3} and the number of sets is
in $O(n/\lfloor \gamma /(2\log n)\rfloor + \log n)$ due to Lemma~\ref{lics:lem:recursiondepth}.
In symbolic algorithms it is of particular interest to optimize symbolic space resources. Note that we obtain
a space-time trade-off when setting the parameter $\gamma$ such that $(2\sqrt{n}+ 2) \log n \leq \gamma \leq (2n +1) \log n$.

For the symbolic space of Algorithm~\ref{lics:alg:metasymmec} notice that the SCC algorithms are in logarithmic symbolic space and the algorithm itself only needs to stores the set $M$ and the current SCC\@. 
Thus, when we immediately output the computed MECs it only requires $O(\log n)$ additional space.

\begin{theorem}\label{lics:thm:symbolicmecTST}
	The MEC decomposition of an MDP can be computed in $O\left(n^{2-\epsilon} \log n\right)$
	symbolic operations and with symbolic space $O\left(n^{\epsilon} \log n \right)$ for $0 < \epsilon \leq 0.5$.
\end{theorem}

By setting $\epsilon = 0.5$ we obtain 
that the MEC decomposition of an MDP can be computed in $\O(n\sqrt{n})$ symbolic operations and with symbolic space $\O(\sqrt{n})$.


\section{Symbolic Qualitative Analysis of Parity Objectives}
In this section, we present symbolic algorithms for the qualitative analysis of parity objectives. 

\begin{theorem}[\cite{CY95,deAlfaroThesis}]
	\begin{sloppypar}
		For all MDPs $P$, and all parity objectives $\parity{p}$, there exists a memoryless strategy 
		$\sigma$ such that for all $v \in \ASW{\parity{p}}$ we have $\Pr_v^\sigma(\parity{p}) = 1$.
	\end{sloppypar}
\end{theorem}

\subsection{Almost-sure Reachability.}
In this section, we present a symbolic algorithm that computes reachability 
objectives $\reach{T}$ in an MDP\@. 
The algorithm is a symbolic version of~\cite[Theorem 4.1]{CDHL16}.

\smallskip\noindent\emph{Symbolic Graph Reachability.}
For a graph $G = (V,E)$ and a set of vertices $S$, the set 
$\GraphReach{S,G}$ is the set of vertices $V$ that can reach a vertex of $S$ within $G$. We compute it 
by repeatedly calling $S \gets \Pre{E}{S}$ until we reach a fixed point.
In the worst case, we add one vertex in each such call and need $|\GraphReach{S,G} \setminus S| +1
= O(n)$ many \Pre{E}{\cdot} operations to reach a fixed point.

\smallskip\noindent\emph{Algorithm Description.}
Given an MDP $P$ we compute the set $\ASW{\reach{T,P}}$ as follows:
First, if player~1 can reach one vertex of a MEC he can reach all the vertices of a MEC and thus 
we can collapse each MEC $M$ to a player-1 vertex.
If $M$ contains a vertex of $T$, we include the collapsed vertex into $T$. 
$P' = (V',E',\langle V'_1,V'_2 \rangle, \delta')$ is the MDP where the MECs of $P$ are collapsed as described above.
Next, we compute the set of vertices $S$ which can reach $T$ in the graph induced by $P'$. 
A vertex in $V' \setminus S$ cannot reach $T$ almost-surely because there is no path to $T$.
Note that a play starting from a vertex in the random attractor $A$ of $V' \setminus S$ might also end up in $V' \setminus S$. We
thus remove $A$ from $V'$ to obtain the set $R$, where vertices can almost-surely
reach $T$ in $P'$. Finally, to transfer the result back to $P$ we include all MECs with a vertex in $R$.

We implement Algorithm~\ref{lics:alg:metasymmec} to compute the MEC decomposition of $P$ but note that we could use any symbolic MEC algorithm. 
To minimize the extra space usage we also assume that the algorithm that computes the
MEC decomposition outputs one MECs after the
other instead of all MECs at once. 
Note that we can easily modify Algorithm~\ref{lics:alg:metasymmec} to output one MEC after the
other by iteratively returning each SCC found at Line~\ref{lics:alg:metasymmec:while2}. 
Moreover, we only require logarithmic space to maintain the state of the SCC
algorithm~\cite{ChatterjeeDHL18}.
\SetKwFunction{SymASReach}{SymASReach}
\begin{algorithm}
\small
\footnotesize
	\KwIn{An MDP $P = (V,E, \langle V_1, V_R \rangle, \delta)$ and a target set $T \subseteq V$}
	\KwOut{$\ASW{\reach{T,P}}$}
	\Procedure{\SymASReach{$T,\gamma,P$}}{
		$V' \gets V, E' \gets E, V'_1 \gets V_1, V'_R \gets V_R,  \delta' \gets \delta$\;
		$P'=(V',E', \langle V'_1, V'_R \rangle, \delta')$\;
		\For{$M \gets \mathsf{ComputeMECs}(P)$}{ 
			$C \gets \CollapseEC{$M,P'$}$\; 
			\lIf{$M \cap T \neq \emptyset$}{$T \gets T \cup C$}
		}
		$S \gets \GraphReach{T,P'}$; $A \gets \Attr{R}{P'}{V'\setminus S}$; $R \gets V' \setminus A$\;

		\For{$M \gets \mathsf{ComputeMECs}(P)$}{
			\lIf(\tcp*[h]{If $v \in M$ can reach $T$, $M$ can reach $T$.}){$M \cap R \neq \emptyset$}{
				$R \gets R \cup M$	
			}
		}

		\Return{$R$}\;
	}
	
	\caption{Symbolic Almost-Sure Reachability}\label{lics:alg:symasreach}
\end{algorithm}

We prove the following two propositions for Algorithm~\ref{lics:alg:symasreach}.
Let $P$ be an MDP, $T$ a set of vertices and $\mathbf{MEC}$ 
be the number of symbolic operations we need to compute the MEC decomposition.
Let $\aspace(\mathbf{MEC})$ denote the space of computing the MEC decomposition.
\begin{sloppypar}
\begin{proposition}[Correctness~\cite{ChatterjeeDHL18}]\label{lics:prop:asreachcorr}
	Algorithm~\ref{lics:alg:symasreach} computes the set $\ASW{\reach{T}}$.
\end{proposition}
\end{sloppypar}

\begin{proposition}[Running time and Space]\label{lics:prop:asreachrunning}
	The total number of symbolic operations of Algorithm~\ref{lics:alg:symasreach} is $O(\mathbf{MEC} +
	n)$. Algorithm~\ref{lics:alg:symasreach} uses $\O(\aspace(\mathbf{MEC}))$ symbolic space.
\end{proposition}

Proposition~\ref{lics:prop:asreachrunning} and Proposition~\ref{lics:prop:asreachcorr} together with Theorem~\ref{lics:thm:symbolicmecTST} yield the following theorem.

\begin{theorem}\label{lics:thm:asreach}
	The set $\ASW{\reach{T}}$ of an MDP can be computed with $\O(n^{2-\epsilon})$ 
	many symbolic operations and $\O(n^\epsilon)$ 
	symbolic space for $0 < \epsilon \leq 0.5$.
\end{theorem}

\subsection{Parity Objectives.}
In this section, we consider the qualitative analysis of MDPs with parity objectives.
We present an algorithm for computing the winning region which is based on the algorithms we present
in the previous sections and the
algorithm presented in~\cite[Section 5]{CH11}.
The algorithm presented in~\cite[Section 5]{CH11} draws ideas from a hierarchical clustering technique~\cite{Tarjan82, KingKV01}.
Without loss of generality, we consider the parity objectives $\parity{p}$ where 
$p: V \rightarrow \{0,1,\dots,2d\}$.
In the symbolic setting, instead of $p$, we get the sets $\P_{\geq i} = \{ v \in V \mid p(v) \geq i \}$ where $(1 \leq i \leq 2d)$ as part of the input.  
We abbreviate the family $\{ \P_{\geq i} \mid 1 \leq i \leq 2d\}$ as $(\P_{\geq k})_{1\leq k\leq 2d}$. 
Let $\P_{\leq m} = V \setminus \P_{\geq m+1}$ and  $\P_{m} = \P_{\geq m} \setminus \P_{\geq m+1}$. 
Given an MDP $P$, let $P_i$ denote the MDP obtained by removing $\Attr{R}{P}{\P_{\leq i-1}}$, the set
of vertices with priority less than $i$ and its random attractor.
A MEC $M$ is a \emph{winning MEC} in $P_i$ if there exists a vertex $u \in M$ such that $p(u) = i$ and $i$ is even, i.e., the smallest priority in the MEC is even. Let $\WE_i$ be the union of
vertices of winning maximal end-components in $P_i$, and let $\WE = \bigcup_{0\leq i \leq 2d} \WE_i$.
Lemma~\ref{lics:lem:paritywinningmec} says that computing $\ASW{\parity{p}}$ is equivalent to computing almost-sure
reachability of $\WE$. Intuitively, player 1 can infinitely often satisfy the parity condition after
reaching an end-component which satisfies the parity condition. 

\begin{lemma}[\cite{CH11}]\label{lics:lem:paritywinningmec}
	We have $\ASW{\parity{p}} = \ASW{\reach{\WE}}$.
\end{lemma}

Thus, we describe in Algorithm~\ref{lics:alg:winpec} how to compute $\WE$ symbolically.

\subsubsection{Algorithm Description.}
The algorithm uses a key idea which we describe first. 
Recall that $P_i$ denotes the MDP obtained by removing $\Attr{R}{P}{\P_{\leq i-1}}$, i.e., the set
of vertices with priority less than $i$ and its random attractor.

\smallskip\noindent\emph{Key Idea.}
If $u,v$ are in a MEC in $P_i$, then they are in the same MEC in $P_{i-1}$. 
The key idea implies that if a vertex is in a winning MEC of $P_i$, it is also in a winning MEC of
$P_{i-1}$.
Intuitively, this holds due to the following two facts: 
(1)~Because $P_{i-1}$ contains all edges and
vertices of $P_i$ the MECs $P_i$ are still strongly connected in $P_{i-1}$. 
(2)~Because $\Attr{R}{P}{\P_{\leq i-1}}$ makes sure that no MEC $M$ in $P_i$ 
has a random vertex with an edge leaving $P_i$ in $M$, the same is true for the set $M$ in $P_{i-1}$. 

\SetKwFunction{WinPEC}{WinParityEC}
\begin{sloppypar}
We next present the recursive algorithm $\WinPEC{$P,(\P_{\geq k})_{1\leq k \leq 2d},i,j$}$ which,
for a MDP $P$, computes the set $\bigcup_{i \leq \ell \leq j} \WE_i$ of winning MECs for priorities between $i$ and $j$.
\end{sloppypar}
\begin{enumerate}
	\item \emph{Base Case:} If $j < i$, return $\emptyset$.
	\item Compute $m \gets \lceil (i+j)/2 \rceil$.
	\item \begin{sloppypar}Compute the MECs of $P_m = V \setminus \Attr{R}{P}{\P_{\leq m-1}}$ and for each MEC $M \in P_m$
			compute the minimal priority $min$ among all vertices in that MEC.\end{sloppypar}
	\item For each MEC $M$:
	  \begin{itemize}
	   \item If $min$ is even then add $M$ to the set $W$ of vertices in winning MECs.
	   \item If $min$ is odd we recursively call $\WinPEC{$P^u,(\P_{\geq k})_{1\leq k \leq 2d},min+1,j$}$ where $P^u$ is the sub-MDP containing only vertices and edges inside $M$.
	   This call applies the key idea and refines the MECs of $P_m$ 
	   and  computes the set $\bigcup_{min+1\leq \ell \leq j} \WE_\ell$.	   
	  \end{itemize}
	\item Call $\WinPEC{$P^\ell,(\P_{\geq k})_{1\leq k \leq 2d},i,m-1$}$ where $P^\ell$ is the MDP where all MECs in $P_m$ are
		collapsed into a single vertex and thus only the edges outside the MECs of $P_m$ are considered.
	This call computes the set $\bigcup_{i\leq k\leq m-1} \WE_k$.
\end{enumerate}
\begin{sloppypar}
We initialize $\WinPEC{$P, (\P_{\geq k})_{1\leq k \leq 2d},i,j$}$ with $\WinPEC{$P,(\P_{\geq k})_{1\leq k \leq 2d},0,2d$}$. 
\end{sloppypar}

Algorithm~\ref{lics:alg:winpec} is the formal version of the sketched algorithm.
\begin{algorithm}
	\Procedure{\WinPEC{$P,(\P_{\geq k})_{1\leq k \leq 2d},i,j$}}{
		\KwIn{$P=(V,E,\langle V_1,V_R \rangle,\delta),(\P_{\geq k})_{1\leq k \leq 2d},i,j$}
		$W \gets \emptyset$\;
		\lIf{$j<i$}{\Return $W$}
		$m \gets \lceil (i+j)/2 \rceil$\;
		$X_m \gets \Attr{R}{P}{\P_{\leq m -1}}$\;\label{lics:alg:winpec:attrxm}
		$Z_m \gets V \setminus X_m$; $E_m \gets E \cap (Z_m \times Z_m)$\; 
		$P' \gets (Z_m,E_m,\langle V_1 \cap Z_m, V_R \cap Z_m\rangle, \delta)$\;
		\For{$M \gets \mathsf{ComputeMECs}(P')$}{\label{lics:alg:winpec:computemecs1}

			$min \gets minPriority(M)$\;\label{lics:alg:winpec:minpriority}
			\If{$min$ is even}{
				$W \gets W \cup M$\;
				}\Else{
				$V^u \gets M \setminus \Attr{R}{P}{\P_{min}}$\;\label{lics:alg:winpec:rmPmin}
				$P^u \gets (V^u,(V^u \times V^u) \cap E,\langle V_1 \cap V^u, V_R \cap V^u\rangle, \delta)$\;
				$W \gets W \cup \WinPEC{$P^u,(\P_{\geq k})_{1\leq k \leq
				2d},min+1,j$}$\label{lics:alg:winpec:rec1}
			}
		}
		\tcc{MDP with MECs collapsed is $P^\ell$ in the text}
		\lFor{$M \gets \mathsf{ComputeMECs}(P')$}{\label{lics:alg:winpec:computemecs2}
			$\CollapseEC{$M,P$}$\label{lics:alg:winpec:collapse}	
		}
		$W \gets W \cup \WinPEC{$P,(\P_{\geq k})_{1\leq k \leq 2d},i, m-1$}$\;\label{lics:alg:winpec:rec2}
		\Return $W$\;
	}
	\caption{Compute WE of an MDP $P$}\label{lics:alg:winpec}
\end{algorithm}
\subsubsection{Correctness and Number of Symbolic Steps.}
In this section, we argue that Algorithm~\ref{lics:alg:winpec} is correct and bound the number of symbolic
steps and the symbolic space usage. 
A key difference in the analysis of Algorithm~\ref{lics:alg:winpec} and~\cite{CH11} 
is that we aim for a symbolic step bound that is independent of the number of edges in $P$
and, thus, we cannot use the argument from~\cite{CH11} which charges the cost of each recursive call
to the edges of $P$.
The key argument in~\cite{CH11} is that the sets of edges in the different branches of the recursions do not overlap.
For vertices, it is not that simple, as we do not entirely remove vertices that appear in a MEC but merge the MEC and represent it by a single vertex.
That is, a vertex can appear in both $P^\ell$ and in $P^u$ corresponding to the MEC\@. 
To accomplish our symbolic step bound we adjusted the algorithm.
At Line~\ref{lics:alg:winpec:rmPmin} we \emph{always} remove the minimum priority vertices
instead of removing the vertices with priority $m$ to ensure that we remove at least one vertex.
Intuitively, by always removing at least one vertex from a MEC we ensure that the total number of vertices processed
at each recursion level does not grow.
Note that these changes of the algorithm do not affect the correctness argument of~\cite{CH11} as we always compute the same sets $\WE_m$ 
but avoid calls to $\WinPEC{$\cdot$}$ with no progress on some MECs.

\begin{proposition}[Correctness]\label{lics:prop:parity:corr}
	\begin{sloppypar}
	Algorithm~\ref{lics:alg:winpec} returns the set of winning end-components $\WE$.
\end{sloppypar}
\end{proposition}
\begin{proof}
	The correctness of the algorithm is by induction on $j-i$
	for the induction hypothesis 
	$\bigcup_{i \leq \ell \leq j} \WE_\ell \subseteq \WinPEC{$P,(\P_{\geq k})_{1\leq k \leq 2d},i,j$} \subseteq \bigcup_{1 \leq \ell \leq 2d} \WE_\ell$. 
	
	First consider the induction base cases:
	If $j>i$, the algorithm correctly returns the empty set.
	Next, consider the induction step. 
	Assume that the results hold for $j-i \leq k$, and we consider $j-i = k+1$. 
	If $m$ is even, then
	\begin{equation*}
		\bigcup_{i \leq k \leq j} \WE_k = \WE_{m} \cup \bigcup_{i \leq k \leq m-1} \WE_k \cup
		\bigcup_{m+1 \leq k \leq j} \WE_k,
	\end{equation*}
	otherwise ($m$ is odd), then
	\begin{equation*}
		\bigcup_{i \leq k \leq j} \WE_k = \bigcup_{i \leq k \leq m-1} \WE_k \cup \bigcup_{m+1 \leq k \leq j} \WE_k,
	\end{equation*}	
	Consider an arbitrary winning MEC $M_k$ in $P_k$, i.e., the lowest even priority is $k$. 
	We consider the following cases.
	\begin{enumerate}
		\item For all $k \geq m$ we have that $M_k$ is contained in a MEC $M_m$ of $P_m$. 
			Additionally, no random vertex in $M_m$ can have an edge leaving $M_m$ and thus no
			random vertex in $M_k$ can have a random edge leaving $M_m$. 
			Moreover, for the minimum priority $min$ of $M_m$ we have $k \geq min \geq m$.
			If $min$ is even then $M_m$ is itself winning and thus $M_k \subseteq
			\WinPEC{$P,(\P_{\geq k})_{1\leq k \leq 2d},i,j$}$ and $M_m \subseteq \bigcup_{1 \leq \ell
				\leq 2d} \WE_\ell$ (note that it might be that $min > j$).
			
			If $min$ is odd 
			$M_k$ is a winning MEC of $P$ iff 
			it is a winning MEC of $P^u$ and thus by the induction
			hypothesis $M_k \subseteq \WinPEC{$P^u,(\P_{\geq k})_{1\leq k \leq
					2d},min+1,j$}$.
            It follows that also $M_k \subseteq \WinPEC{$P,(\P_{\geq k})_{1\leq k \leq 2d},i,j$}$.
		\item For $k < m$ consider a MEC $M_k$ in $P_k$. If $M_k$ contains a vertex $v$ that
			belongs to a MEC $M_m$ of $P_m$, then $M_m \subset M_k$ (i.e., all vertices of the
			MEC in $P_m$ of $v$ also belong to $M_k$ and $M_k$ has at least one additional vertex
			with priority $<m$).
			We thus have that for $k \leq m$  the winning MECs $M_k$ in $P_k$ are in one-to-one correspondence
			with the winning MECs $M'_k$ of the modified MDP where all MECs of $P_m$ are collapsed.
			From the induction hypothesis 
			it follows that $\bigcup_{i\leq k \leq m-1} \WE_k = \WinPEC{$P,(\P_{\geq k})_{1\leq k \leq 2d},i, m-1$}$
	\end{enumerate}
	
	Hence, $\bigcup_{i \leq k \leq j} \WE_k  \subseteq \WinPEC{$P,(\P_{\geq k})_{1\leq k \leq 2d},i,j$} \subseteq  
	\bigcup_{1 \leq \ell \leq 2d} \WE_\ell$. The statements follows from setting $i=0$ and $j=2d$.
\end{proof}

	\begin{proposition}[Symbolic Steps]\label{lics:prop:parity:symbsteps}
	The total number of symbolic operations for Algorithm~\ref{lics:alg:winpec} is $O(\mathbf{MEC} \cdot \log d)$ for
	$0 < \epsilon \leq 0.5$.
\end{proposition}

\begin{proof}
	Given an MDP $P$ with $n$ vertices and $d$ priorities, let us denote by
	$T(n,x)$ the number of symbolic steps of Algorithm~\ref{lics:alg:winpec} at recursion depth $x$ 
	and with $T_M(n)$ the number of symbolic steps incurred by the symbolic MEC Algorithm. 
	As shown in~\cite{CH11}, note that the recursion depth of Algorithm~\ref{lics:alg:winpec} 
	is in $O(\log d)$ because we recursively consider either $(min + 1,j)$ or $(i, m-1)$ where $min
	\geq m = \lceil(i+j/2)\rceil$ until $j>i$, where $j=2d$ initially.
	First, we argue why there exists $c > 0$ such that
	\begin{align*}
	T(n,x) \leq c \cdot T_M(n) + \left(\sum_{i=1,\dots,t} T\left(n_i, x-1\right)\right) \\ +T\left(n- \left(\sum_{i=1,\dots,t}
	n_i\right) + t , x-1\right)  \text{ if } x > 1\\ 
	T(n,0) \leq c. 
	\end{align*}
	The attractors computed at Line~\ref{lics:alg:winpec:attrxm} and Line~\ref{lics:alg:winpec:rmPmin} 
	can be done in $O(n)$ symbolic steps as the set of vertices in the attractors are all disjunct.
	Clearly, this is cheaper than computing the MEC decomposition. 
	To extract the minimum priority of a set of nodes $X
	\subseteq V$ we apply a binary search procedure which takes $O(\log d)$ symbolic steps at Line~\ref{lics:alg:winpec:minpriority}.
	Note that when $\log d > n$ we cannot charge the cost to computing the MEC decomposition.
	Thus, we argue in Claim~\ref{lics:claim:minpriority} that the total number of symbolic steps for Line~\ref{lics:alg:winpec:minpriority} in Algorithm~\ref{lics:alg:winpec} is less than $O(n \log d)$.
	The rest of the symbolic steps in Algorithm~\ref{lics:alg:winpec}, (except the recursive calls and
	computing the MEC decomposition) in Algorithm~\ref{lics:alg:winpec} can be done in a constant amount of symbolic steps.

	Note that when $x = 0$, i.e., in the case $j < i$, we only need a constant amount of symbolic steps.	

	Let $t$ be the number of MECS in $P'$. 
	When $x > 0$, consider the following argumentation for the
	number of symbolic steps of the recursive calls:
	\begin{itemize}
		\item $\WinPEC{$P^u,(\P_{\geq k})_{1\leq k \leq 2d},min,j$} $:
			We perform the recursive call for each MEC $M_i \in P'$ $(1 \leq i \leq t)$ 
			where the vertex with minimum priority is odd.
			The total cost incurred by all such recursive calls is $\sum_{i=1,\dots,t} T(n_i,
			x-1)$ where $n_i \leq |M_i|-1$ because we always remove the
			vertices with minimum priority at Line~\ref{lics:alg:winpec:rmPmin}.
		\item $\WinPEC{$P^\ell,(\P_{\geq k})_{1\leq k \leq 2d},i, m-1$}$:
			$P^\ell$ consists of the vertices representing the collapsed MECs, the vertices not
			in $P'$ and the vertices which are not in a MEC of $P'$. 
			The number of vertices in $P^\ell$ is thus $n_\ell = n - \sum_{i=1,\dots,t} |M_i| + t$
			and we obtain $T(n_\ell,x-1)$. 
	\end{itemize}
	Note that $\sum_{i=1, \dots, t} n_i + n_\ell 
	\leq n$.
	We choose $c$ such that $c T_M(n)$ is greater than the number of symbolic steps for computing
	the MECs twice and the rest of the work in the current iteration of Algorithm~\ref{lics:alg:winpec}.
	It is straightforward to show that
	$T(n,d) = O(T_M(n) \log d)$.
	
	The following claim shows that the total number of symbolic steps incurred by
	Line~\ref{lics:alg:winpec:minpriority} for all calls to $\WinPEC{$\cdot$}$ is only $O(n \log d)$.
	\begin{claim}\label{lics:claim:minpriority}
		The total amount of symbolic steps used by Line~\ref{lics:alg:winpec:minpriority} is in
		$O(n \log d)$.
	\end{claim}
	\begin{proof}
		To obtain the set of vertices with minimum priority from a set of vertices $X \subseteq V$ the function
		$\minPriority{X}$ performs a binary search using the sets $(\P_{\geq k})_{1\leq k \leq 2d}$.
		This can be done in $O(\log d)$ many symbolic steps.
		To prove that the number of symbolic steps used by Line~\ref{lics:alg:winpec:minpriority} in
		total is in
		$O(n \log d)$ note that each time the function is performed we either:
		(i) Remove all
		vertices in $M$, and we never perform the function on the vertices in $M$ again. We charge the cost to an arbitrary vertex in $M$.
		(ii) Remove at least one vertex at Line~\ref{lics:alg:winpec:rmPmin} and we never perform
				the function on a MEC containing this vertex again. We charge the cost to this vertex.
		As there are only $n$ vertices we obtain that the total amount of symbolic steps used by
		Line~\ref{lics:alg:winpec:minpriority} is in $O(n \log d)$.
	\end{proof}
	Using Claim~\ref{lics:claim:minpriority} we conclude the proof for the symbolic step bound of Algorithm~\ref{lics:alg:winpec}.
\end{proof}
	
	\begin{proposition}\label{lics:prop:parity:symbspace}
		Algorithm~\ref{lics:alg:winpec} uses $O(\aspace(\mathbf{MEC}) + \log n \log d)$ space.
	\end{proposition}
	\begin{proof}
		Let $P$ be an MDP $n$ vertices and a parity objective with $d$ priorities.
		We denote with $\aspace(\mathbf{MEC})$ the symbolic space used by the algorithm that computes the MEC decomposition.
		Observe that all computation steps in Algorithm~\ref{lics:alg:winpec} need constant space except
		for the recursions and computing the MEC decomposition. 
		Both at Line~\ref{lics:alg:winpec:computemecs1} and Line~\ref{lics:alg:winpec:computemecs2}
		we first compute the MEC decomposition and then, to minimize extra space, 
		we output one MEC after the other by returning each SCC found given the set of vertices in nontrivial MECs.
		Note that we only require logarithmic space to maintain the state of the SCC algorithm~\cite{ChatterjeeDHL18}.  
		As argued in~\cite{CH11} the recursion depth of Algorithm~\ref{lics:alg:winpec} is $O(\log d)$.
		Thus, we need $O(\log n \log d)$ space for maintaining the state of the SCC algorithm 	
		at Line~\ref{lics:alg:winpec:attrxm} until we reach a leaf of the recursion tree. 
		At each
		recursive call, we need additive $O(\aspace(\mathbf{MEC}))$ space to compute the MEC
		decomposition of $P$. This yields the claimed space bound.
	\end{proof}

	Given an MDP, we first compute the set $\WE$ with Algorithm~\ref{lics:alg:winpec} which is correct
	due to Proposition~\ref{lics:prop:parity:corr}.
	We instantiate $\mathbf{MEC}$ and $\aspace(\mathbf{MEC})$ in Proposition~\ref{lics:prop:parity:symbsteps} and
	Proposition~\ref{lics:prop:parity:symbspace} respectively with Theorem~\ref{lics:thm:symbolicmecTST} and thus need 
	$O(n^{2-\epsilon} \log n \log d)$ symbolic steps and $O(n^\epsilon \log n + \log n \log d)$
	(where $0 < \epsilon \leq 0.5$)	symbolic space for computing $\WE$.
	Then, we compute almost-sure reachability of $\WE$ with Theorem~\ref{lics:thm:asreach}. 
	Finally, using Lemma~\ref{lics:lem:paritywinningmec} we obtain the following theorem.
	\begin{theorem}
		The set $\ASW{\parity{p}}$ of an MDP $P$ can be computed with $\O(n^{2-\epsilon})$ 
		many symbolic operations and $\O(n^{\epsilon})$ symbolic space for $0 < \epsilon \leq 0.5$.
	\end{theorem}

\section{Conclusion}
We present a faster symbolic algorithm for the MEC decomposition. Furthermore, we 
improve the fastest symbolic algorithm for verifying MDPs with $\omega$-regular properties. There are
several interesting directions for future work. On the practical side, implementations and
experiments with case studies is an interesting direction.
On the theoretical side, improving upon the $\O(n^{1.5})$ bound for MECs is an interesting open
question which would also, using our work, improve the presented algorithm for verifying
$\omega$-regular properties of MDPs.

	\chapter{Conclusion}\label{cha:conclusion}
	In the thesis, we examine instances of central problems in model-checking and reactive synthesis.
	Chapters~\ref{cha:mfcs}--\ref{cha:icaps} provide \emph{explicit algorithms} for problems with
	widely-considered objectives like mean-payoff parity objectives, Streett objectives, 
	bounded liveness objectives and variants of reachability objectives. 
	Chapters~\ref{cha:lpar}--\ref{cha:lics} provide \emph{symbolic algorithms} for 
	problems with parity objectives in game graphs and MDPs. 
	The careful transfer of sophisticated modern graph algorithmic techniques to instances of these central problems
	provides the new improved algorithms.

	We conclude with concrete ideas for follow-up work.

	\para{Implementation and Experiments.}
	Even though the discovery of improved theoretical algorithms is an important problem-solving challenge we must also implement
them: 
	The implementation of a theoretic algorithm removes
	semantic gaps between code and pseudocode and meaningful experiments offer valuable insights
	into how an algorithm performs in the real world~\cite{Sanders09}. 
	\begin{compactitem}
		\item For Streett objectives in graphs and MDPs and parity objectives in games implementations 
			of explicit algorithms~\cite{dijk2018oink, BruyerePRT19} and symbolic
			algorithms~\cite{SanchezWW18,abs-2009-10876,CHLOT18} exist. An interesting direction for
			future work is to implement the algorithms for parity objectives in MDPs 
			introduced in Chapter~\ref{cha:lics}.
		\item For problems with mean-payoff parity objectives in games our theoretical result in
			Chapter~\ref{cha:mfcs} is improved to a
			pseudo-quasi-polynomial running time~\cite{DaviaudJL18} 
			and it is important future work to determines which algorithms perform best in practice
			in both the explicit model and the symbolic model of computation.
		\item For bounded liveness objectives, implementations of the algorithms is future work. 
		\item For the MEC-decomposition in MDPs, implementations for explicit
			algorithms exist~\cite{WijsKB16} but for symbolic algorithms it is important future work.
	\end{compactitem}

	\para{Theoretical follow-up question.}
	Recently, a breakthrough for deterministic dynamic algorithms~\cite{SaranurakW19,ChuzhoyGLNPS20,BernsteinGS20} 
	yielded faster deterministic dynamic algorithms for many central graph-theoretic problems and it is
	an important open question if the techniques are useful for improved \emph{deterministic} algorithms 
	for MEC decomposition in MDPs or Streett objectives in graphs. 
	For symbolic algorithms improving upon the $\O(n^{1.5})$ bound for MEC decomposition is an
	interesting open question. 
	For explicit algorithms any improvement upon the $\O(n^{2.5})$ time algorithm in graphs and the
	$O(n^2d)$ time algorithms in games for bounded B\"uchi objectives is interesting.
	Sub-cubic time \emph{deterministic} algorithm for bounded B\"uchi
	objectives in graphs discussed in Chapter~\ref{cha:icalp} are interesting future work. 
	Additionally, any conditional lower bounds for bounded B\"uchi objectives in
	games would separate
	the objectives from B\"uchi objectives. Any improvement upon the long-standing
	\emph{deterministic} $O(mn^{2/3})$ time algorithm for MEC decomposition~\cite{CH14} is interesting.
	Finally, providing a \emph{polynomial time} algorithm for computing the winning set of parity or
	mean-payoff objectives in games is a major theoretical open problem and a formidable task for
	future work.

	\printbibliography

	\backmatter

	\thispagestyle{empty}

\end{document}